\documentclass[11pt,a4paper,USenglish]{article} 
\date{}
\usepackage{comment}

\usepackage[margin=1in]{geometry}

\usepackage{subfig}

\usepackage{microtype}
\usepackage{complexity}
\usepackage{amsmath,amssymb}
\usepackage{thm-restate}
\usepackage{hyperref}
\usepackage{xcolor}
\usepackage{lineno}
\usepackage{multirow}
\usepackage{ntheorem}
\theoremstyle{plain}
\theoremsymbol{}
\theoremprework{}
\theorempostwork{}
\theoremseparator{.}

\newtheorem{property}{Property}
\newtheorem{fact}{Fact}
\newcommand{\remove}[1]{}
\newtheorem{observation}{Observation}
\newtheorem{theorem}{Theorem}
\newtheorem{lemma}{Lemma}

\newtheorem{claim}{Claim}
\newtheorem{corollary}{Corollary}
\newtheorem{definition}{Definition}
\newtheorem{remark}{Remark}

\usepackage[nameinlink,capitalise]{cleveref}
\Crefname{observation}{Observation}{Observations}
\Crefname{algorithm}{Algorithm}{Algorithms}
\Crefname{}{Section}{Sections}
\Crefname{observation}{Observation}{Observations}
\Crefname{lemma}{Lemma}{Lemmas}
\Crefname{claim}{Claim}{Claims}
\Crefname{figure}{Fig.}{Figs.}
\Crefname{figure}{Fig.}{Figs.}
\Crefname{enumi}{Condition}{Conditions}
\Crefname{property}{Property}{Properties}
\Crefname{remark}{Remark}{Remarks}
\Crefname{fact}{Fact}{Facts}

\definecolor{lipicsblue}{rgb}{0.08235294118,0.3098039216,0.537254902}
\definecolor{ourred}{rgb}{1,0.3,0.3}
\definecolor{green}{rgb}{0,0.588,0.509}

\hypersetup{colorlinks,breaklinks,colorlinks=lipicsblue,urlcolor=lipicsblue,citecolor=lipicsblue,linkcolor=lipicsblue}

\renewcommand{\emph}[1]{\textcolor{lipicsblue}{\em #1}}

\usepackage{hhline}

\usepackage{microtype}
\usepackage{xspace}
\usepackage{color}
\usepackage{paralist}
\usepackage{dsfont}

\usepackage[color=yellow]{todonotes}
\usepackage{gensymb}
\usepackage{wrapfig}

\graphicspath{{figures/}}

\renewcommand{\paragraph}[1]{\noindent{\bf #1}\xspace}

\newcommand{\red}[1]{\textcolor{ourred}{#1}\xspace} 
\newcommand{\blue}[1]{\textcolor{lipicsblue}{#1}\xspace}

\makeatletter
\renewcommand{\todo}[2][]{\@bsphack\@todo[#1]{\textcolor{black}{#2}}\@esphack\ignorespaces}
\makeatother

\graphicspath{{figs/}{img/}}

\newcommand*{\qed}{\hfill\ensuremath{\square}}

\newenvironment{fullproof}{\noindent{\em Proof.~}}{\hspace*{\fill}${\small\qed}$\vspace{2mm}}
\newenvironment{proof}{\noindent{\em Proof.~}}{\hspace*{\fill}${\small\qed}$\vspace{2mm}}

\newif\ifshortversion

\shortversionfalse

\ifshortversion
  \usepackage{environ}
  \NewEnviron{hide}{}

\else
\fi

\renewcommand{\NPC}{\mbox{\NP-complete}\xspace}

\newcommand{\NPCN}{\mbox{\NP-completeness}\xspace}
\newcommand{\NPHN}{\mbox{\NP-hardness}\xspace}

\newcommand{\backbone}{${\cal B_\mu}$\xspace} 

\newcommand{\btpbeP}{{\sc Bipartite 2-Page Book Embedding}\xspace}
\newcommand{\btpbe}{{\sc B2BE}\xspace}
\newcommand{\btpbescP}{{\sc Bipartite 2-Page Book Embedding with Spine Crossings}\xspace}

\newcommand{\btpbef}{{\sc B2BEFO}\xspace}
\newcommand{\btpbefP}{{\sc Bipartite 2-Page Book Embedding with Fixed Order}\xspace}

\newcommand{\ctype}[1]{BP{#1}\xspace}

\newcommand{\gtype}[1]{BB{#1}\xspace}
\newcommand{\atype}[1]{RE{#1}\xspace}
\newcommand{\dtype}[1]{RF{#1}\xspace}

\newcommand{\skel}{sk}
\newcommand{\pert}[1]{\ensuremath{H_{#1}}}
\newcommand{\rest}[1]{\ensuremath{\overline{H}_{#1}}}

\newcommand{\cpllong}{Clustered Planarity with Linear Saturators\xspace}
\newcommand{\cpl}{{\sc cpls}\xspace}

\newif\ificons
\iconstrue

\newcommand{\BF}{\ificons\raisebox{-2pt}{\includegraphics[page=17,scale=0.35]{figs/icons-bf.pdf}}\hspace{1pt}\else$\square$\fi}
\newcommand{\BE}{\ificons\raisebox{-2pt}{\includegraphics[page=1,scale=0.35]{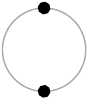}}\hspace{1pt}\else$\square$\fi}
\renewcommand{\BP}{\ificons\raisebox{-2pt}{\includegraphics[page=1,scale=0.35]{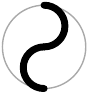}}\hspace{1pt}\else$\square$\fi}
\newcommand{\BPi}{\ificons\raisebox{-2pt}{\includegraphics[page=1,scale=0.35]{figs/icons-bp.pdf}}\hspace{1pt}\else$\square$\fi}
\newcommand{\BPii}{\ificons\raisebox{-2pt}{\includegraphics[page=2,scale=0.35]{figs/icons-bp.pdf}}\hspace{1pt}\else$\square$\fi}
\newcommand{\BPiii}{\ificons\raisebox{-2pt}{\includegraphics[page=3,scale=0.35]{figs/icons-bp.pdf}}\hspace{1pt}\else$\square$\fi}
\newcommand{\BPiv}{\ificons\raisebox{-2pt}{\includegraphics[page=4,scale=0.35]{figs/icons-bp.pdf}}\hspace{1pt}\else$\square$\fi}
\newcommand{\BPv}{\ificons\raisebox{-2pt}{\includegraphics[page=5,scale=0.35]{figs/icons-bp.pdf}}\hspace{1pt}\else$\square$\fi}

\renewcommand{\RE}{\ificons\raisebox{-2pt}{\includegraphics[page=1,scale=0.35]{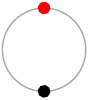}}\hspace{1pt}\else$\square$\fi}
\newcommand{\RF}{\ificons\raisebox{-2pt}{\includegraphics[page=1,scale=0.35]{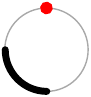}}\hspace{1pt}\else$\square$\fi}
\newcommand{\RFNz}{\ificons\raisebox{-2pt}{\includegraphics[page=1,scale=0.35]{figs/icons-rf.pdf}}\hspace{1pt}\else$\square$\fi}
\newcommand{\RFNza}{\ificons\raisebox{-2pt}{\includegraphics[page=16,scale=0.35]{figs/icons-rf.pdf}}\hspace{1pt}\else$\square$\fi}
\newcommand{\RFNzb}{\ificons\raisebox{-2pt}{\includegraphics[page=1,scale=0.35]{figs/icons-rf.pdf}}\hspace{1pt}\else$\square$\fi}
\newcommand{\RFNi}{\ificons\raisebox{-2pt}{\includegraphics[page=2,scale=0.35]{figs/icons-rf.pdf}}\hspace{1pt}\else$\square$\fi}
\newcommand{\RFNia}{\ificons\raisebox{-2pt}{\includegraphics[page=17,scale=0.35]{figs/icons-rf.pdf}}\hspace{1pt}\else$\square$\fi}
\newcommand{\RFNib}{\ificons\raisebox{-2pt}{\includegraphics[page=18,scale=0.35]{figs/icons-rf.pdf}}\hspace{1pt}\else$\square$\fi}
\newcommand{\RFNic}{\ificons\raisebox{-2pt}{\includegraphics[page=20,scale=0.35]{figs/icons-rf.pdf}}\hspace{1pt}\else$\square$\fi}
\newcommand{\RFNid}{\ificons\raisebox{-2pt}{\includegraphics[page=19,scale=0.35]{figs/icons-rf.pdf}}\hspace{1pt}\else$\square$\fi}
\newcommand{\RFNii}{\ificons\raisebox{-2pt}{\includegraphics[page=21,scale=0.35]{figs/icons-rf.pdf}}\hspace{1pt}\else$\square$\fi}
\newcommand{\RFIz}{\ificons\raisebox{-2pt}{\includegraphics[page=11,scale=0.35]{figs/icons-rf.pdf}}\hspace{1pt}\else$\square$\fi}

\newcommand{\RFIia}{\ificons\raisebox{-2pt}{\includegraphics[page=13,scale=0.35]{figs/icons-rf.pdf}}\hspace{1pt}\else$\square$\fi}
\newcommand{\RFIib}{\ificons\raisebox{-2pt}{\includegraphics[page=12,scale=0.35]{figs/icons-rf.pdf}}\hspace{1pt}\else$\square$\fi}

\newcommand{\RFIi}{\ificons\raisebox{-2pt}{\includegraphics[page=8,scale=0.35]{figs/icons-rf.pdf}}\hspace{1pt}\else$\square$\fi}
\newcommand{\RFIii}{\ificons\raisebox{-2pt}{\includegraphics[page=10,scale=0.35]{figs/icons-rf.pdf}}\hspace{1pt}\else$\square$\fi}
\newcommand{\BFNz}{\ificons\raisebox{-2pt}{\includegraphics[page=1,scale=0.35]{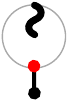}}\hspace{1pt}\else$\square$\fi}
\newcommand{\BFNza}{\ificons\raisebox{-2pt}{\includegraphics[page=18,scale=0.35]{figs/icons-bf.pdf}}\hspace{1pt}\else$\square$\fi}
\newcommand{\BFNzb}{\ificons\raisebox{-2pt}{\includegraphics[page=19,scale=0.35]{figs/icons-bf.pdf}}\hspace{1pt}\else$\square$\fi}
\newcommand{\BFNi}{\ificons\raisebox{-2pt}{\includegraphics[page=2,scale=0.35]{figs/icons-bf.pdf}}\hspace{1pt}\else$\square$\fi}
\newcommand{\BFNia}{\ificons\raisebox{-2pt}{\includegraphics[page=2,scale=0.35]{figs/icons-bf.pdf}}\hspace{1pt}\else$\square$\fi}
\newcommand{\BFNib}{\ificons\raisebox{-2pt}{\includegraphics[page=20,scale=0.35]{figs/icons-bf.pdf}}\hspace{1pt}\else$\square$\fi}
\newcommand{\BFNic}{\ificons\raisebox{-2pt}{\includegraphics[page=21,scale=0.35]{figs/icons-bf.pdf}}\hspace{1pt}\else$\square$\fi}
\newcommand{\BFNid}{\ificons\raisebox{-2pt}{\includegraphics[page=22,scale=0.35]{figs/icons-bf.pdf}}\hspace{1pt}\else$\square$\fi}
\newcommand{\BFNii}{\ificons\raisebox{-2pt}{\includegraphics[page=3,scale=0.35]{figs/icons-bf.pdf}}\hspace{1pt}\else$\square$\fi}
\newcommand{\BFIza}{\ificons\raisebox{-2pt}{\includegraphics[page=4,scale=0.35]{figs/icons-bf.pdf}}\hspace{1pt}\else$\square$\fi}
\newcommand{\BFIzb}{\ificons\raisebox{-2pt}{\includegraphics[page=5,scale=0.35]{figs/icons-bf.pdf}}\hspace{1pt}\else$\square$\fi}

\newcommand{\BFIia}{\ificons\raisebox{-2pt}{\includegraphics[page=23,scale=0.35]{figs/icons-bf.pdf}}\hspace{1pt}\else$\square$\fi}
\newcommand{\BFIib}{\ificons\raisebox{-2pt}{\includegraphics[page=24,scale=0.35]{figs/icons-bf.pdf}}\hspace{1pt}\else$\square$\fi}
\newcommand{\BFIic}{\ificons\raisebox{-2pt}{\includegraphics[page=25,scale=0.35]{figs/icons-bf.pdf}}\hspace{1pt}\else$\square$\fi}
\newcommand{\BFIid}{\ificons\raisebox{-2pt}{\includegraphics[page=26,scale=0.35]{figs/icons-bf.pdf}}\hspace{1pt}\else$\square$\fi}
\newcommand{\BFIie}{\ificons\raisebox{-2pt}{\includegraphics[page=27,scale=0.35]{figs/icons-bf.pdf}}\hspace{1pt}\else$\square$\fi}
\newcommand{\BFIif}{\ificons\raisebox{-2pt}{\includegraphics[page=29,scale=0.35]{figs/icons-bf.pdf}}\hspace{1pt}\else$\square$\fi}
\newcommand{\BFIig}{\ificons\raisebox{-2pt}{\includegraphics[page=28,scale=0.35]{figs/icons-bf.pdf}}\hspace{1pt}\else$\square$\fi}
\newcommand{\BFIiia}{\ificons\raisebox{-2pt}{\includegraphics[page=11,scale=0.35]{figs/icons-bf.pdf}}\hspace{1pt}\else$\square$\fi}
\newcommand{\BFIiib}{\ificons\raisebox{-2pt}{\includegraphics[page=12,scale=0.35]{figs/icons-bf.pdf}}\hspace{1pt}\else$\square$\fi}

\newcommand{\BFIz}{\ificons\raisebox{-2pt}{\includegraphics[page=14,scale=0.35]{figs/icons-bf.pdf}}\hspace{1pt}\else$\square$\fi}
\newcommand{\BFIi}{\ificons\raisebox{-2pt}{\includegraphics[page=15,scale=0.35]{figs/icons-bf.pdf}}\hspace{1pt}\else$\square$\fi}
\newcommand{\BFIii}{\ificons\raisebox{-2pt}{\includegraphics[page=16,scale=0.35]{figs/icons-bf.pdf}}\hspace{1pt}\else$\square$\fi}
\newcommand{\BB}{\ificons\raisebox{-2pt}{\includegraphics[page=1,scale=0.35]{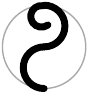}}\hspace{1pt}\else$\square$\fi}

\newcommand{\BBii}{\ificons\raisebox{-2pt}{\includegraphics[page=2,scale=0.35]{figs/icons-bb.pdf}}\hspace{1pt}\else$\square$\fi}
\newcommand{\BBiii}{\ificons\raisebox{-2pt}{\includegraphics[page=3,scale=0.35]{figs/icons-bb.pdf}}\hspace{1pt}\else$\square$\fi}
\newcommand{\BBiv}{\ificons\raisebox{-2pt}{\includegraphics[page=4,scale=0.35]{figs/icons-bb.pdf}}\hspace{1pt}\else$\square$\fi}
\newcommand{\BBv}{\ificons\raisebox{-2pt}{\includegraphics[page=5,scale=0.35]{figs/icons-bb.pdf}}\hspace{1pt}\else$\square$\fi}

\author{
Patrizio Angelini\\
{\em John Cabot University, Rome}\\
\href{mailto:pangelini@johncabot.edu}{pangelini@johncabot.edu}
\medskip\\
Giordano {Da Lozzo}, Giuseppe {Di Battista}, Fabrizio {Frati}, and Maurizio {Patrignani}\\
{\em Roma Tre University, Rome, Italy}\\
\href{mailto:giordano.dalozzo@uniroma3.it,giuseppe.dibattista@uniroma3.it,fabrizio.frati@uniroma3.it,maurizio.patrignani@uniroma3.it}{
\{giordano.dalozzo,giuseppe.dibattista,fabrizio.frati,maurizio.patrignani\}@uniroma3.it}
}

\title{2-Level Quasi-Planarity\\or How Caterpillars Climb (SPQR-)Trees}

\begin{document}

\maketitle

\begin{abstract}
Given a bipartite graph $G=(V_b,V_r,E)$, the {\sc $2$-Level Quasi-Planarity} problem asks for the existence of a drawing of $G$ in the plane such that the vertices in $V_b$ and in $V_r$ lie along two parallel lines $\ell_b$ and $\ell_r$, respectively, each edge in $E$ is drawn in the unbounded strip of the plane delimited by $\ell_b$ and $\ell_r$, and no three edges in $E$ pairwise cross.

We prove that the {\sc $2$-Level Quasi-Planarity} problem is \NPC. This answers an open question of Dujmovi\'c, P{\'{o}}r, and Wood. Furthermore, we show that the problem becomes linear-time solvable if the ordering of the vertices in $V_b$ along $\ell_b$ is prescribed. 
Our contributions provide the first results on the computational complexity of recognizing quasi-planar graphs, which is a long-standing open question. 

Our linear-time algorithm exploits several ingredients, including a combinatorial characterization of the positive instances of the problem in terms of the existence of a planar embedding with a caterpillar-like structure, and an SPQR-tree-based algorithm for testing the existence of such a planar embedding. Our algorithm builds upon a classification of the types of embeddings with respect to the structure of the portion of the caterpillar they contain and performs a computation of the realizable embedding types based on a succinct description of their features by means of constant-size gadgets.

\end{abstract}

\clearpage
\vspace{-15mm}
\setcounter{tocdepth}{3}
\tableofcontents

\clearpage


\section{Introduction}

Planarity is a key concept in graph theory~\cite{bm-gtwa-76,DETT99,d-gt-97,nc-pgta-88,nr-pgd-lnsc-04,w-igt-01}. It has long been known that the planarity of a graph can be tested in polynomial time~\cite{ap-igs-61,b-dtp-64} and, in fact, even in linear time~\cite{ht-ept-74}. It was at first surprising, and it appears now evident, that testing whether a graph is ``almost'' planar is difficult. For example, it is \NP-hard to test whether a graph is \emph{$k$-planar}~\cite{cm-cnwc-11,km-mo1p-13}, i.e., whether it admits a drawing in which each edge has at most $k$ crossings, even for $k=1$. It is also \NP-hard to recognize \emph{$k$-apex} graphs~\cite{ly-ndphp-11} or \emph{$k$-skewness} graphs~\cite{lg-dneg-79}, i.e., graphs that become planar if $k$ vertices or edges are allowed to be removed, respectively. Further, it is \NP-hard to decide whether a graph that consists of a planar graph plus a single edge admits a drawing with at most $k$ crossings~\cite{cm-aoe-13}. See~\cite{dlm-sgdbp-19} for a survey on the variants of ``almost'' planarity. 

There is one notorious notion of ``almost'' planarity that has so far eluded the efforts to establish its computational complexity; this is called \emph{quasi-planarity}. A graph is \emph{quasi-planar} if it admits a \emph{quasi-planar drawing}, i.e., a drawing in which no three edges pairwise cross. Quasi-planar graphs have been studied thoroughly. For example, it is known that an $n$-vertex quasi-planar graph has at most $8n-O(1)$ edges and some quasi-planar graphs have $7n-O(1)$ edges~\cite{at-mneqpg-07}; therefore, the graphs that can be represented with no three pairwise crossing edges can be much denser than those that can be represented planarly. Further, the class of quasi-planar graphs is known to include the one of $2$-planar graphs~\cite{DBLP:conf/wg/AngeliniBBLBDLM17}.
Despite these combinatorial results, from an algorithmic perspective neither efficient algorithms nor hardness results are known (see, e.g.,~\cite[Problem 5]{dlm-sgdbp-19}), even in the generalization in which no $k$ edges are allowed to cross.

In this paper, we show the first complexity results on the problem of recognizing quasi-planar graphs. Given a bipartite graph $G=(V_b,V_r,E)$, where the vertices in $V_b$ and $V_r$ are called \emph{black} and \emph{red}, respectively, the {\sc $2$-Level Quasi-Planarity} problem asks for the existence of a \emph{$2$-level quasi-planar drawing} of $G$, that is, a quasi-planar drawing such that the black and red vertices lie along two parallel lines $\ell_b$ and $\ell_r$, respectively, and each edge in $E$ is drawn in the unbounded strip of the plane delimited by $\ell_b$ and $\ell_r$. We show that the {\sc $2$-Level Quasi-Planarity} problem is \NPC. Further, we show how to solve the problem in linear time if the order of the black vertices along $\ell_b$ is part of the input; this version of the problem is called {\sc $2$-Level Quasi-Planarity with Fixed Order}. 

The study of $2$-level drawings of bipartite graphs is a classical research topic. A bipartite graph admits a $2$-level planar drawing if and only if it is a forest of caterpillars~\cite{emw-oacp-86,hs-ncn-72}, where a \emph{caterpillar} is a tree whose non-leaf nodes induce a path, called \emph{backbone}. The problems of constructing $2$-level $k$-planar drawings, $2$-level RAC drawings, and $2$-level fan-planar drawings have been studied in~\cite{adfs-2lkp-20,DBLP:journals/corr/abs-2002-09597,DBLP:journals/jgaa/BinucciCDGKKMT17,d-dsp-13,DBLP:journals/algorithmica/GiacomoDEL14}. Further, minimizing the number of crossings in a $2$-level drawing is an \NPC~problem (for a discussion on the matter, see~\cite{ms-gcn-13}), even if the order of the vertices in one level is prescribed~\cite{emw-oacp-86,ew-ecdbg-94}; the fact that the problem's complexity remains unchanged even if the order of the vertices in one level is prescribed contrasts with what we show to happen for $2$-level quasi-planarity. Many approximation algorithms and heuristics for the crossing-minimization problem for $2$-level drawings have also been designed; see, e.g.,~\cite{ew-ecdbg-94,jm-scm-97,pllmd-amp-20}. These are commonly used as building blocks in the notorious Sugiyama framework~\cite{stt-mvuhs-81}, that inspired a substantial number of graph drawing algorithms; see, e.g.,~\cite{bm-ldd-01,DETT99}. In the Sugiyama framework, one is interested in the construction of a drawing on more than $2$ levels; this is done by fixing the ordering of the vertices on the first level and by then repeatedly solving a $2$-level drawing problem to establish the order of the vertices on a level with an already fixed order of the vertices on the previous level. A possible practical application of our linear-time algorithm is to 
be employed as the $2$-level drawing algorithm within the Sugiyama framework (only for pairs of levels that admit a quasi-planar drawing). 

To obtain our results, we leverage on the equivalence of the {\sc $2$-Level Quasi-Planarity} problem with the \btpbeP problem~\cite{DBLP:journals/jgaa/AngeliniLBFPR17}. Given a bipartite graph $G$, a \emph{bipartite $2$-page book embedding} of $G$ is a planar drawing of $G$ in which the vertices are placed along a Jordan curve $\ell$, called \emph{spine}, the black vertices occur consecutively along $\ell$ (and, thus, so do the red vertices), and each edge is entirely drawn in one of the two regions of the plane, called \emph{pages}, delimited by $\ell$. The \btpbeP problem asks whether a bipartite $2$-page book embedding exists for a given bipartite graph. The equivalence between the two problems descends from the equivalence of both problems with the {\sc $(2,2)$-Track Graph Recognition} problem, which asks, for a given bipartite graph $G$, whether a $2$-level drawing of $G$ and a $2$-coloring of the edges of $G$ exist such that no two edges of $G$ with the same color cross. These equivalences easily follow from results by Dujmovi\'c, P{\'{o}}r, and Wood~\cite{dpw-tlg-04,dw-sqtlg-05}. See \cref{fig:problems}. We prove that the \btpbeP problem is \NPC, which implies that the {\sc $2$-Level Quasi-Planarity} and the {\sc $(2,2)$-Track Graph Recognition} problems are also \NPC. The latter result solves an open question by Dujmovi\'c et al.~\cite{dpw-tlg-04}.  

\begin{figure}[tb!]
	\centering
	\subfloat[]{
		\includegraphics[width=.27\textwidth,page=1]{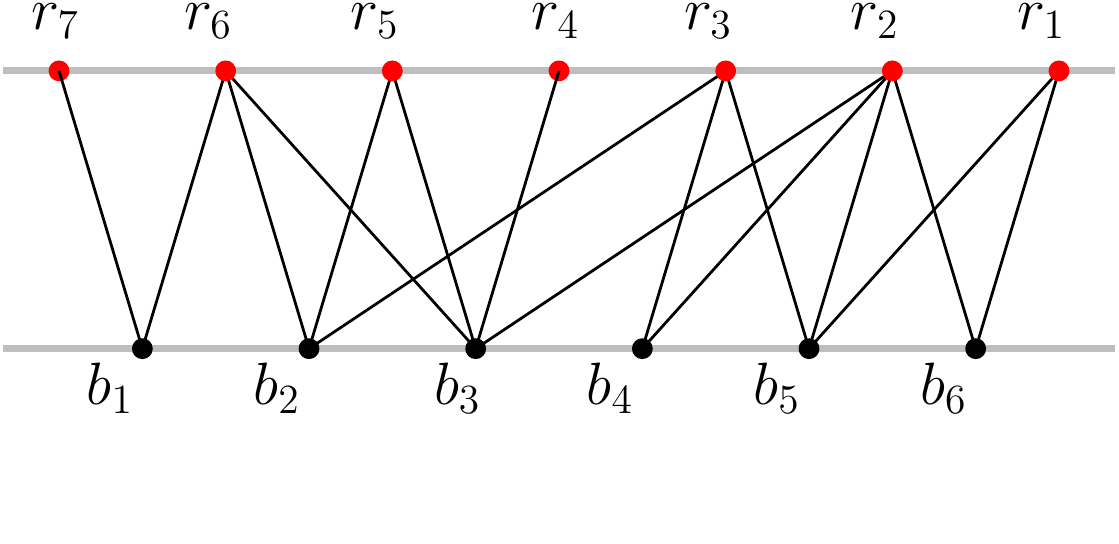} 
		\label{fig:problems-quasi}
	}\hfil
	\subfloat[]{
		\includegraphics[width=.27\textwidth,page=2]{problemi} 
		\label{fig:problems-track}
	}\hfil
	\subfloat[]{
		\includegraphics[width=.35\textwidth,page=3]{problemi} 
		\label{fig:probelms-2-page}
	}
	\caption{A bipartite graph $G$ and (a) a $2$-level quasi-planar drawing of $G$, (b) a $(2,2)$-track layout of $G$, and (c) a bipartite $2$-page book embedding of $G$.}
	\label{fig:problems}
\end{figure}

 
The linear-time testing algorithm for {\sc $2$-Level Quasi-Planarity with Fixed Order} is obtained by studying the \btpbefP problem. Given a bipartite graph $G$ and a total order $\pi_b$ of its black vertices, the problem asks whether $G$ admits a bipartite $2$-page book embedding in which the black vertices appear (consecutively) along the spine in the order $\pi_b$. We show a linear-time algorithm to solve the \btpbefP problem. The outline of the algorithm is as follows.

First, we augment $G$ with a path $P$ connecting the black vertices in the order $\pi_b$; let $H$ be the resulting graph. The basis of our algorithm is a structural characterization of the planar embeddings of $H$ that allow for a bipartite $2$-page book embedding with fixed order of $\langle G,\pi_b\rangle$; these are called \emph{good embeddings}. Namely, for a given planar embedding $\cal E$ of $H$, we construct an auxiliary graph $A(\mathcal E)$ whose vertex set consists of the red vertices of $H$ and of the \emph{red faces} of $\cal E$, which are those faces incident to at least two red vertices, and whose edges connect red vertices to their incident red faces. We show that $\langle G,\pi_b\rangle$ admits a bipartite $2$-page book embedding with fixed order in which the planar embedding of $H$ is $\mathcal E$ if and only if $A(\mathcal E)$ is a caterpillar whose backbone starts and ends ``close'' to the end-vertices of $P$; when this happens, $\mathcal E$ is a good embedding. Our problem now becomes the one of testing whether $H$ admits a good embedding. It is interesting that caterpillars, that characterize $2$-level planarity, also show up, coincidentally or not, in our characterization of $2$-level quasi-planarity. 

Second, we show that the problem of testing whether $H$ admits a good embedding can be solved independently for each \emph{$rb$-augmented component} of $H$; this is a maximal biconnected component of $H$ together with edges connecting its black vertices with degree-$1$ red vertices. 

Third, we consider an $rb$-augmented component, which for simplicity we denote again by $H$. We decompose $H$ along its separating pairs of vertices. For each separation pair $\{u,v\}$, we classify the subgraphs separated by $\{u,v\}$, usually called \emph{triconnected components}, according to the color of $u$ and $v$, and according to the portion of $P$ they contain. We also classify the types of embeddings of the triconnected components according to several features, related to the structure of the subgraph of the caterpillar $A(\mathcal E)$ they contain. We then show how to test whether an embedding type is realizable by a graph, based on the embedding types that are realizable by its triconnected components. We represent the decomposition of $H$ into its triconnected components via an SPQR-tree $\cal T$~\cite{dt-opl-96}. Our algorithm performs a bottom-up traversal of $\cal T$, while computing, at each node $\mu$ of $\cal T$, the set of realizable embedding types for the triconnected component associated with $\mu$. This is done using dynamic programming, by combining the embedding types that can be realized by the triconnected components associated with the children of $\mu$. Since there might be exponentially-many possible combinations, we need to argue that it suffices to consider ``few'' of them, without losing any realizable type of embedding. In order to achieve linear running time, with a methodology that resembles Thevenin's theorem~\cite{t-83}, which represents an arbitrarily complex electrical circuit with an equivalent circuit consisting only of a resistance and a source voltage, we represent an embedding of a triconnected component with a constant-size gadget that has the same embedding type as the triconnected component of the embedding it substitutes. 

Our algorithm is constructive and returns a bipartite $2$-page book embedding (which can also be easily transformed into a $2$-level quasi-planar drawing) of $\langle G,\pi_b\rangle$, if it exists.

The paper is organized as follows. In \cref{se:preliminaries} we introduce some preliminaries. In \cref{se:fixed-embedding} we present the characterization for graphs with a fixed planar embedding. In \cref{se:simply} we reduce the problem to $rb$-augmented components. In \cref{se:biconnected} we classify the types of nodes of the SPQR-tree and the types of embeddings of the triconnected components associated with such nodes, and we present related structural results. In \cref{se:spqr-tree}  we present our linear-time algorithm. Finally, in \cref{se:conclusions}, we conclude and present some open problems.

\section{Preliminaries and Relationships with Other Problems} \label{se:preliminaries}
In the paper we denote the vertex and edge sets of a graph $G$ by $V(G)$ and $E(G)$, respectively. For simplicity, we denote $|G|:=|V(G)|+|E(G)|$; note that, whenever $G$ is planar, we have $|G|\in O(|V(G)|)$. A \emph{drawing} of a graph maps each vertex to a point in the plane and each edge to a Jordan arc between its end-vertices. A drawing is \emph{planar} if no two edges cross. A graph is \emph{planar} if it admits a planar drawing. A planar drawing partitions the plane into connected regions, called \emph{faces}. The unbounded face is the \emph{outer face}, while all the other faces are \emph{internal}. Two planar drawings of a connected planar graph are \emph{equivalent} if the clockwise order of the edges incident to each vertex is the same in both drawings. An equivalence class of planar drawings is called a \emph{planar embedding} or, sometimes, just \emph{embedding}.  We often talk about planar embeddings as if they were actual planar drawings; when this happens, we are referring to any planar drawing within the equivalence class. This happens frequently in this paper because the problems we study are topological and the actual geometry of the drawings does not matter. For example, we often talk about a \emph{face of a planar embedding}, meaning a face in any planar drawing within that equivalence class. In a planar embedding, an \emph{internal vertex} is not incident to the outer face.

\subsection{Equivalence Between $2$-Level Quasi-Planarity and Related Problems}

We observe that the {\sc $2$-Level Quasi-Planarity} problem is linear-time equivalent to two problems in the area of linear layouts, namely the problem of recognizing graphs that can be drawn on two tracks, called {\sc $(2,2)$-Track Graph Recognition} problem~\cite{dpw-tlg-04}, and a variant of the $2$-page book embedding problem for bipartite graphs, called \btpbeP problem~\cite{DBLP:journals/jgaa/AngeliniLBFPR17}. Using this equivalence, by proving that \btpbeP is \NP-complete, we obtain analogous results for the {\sc $2$-Level Quasi-Planarity} and the {\sc $(2,2)$-Track Graph Recognition} problems.

We start by formally defining the above mentioned problems and further related problems. 

A \emph{$2$-page book embedding} of a planar graph $H$ is a planar drawing of $H$ in which the vertices are placed along a Jordan curve $\ell$, called \emph{spine}, and each edge is entirely drawn in one of the two regions of the plane delimited by $\ell$, which we call \emph{pages}. The {\sc $2$-Page Book Embedding} problem asks whether a $2$-page book embedding exists for a given graph. 

Now, consider a bipartite graph $G=(V_b,V_r,E)$. The vertices in $V_b$ are \emph{black} vertices and those in $V_r$ are \emph{red} vertices. A \emph{bipartite $2$-page book embedding} of $G$ is a $2$-page book embedding such that all the vertices in $V_b$ occur consecutively along the spine (and thus all the vertices in $V_r$ occur consecutively as well). We call the corresponding decision problem \btpbeP. For simplicity, we use \btpbe as an abbreviation for both a bipartite $2$-page book embedding and for the corresponding decision problem.

We will also consider a version of the problem in which, together with a bipartite graph $G=(V_b,V_r,E)$, the input also contains a total order $\pi_b$ of the vertices in $V_b$. The question is then whether $G$ admits a \btpbe in which the vertices in $V_b$ appear (consecutively) along the spine in the order $\pi_b$. We call the corresponding decision problem \btpbefP.
For simplicity, we use \btpbef as an abbreviation for both a bipartite $2$-page book embedding with fixed order and for the corresponding decision problem.

A \emph{$(k,t)$-track layout} $\langle \gamma,\Sigma,\omega\rangle$ of a graph $G$ consists of a proper vertex $t$-coloring $\gamma: V(G) \rightarrow \{1,2,\dots,t\}$ of $G$, of total orders $\Sigma=\langle\xi_1,\xi_2,\dots,\xi_t\rangle$ for the vertices in each color class, and of an edge $k$-coloring $\omega: E(G) \rightarrow \{1,2,\dots,k\}$ such that there exist no two edges $e'=(u,v)$ and $e''=(w,z)$ with $\gamma(u)=\gamma(w)$, $\gamma(v)=\gamma(z)$, $\omega(e')=\omega(e'')$, $u \prec_{\xi_{C(u)}} w$, and $z \prec_{\xi_{C(z)}} v$.
A graph is a \emph{$(k,t)$-track graph} if it admits a $(k,t)$-track layout.
The {\sc $(k,t)$-Track Graph Recognition} problem takes as input a graph and asks whether it is a $(k,t)$-track graph~\cite{dpw-tlg-04}. 

%
%
%

The following lemmata were proved\footnote{In~\cite{dpw-tlg-04,dw-sqtlg-05} more general versions of the results we state here are actually proved, dealing with $(k,t)$-track layouts and $k$-page book embeddings, where $k$ and $t$ might be larger than $2$.} by Dujmovi\'c, P{\'{o}}r, and Wood~\cite{dpw-tlg-04,dw-sqtlg-05}. Refer to \cref{fig:problems}.

\begin{lemma}\label{le:structural-equivalence-1}~\cite[Lemma 2]{dpw-tlg-04}
Let $G=(V_b,V_r,E)$ be a bipartite graph and let $\xi_1$ and $\xi_2$ be total orders of $V_b$ and $V_r$, respectively. The following statements are equivalent.
\begin{itemize}
\item The $2$-level drawing in which the vertices in $V_b$ lie along a horizontal line $\ell_b$ in left-to-right order $\xi_1$, the vertices in $V_r$ lie along a distinct horizontal line $\ell_r$ in left-to-right order $\xi_2$, and the edges in $E$ are straight-line segments is quasi-planar. 
\item There exists an edge $2$-coloring $\omega: E \rightarrow \{1,2\}$ such that $\langle \gamma,\langle\xi_1,\xi_2\rangle,\omega\rangle$ is a $(2,2)$-track layout of $G$, where $\gamma$ is the vertex $2$-coloring of $G$ that defines the color classes $V_b$ and $V_r$.
\end{itemize}
\end{lemma}

\begin{lemma}\label{le:structural-equivalence-2}~\cite[Lemma 1]{dw-sqtlg-05}
	Let $G=(V_b,V_r,E)$ be a bipartite graph, let $\xi_1$ and $\xi_2$ be total orders of $V_b$ and $V_r$, respectively, and let $\omega: E \rightarrow \{1,2\}$ be an edge $2$-coloring of $G$. The following statements are equivalent.
	\begin{itemize}
		\item The total order $\xi_1 \circ\overleftarrow{\xi_2}$ of the vertices of $G$ and the page assignment $\omega$ for the edges of $G$ define a bipartite $2$-page book embedding, where $\overleftarrow{\xi_2}$ is the reverse order with respect to $\xi_2$.
		\item $\langle \gamma,\langle\xi_1,\xi_2\rangle,\omega\rangle$ is a $(2,2)$-track layout of $G$, where $\gamma$ is the vertex $2$-coloring of $G$ that defines the color classes $V_b$ and $V_r$.
	\end{itemize}
\end{lemma}

From~\cref{le:structural-equivalence-1,le:structural-equivalence-2} we get the following corollaries.

\begin{corollary}\label{cor:equivalence-variable}
	The following problems are linear-time equivalent:
	\begin{inparaenum}[(i)]
		\item\label{i:quasi} {\sc $2$-Level Quasi-Planarity},
		\item\label{i:btpbe} \btpbeP, and 
		\item\label{i:track} {\sc $(2,2)$-Track Graph Recognition}.
	\end{inparaenum}
\end{corollary}

\begin{corollary}\label{cor:equivalence-fixed}
	The following problems are linear-time equivalent:
	\begin{inparaenum}[(i)]
		\item\label{i:quasi-fixed} {\sc $2$-Level Quasi-Planarity with Fixed Order},
		\item\label{i:btpbe-fixed} \btpbefP, and 
		\item\label{i:track-fixed} {\sc $(2,2)$-Track with Fixed Order Graph Recognition}.
	\end{inparaenum}
\end{corollary}

Although we did not introduce the {\sc $(2,2)$-Track with Fixed Order Graph Recognition} problem, its definition can be easily derived from the other problems presented in this section. 

In view of \cref{cor:equivalence-variable,cor:equivalence-fixed}, in the remainder of the paper we study the complexity of the {\sc $2$-Level Quasi-Planarity} and of {\sc $2$-Level Quasi-Planarity with Fixed Order} under the notation of the \btpbeP and \btpbefP problems, respectively. In the following we discuss properties concerning this problem and introduce further definitions.

\subsection{Properties and Definitions for Bipartite $2$-Page Book Embeddings}\label{sse:tools}

Let $G=(V_b,V_r,E)$ be a bipartite planar graph. 
Let $G^+$ be a planar supergraph of $G$ whose vertex set is $V_b \cup V_r$ and whose edge set is $E \cup E(\mathcal{C})$, where $\cal C$ is a Hamiltonian cycle that traverses all the vertices of $V_b$ (and, thus, also of $V_r$) consecutively. 
We say that a cycle $\cal C$ satisfying the above condition is a \emph{saturator} of $G$.
An edge of $\cal C$ is a \emph{saturating edge}.
Also, a saturating edge  is \emph{black} if it connects two vertices in $V_b$, and it is \emph{red} if it connects two vertices in $V_r$.  
Note that the saturator $\cal C$ consists of four paths: A path $P=(b_1,\dots,b_m)$ consisting of black saturating edges, a path $R=(r_1,\dots,r_p)$ consisting of red saturating edges, and two edges, namely, the edge $(b_m,r_1)$ and the edge $(r_p,b_1)$. 
The end-vertices of $P$ and the end-vertices of $R$ are the \emph{black} and the \emph{red end-vertices} of $\cal C$, respectively. We formalize a simple property of saturators.

\begin{property}\label{prop:edpointsAREclose}
In any planar embedding of $G^+$, or in any bipartite $2$-page book embedding of $G$, each red end-vertex of~$\cal C$ shares a face with a black end-vertex of $\cal C$, and vice versa.
\end{property}

The {\sc $2$-Page Book Embedding} problem boils down to determining whether the input graph admits a Hamiltonian planar spanning supergraph~\cite{bk-btg-79} and it was proved \NPC by Wigderson~\cite{w-chcpmpg-82}. The next lemma shows a similar characterization for the \btpbe problem.

\begin{lemma}\label{le:characterization-btpbe}
	A bipartite graph $G=(V_b,V_r,E)$ admits a bipartite $2$-page book embedding if and only if it admits a saturator.
\end{lemma}

\begin{fullproof}
Let $\Gamma$ be a \btpbe of $G$ and let $\ell$ be the closed curve traversing the vertices of $V_r$ and the vertices of $V_b$ consecutively in $\Gamma$. We construct a saturator $\cal C$ of $G$ by joining with an edge any two vertices that are consecutive along $\ell$. 

Conversely, let $\cal C$ be a saturator of $G$. By the definition of saturator, $G \cup E(\cal C)$ is planar. Let $\Gamma$ be a planar drawing of $G \cup E(\cal C)$. By interpreting $\cal C$ as a closed curve, we have that $\Gamma$ is a \btpbe of $G$.
\end{fullproof}

The following will turn useful.

\begin{lemma}\label{le:quadrangles}
	Let $G=(V_b,V_r,E)$ be a bipartite planar graph that admits a \btpbe and that has a unique (up to a flip) planar embedding $\cal E$. Also, let $f$ be a face of $\cal E$ bounded by a length-$4$ cycle $c_f = (v_1,v_2,v_3,v_4)$ of $G$ and let $\cal C$ be a saturator of $G$. 
	We have that 
	\begin{enumerate}[a)]
	\item \label{atleastone} if $(v_1,v_3) \notin \cal C$, then $(v_2,v_4) \in \cal C$, and vice versa, and
	\item\label{oneinside}  in any planar drawing $\Gamma$ of $G \cup E(\cal C)$  either $(v_1,v_3)$ or $(v_2,v_4)$ lies inside~$c_f$.
	\end{enumerate}
\end{lemma}

\begin{fullproof}
Let $\cal C$ be a saturator of $G$ and assume, w.l.o.g., that $v_1,v_3 \in V_b$ and $v_2,v_4 \in V_r$. 
Observe that, since $G$ has a unique (up to a flip) planar embedding $\cal E$, then the restriction to $G$ of any planar drawing $\Gamma$ of $G \cup E(\cal C)$ induces the embedding $\cal E$ (or its flip).

Clearly, since $\cal C$ is a saturator, we have that $v_1$ and $v_3$ are connected by a subpath $P_b$ of $\cal C$ consisting of black saturating edges, and that $v_2$ and $v_4$ are connected by a subpath $P_r$ of $\cal C$ consisting of red saturating edges. Also, if $(v_1,v_3)$ exists in $\cal C$, then $P_b=(v_1,v_3)$; similarly, if $(v_2,v_4)$ exists in $\cal C$, then $P_r=(v_2,v_4)$.

Suppose, for a contradiction, that neither $(v_1,v_3)$ nor $(v_2,v_4)$ belongs to $\cal C$. Then $P_b$ and $P_r$ both contain some vertices different from $v_1$, $v_2$, $v_3$, and $v_4$, hence they both lie outside $c_f$ (thus, $c_f$ also bounds a face of $\Gamma$). Since the end-vertices of $P_b$ and $P_r$ alternate along $c_f$ and since they do not share any vertex, they cross in $\Gamma$, a contradiction which proves \cref{atleastone}. Thus, one of $P_b$ and $P_r$, say $P_b$, lies in the interior of $f$.
We prove that $P_b$ contains no internal vertices, that is, $P_b = (v_1,v_3)$.
Suppose, for a contradiction, that $(v_1,\dots,x,\dots,v_3)$, where $x \in V_b$ and $x \notin \{v_1,v_3\}$. However, this implies that $x$ is incident to $f$ and hence that $f$ is not bounded by the simple cycle $c_f$, a contradiction which proves \cref{oneinside}.
\end{fullproof}

\subsection{SPQR-trees and Planar Embeddings}\label{se:spqr-trees}
Let $G$ be a graph. A \emph{cut-vertex} in a graph $G$ is a vertex whose removal disconnects $G$. A \emph{separation pair} in $G$ is a pair of vertices whose removal disconnects $G$. Graph $G$ is \emph{biconnected} (\emph{triconnected}) if it has no cut-vertex (resp.\ no separation pair). A \emph{biconnected component} (or {\em block}) of $G$ is a maximal (in terms of vertices and edges) biconnected subgraph of $G$. If $G$ contains the vertices $s$ and $t$, we say that it is \emph{st-biconnectible} if the graph $G \cup (s,t)$ is biconnected.

Let $H$ be an $n$-vertex biconnected planar graph. 
A \emph{split pair} of $H$ is either a separation pair or a pair of adjacent vertices of $H$. 
A \emph{maximal split component} of $H$ with respect to a split pair $\{u,v\}$ (or, simply, a \emph{maximal split component} of $\{u,v\}$) is either an edge $(u,v)$ or a maximal subgraph $H'$ of $H$ such that $H'$ contains $u$ and $v$ and $\{u,v\}$ is not a split pair of $H'$. A vertex $w$ distinct from $u$ and $v$ belongs to exactly one maximal split component of $\{u,v\}$. We define a \emph{split component} of $\{u,v\}$ as the union of any number of maximal split components of $\{u,v\}$.
A split pair $\{u,v\}$ is \emph{maximal}, if there is no distinct split pair $\{w,z\}$ in $H$ such that $\{u, v\}$ is contained in a split component of $\{w,z\}$.

\begin{figure}[tb!]
	\centering
	\includegraphics[height=.33\textwidth]{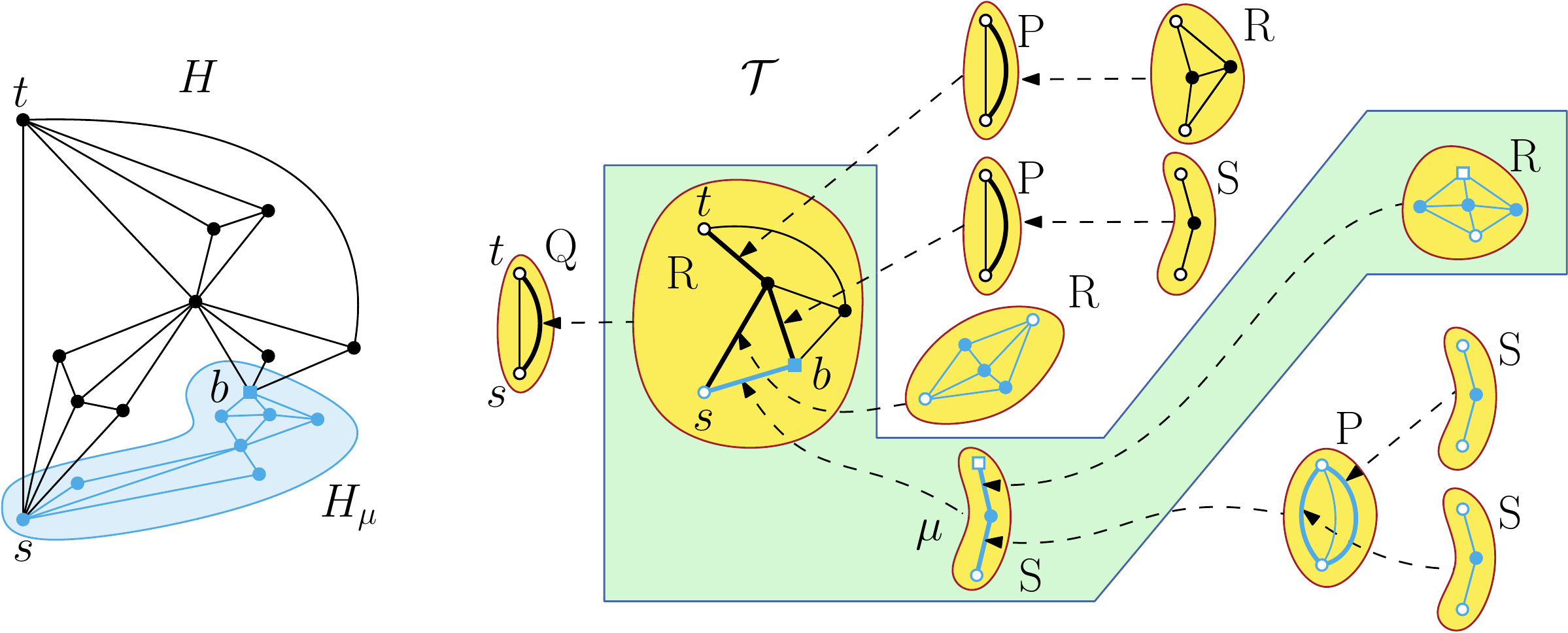}
	\caption{(left) A biconnected planar graph $H$ containing the edge $(s,t)$ and (right) the SPQR-tree $\cal T$ of $H$ rooted at the Q-node representing $(s,t)$. The leaves of $\mathcal T$, that is, the Q-nodes different from the root of $\mathcal T$, are omitted. The skeleton of each node of $\mathcal T$ that is not a leaf of $\mathcal T$ is represented inside a yellow region, corresponding to the node itself. Virtual edges corresponding to omitted Q-nodes are drawn thin, whereas virtual edges corresponding to S-, P-, and R-nodes are drawn thick. The pertinent graph $H_\mu$ of the S-node $\mu$ is enclosed in the blue shaded region.
	The allocation nodes of the vertex $b$ (squared vertex) are in the green-shaded region; the child of the root of $\cal T$ is the proper allocation node of $b$. The poles of the skeletons are filled white.
	}
	\label{fig:SPQR}
\end{figure}

The SPQR-tree $\mathcal T$ of $H$, defined in~\cite{dt-opl-96}, describes a recursive decomposition of $H$ with respect to its split pairs and represents succinctly all the planar embeddings of $H$. The tree $\mathcal T$ is a rooted tree with four types of nodes: S, P, Q, and R; refer to~\cref{fig:SPQR}.  Any node $\mu$ of $\mathcal T$ is associated with a planar $uv$-biconnectible graph, called \emph{skeleton} of $\mu$, which might contain multiple edges and which we denote by $\skel(\mu)$. 
The tree $\mathcal T$ is recursively defined as follows. 

Let $(s,t)$ be an edge of $H$, called \emph{reference edge}. We initialize $\mathcal T$ with a Q-node $\rho$ representing the edge $(s,t)$; $\rho$ is the \emph{root} of $\mathcal T$. The skeleton of $\rho$ consists of two parallel edges $(s,t)$, namely a \emph{real} edge $(s,t)$ and a \emph{virtual} edge $(s,t)$. Further, we insert into $\mathcal T$ a child $\tau$ of $\rho$.

Now assume that we are given a node $\mu$ of $\mathcal T$, a split component $H_{\mu}$ of $H$, and a pair of vertices $\{u,v\}$ of $H_\mu$, where (i) $H_\mu$ is a planar $uv$-biconnectible graph called \emph{pertinent graph} of $\mu$, and (ii) $u$ and $v$ are two vertices of $H_\mu$ called \emph{poles} of $\mu$. 

In order to meet this assumption after the initialization of $\mathcal T$, we let $\mu=\tau$, we let $H_\mu$ be the graph obtained from $H$ by removing the edge $(s,t)$, and we let $\{u,v\}=\{s,t\}$. 

\begin{itemize}
	\item {\em Trivial case}. If $H_\mu$ consists of a single edge $(u,v)$, we have that $\mu$ is a Q-node and is a leaf of $\mathcal T$; further, $\skel(\mu)$ also coincides with the edge $(u,v)$.
	
	\item {\em Series case}. If $H_\mu$ is not a single edge and is not biconnected, we have that $\mu$ is an S-node. Let $c_1, \dots, c_{k-1}$ (for some $k \geq  2$) be the cut-vertices of $H_\mu$, in the order in which they appear in any simple path in $H_{\mu}$ from $c_0=u$ to $c_k=v$. Then $\skel(\mu)$ is a path $(c_0,c_1,\dots,c_k)$. We insert $k$ children $\mu_1,\dots,\mu_k$ of $\mu$ in $\mathcal T$. For $i=1,\dots,k$, the pertinent graph $H_{\mu_i}$ of $\mu_i$ is the biconnected component of $H_\mu$ containing the vertices $c_{i-1}$ and $c_i$. Further, the poles of $\mu_i$ are the vertices $c_{i-1}$ and $c_i$. 
	
	\item {\em Parallel case}. If $H_\mu$ is not a single edge, if it is biconnected, and if $\{u,v\}$ is a split pair of $H_{\mu}$ defining $k$ split components of $H_\mu$, we have that $\mu$ is a P-node. Then $\skel(\mu)$ consists of $k$ parallel edges $(u,v)$. We insert $k$ children $\mu_1,\dots,\mu_k$ of $\mu$ in $\mathcal T$. The pertinent graphs $H_{\mu_1},\dots,H_{\mu_k}$ of $\mu_1,\dots,\mu_k$ are the split components of $H_\mu$; these are planar $uv$-biconnectible graphs. Further, the poles of $\mu_i$ are $u$ and $v$, for every $i=1,\dots,k$. 
	
	\item {\em Rigid case}. If $H_\mu$ is not a single edge, if it is biconnected, and if $\{u,v\}$ is not a split pair of $H_{\mu}$, we have that $\mu$ is an R-node. Let $\{u_1,v_1\},\dots,\{u_k,v_k\}$ be the maximal split pairs of $H_{\mu}$. We insert $k$ children $\mu_1,\dots,\mu_k$ of $\mu$ in $\mathcal T$. For $i=1,\dots,k$, the pertinent graph  $H_{\mu_i}$ of $\mu_i$ is the union of all the split components of $\{u_i,v_i\}$; then $H_{\mu_i}$ is a planar $uv$-biconnectible graph. The graph $\skel(\mu)$ is obtained from $H_{\mu}$ by replacing each subgraph $H_{\mu_i}$ with an edge $(u_i,v_i)$. Then $\skel(\mu)$ is a planar $uv$-biconnectible graph that becomes triconnected if the edge $uv$ is added to it.
\end{itemize}

For each node $\mu$ of $\mathcal T$ that is not a Q-node, the edges of $\skel(\mu)$ are called \emph{virtual edges}, as they do not correspond to real edges of $H$, but rather to subgraphs of $H$. Indeed, every virtual edge $e_i$ in $\skel(\mu)$ is associated with a child $\mu_i$ of $\mu$ and thus corresponds to the pertinent graph $\pert{\mu_i}$, which in fact we also denote by $\pert{e_i}$. On the other hand, for a Q-node different from the root of $\mathcal T$, the only edge of $\skel(\mu)$ is not virtual, but rather is a \emph{real edge}. 

Let~$w$ be a vertex of $H$. The \emph{allocation nodes} of $w$ are the nodes of $\mathcal T$ whose skeletons contain~$w$. Note that~$w$ has at least one allocation node. The lowest common ancestor of the allocation nodes of~$w$ is itself an allocation node of~$w$, called the \emph{proper allocation node} of~$w$. 
We have the following.

\begin{remark}\label{remark:pan}
Let $\mu$ be a non-root node of $\mathcal T$ and let $w$ be a vertex of $\skel(\mu)$. Then, $\mu$ is the proper allocation node of $w$ if and only if $w$ is not a pole of $\mu$.
\end{remark}


If $H$ has $n$ vertices, then  $\mathcal T$ has $O(n)$ nodes and the total number of virtual edges in the skeletons of the nodes of $\mathcal T$ is in $O(n)$.  From a computational complexity perspective, $\mathcal T$ can be constructed in $O(n)$ time~\cite{DBLP:conf/gd/GutwengerM00}. To ease the description of our embedding algorithms, we use the slightly modified version of SPQR-trees defined in \cite{DBLP:journals/siamdm/DidimoGL09}, where each S-node has exactly two children. The SPQR-trees defined in this way can still be constructed in $O(n)$ time.

The SPQR-tree $\mathcal T$ allows to succinctly and recursively construct all the planar embeddings of $H$. We now explain this fact. 

First, for any node $\mu$ of $\mathcal T$, the restriction of any planar embedding $\mathcal E$ of $H$ to the vertices and edges of $H_{\mu}$ is a planar embedding $\mathcal E_{\mu}$ of $H_{\mu}$ in which the poles of $\mu$ are incident to a common face. This is best seen by assuming\footnote{We remark that, in all the problems we address in this paper, the existence of an outer face is not relevant other than for ease of description. Then it can be assumed without loss of generality, as we do in the following, that a certain edge is incident to the outer face of every planar embedding we consider.} that the reference edge $(s,t)$ of $H$ is incident to the outer face of $\mathcal E$; then the poles of $\mu$ are incident to the outer face of $\mathcal E_{\mu}$. A planar embedding $\mathcal E$ of $H$ also defines a \emph{corresponding} planar embedding $\mathcal S_{\mu}$ for the skeleton $\skel(\mu)$ of each node $\mu$ of $\mathcal T$ different from the root. The embedding $\mathcal S_{\mu}$ can be obtained from the embedding $\mathcal E_{\mu}$ of $H_{\mu}$ in $\mathcal E$ by replacing the pertinent graph of each virtual edge $e_i$ of $\skel(\mu)$ with $e_i$. 

Second, every planar embedding of $H$ in which $(s,t)$ is incident to the outer face can be obtained by bottom-up traversing $\mathcal T$ and constructing, at each node $\mu$ of $\mathcal T$ different from the root, a planar embedding of $H_{\mu}$ in which the poles of $\mu$ are incident to the outer face. This is done by performing two choices at $\mu$: (i) a planar embedding $\mathcal S_{\mu}$ for $\skel(\mu)$ such that the poles of $\mu$ are incident to the outer face; and (ii)  for each virtual edge $e_i$ of $\skel(\mu)$, whether to flip or not the already constructed embedding of the pertinent graph of $e_i$, where a \emph{flip} is a reversal of the adjacency list of each vertex. From these choices, a planar embedding of $H_{\mu}$ in which the poles of $\mu$ are incident to the outer face is constructed starting from the chosen planar embedding $\mathcal S_{\mu}$ of $\skel(\mu)$ and by replacing each virtual edge $e_i$ with the already constructed embedding of the pertinent graph of $e_i$ or with its flip, as chosen. Different types of nodes of $\mathcal T$ allow for different choices for the planar embedding $\mathcal S_{\mu}$ of $\skel(\mu)$. If a node $\mu$ of $\mathcal T$ is an S-node, then $\skel(\mu)$ is a path and has a unique planar embedding, hence there is nothing to choose. If $\mu$ is an R-node, then its skeleton $\skel(\mu)$ also has a unique planar embedding in which the poles of $\mu$ are incident to the outer face~\cite{Whitney33}, up to a flip, and again there is nothing to choose. Conversely, if $\mu$ is a P-node, then $\mathcal S_{\mu}$ can be chosen as any permutation of its virtual edges. After bottom-up traversing $\mathcal T$ up to the child $\tau$ of the root, a planar embedding of $H$ is constructed by inserting the edge $(s,t)$ in the outer face of the constructed embedding of $H_{\tau}$.

In the above bottom-up construction of an embedding $\mathcal E$ of $H$, an injection is naturally defined from the internal faces of $\mathcal S_{\mu}$ to the faces of $\mathcal E$, where $\mathcal S_{\mu}$ is the embedding of the skeleton $\skel(\mu)$ of a node $\mu$ of $\mathcal T$ chosen in the construction of $\mathcal E$. Indeed, for every internal face $f$ of $\mathcal S_{\mu}$ incident to virtual edges $e_1,\dots,e_k$, there is a distinct and unique face $g$ of $\mathcal E$ that is delimited by edges of the pertinent graphs of all of $e_1,\dots,e_k$; then $g$ is the face of $\mathcal E$ that \emph{corresponds to} $f$ (see \cref{fig:SPQR2} for an example). This correspondence also extends to the outer face(s) of $\mathcal S_{\mu}$ in a slightly less obvious way. Note that the restriction $\mathcal E_{\mu}$ of $\mathcal E$ to the pertinent graph $H_{\mu}$ of a non-root node $\mu$ of $\mathcal T$ has an outer face that is delimited by two paths between the poles of $\mu$, possibly sharing some vertices and edges. These two paths are incident to distinct faces of $\mathcal E$. These are called \emph{outer faces of}  $\mathcal E_{\mu}$; in fact, we often refer to the outer faces of an embedding  $\mathcal E_{\mu}$ of  the pertinent graph $H_{\mu}$ of a non-root node $\mu$ of $\mathcal T$ even when an entire embedding of $H$ is not specified. By convention, if $u$ and $v$ denote the poles of $\mu$, we call \emph{left outer face} $\ell(\mathcal E_{\mu})$ of $\mathcal E_{\mu}$ (\emph{right outer face} $r(\mathcal E_{\mu})$ of $\mathcal E_{\mu}$) the outer face that is delimited by the path obtained by walking in clockwise direction (resp.\ in counter-clockwise direction) from $u$ to $v$ along the boundary of the outer face of $\mathcal E_{\mu}$. The terms left outer face and right outer face come from the fact that we usually think about $\mathcal E_{\mu}$ as having a pole $u$ of $\mu$ at the bottom and the other pole $v$ of $\mu$ at the top. We also talk about left and right outer faces $\ell(\mathcal S_{\mu})$ and $r(\mathcal S_{\mu})$, respectively, of the embedding $\mathcal S_{\mu}$ of $\skel(\mu)$. These can be defined by adding an edge $(u,v)$ in the outer face of  $\mathcal S_{\mu}$; then the \emph{left outer face} of $\mathcal S_{\mu}$ (\emph{right outer face} of $\mathcal S_{\mu}$) is the face delimited by the cycle composed of $(u,v)$ and of the path obtained by walking in clockwise direction (resp.\ in counter-clockwise direction) from $u$ to $v$ along the boundary of the outer face of $\mathcal S_{\mu}$. 

\begin{figure}[tb!]
	\centering
	\includegraphics[width=.85\textwidth]{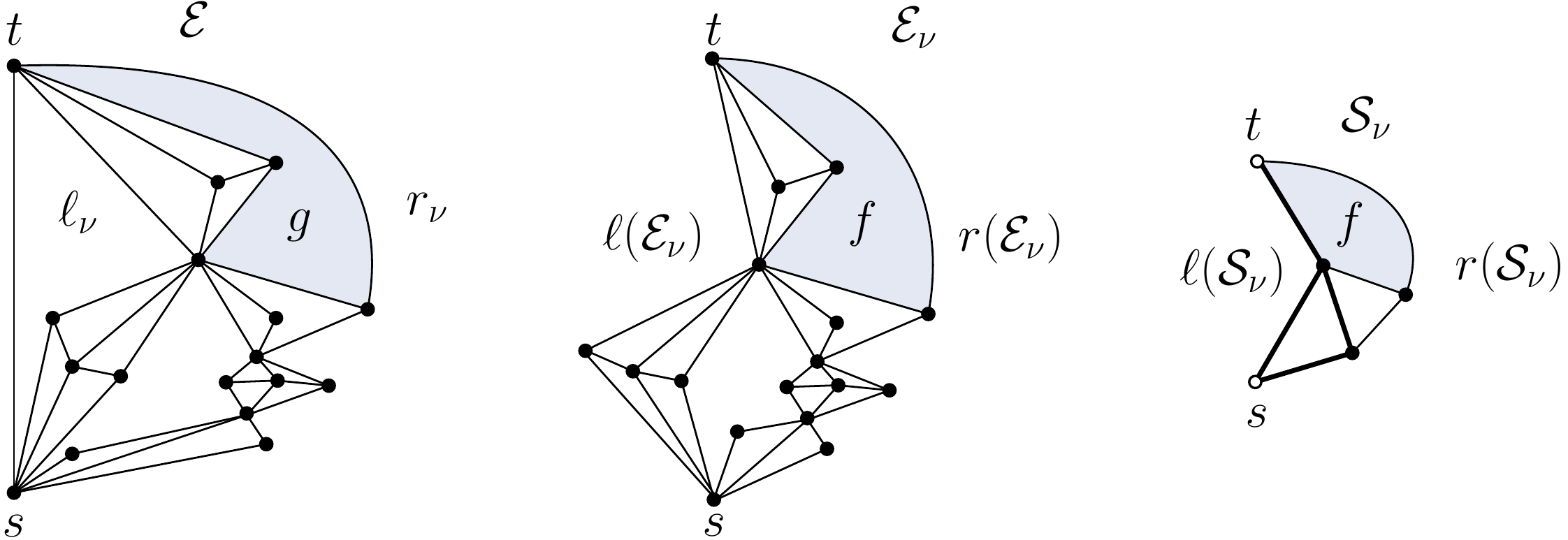}
	\caption{Consider the R-node $\nu$ child of the root of the SPQR-tree of \cref{fig:SPQR}. The internal face $f$ of $\mathcal{E}_{\nu}$ corresponds to face $g$ of $\mathcal E$. Faces $\ell(\mathcal{E}_\nu)$ and $r(\mathcal{E}_\nu)$ are the left outer face and the right outer face of $\mathcal{E}_{\nu}$, respectively, and correspond to faces $\ell_\nu$ and $r_\nu$ of $\mathcal{E}$, respectively.}
	\label{fig:SPQR2}
\end{figure}


\section{Complexity}\label{se:complexity}
In this section, we show that the \btpbeP{} (\btpbe) problem is \NPC. 
This result and \cref{cor:equivalence-variable} then imply that the {\sc 2-Level Quasi-Planarity} and the {\sc $(2,2)$-Track Graph Recognition} problems are also \NPC.

The membership in \NP~is trivial. Therefore, in the remainder of the section we will focus on proving the \NP-hardness of the problem.

\medskip
\paragraph{Leveled Planarity.} Our \NPHN proof is based on a polynomial-time reduction from an \NP-complete problem called {\sc Leveled Planarity}~\cite{DBLP:journals/siamcomp/HeathR92}.  We start with a definition. 
\begin{definition}\label{def:leveled-planar}
A \emph{leveled planar drawing} of a graph $H=(V,E)$ is a triple $\langle k,\gamma,\Sigma\rangle$ containing:
\begin{enumerate}
\item \label{con:integer} an integer $k \leq |V|$, 
\item \label{con:i-leveled} a function $\gamma: V \rightarrow \{1,\dots,k\}$ such that, for each edge $(u,v) \in E$, it holds $\gamma(v)= \gamma(u) \pm 1$, and
\item \label{cond:orders} a sequence of orders $\Sigma=\langle\xi_1,\xi_2,\dots,\xi_k\rangle$, where $\xi_i: U_i \rightarrow \{1,\dots,|U_i|\}$ with $U_i:= \{v \in V: \gamma(v)=i\}$, such that, for any two edges $(u,v),(p,q) \in E$ with $a=\gamma(u)=\gamma(p)$ and $a+1=\gamma(v)=\gamma(q)$, it holds that $\xi_a(u) < \xi_a(p)$ if and only if $\xi_{a+1}(v) < \xi_{a+1}(q)$.
\end{enumerate}

A graph $H=(V,E)$ is \emph{leveled planar} if at admits a leveled planar drawing $\langle k,\gamma,\Sigma\rangle$. 
\end{definition}

We say that the subset $U_i$ of $V$ is \emph{level} $i$ of $\langle k,\gamma,\Sigma\rangle$, for $i=1,\dots,k$. 
Observe that bipartiteness is necessary for a graph to be leveled planar, by \cref{con:i-leveled} of \cref{def:leveled-planar}. 

The {\sc Leveled Planarity} problem, proved \NPC by Heath and Rosenberg~\cite{DBLP:journals/siamcomp/HeathR92}, asks whether a graph is leveled planar.
In fact, in~\cite{DBLP:journals/siamcomp/HeathR92}, it is shown that {\sc Leveled Planarity} is \NPC also when the following two properties hold: 
\begin{enumerate}[(i)]
	\item the instance is connected; and 
	\item the input also specifies the number $k$ of levels and two special vertices $v_1$ and $v_k$ such that
	$v_1$ is the only vertex on level $1$ and $v_k$ is the only vertex on level $k$, i.e.,
	$\gamma^{-1}(1)=\{v_1\}$ and $\gamma^{-1}(k)=\{v_k\}$. 
\end{enumerate}
Furthermore, the proof in~\cite{DBLP:journals/siamcomp/HeathR92} can be easily adapted to show that the problem remains \NPC even when $k$ is constrained to be odd. We call the resulting problem {\sc Odd-Leveled Planarity} and denote its input as a quadruple $\langle H, k, v_1, v_k \rangle$, where $H$ is a connected bipartite graph, $k$ is an odd integer, and $v_1, v_k \in V(H)$.

\smallskip
\paragraph{Proof strategy.} 
To prove the \NP-hardness of \btpbe we proceed as follows. A \emph{subdivision} of a graph $K$ is a graph $K'$ obtained by replacing some edges of $K$ with paths whose internal vertices have degree $2$ in $K'$; these are called \emph{subdivision vertices}. Let $\langle H, k, v_1, v_k \rangle$ be an instance of {\sc Odd-Leveled Planarity}.
First, we construct a graph $F_k$, called \emph{frame}, that has the following properties (refer to \cref{fig:frame}):
\begin{enumerate}[(1)]
	\item $F_k$ is a bipartite graph which is a subdivision of a triconnected planar graph (and, thus, it has a unique planar embedding, up to a flip~\cite{Whitney33});
	\item $F_k$ has a unique face $f_{in}$ consisting of $2k$ vertices, while all the other faces have $4$ vertices (see \cref{fig:frame-a}); and
	\item $F_k$ admits a unique saturator $\cal C$; in such a saturator the edges of $\cal C$ traversing $f_{in}$ form a non-crossing matching $M=\{(v_i,x_i): 1 \leq i \leq k-1\}$ of alternating red and black edges (see \cref{fig:frame-b}). 
\end{enumerate}  
Then we combine $H$ with $F_k$ into a new graph $G$, by identifying $v_1$ and $v_k$ with the corresponding vertices of $f_{in}$; this ensures that $G$ is bipartite and that $H$ lies inside $f_{in}$ in any planar embedding of $G$ (and thus in any \btpbe  of $G$).  
Finally, we show that any saturator $\mathcal C_G$ of $G$ is obtained from the unique saturator $\cal C$ of $F_k$ by subdividing the edges of $M$ into a set $P$ of paths; this is done by suitably inserting all the vertices of $H$ different from $v_1$ and $v_k$ into the edges of $M$ (see \cref{fig:frame-c}).
This allows us to establish a one-to-one correspondence between a saturator of $G$ and a solution $\langle k,\gamma,\Sigma\rangle$ of {\sc Odd-Leveled Planarity} for $\langle H, k, v_1, v_k \rangle$, where the assignment $\gamma$ of the vertices of $H$ to the levels and the order $\xi_i\in \Sigma$ of the vertices of each level $i$ are those defined by the paths in $P$.

\begin{figure}[t!]
		\centering
\subfloat[]{
\includegraphics[page=1,height=.5\textwidth]{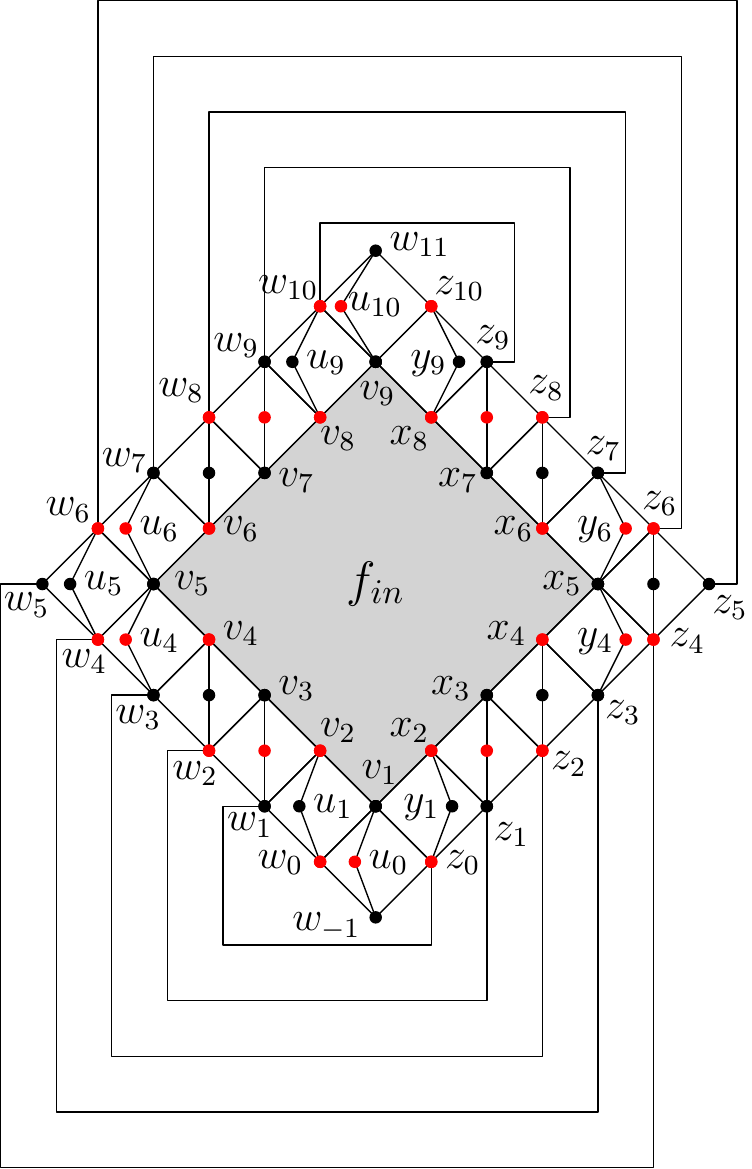}
\label{fig:frame-a}
}
\subfloat[]{
\includegraphics[page=2,height=.5\textwidth]{nuovo-gadget-rod.pdf}
\label{fig:frame-b}
}
\subfloat[]{
\includegraphics[page=3,height=.5\textwidth]{nuovo-gadget-rod.pdf}
\label{fig:frame-c}
}
\caption{(a) The frame $F_9$; the face $f_{in}$ bounded by the inner cycle $c_{in}$ is shaded gray.
(b) The graph $F_9$ together with the edges of its unique saturator; the red and black edges of the saturator are thick. The face bounded by the $4$-cycle $c_4 = (v_3,v_4,x_4,x_3,v_3)$ is shaded gray.
(c) A planar drawing of $G \cup E(\mathcal C_G)$, where $G= F_9 \cup H$, the graph $H$ (in the shaded-gray region) is a positive instance of {\sc Odd-Leveled Planarity} with $k=9$ and $\mathcal C_G$ is a saturator~of~$G$.}
\label{fig:frame}
\end{figure}

\smallskip
\paragraph{Frame gadget.} We describe the frame $F_k$; refer to \cref{fig:frame-a}. 

We initialize $F_k$ to the union of an \emph{inner cycle} $c_{in}=({v_1},v_2,\dots,v_{k-1},{v_k},x_{k-1},\dots,x_{2},{v_1})$ and an \emph{outer cycle} $({w_{-1}},w_0,\dots,w_{k+1},{w_{k+2}},z_{k+1},\dots,z_0, {w_{-1}})$.
For $i=0,\dots,\frac{k-1}{2}$, we add the edges $(w_i,v_{i+1})$ and $(z_i,x_{i+1})$, where $x_1=v_1$; also, for $i=\frac{k+1}{2},\dots,k$, we add the edges $(v_i,w_{i+1})$ and $(x_i,z_{i+1})$, where $x_k=v_k$.
Then, for $i=1,\dots,\frac{k-1}{2}$, we add a vertex~$u_i$ together with the edges $(w_{i-1},u_i)$ and $(u_i,v_{i+1})$, and a vertex $y_i$ together with the edges $(z_{i-1},y_i)$ and $(y_i,x_{i+1})$. Also, for $i=\frac{k+3}{2},\dots,k$, 
we add a vertex $u_i$ together with the edges $(w_{i+1},u_i)$ and $(u_i,v_{i-1})$, and a vertex $y_i$ together with the edges $(z_{i+1},y_i)$ and $(y_i,x_{i-1})$.
Further, we add a vertex $u_0$ together with the edges $(w_{-1},u_0)$ and $(u_0,v_1)$, a vertex $u_{k+1}$ together with the edges $(w_{k+2},u_{k+1})$ and $(u_{k+1},v_k)$, a vertex $u_{\frac{k+1}{2}}$ together with the edges $(w_{\frac{k-1}{2}},u_\frac{k+1}{2})$ and $(u_\frac{k+1}{2},w_{\frac{k+3}{2}})$, and a vertex $y_{\frac{k+1}{2}}$ together with the edges $(z_{\frac{k-1}{2}},y_\frac{k+1}{2})$ and $(y_\frac{k+1}{2},z_{\frac{k+3}{2}})$. 
Finally, for $i=0,\dots,k$, we add the edges $(z_i,w_{i+1})$.

By construction, the frame is bipartite and is a subdivision of a triconnected planar graph (with subdivision vertices $u_0,\dots,u_{k+1}$ and $y_1,\dots,y_k$). Hence, it has a unique planar embedding $\mathcal E$ (up to a flip), whose faces all have  length $4$, except for the face $f_{in}$ bounded by $c_{in}$, which has length $2k$. We now need to prove that $\cal C$ has a unique saturator.
We start with the following.

\begin{lemma}\label{lem:two-red-edges}
Any saturator $\cal C$ of the frame $F_k$ contains the saturating red edges $(w_0,u_0)$, $(u_0,z_0)$, $(w_{k+1},u_{k+1})$, and $(u_{k+1},z_{k+1})$.
\end{lemma}

\begin{proof}
Let $\cal E'$ be any planar embedding of $F_k \cup E(\cal C)$.
In the following, we show that $\cal C$ contains the saturating red edges $(w_0,u_0)$ and $(u_0,z_0)$; the proof that $\cal C$ contains the saturating red edges $(w_{k+1},u_{k+1})$ and $(u_{k+1},z_{k+1})$ is analogous. 

The only red vertices $u_0$ shares faces of~$\cal E$ with are $w_0$ and $z_0$; hence, at least one of the saturating edges $(w_0,u_0)$ and $(u_0,z_0)$ must belong to $\cal C$. 


First, suppose, for a contradiction, that $(w_0,u_0) \notin E({\cal C})$ and $(u_0,z_0) \in E({\cal C})$; refer to \cref{fig:lemma-two-red-edges-NO-w_0u_0-global}.
Since $(w_0,u_0) \notin E({\cal C})$, by \cref{le:quadrangles} we have that $\cal C$ contains the black saturating edge $(w_{-1},v_1)$; further, the edge $(w_{-1},v_1)$ lies in the interior of the $4$-cycle $(w_{-1}, w_0, v_{1}, u_{0})$ in~$\cal E'$,  as this cycle bounds a face of~$\cal E$; refer to \cref{fig:lemma-two-red-edges-NO-w_0u_0}. 
 
Since the only red vertices $u_0$ shares a face of~$\cal E$ with are $w_0$ and $z_0$, our assumption implies that $u_0$ is a red end-vertex of~$\cal C$.  Also, since the only black vertices $u_0$ shares a face of~$\cal E$ with are $v_1$ and $w_{-1}$, by \cref{prop:edpointsAREclose} one of such vertices must be a black end-vertex of~$\cal C$. Note that $w_{-1}$ and $v_1$ are not both black end-vertices of~$\cal C$, as otherwise $\cal C$ would not span all the black vertices of $F_k$, given that the path $P$ of black saturating edges would only consist of the edge $(w_{-1},v_1)$.


\begin{itemize}
\item Suppose first that $w_{-1}$ is a black end-vertex of~$\cal C$ and that $v_1$ is not; refer to \cref{fig:lemma-two-red-edges-NO-w_0u_0-AND-w_-1-endvertex,fig:lemma-two-red-edges-NO-w_0u_0-AND-w_-1-endvertex-first-case}. Since $w_{-1}$ is a black end-vertex of~$\cal C$ and since the edge $(w_{-1},v_1)$ belongs to $\cal C$, the black saturating edge $(w_{-1},w_1)$ does not belong to $\cal C$. Therefore, by \cref{le:quadrangles}, $\cal C$ contains the red saturating edge $(w_0,z_0)$, which lies in the interior of the $4$-cycle $(w_{-1}, z_0, w_{1}, w_{0})$ in~$\cal E'$, as this cycle bounds a face of~$\cal E$.
Also, since $z_0$ is adjacent to $u_0$ and $w_0$ in $\cal C$, we have that $\cal C$ contains neither the red saturating edge $(z_0,x_2)$ nor the red saturating edge $(z_0,w_2)$. Thus, by \cref{le:quadrangles}, $\cal C$ contains the black saturating edges $(v_1,y_1)$, $(y_1,z_1)$, and $(z_1,w_1)$, which lie in the interior of the $4$-cycles $(z_0,v_1,x_2,y_1)$, $(z_0,y_1,x_2,z_1)$, and $(z_0,z_1,w_2,w_1)$ in~$\cal E'$, respectively, since such cycles bound faces of~$\cal E$; refer to \cref{fig:lemma-two-red-edges-NO-w_0u_0-AND-w_-1-endvertex}.

\begin{figure}[tb!]
	\centering
	\subfloat[]{
		\includegraphics[page=1,height=.18\textwidth]{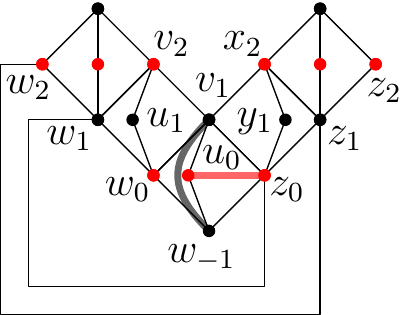}\label{fig:lemma-two-red-edges-NO-w_0u_0}
	}\hfil
	\subfloat[]{
		\includegraphics[page=2,height=.18\textwidth]{lemma-two-red-edges.pdf}\label{fig:lemma-two-red-edges-NO-w_0u_0-AND-w_-1-endvertex}
	}\hfil
	\subfloat[]{
		\includegraphics[page=3,height=.18\textwidth]{lemma-two-red-edges.pdf}\label{fig:lemma-two-red-edges-NO-w_0u_0-AND-w_-1-endvertex-first-case}
	}
	\hfil
	\subfloat[]{
		\includegraphics[page=5,height=.18\textwidth]{lemma-two-red-edges.pdf}\label{fig:lemma-two-red-edges-NO-w_0u_0-AND-v_1-endvertex}
	}
	\caption{Illustration for the proof of \cref{lem:two-red-edges}, when $(w_0,u_0) \notin E({\cal C})$ and $(u_0,z_0) \in E({\cal C})$. In (b) and (c), the vertex $w_{-1}$ is supposed to be a black end-vertex of~$\cal C$. In (d), the vertex $v_1$ is supposed to be a black end-vertex of~$\cal C$.}
	\label{fig:lemma-two-red-edges-NO-w_0u_0-global}
\end{figure}

Further, since $v_1$ is adjacent to $w_{-1}$ and $y_1$ in $\cal C$, we have that $\cal C$ does not contain the black saturating edge $(v_1,u_1)$. Thus, by \cref{le:quadrangles}, $\cal C$ contains the red saturating edge $(w_0,v_2)$, which lies in the interior of the $4$-cycle $(w_0,u_1,v_2,v_1)$ in~$\cal E'$. Again by \cref{le:quadrangles}, we have that $\cal C$ contains the black saturating edge $(w_1,u_1)$, which lies in the interior of the cycle $(w_0,w_1,v_2,u_1)$, as this cycle bounds a face in~$\cal E$ (see \cref{fig:lemma-two-red-edges-NO-w_0u_0-AND-w_-1-endvertex-first-case}). Since the only black vertices the vertex $u_1$ shares a face of~$\cal E$ with are $w_1$ and $v_1$, and since $\cal C$ does not contain the black saturating edge $(v_1,u_1)$, it follows that $u_1$ is a black end-vertex of~$\cal C$. However, the path $(w_{-1},v_1,y_1,z_1,w_1,u_1)$ of black saturating edges of $\cal C$ does not span all the black vertices, which contradicts the assumption that $\cal C$ is a saturator.

%

\item Suppose next that $v_{1}$ is a black end-vertex of~$\cal C$ and that $w_{-1}$ is not; refer to \cref{fig:lemma-two-red-edges-NO-w_0u_0-AND-v_1-endvertex}. Since $v_{1}$ is a black end-vertex of~$\cal C$, the black saturating edges $(u_1,v_1)$ and $(v_1,y_1)$ do not belong to $\cal C$. Therefore, by \cref{le:quadrangles}, $\cal C$ contains the red saturating edges $(w_0,v_2)$ and $(z_0,x_2)$, which lie in the interior of the $4$-cycles $(w_0,u_1,v_2,v_1)$ and $(z_0,v_1,x_2,y_1)$ in~$\cal E'$, respectively, as these cycles bound faces of~$\cal E$. Consequently, $\cal C$ does not contain the black saturating edges $(v_1,u_1)$ and $(v_1,y_1)$. Since the only black vertices the vertex $u_1$ shares a face of~$\cal E$ with are $w_1$ and $v_1$, and since $\cal C$ does not contain the black saturating edge $(v_1,u_1)$, it follows that $u_1$ is a black end-vertex of~$\cal C$. Analogously, since the only black vertices the vertex $y_1$ shares a face of~$\cal E$ with are $z_1$ and $v_1$, and since $\cal C$ does not contain the black saturating edge $(v_1,y_1)$, it follows that $y_1$ is a black end-vertex of~$\cal C$. Therefore, $\cal C$ contains three black end-vertices, namely $u_1$, $y_1$, and $v_1$, a contradiction.
\end{itemize}

Suppose now that $(w_0,u_0) \in E({\cal C})$ and $(u_0,z_0) \notin E({\cal C})$; refer to \cref{fig:lemma-two-red-edges-NO-u_0z_0-global}. Since $(u_0,z_0) \notin E({\cal C})$, by \cref{le:quadrangles} we have that $\cal C$ contains the black saturating edge $(w_{-1},v_1)$, which lies in the interior of the $4$-cycle $(w_{-1}, u_0, v_{1}, z_{0})$ in~$\cal E'$, as this cycle bounds a face of~$\cal E$; refer to \cref{fig:lemma-two-red-edges-NO-u_0z_0}. 

Since the only red vertices $u_0$ shares a face of~$\cal E$ with are $w_0$ and $z_0$, our assumption implies that $u_0$ is a red end-vertex of~$\cal C$. 
Also, since the only black vertices $u_0$ shares a face of~$\cal E$ with are $v_1$ and $w_{-1}$, by \cref{prop:edpointsAREclose} one of such vertices must be a black end-vertex of~$\cal C$. Note that $w_{-1}$ and $v_1$ are not both black end-vertices of~$\cal C$, as otherwise $\cal C$ would not span all the black vertices of $F_k$, given that the path $P$ of black saturating edges would only consist of the edge $(w_{-1},v_1)$.

\begin{figure}[tb!]
\centering
\subfloat[]{
\includegraphics[page=6,height=.18\textwidth]{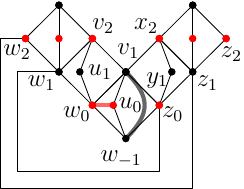}\label{fig:lemma-two-red-edges-NO-u_0z_0}
}\hfil
\subfloat[]{
\includegraphics[page=7,height=.18\textwidth]{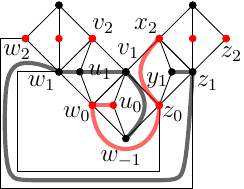}\label{fig:lemma-two-red-edges-NO-u_0z_0-AND-w_-1-endvertex}
}\hfil
\subfloat[]{
\includegraphics[page=8,height=.18\textwidth]{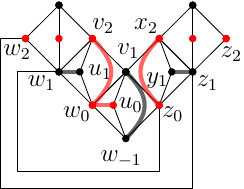}\label{fig:lemma-two-red-edges-NO-u_0z_0-AND-v_1-endvertex}
}
\caption{Illustration for the proof of \cref{lem:two-red-edges}, when $(w_0,u_0) \in E({\cal C})$ and $(u_0,z_0) \notin E({\cal C})$. In (b), the vertex $w_{-1}$ is supposed to be a black end-vertex of~$\cal C$. In (c), the vertex $v_1$ is supposed to be a black end-vertex of~$\cal C$.}
\label{fig:lemma-two-red-edges-NO-u_0z_0-global}
\end{figure}

\begin{itemize}
	\item Suppose that $w_{-1}$ is a black end-vertex of~$\cal C$ and that $v_1$ is not; refer to \cref{fig:lemma-two-red-edges-NO-u_0z_0-AND-w_-1-endvertex}. Since $w_{-1}$ is a black end-vertex of~$\cal C$ and since the edge $(w_{-1},v_1)$ belongs to $\cal C$, the black saturating edge $(w_{-1},w_1)$ does not belong to $\cal C$. Therefore, by \cref{le:quadrangles}, $\cal C$ contains the red saturating edge $(w_0,z_0)$, which lies in the interior of the $4$-cycle $(w_{-1}, z_0, w_{1}, w_{0})$ in~$\cal E'$, as this cycle bounds a face of~$\cal E$.
	Also, since $w_0$ is adjacent to $u_0$ and $z_0$ in $\cal C$, we have that $\cal C$ does not contain the red saturating edge $(w_0,v_2)$. Thus, by \cref{le:quadrangles}, $\cal C$ contains the black saturating edges $(w_1,u_1)$ and $(u_1,v_1)$, which lie in the interior of the $4$-cycles $(w_0,w_1,v_2,u_1)$ and $(w_0,u_1,v_2,v_1)$ in~$\cal E'$, respectively, as these cycles bound faces of~$\cal E$. 
	This, in turn, implies that $v_1$ is adjacent to $u_1$ and $w_{-1}$ in $\cal C$, and thus $\cal C$ does not contain the black saturating edge $(v_1,y_1)$. Therefore, by \cref{le:quadrangles}, $\cal C$ contains the red saturating edge $(z_0,x_2)$, which lies in the interior of the $4$-cycle $(z_0,v_1,x_2,y_1)$ in~$\cal E'$, as this cycle bounds a face in~$\cal E$. 
	Then \cref{le:quadrangles} implies that $\cal C$ contains the black saturating edge $(y_1,z_1)$, which lies in the interior of the $4$-cycle $(z_0,y_1,x_2,z_1)$ in~$\cal E'$, as this cycle bounds a face in~$\cal E$. 
	Since the only black vertices $y_1$ shares a face of~$\cal E$ with are $v_1$ and $z_1$, we have that $y_1$ is a black end-vertex of~$\cal C$.
	Also, since $z_0$ is adjacent to $x_2$ and $w_0$ in $\cal C$, the red saturating edge $(z_0,w_2)$ does not belong $\cal C$, and thus by \cref{le:quadrangles}, $\cal C$ contains the black saturating edge $(z_1,w_1)$, which lies in the interior of the $4$-cycle $(z_0,z_1,w_2,w_1)$ in~$\cal E'$, as this cycle bounds a face of~$\cal E$.  However, the path $(w_{-1},v_1,u_1,w_1,z_1,y_1)$ of black saturating edges of $\cal C$ does not span all the black vertices, which contradicts the assumption that $\cal C$ is a saturator.
	\item In the case in which $v_{1}$ is a black end-vertex of~$\cal C$ and $w_{-1}$ is not, a contradiction can be derived exactly as in the case in which the same assumptions are satisfied with $(u_0,z_0) \in E({\cal C})$ and $(w_0,u_0) \notin E({\cal C})$; refer also to \cref{fig:lemma-two-red-edges-NO-u_0z_0-AND-v_1-endvertex}.
\end{itemize}

Since a contradiction has been obtained in every case, it follows that $\cal C$ contains both the red saturating edges $(w_0,u_0)$ and $(u_0,z_0)$. This concludes the proof of the lemma.
\end{proof}

We continue the investigation of the structure of a saturator of $F_k$ with the following.

\begin{lemma}\label{lem:end-vertices-and-saturators}
Let $\cal C$ be any saturator of the frame $F_k$. Then the following hold:
\begin{enumerate}[(a)]
	\item \label{end-vertices} the vertices $w_0$ and $z_{k+1}$ are the red end-vertices of~$\cal C$, and the vertices $w_{-1}$ and $w_{k+2}$ are the black end-vertices of~$\cal C$,
	\item \label{edges:1} $\cal C$ contains the black saturating edges $(w_{-1},w_1)$ and $(w_{k+2},z_k)$, and the red saturating edges $(z_0,w_2)$ and $(w_{k+1},z_{k-1})$, and
	\item \label{edges:2} $\cal C$ contains the saturating edges $(w_i,u_i)$, $(u_i,v_i)$, $(x_i,y_i)$, 
	$(y_i,z_i)$, $(z_i,w_{i+2})$, for $i=1,\dots, k$, where $x_1=v_1$ and $x_k = v_k$.
\end{enumerate}
\end{lemma}

\begin{proof}
Let $\cal E'$ be any planar embedding of $F_k \cup E(\cal C)$. We first prove Properties~(a) and~(b).

%
%
%
%

\begin{itemize}
	\item The only black vertices the vertex $w_{-1}$ shares a face of~$\cal E$ with are $w_1$ and $v_1$. By \cref{lem:two-red-edges}, $\cal C$ contains the saturating red edges $(w_0,u_0)$ and $(u_0,z_0)$, hence it does not contain the edge $(w_{-1},v_1)$. It follows that $w_{-1}$ is a black end-vertex of $\cal C$ and that $\cal C$ contains the black saturating edge $(w_{-1},w_1)$, which lies in the interior of the $4$-cycle $(w_{-1},z_0,w_1,w_0)$, as this cycle bounds a face in $\cal E$. An analogous proof shows that $w_{k+2}$ is the other black end-vertex of $\cal C$ and that $\cal C$ contains the black saturating edge $(w_{k+2},z_k)$.
	\item Next, we show that the vertex $w_0$ is a red end-vertex of $\cal C$; the proof that $z_{k+1}$ is the other red end-vertex of $\cal C$ is analogous.
	
	Since $w_{-1}$ is a black end-vertex of~$\cal C$, since the only red vertices $w_{-1}$ shares a face of~$\cal E$ with are $w_0$, $u_0$, and $z_0$, and since, by \cref{lem:two-red-edges}, $u_0$ is adjacent to both $w_0$ and $z_0$ in $\cal C$, by \cref{prop:edpointsAREclose} we have that at least one of $w_0$ and $z_0$ is a red end-vertex of~$\cal C$.

	\begin{figure}[tb!]
		\centering
		\subfloat[]{\includegraphics[page=6,height=.18\textwidth]{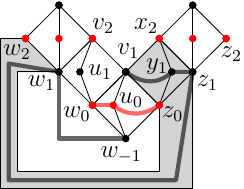}\label{fig:fact-end-vertices-1}}\hfil
		\subfloat[]{\includegraphics[page=8,height=.18\textwidth]{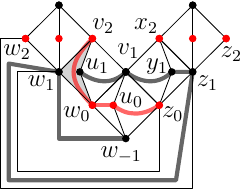}\label{fig:fact-end-vertices-2}}\hfil
		\subfloat[]{\includegraphics[page=5,height=.18\textwidth]{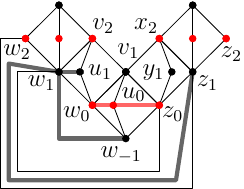}\label{fig:fact-end-vertices-4}}
		\caption{Illustrations for the proof of Properties~(a) and~(b) of \cref{lem:end-vertices-and-saturators}.}
		\label{fig:fact-end-vertices-global-2}
	\end{figure}
		
The vertices $w_0$ and $z_0$ cannot be both red end-vertices of~$\cal C$, as the path $(w_0,u_0,z_0)$ does not span all the red vertices. Suppose, for a contradiction, that $z_0$ is a red end-vertex of~$\cal C$ and that $w_0$ is not. Then $\cal C$ contains neither the red saturating edge $(z_0,x_2)$ nor the red saturating edge $(z_0,w_2)$. Therefore, by \cref{le:quadrangles}, we have that $\cal C$ contains the black saturating edges $(v_1,y_1)$, $(y_1,z_1)$, and $(z_1,w_1)$, which lie in the interior of the $4$-cycles $(z_0,v_1,x_2,y_1)$, $(z_0,y_1,x_2,z_1)$, and $(z_0,z_1,w_2,w_1)$ in~$\cal E'$ (shaded gray in \cref{fig:fact-end-vertices-1}), respectively, as these cycles bound faces of~$\cal E$. Further, since $w_1$ is adjacent to $w_{-1}$ and $z_1$ in $\cal C$, we have that $\cal C$ does not contain the black saturating edge $(w_1,u_1)$. Thus, by \cref{le:quadrangles}, $\cal C$ contains the red saturating edges $(w_0,v_2)$, which lies in the interior of the $4$-cycle $(w_0,w_1,v_2,u_1)$, as this cycle bounds a face of~$\cal E$ (shaded gray in \cref{fig:fact-end-vertices-2}). Then, by \cref{le:quadrangles}, $\cal C$ contains the black saturating edge $(u_1,v_1)$, which lies in the interior of the $4$-cycle $(w_0,u_1,v_2,v_1)$ in~$\cal E'$, as this cycle bounds a face of $\cal E$. Since the only black vertices $u_1$ shares a face of~$\cal E$ with are $w_1$ and $v_1$ and since $\cal C$ does not contain the black saturating edge $(w_1,u_1)$, it follows that $u_1$ is a black end-vertex of~$\cal C$. However, the path $(w_{-1},w_1,z_1,y_1,v_1,u_1)$ of black saturating edges of $\cal C$ does not span all the black vertices, which contradicts the assumption that $\cal C$ is a saturator. This contradiction proves that $w_0$ is a red end-vertex of $\cal C$.
	\item Finally, we show that $\cal C$ contains the red saturating edge $(z_0,w_2)$; the proof that $\cal C$ contains the red saturating edge $(z_{k-1},w_{k+1})$ is analogous; refer to \cref{fig:fact-end-vertices-4}. 
	Suppose, for a contradiction, that $\cal C$ does not contain the red saturating edge $(z_0,w_2)$. Then, by \cref{le:quadrangles}, we have that $\cal C$ contains the black saturating edge $(z_1,w_1)$, which lies in the interior of the $4$-cycle $(z_0,z_1,w_2,w_1)$ in~$\cal E'$, as this cycle bounds a face of $\cal E$. 
	Moreover, since $w_0$ is a red end-vertex of $\cal C$ and since, by \cref{lem:two-red-edges}, $w_0$ is adjacent to $u_0$ in $\cal C$, we have that $\cal C$ does not contain the red saturating edge $(w_0,v_2)$. \cref{le:quadrangles} then implies that $\cal C$ contains the black saturating edge $(w_1,u_1)$, which lies in the interior of the $4$-cycle $(w_0,w_1,v_2,u_1)$ in $\cal E'$, as this cycle bounds a face of $\cal E$. Therefore, we have that $w_1$ is adjacent to $u_1$, $z_1$, and $w_{-1}$ in $\mathcal C$, which contradicts the black saturating edges induce a path in $\cal C$. 
\end{itemize}

It remains to prove Property~(c) of the statement.

\begin{itemize}
	\item First, we show that $\cal C$ contains the saturating edges $(w_i,u_i)$, $(u_i,v_i)$, $(x_i,y_i)$, and $(y_i,z_i)$, for $i=1,2,\dots, k$, where $x_1=v_1$ and $x_k = v_k$. Observe that the only vertices of the color class of $u_i$ which share a face of $\cal E$ with $u_i$ are $w_i$ and $v_i$. Hence, $\cal C$ contains at least one of the edges $(w_i,u_i)$ and $(u_i,v_i)$. Further, if $\cal C$ contains just one of these two edges, then $u_i$ is an end-vertex of $\cal C$ (for its color class), while the end-vertices of $\cal C$ are $w_{-1}$, $w_0$, $z_{k+1}$, and $w_{k+2}$, by Property~(a). The proof that $\cal C$ contains both the edges $(x_i,y_i)$ and $(y_i,z_i)$, for $i=1,2,\dots, k$, is analogous.
	\item Second, we prove that $\cal C$ contains the edge $(z_i,w_{i+2})$, for $i=1,\dots, k$; refer to \cref{fig:proof-induction}. Assume that $(z_{i-1},w_{i+1}) \in E(\cal C)$, for some $i\in \{1,\dots,k\}$; we prove that $(z_{i},w_{i+2}) \in E(\cal C)$ as well. This is enough to prove that $\cal C$ contains all the edges $(z_1,w_{3}),(z_2,w_{4}),\dots,(z_k,w_{k+2})$, given that $(z_{0},w_{2}) \in E(\cal C)$, by Property~(b). Suppose, for a contradiction, that $(z_{i},w_{i+2}) \notin E(\cal C)$. By \cref{le:quadrangles}, we have that $\cal C$ contains the red saturating edge $(z_{i+1},w_{i+1})$, which lies in the interior of the $4$-cycle $(z_i,z_{i+1},w_{i+2},w_{i+1})$ in~$\cal E'$, as this cycle bounds a face of~$\cal E$. However, this implies that $w_{i+1}$ is adjacent to $z_{i+1}$, to $u_{i+1}$ (as proved in the first item of this list), and to $z_{i-1}$ (by assumption) in $\cal C$, a contradiction. 
\end{itemize}

	\begin{figure}[tb!]
		\centering
		\includegraphics[page=18,width=.35\textwidth]{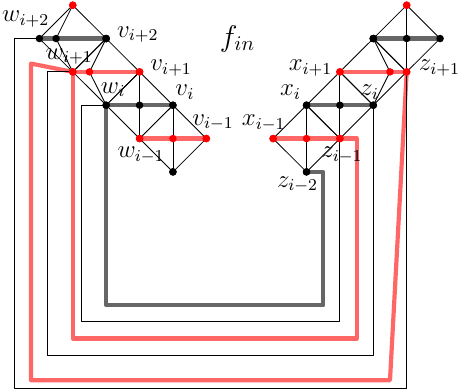}
		\caption{Illustrations for the proof of Property~(c) of \cref{lem:end-vertices-and-saturators}.}
		\label{fig:proof-induction}
	\end{figure}

This completes the proof of Property~(c) and hence of the lemma.
\end{proof}

We now turn our attention to the edges of the saturator that lie in $f_{in}$.

\begin{lemma}\label{lem:saturating-matching}
Any saturator $\cal C$ of the frame $F_k$ contains the saturating edges of the matching $M:=\{(v_i,x_i): 2 \leq i \leq k-1\}$; see \cref{fig:frame-b}. Moreover, all and only the saturating edges of~$M$ lie in the interior of the region delimited by the inner cycle~$c_{in}$ in any planar embedding of~$F_k \cup E(\cal C)$.
\end{lemma}

\begin{fullproof}
Let $\cal E'$ be any planar embedding of $F_k \cup E(\cal C)$. Let $X$ be the set of saturating edges belonging to~$\cal C$ listed in the statements of \cref{lem:two-red-edges,lem:end-vertices-and-saturators}. Namely, the set $X$ contains:
\begin{inparaenum}
\item  the red saturating edges $(w_0,u_0)$, $(u_0,z_0)$, $(w_{k+1},u_{k+1})$, and $(u_{k+1},z_{k+1})$ (from \cref{lem:two-red-edges}),~and 
\item the black saturating edges $(w_{-1},w_1)$ and $(w_{k+2},z_k)$, the red saturating edges $(z_0,w_2)$ and $(w_{k+1},z_{k-1})$, and the saturating edges $(w_i,u_i)$, $(u_i,v_i)$, $(x_i,y_i)$, $(y_i,z_i)$, $(z_i,w_{i+2})$, for $i=1,\dots, k$, where $x_1=v_1$ and $x_k = v_k$ (from \cref{lem:end-vertices-and-saturators}).
\end{inparaenum}

We first observe that, for any planar embedding $\mathcal E_k$ of $F_k \cup X$, there is a face of $\mathcal E_k$ which is delimited by $c_{in}$, while every other face of $\mathcal E_k$ is delimited by a $3$-cycle. This observation descends from the existence of a bijection between the edges of $X$ and the faces of $\cal E$ delimited by $4$-cycles. That is, for every edge $(a,b) \in X$ there is a unique face of $\cal E$ incident to both $a$ and $b$; further, such a face is delimited by a $4$-cycle. Conversely, for every face of $\cal E$ that is delimited by a $4$-cycle, there is a unique edge $(a,b) \in X$ whose end-vertices are both incident to $f$. Therefore, since all the faces of $\cal E$, except for $f_{in}$, are delimited by $4$-cycles, all the faces of $\mathcal E_k$ are delimited by $3$-cycles, except for the one delimited by the inner cycle $c_{in}$, which bounds $f_{in}$ in~$\cal E$.



Let $M$ be a set consisting of the red saturating edges of $\cal C$ and of the black saturating edges of $\cal C$ that are not in $X$. In the following, we show that:
\begin{enumerate}[(A)]
\item the edges in $M$ form a matching that, in any planar embedding of $F_k \cup E(\cal C)$, lies in the region delimited by the inner cycle $c_{in}$, and
\item $M=\{(v_i,x_i): 2 \leq i \leq k-1\}$.
\end{enumerate}

{\noindent \em Proof of (A)}. Since in any planar embedding $\mathcal E_k$ of $F_k \cup X$, there is a face of $\mathcal E_k$ which is delimited by $c_{in}$, while every other face of $\mathcal E_k$ is delimited by a $3$-cycle, it follows that the edges of $M$ lie in the interior of $c_{in}$ in $\cal E'$ and that, therefore, their end-vertices belong~to~$c_{in}$.

Now consider the vertices $v_i$'s and $x_i$'s. By \cref{lem:end-vertices-and-saturators}, we have that: 
\begin{enumerate}
	\item $v_1$ and $v_k$ are incident to two saturating edges in $X$, and
	\item for $i=2,\dots,k-1$, each of $v_i$ and $x_i$ is incident to one saturating edge in $X$.
\end{enumerate}
By \cref{lem:end-vertices-and-saturators}, we have that no vertex $v_i$ and no vertex $x_i$ is a black or red end-vertex of $\mathcal C$. Hence, $v_1$ and $v_k$ have degree $0$ in $M$, while $v_i$ and $x_i$ have degree $1$ in $M$, for $i=2,\dots,k-1$. This implies that the edges in $M$ form a matching.

\smallskip
{\noindent \em Proof of (B)}. Consider the subpaths $P_\ell=(v_2,v_3,\dots,v_{k-1})$ and $P_r=(x_2,x_3,\dots,x_{k-1})$ of~$c_{in}$.

First, we show that there exists no edge $(v_i,v_j)$ in $M$, with $2\leq i<j\leq k-1$; the proof that there exists no edge $(x_i,x_j)$ in $M$, with $2 \leq i<j \leq k-1$, is analogous. 
Suppose, for a contradiction, that $M$ contains an edge $(v_i,v_j)$ with $2 \leq i < j \leq k-1$. Assume, w.l.o.g., that $j-i$ is minimum among all the edges $(v_i,v_j)$ in $M$ with $2 \leq i < j \leq k-1$. Since $v_i$ and $v_j$ both belong to $V_r$ or both belong to $V_b$, and since $F_k$ is bipartite, we have that the subpath of $P_{\ell}$ between $v_i$ and $v_j$ contains at least one internal vertex, namely $v_{i+1}$. By the minimality of $j-i$, the edge in $M$ incident to $v_{i+1}$ is also incident to a vertex of $c_{in}$ that does not belong to the subpath of $P_{\ell}$ between $v_i$ and $v_j$. 
Since all the edges in $M$ lie in the interior of $c_{in}$ in $\cal E'$, such an edge crosses $(v_i,v_j)$, which contradicts the planarity of $\cal E'$.

It follows that every edge in $M$ is of the form $(v_i,x_j)$, with $2 \leq i \leq k-1$ and $2 \leq j \leq k-1$. We show that no edge $(v_i,x_j)$ in $M$ is such that $i \neq j$; this completes the proof that $M=\{(v_i,x_i): 2 \leq i \leq k-1\}$. Indeed, if an edge $(v_i,x_j)$  with $i \neq j$ is in $M$, then the subpath of $c_{in}$ connecting $v_i$ and $x_j$ and passing through $v_k$ contains a different number of internal vertices belonging to $P_{\ell}$ and to $P_r$. Therefore, $M$ contains either an edge crossing $(v_i,x_j)$ in $\cal E'$, contradicting the planarity of $\cal E'$, or an edge $(v_{i'},v_{j'})$ with $2 \leq i' < j' \leq k-1$, or an edge $(x_{i'},x_{j'})$ with $2 \leq i' < j' \leq k-1$, where the last two cases have been ruled out above. 
\end{fullproof}

Altogether \cref{lem:two-red-edges,lem:end-vertices-and-saturators,lem:saturating-matching} imply the following.

\begin{lemma}\label{thm:unique-saturator}
The frame $F_k$ admits a unique saturator $\cal C$ containing the set $M=\{(v_i,x_i): 2 \leq i \leq k-1\}$ of saturating edges. Further, in any planar embedding of $F_k \cup E(\cal C)$, all and only the saturating edges in $M$ lie in the interior of the region delimited by the inner cycle $c_{in}$.
\end{lemma}

We are ready to prove the following main result of this section.

\begin{theorem}\label{th:btpbeP-NPC}
The \btpbeP problem is \NPC.
\end{theorem}

\begin{fullproof}
Let $\langle H, k, v_1, v_k \rangle$ be an instance of {\sc Odd-Leveled Planarity}, where $H$ is a connected bipartite graph, $k$ is an odd integer, and $v_1, v_k \in V(H)$.
We show how to construct in polynomial time a bipartite graph $G=(V_b,V_r,E)$ that admits a \btpbe if and only if $\langle H, k, v_1, v_k \rangle$ admits a leveled planar drawing $\langle k,\gamma,\Sigma\rangle$ such that $\gamma^{-1}(1)=\{v_1\}$ and $\gamma^{-1}(k)=\{v_k\}$. 

We construct the graph $G$ as follows. First, we initialize $G$ to the union of the frame~$F_k$ and of the graph $H$; note that $G$ now contains two copies of the vertex $v_1$ and two copies of the vertex $v_k$ (where a copy of each vertex belongs to $F_k$ and another copy to $H$). Then we identify the two copies of $v_1$ and we identify the two copies of $v_k$.
This concludes the construction of~$G$, which can clearly be performed in polynomial time.

Note that $G$ is a bipartite graph. This is due to the fact that $F_k$ and $H$ are bipartite and that, since $k$ is odd, the vertices $v_1$ and $v_k$ belong to the same part of the bipartition of the vertex set of each of $F_k$ and $H$.
We let $V_b$ consist of all the vertices of $H$ that belong to the same part as $v_1$ (and $v_k$), and we let $V_r$ consist of the remaining vertices of $H$.

We now show the equivalence between $\langle H, k, v_1, v_k \rangle$ and $G=(V_b,V_r,E)$.

($\Longrightarrow$) Suppose first that $\langle H, k, v_1, v_k \rangle$ admits a leveled planar drawing $\langle k,\gamma,\Sigma\rangle$ such that $\gamma^{-1}(1)=\{v_1\}$ and $\gamma^{-1}(k)=\{v_k\}$.
We show that $G$ admits a saturator $\mathcal C_G$, and thus, by \cref{le:characterization-btpbe}, it~admits~a~\btpbe{}.

We construct $\mathcal C_G$ as follows. 
By \cref{thm:unique-saturator}, we have that $F_k$ admits a saturator $\cal C$ containing the edges of the matching $M=\{(v_i,x_i): 2 \leq i \leq k-1\}$.
We initialize $\mathcal C_G$ to the edges of~$\cal C$, except for the edges of $M$.
For $i=2,\dots,k-1$, let $\alpha^i_1,\alpha^i_2,\dots,\alpha^i_{|U_i|}$ be the vertices in the level $i$ of $\langle k,\gamma,\Sigma\rangle$, in the linear order determined by $\xi_i\in \Sigma$; that is, for each $j=1,\dots,|U_i|$, we have $\xi_i(\alpha^i_j)=j$.
We add the path of saturating edges $P_i = (v_i,\alpha^i_1,\alpha^i_2,\dots,\alpha^i_{|U_i|},x_i)$~to~$\mathcal C_G$. 

In the following, we show that $\mathcal C_G$ is a saturator of $G$.

We start by proving that~$\mathcal C_G$ is a cycle that traverses all the vertices of $V_r$ (and all the vertices of $V_b$) consecutively. 
Note that $\mathcal C_G$ is obtained from the saturator~$\mathcal C$ of $F_k$ by subdividing each edge $(v_i,x_i)$ of $M$ with the vertices in $U_i$, for $i=2,\dots,k-1$. This implies that $\mathcal C_G$ is a cycle that spans all the vertices of $G$.
Thus, it remains to show that the vertices $v_i$ and $x_i$ belong to the same part of the bipartition of $V(G)$ as the vertices of $U_i$.
First, we note that the vertices in each set $U_i$ with $i$ odd (resp., with $i$ even) all belong to $V_b$ (resp., to $V_r$), given that $v_1$ belongs to $V_b$, by construction. Second, the vertices $v_i$ and $x_i$ of $F_k$ belong to $V_b$ if $i$ is odd, or to $V_r$ otherwise, by the construction of $F_k$.
%

Next, we prove that $G \cup E(\mathcal C_G)$ is planar. To this aim, we construct a planar drawing $\Gamma$ of it in three steps as follows.

First, we initialize $\Gamma$ to any planar drawing of $F_k \cup E(\mathcal C)$; see \cref{fig:frame-b}.
This is possible since the graph $F_k \cup E(\mathcal C)$ is planar, due to the fact that $\cal C$ is a saturator of $F_k$. By \cref{thm:unique-saturator}, all and only the saturating edges in $M$ lie in the interior of the region delimited by the inner cycle $c_{in}$ in $\Gamma$. In fact, these edges ``split'' the interior of the cycle $c_{in}$ into $k-1$ faces. In particular, for $i=2,\dots,k-2$, the $4$-cycle $c_i=(v_i,v_{i+1},x_{i+1},x_{i})$ bounds a face of~$\Gamma$, and the $3$-cycles $c_1 = (v_1,v_2,x_2)$ and $c_{k-1} = (v_k,x_{k-1},v_{k-1})$ bound two faces of~$\Gamma$.


Second, we extend $\Gamma$ to a  planar drawing of $F_k \cup E(\mathcal C_G)$ as follows.
For $i=2,\dots,k-1$, we replace the drawing of the edge $(v_i,x_i)$ in $\Gamma$ with a drawing of the path $P_i = (v_i,\alpha^i_1,\alpha^i_2,\dots,\alpha^i_{|U_i|},x_i)$, by subdividing the curve representing $(v_i,x_i)$ in $\Gamma$ with the vertices $\alpha^i_1,\alpha^i_2,\dots,\alpha^i_{|U_i|}$, in this order from $v_i$ to $x_i$. Clearly, $\Gamma$ remains planar after this modification. For $i=1,\dots,k-1$, let $c'_i$ be the cycle of $F_k \cup E(\mathcal C_G)$ corresponding to the cycle $c_i$ of $F_k \cup E(\mathcal C)$; that is, $c'_i$ is the cycle obtained from $c_i$ by subdividing its edges as described above.


Third, we extend $\Gamma$ to a  planar drawing of $G \cup E(\mathcal C_G)$ by inserting the edges of $H$ in~$\Gamma$ without introducing any crossings. For $i=1,\dots,k-1$, we draw the edges of~$H$ with an end-vertex in $U_i$ and an end-vertex in $U_{i+1}$ in the interior of the face of $\Gamma$ bounded by the cycle $c'_i$.
Clearly, the edges incident to $v_1$ and the edges incident to $v_k$ can easily be drawn in the interior of~$c'_1$ and of~$c'_{k-1}$, respectively, without introducing any crossings, as these edges are all incident to a common vertex.
For $i=2,\dots,k-2$, the edges in $E(H) \cap (U_i \times U_{i+1})$ can also be drawn without crossings in the interior of $c'_i$, as for any 
two edges~$(\alpha^i_a,\alpha^{i+1}_b)$ and $(\alpha^i_c,\alpha^{i+1}_d)$, we have that $a < c$ if and only if $b < d$, i.e., the end-vertices of $(\alpha^i_a,\alpha^{i+1}_b)$ and $(\alpha^i_c,\alpha^{i+1}_d)$ do not alternate in the cyclic ordering of the vertices of $c'_i$, given that $\langle k,\gamma,\Sigma\rangle$ is a leveled planar drawing of $\langle H, k, v_1, v_k \rangle$. This concludes the proof that $\mathcal C_G$ is a saturator of $G$.


($\Longleftarrow$) Suppose now that $G$ admits a \btpbe{}, and thus, by \cref{le:characterization-btpbe}, it admits a saturator $\mathcal C_G$. We show that $\langle H, k, v_1, v_k \rangle$ admits a leveled planar drawing $\langle k,\gamma,\Sigma\rangle$ such that $\gamma^{-1}(1)=\{v_1\}$ and $\gamma^{-1}(k)=\{v_k\}$.


Consider any planar drawing $\Gamma$ of $G \cup E(\mathcal C_G)$, which exists since $\mathcal C_G$ is a saturator of $G$; see \cref{fig:frame-c}. Recall that $F_k$ admits a unique (up to a flip) planar embedding $\cal E$ and that a face $f_{in}$ of $\cal E$ is bounded by the inner cycle $c_{in}$. Since $G$ is a supergraph of $F_k$, we have that restriction of $\Gamma$ to $F_k$ also contains a face bounded by the inner cycle $c_{in}$, which we again call  $f_{in}$. 
Further, since $H$ is connected, since $v_1,v_k \in V(H)$, and since $v_1$ and $v_k$ only share the face $f_{in}$ of $\cal E$, we have that the 
all the vertices in $V(H) \setminus \{v_1,v_k\}$ lie in the interior of $c_{in}$ in $\Gamma$.

We now prove that $\mathcal C_G$ contains a set of vertex-disjoint paths $P_2,\dots,P_{k-1}$ spanning $V(H)$ where, for $i=2,\dots,k-1$, the path $P_i$ satisfies the following properties: (i) the end-vertices of $P_i$ are $v_i$ and $x_i$, and every internal vertex of $P_i$ is in $V(H)$; and (ii) $P_i$ is either composed of all red saturating edges or of all black saturating edges. 

By definition of saturator and since every edge of $H$ connects two vertices in different color classes of $G$, it follows that $\mathcal C_G$ contains at most two edges of $H$ (in fact, we will argue later that $\mathcal C_G$ contains no edge of $H$). Since $\mathcal C_G$ is a cycle spanning the vertices of $G$, it follows that the vertices of $H$, together with their incident edges in $\mathcal C_G$, define a set of subpaths $Q_1,Q_2,\dots$ of $\mathcal C_G$ satisfying the following properties: (a) the end-vertices of each path $Q_i$ are two vertices of $c_{in}$; (b) the internal vertices of each path $Q_i$ are vertices of $H$; (c) the paths $Q_1,Q_2,\dots$ span $V(H)$; and (d) any two distinct paths $Q_i$ and $Q_j$ do not share any internal vertex. Note that Properties (a)--(d) do not exclude that two distinct paths $Q_i$ and $Q_j$ share an end-vertex (although we will prove later that they never do). 

Since all the vertices in a color class of $G$ appear consecutively in $\mathcal C_G$, it follows that each path $Q_i$ contains at most one edge that is neither a red saturating edge nor a black saturating edge (i.e., an edge of $H$). Hence, by removing from $\Gamma$ all the edges of $H$ that are not in $\mathcal C_G$ and by then ``flattening'' each path $Q_i$ into an edge $q_i$ (i.e., by removing every internal vertex of $Q_i$ from $\Gamma$ and by interpreting the drawing of $Q_i$ in $\Gamma$ as the drawing of an edge $q_i$ between its end-vertices), we obtain a planar drawing of $F_k \cup E(\mathcal  C)$, where $\mathcal  C$ is a saturator of $F_k$. 


By \cref{thm:unique-saturator}, we have that there are $k-2$ edges $q_i$; we rename these edges to $e_2,\dots,e_{k-1}$ where (again by \cref{thm:unique-saturator}) for $i=2,\dots,k-1$, the edge $e_i$ is incident to $v_i$ and $x_i$. For $i=2,\dots,k-1$, we also rename to $P_i$ the path that has been flattened in order to obtain $e_i$. By Properties~(b)--(d) of the paths $Q_1,Q_2,\dots$, now renamed to $P_2,\dots,P_{k-1}$, and since the end-vertices of $P_i$ are $v_i$ and $x_i$, for $i=2,\dots,k-1$, it follows that the paths  $P_2,\dots,P_{k-1}$ are vertex-disjoint and span $V(H)$. Further, the end-vertices of $P_i$ are $v_i$ and $x_i$, and every internal vertex of $P_i$ is in $V(H)$, thus satisfying (i). Finally, since $v_i$ and $x_i$ belong to the same color class of $G$ and since every path $P_i$ contains at most one edge between vertices in different color classes of $G$, it follows that $P_i$ actually contains no edge between vertices in different color classes of $G$. Hence, either every edge in $P_i$ is a red saturating edge or every edge in $P_i$ is a black saturating edge, thus $P_i$ satisfies (ii). Note that the latter implies that $\mathcal C_G$ contains no edge of $H$.

We define $\gamma$ and the sequence of orders $\Sigma=\langle\xi_1,\xi_2,\dots,\xi_k\rangle$ as follows. First, we set $\gamma(v_1)=1$ and $\gamma(v_k)=k$. Then, for each vertex $w \in V(H) \setminus \{v_1,v_k\}$, we set $\gamma(w)=i$ if and only if $w$ belongs to $P_i$. 
We set $\xi_1(v_1)=\xi_k(v_k)=1$. Also, for $i=2,\dots,k-1$, we define the ordering $\xi_i$ of the vertices in $\gamma^{-1}(i)$ according to the order in which such vertices are encountered when traversing $P_i$ from $v_i$ to $x_i$. Since the paths $P_2,\dots,P_{k-1}$ span $V(H)$ and are vertex-disjoint, we have that $\gamma$ and $\Sigma$ are well-defined.

We now show that $\langle k,\gamma,\Sigma\rangle$ is a leveled planar drawing of $ \langle H, k, v_1, v_k \rangle$. Note that $\gamma^{-1}(1)=\{v_1\}$ and $\gamma^{-1}(k)=\{v_k\}$, by construction. 

We prove that \cref{con:i-leveled} of \cref{def:leveled-planar} holds. By Property~(ii) of the path $P_i$ we have that all the vertices in $P_i$, and thus all the vertices in $\gamma^{-1}(i)$, belong to the same color class of $G$, hence no edge $(u,v)$ of $H$ is such that $\gamma(u)=\gamma(v)$. Suppose, for a contradiction, that there exists an edge $(u,v)$ of $H$ such that $\gamma(u)=i$ and $\gamma(v)=i+h$ with $h>1$. Recall that $\Gamma$ is a planar drawing of $G\cup E(\mathcal C_G)$. Then the path $P_{i+1}$ splits the region of $\Gamma$ delimited by $c_{in}$ into two subregions $R_u$ and $R_v$ containing $P_i$ (and thus $u$) and $P_{i+h}$ (and thus $v$) in their interior, respectively. This implies that $(u,v)$ crosses $P_{i+1}$, which contradicts the planarity of $\Gamma$.

\begin{figure}[h!]
\centering
\includegraphics[page=4,height=.3\textwidth]{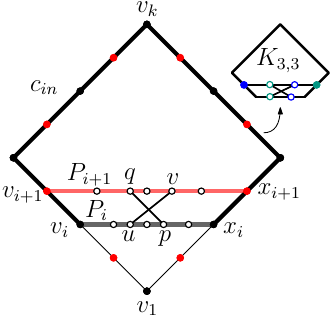}
\caption{Illustration for the proof of \cref{th:btpbeP-NPC}: If there exist two edges $(u,v),(p,q) \in E(H)$ with $i=\gamma(u)=\gamma(p)$ and $i+1=\gamma(v)=\gamma(q)$, such that $\xi_i(u) < \xi_i(p)$ and $\xi_{i+1}(v) > \xi_{i+1}(q)$, then $G \cup E(\mathcal C_G)$ contains a subdivision of $K_{3,3}$.}
\label{fig:k33}
\end{figure}
Suppose, for a contradiction, that \cref{cond:orders} of \cref{def:leveled-planar} does not hold, i.e., there exist two edges $(u,v),(p,q) \in E(H)$ with $i=\gamma(u)=\gamma(p)$ and $i+1=\gamma(v)=\gamma(q)$, such that $\xi_i(u) < \xi_i(p)$ and $\xi_{i+1}(v) > \xi_{i+1}(q)$. Then, the subgraph of $G \cup E(\mathcal C_G)$ consisting of the subpath of $c_{in}$ between $v_i$ and $x_i$ containing $v_k$, of the paths $P_i$ and $P_{i+1}$, and of the edges $(u,v)$ and $(p,q)$ is a subdivision of $K_{3,3}$; refer to \cref{fig:k33}. This contradicts the fact that $G \cup E(\mathcal C_G)$ \mbox{is planar and concludes the proof.}\end{fullproof}

\cref{th:btpbeP-NPC}, together with~\cref{cor:equivalence-variable}, allows us to answer, in \cref{cor:two-two}, an open question by Dujmovi\'c, P{\'{o}}r and Wood~\cite{dpw-tlg-04}, and to provide, in \cref{cor:quasi}, the first \NPCN proof for a natural constrained version of the {\sc Quasi-Planarity} problem.

\begin{corollary}\label{cor:two-two}
	The {\sc $(2,2)$-Track Graph Recognition} problem is \NPC.
\end{corollary}

\begin{corollary}\label{cor:quasi}
	The {\sc 2-Level Quasi-Planarity} problem is \NPC.
\end{corollary}

\section{Graphs with a Fixed Planar Embedding}\label{se:fixed-embedding}

From this section on, we work on the \btpbefP (\btpbef) problem. Recall that the input of the problem is a pair $\langle G, \pi_b\rangle$, where
$G=(V_b,V_r,E)$ is a bipartite planar graph 
and 
$\pi_b=\langle b_1,\dots,b_m\rangle$ is a prescribed linear ordering of the vertices in $V_b$. If we connect with an edge every pair $\{b_i,b_{i+1}\}$ of black vertices, for $i=1,\dots,m-1$, we obtain a path $P$, which we call the \emph{black path} of $\langle G, \pi_b\rangle$.
We denote by $H$ the graph $(V_b,V_r,E \cup E(P))$, which we call the \emph{black saturation} of $\langle G, \pi_b\rangle$. 
The notion of black and red vertices extends to $H$. 
We have the following simple observation.

\begin{observation}\label{obs:gprime-planar}
	If $H$ is not planar, then $\langle G,\pi_b\rangle$ is a negative instance of \btpbef.
\end{observation}

Indeed, if $\langle G,\pi_b\rangle$ admits a \btpbef $\Gamma_G$, then the edges of $P$ can be drawn along the spine of $\Gamma_G$ without affecting its planarity, yielding a planar drawing $\Gamma_H$ of $H$. We say that the planar embedding of $H$ induced by $\Gamma_H$ is \emph{associated with} $\Gamma_G$.

By \cref{obs:gprime-planar}, we can assume that $H$ is planar. Next, we present some additional simplifying assumptions on $H$, all of which can be made without loss of generality.

\begin{enumerate}[({A}1)]
\item \label{A1} 
Every black vertex of $H$ has at least one red neighbor and vice versa. In fact, a black vertex of $H$ with no red neighbor corresponds to an isolated vertex of $G$, hence its presence does not affect the existence of a \btpbef of $\langle G,\pi_b\rangle$, and we can safely remove it from $G$. The same argument holds for the red vertices.

\item  \label{A2} 
The graph $H$ contains at least three black vertices and at least three red vertices.
In fact, if there are at most two black vertices in $G$, then a \btpbef of $\langle G,\pi_b\rangle$ can be constructed by placing all the red vertices on the spine in any order, and by embedding all the edges incident to $b_1$ in one page and all the edges incident to $b_2$, if such a vertex exists, in the other page.
The same argument holds for the red vertices.

%
\end{enumerate}

Notice that Assumption \blue{A\ref{A1}} implies that $H$ is connected.

The aim of this section is to characterize the existence of a \btpbef of $\langle G,\pi_b\rangle$  based on the existence of a planar embedding of $H$ that satisfies certain topological properties. Hence, in the following, we consider $H$ as equipped with a planar embedding $\mathcal E$, and we study which properties of $\cal E$ allow for the construction of a \btpbef $\Gamma_G$ of $\langle G,\pi_b\rangle$ such that $\mathcal E$ is associated with $\Gamma_G$.

A face of $\mathcal E$ is \emph{red} if it is incident to at least two red vertices. We construct an auxiliary bipartite graph $A({\mathcal E})$ as follows; refer to~\cref{fig:a-di-e}. The vertex set of $A({\mathcal E})$ contains a vertex for each red vertex of $H$ and a vertex for each red face of $\cal E$. The edge set of $A({\mathcal E})$ contains an edge $(v,f)$ for each red vertex $v$ of $H$ incident to a red face $f$ of $\cal E$. 
The following characterization is the basis of our linear-time algorithm to test for the existence of a \btpbef of $\langle G,\pi_b\rangle$.

\begin{figure}[t!]
	\centering
	\subfloat[]
	{\includegraphics[width=.4\textwidth,page=5]{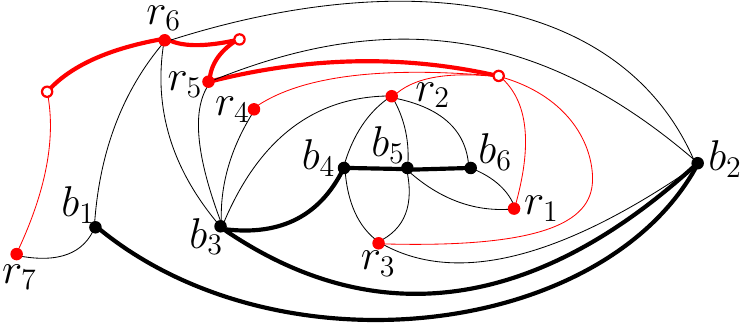}}
	\hfil
	\subfloat[]
	{\includegraphics[width=.4\textwidth,page=6]{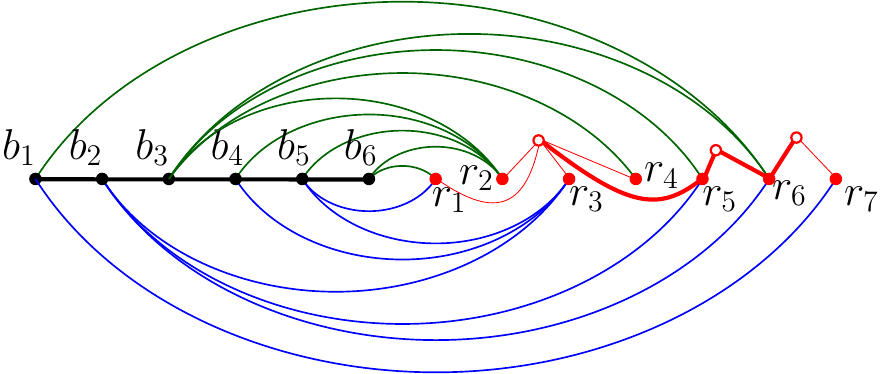}}
	\caption{Two different drawings of the auxiliary bipartite graph $A(\mathcal{E})$ for a planar embedding $\mathcal E$ of the graph in \cref{fig:problems}. Vertices of $A(\mathcal{E})$ corresponding to red faces of $\cal E$ are filled white; edges of $A(\mathcal E)$ are red. The black path $P$ and the backbone of $A(\mathcal{E})$ are thick.}
	\label{fig:a-di-e}
\end{figure}

\begin{lemma} \label{le:characterization-triconnected-onesidefixed}
Let $\langle G,\pi_b\rangle$ be an instance of the \btpbefP problem whose black saturation $H$ is a planar graph satisfying Assumptions \blue{A\ref{A1}} and \blue{A\ref{A2}}. Further, let $\cal E$ be a planar embedding of~$H$. 
There exists a \btpbef $\Gamma_G$ of $\langle G,\pi_b\rangle$
such that the planar embedding of $H$ associated with $\Gamma_G$ is $\cal E$ if and only if:
\begin{enumerate}[(C1)]
	\item\label{condition:caterpillar}  $A({\mathcal E})$ is a caterpillar whose backbone $\mathcal B=(f_1,v_2,f_2,\dots,v_k,f_k)$ spans all the red faces of $\cal E$.
	\item\label{condition:endvertices} There exist two distinct red vertices $r'$ and $r''$ that are leaves of $A({\mathcal E})$, whose neighbors in $A({\mathcal E})$ are $f_1$ and $f_k$, respectively, and such that $r'$ and $b_1$ are incident to the same face of $\cal  E$, and $r''$ and $b_m$ are incident to the same face of $\cal  E$.
\end{enumerate}
\end{lemma}

\begin{fullproof}	
$(\Longrightarrow)$ Assume that $\langle G,\pi_b\rangle$ admits a \btpbef $\Gamma_G$ such that the planar embedding of $H$ associated with $\Gamma_G$ is $\cal E$. We prove that \blue{Conditions C\ref{condition:caterpillar}} and \blue{C\ref{condition:endvertices}} are satisfied. Let $\pi_r=\langle r_1,\dots,r_p\rangle$ be the linear ordering of the red vertices in $\Gamma_G$. 

We start with \blue{Condition C\ref{condition:caterpillar}}. 
\begin{figure}[t!]
	\centering
	\subfloat[\label{fi:FixedEmbedding-NecessityCycles}]
	{\includegraphics[scale=0.8]{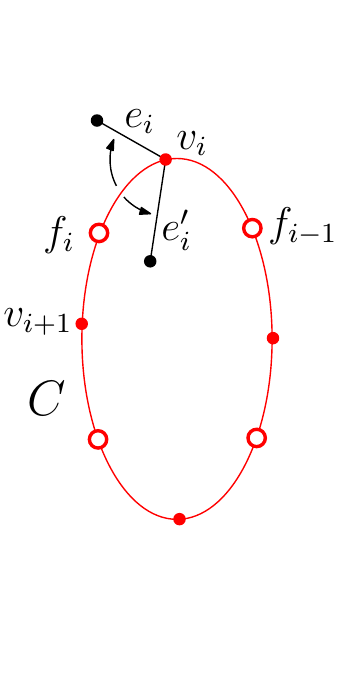}}
	\hfil
	\subfloat[\label{fi:FixedEmbedding-NecessityCaterpillar1}]
	{\includegraphics[scale=0.8]{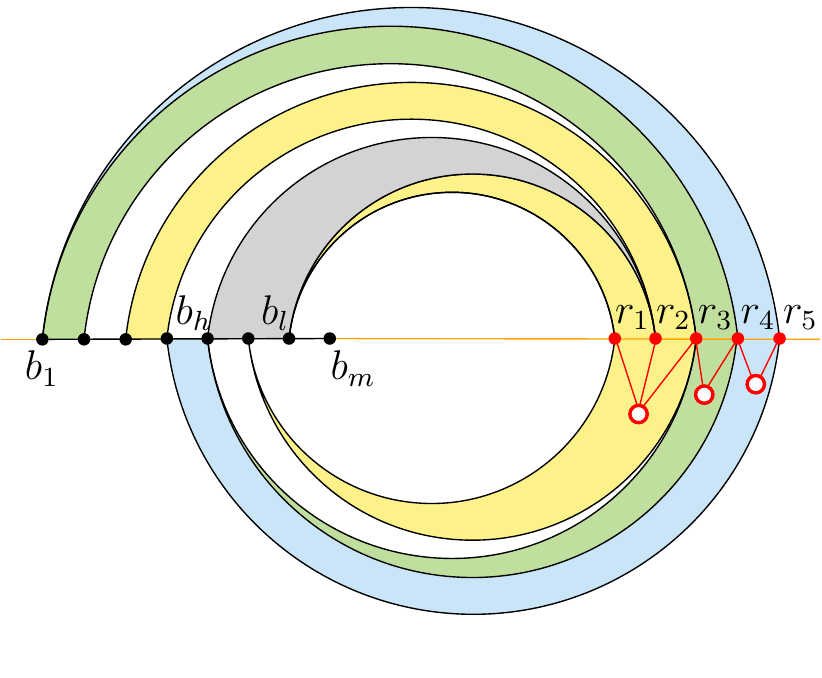}}
	\hfil
	\subfloat[\label{fi:FixedEmbedding-NecessityCaterpillar2}]
	{\includegraphics[scale=0.8]{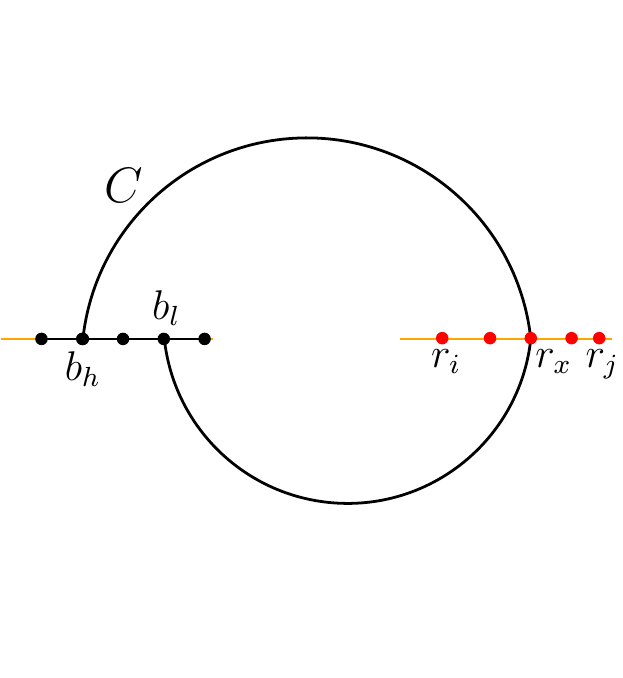}}
	\caption{(a) Illustration for the proof that $A({\mathcal E})$ is acyclic. 
	Illustrations for the proofs of (b) \cref{cl:red-degree} and of (c) \cref{cl:incident-consecutive}, respectively.}
	\label{FixedEmbedding-Necessity}
\end{figure}
First, we prove that $A({\mathcal E})$ is acyclic. Suppose, for a contradiction, that $A({\mathcal E})$ contains a cycle $C$. Since $A(\mathcal{E})$ is bipartite, we have that $C=(v_1,f_1,v_2,f_2,\dots,v_h,f_h)$, where $v_i$ is a red vertex and $f_i$ is a red face, for $i=1,\dots,h$; refer to \cref{fi:FixedEmbedding-NecessityCycles}. 
We planarly embed $C$ into $\cal E$ as follows. 
For $i=1,\dots,h$, we insert in the interior of the face $f_i$ a point representing $f_i$ and two curves representing the edges $(v_i,f_i)$ and $(f_i,v_{i+1})$, where $v_{h+1}=v_1$; this can be done since both $v_i$ and $v_{i+1}$ are incident to $f_i$. 
Now, for any $i \in \{1,\dots,h\}$, consider the two edges $e_i$ and $e'_i$ of $H$ that are incident to $v_i$ and that immediately follow the edge $(v_i,f_i)$ around $v_i$ in clockwise and counter-clockwise direction, respectively. Such edges exist and are distinct, given that $f_i\neq f_{i-1}$, and they lie on different sides of $C$, by construction. Furthermore, the end-vertices of $e_i$ and $e'_i$ different from $v_i$ are black, given that there are no edges between red vertices in $H$. It follows that $C$ has black vertices on both sides; by the Jordan curve theorem, there is a crossing between an edge of the black path $P$ of $H$ and an edge of $C$, which contradicts the fact that $C$ is planarly embedded in $\cal E$.

Second, we prove that $A({\mathcal E})$ is connected. For $i=1,\dots,p-1$, the vertices $r_i$ and $r_{i+1}$ are both incident to a red face $f_j$ in ${\mathcal E}$ (the face that contains the part of the spine between $r_{i}$ and $r_{i+1}$ in $\Gamma_G$), hence $r_{i}$ and $r_{i+1}$ are both connected to $f_j$ in $A({\mathcal E})$. This implies that all the red vertices of $H$ belong to the same connected component of $A({\mathcal E})$; since each red face of ${\mathcal E}$ is adjacent to at least two red vertices in $A({\mathcal E})$, the connectivity of $A({\mathcal E})$ follows. 

We now prove that $A({\mathcal E})$ is a caterpillar. We exploit the following claim; refer to \cref{fi:FixedEmbedding-NecessityCaterpillar1}.

\begin{claim} \label{cl:red-degree}
Every red vertex $r_i$ has degree either $1$ or $2$ in $A(\mathcal E)$. In particular, three cases are possible:
\begin{enumerate}[(i)]
\item\label{degree-i} If $i=1$ or $i=p$, then $r_i$ has degree $1$ in $A(\mathcal E)$ (see $r_1$ and $r_5$ in \cref{fi:FixedEmbedding-NecessityCaterpillar1}). 
\item\label{degree-ii} If $1 < i < p$  and all the edges of $G$ incident to $r_i$ lie in the same page of $\Gamma_G$, then $r_i$ has degree $1$ in $A(\mathcal E)$ (see $r_2$ in \cref{fi:FixedEmbedding-NecessityCaterpillar1}). 
\item\label{degree-iii} If $1 < i < p$ and $r_i$ has incident edges of $G$ lying in both the pages of $\Gamma_G$, then $r_i$ has degree $2$ in $A(\mathcal{E})$ (see $r_3$ and $r_4$ in \cref{fi:FixedEmbedding-NecessityCaterpillar1}).
\end{enumerate}
\end{claim}

\begin{proof} 
By Assumption \blue{A\ref{A2}} and by the connectivity of $A(\mathcal E)$, we have that $r_i$ has degree at least~$1$ in $A(\mathcal E)$.

Consider any red vertex $r_i$ of $H$. There might be two types of faces of $\mathcal E$ incident to $r_i$. A face of $\mathcal E$ incident to $r_i$ is of the \emph{first type} if the intersection of its interior with a sufficiently small disk centered at $r_i$ does not contain any part of the spine of $\Gamma_G$. Conversely, a face of $\mathcal E$ incident to $r_i$ is of the \emph{second type} if the intersection of its interior with any disk centered at $r_i$ contains a part of the spine of $\Gamma_G$. Any face $f$ of the first type is delimited by (a) two edges $e_h=(r_i,b_h)$ and $e_l=(r_i,b_l)$ that are incident consecutively around $r_i$ and that lie in the same page of $\Gamma_G$, and by (b) the subpath of $P$ between $b_h$ and $b_l$ (see the gray face incident to $r_2$ in \cref{fi:FixedEmbedding-NecessityCaterpillar1}). Any face of the first type is not incident to any red vertex of $H$ other than $r_i$, hence it is not red. 

If $1<i<p$ and all the edges of $G$ incident to $r_i$ lie in the same page of $\Gamma_G$, then there exists exactly one face of $\cal E$ incident to $r_i$ of the second type (see the yellow face incident to $r_2$ in \cref{fi:FixedEmbedding-NecessityCaterpillar1}). This proves that $r_i$ has degree $1$ in $A(\mathcal E)$, which implies Case \blue{(\ref{degree-ii})}. 

If $1 < i < p$ and $r_i$ is incident to edges that lie in both the pages of $\Gamma_G$, then there exist two faces of $\cal E$ incident to $r_i$ of the second type (the faces ``to the left'' and ``to the right'' of $r_i$). Both such faces are red; in fact, one of them contains both $r_i$ and $r_{i-1}$, while the other one contains both $r_i$ and $r_{i+1}$ (see, for example, the green and blue faces incident to $r_4$ in \cref{fi:FixedEmbedding-NecessityCaterpillar1}). This proves that $r_i$ has degree $2$ in $A(\mathcal E)$, which implies Case \blue{(\ref{degree-iii})}. 

If all the edges incident to $r_1$ lie in the same page of $\Gamma_G$, then $r_1$ has degree $1$ in $A(\mathcal E)$; this can be proved exactly as in the case in which $1<i<p$. Suppose that $r_1$ is incident to edges that lie in both the pages of $\Gamma_G$. Then one of the faces of the second type incident to $r_1$ also contains $r_2$ (see the yellow face in \cref{fi:FixedEmbedding-NecessityCaterpillar1}), whereas the other one contains no red vertex other than $r_1$, and thus it is not red. Hence $r_1$ has degree $1$ in $A(\mathcal E)$ also in this case. The proof that $r_p$ has degree $1$ in $A(\mathcal E)$ is symmetric. This proves Case~\blue{(\ref{degree-i})} and concludes the proof.
\end{proof} 
	
In order to prove that $A(\mathcal{E})$ is a caterpillar, it remains to prove that every vertex of~$A({\cal E})$ corresponding to a red face of $\mathcal E$ has at most two neighbors of degree greater than $1$. The next claim proves an even stronger property.

\begin{claim} \label{cl:incident-consecutive}
Let $f$ be any red face of $\cal E$. Then the red vertices incident to $f$ 
\begin{enumerate}[(a)]
\item \label{cond:a} are consecutive vertices in $\pi_r$ and 
\item \label{cond:b} have degree~$1$ in $A(\mathcal{E})$ except, possibly, for the leftmost and the rightmost of them in $\pi_r$.
\end{enumerate}
\end{claim}

\begin{proof}
We prove the statement by showing that, for any two vertices $r_i$ and $r_j$ with $i<j$ that are incident to the same red face $f$ of $\cal E$, it holds true that each of $r_{i+1},r_{i+2},\dots,r_{j-1}$ has degree $1$ in $A(\mathcal{E})$ and is incident to $f$. 

Suppose, for a contradiction, that some vertex $r_x$ with $i<x<j$ has degree $2$ in $A(\cal E)$. By Cases~\blue{(\ref{degree-ii})} and~\blue{(\ref{degree-iii})} of \cref{cl:red-degree}, $r_x$ has two incident edges $(r_x,b_h)$ and $(r_x,b_l)$ in different pages of $\Gamma_G$. Consider the cycle $C$ composed of the edges $(r_x,b_h)$ and $(r_x,b_l)$ and of the subpath of $P$ between $b_h$ and $b_l$; see \cref{fi:FixedEmbedding-NecessityCaterpillar2}. The vertices $r_{i}$ and $r_{j}$ are on different sides of $C$ in $\mathcal E$, which implies that they are not incident to a common red face $f$, a contradiction. 

We can now assume that each of $r_{i+1},r_{i+2},\dots,r_{j-1}$ has degree~$1$ in $A(\cal E)$. By Cases~\blue{(\ref{degree-ii})} and~\blue{(\ref{degree-iii})} of \cref{cl:red-degree}, all the edges incident to each of such vertices lie in the same page of $\Gamma_G$. Therefore, $r_i,r_{i+1},r_{i+2},\dots,r_{j-1},r_j$ are all incident to the same red face of $\cal E$.
\end{proof}

Property~\blue{(\ref{cond:b})} of \cref{cl:incident-consecutive} concludes the proof that $A(\cal E)$ is a caterpillar. 
Since every red face of $\cal E$ has degree at least $2$ in $A(\mathcal E)$ by definition, it follows that all the leaves of $A(\mathcal E)$ correspond to red vertices of $H$, hence the backbone $\mathcal B=(f_1,r_2,f_2,\dots,r_k,f_k)$ of $A(\mathcal E)$ spans all the red faces of $\mathcal E$. This proves \blue{Condition C\ref{condition:caterpillar}}. 



In the following, we prove that $r'=r_p$ and $r''=r_1$ satisfy \blue{Condition~C\ref{condition:endvertices}}. By Case \blue{(\ref{degree-i})} of \cref{cl:red-degree}, we have that $r_1$ and $r_p$ are leaves of $A(\mathcal E)$. Further, by construction, $r_p$ and $b_1$ are both incident to the face of $\cal E$ that contains the part of the spine between them in $\Gamma_G$. Similarly, $r_1$ and $b_m$ are both incident to the face of $\cal E$ that contains the part of the spine between them in $\Gamma_G$. It remains to prove that $r_1$ is incident to $f_1$; the proof that $r_p$ is incident to $f_k$ is analogous. 

Consider the red face $f$ of $\cal E$ that is incident to $r_1$. By Property~\blue{(\ref{cond:a})} of \cref{cl:incident-consecutive}, the red vertices incident to $f$ form a prefix of $\pi_r$. Further, by Case \blue{(\ref{degree-i})} of \cref{cl:red-degree} and by Property~\blue{(\ref{cond:b})} of \cref{cl:incident-consecutive}, each of such red vertices, except possibly for the last one, has degree $1$ in $A(\mathcal E)$. It follows that $f$ has at most one neighbor of degree larger than $1$ in $A(\mathcal E)$, hence $f \in \{f_1,f_k\}$. The proof that $r_p$ is incident to one of $f_1$ and $f_k$ is analogous, with the only difference that it exploits a suffix of $\pi_r$ rather than a prefix. This implies that $r_1$ and $r_p$ are incident to different faces, when $f_1 \neq f_k$.


$(\Longleftarrow)$ Given the planar embedding $\mathcal E$ of $H$, suppose that $A({\mathcal E})$ satisfies \blue{Conditions C\ref{condition:caterpillar}} and \blue{C\ref{condition:endvertices}}. In order to prove that $\langle G,\pi_b\rangle$ admits a \btpbef $\Gamma_G$ in which the planar embedding of $H$ associated with $\Gamma_G$ is $\mathcal E$, we show how to add edges to $\mathcal E$ so to augment it to an embedded planar graph $\mathcal{E}_R$ that contains a Hamiltonian cycle passing consecutively through the vertices in $V_b$, in the ordering $\pi_b$, and then passing consecutively through the vertices in $V_r$. Recall that, by \blue{Condition~C\ref{condition:caterpillar}}, we have that $A({\mathcal E})$ is a caterpillar whose backbone $\mathcal B=(f_1,v_2,f_2,\dots,v_k,f_k)$ spans all the red faces of~$\cal E$.

First, we augment $\mathcal E$ to the embedded planar graph $\mathcal{E}^+$ obtained by merging $\mathcal E$ and $A(\cal E)$. In order to do that, we insert each vertex $f_i$ of $\mathcal B$ into the corresponding face of $\mathcal E$; further, we connect the vertex $f_i$ with the red vertices incident to the face $f_i$ by means of non-crossing arcs lying in the interior of the face. 

Second, we show how to augment $\cal E^+$ by introducing the edges $(b_1,r')$ and $(b_m,r'')$, in such a way that the resulting embedded graph (which we still denote by $\cal E^+$) is planar; in some cases, this requires redefining $r'$ and $r''$ while ensuring that \blue{Condition~C\ref{condition:endvertices}} still holds (refer to \cref{fi:add-b1-r1}). 

Suppose first that $b_1$ and $r'$ are both incident to a face $f$ of $\cal  E$ that is not red (see \cref{fi:add-b1-r1-A}). Since $f$ is not red, no edge and no vertex of $A(\cal E)$ has been added inside $f$ during the merge of $\mathcal E$ and $A(\cal E)$; hence, $f$ is also a face of $\cal E^+$ and we can embed the edge $(b_1,r')$ inside $f$ planarly. Analogously, if $b_m$ and $r''$ are both incident to a face $f$ of $\cal  E$ that is not red, then we can embed the edge $(b_m,r'')$ inside $f$ planarly.

Suppose next that $b_1$ and $r'$ are both incident to a red face $f$ of $\cal  E$, and that $b_m$ and $r''$ are both incident to a face different from $f$ (see \cref{fi:add-b1-r1-B}). Note that $r'$ is not incident to any other red face of $\cal  E$, given that it is a leaf of $A(\cal E)$, hence, by \blue{Condition~C\ref{condition:endvertices}}, we have $f=f_1$. The face $f_1$ has been split into at least two faces of $\cal E^+$ by the merge of $\mathcal E$ and $A(\cal E)$. Let $f^{\triangle}_1$ be any of those faces that is incident to $b_1$. We redefine $r'$ as any red vertex that is incident to $f^{\triangle}_1$ and that is a leaf of $A(\cal E)$; such a vertex always exists, since there are two red vertices incident to $f^{\triangle}_1$ and at most one of them is not a leaf of $A(\cal E)$, given that $f_1$ is an end-vertex of $\mathcal B$. We can now embed the edge $(b_1,r')$ inside $f^{\triangle}_1$. Analogously, if $b_m$ and $r''$ are both incident to a red face $f$ of $\cal  E$, and  $b_1$ and $r'$ are both incident to a face different from $f$, then we have $f=f_k$ and  the edge $(b_m,r'')$ can be planarly embedded inside a face $f^{\triangle}_k$ that has been obtained from $f_k$ by the merge of $\mathcal E$ and $A(\cal E)$, possibly after a redefinition of $r''$. 

The only situation that remains to discuss is the one in which $b_1$ and $r'$ are both incident to a red face $f$ of $\cal  E$, and $b_m$ and $r''$ are also incident to $f$. By \blue{Condition~C\ref{condition:endvertices}}, we have that $r'$ and $r''$ are distinct vertices, hence $f$ is red. It follows that $f=f_1=f_k$ is the unique red face of $\mathcal E$, and $b_1$, $b_m$, $r'$ and $r''$ are all incident to it. If $b_1$ and $b_m$ are respectively incident to distinct faces $f^{\triangle}_1$ and $f^{\triangle}_k$ of $\cal E^+$ obtained from $f$ by the merge of $\mathcal E$ and $A(\cal E)$ (see \cref{fi:add-b1-r1-C}), then $r'$ can be redefined as any red vertex incident to $f^{\triangle}_1$ and $r''$ can be redefined as any red vertex incident to $f^{\triangle}_k$ and different from $r'$; this is indeed possible since there are two red vertices incident to $f^{\triangle}_k$; moreover, \blue{Condition~C\ref{condition:endvertices}} still holds with the new choice of $r'$ and $r''$, since every red vertex of $H$ is a leaf of $A(\cal E)$ adjacent to $f_1=f_k$. Now the edges $(b_1,r')$ and $(b_m,r'')$ can be planarly embedded inside $f^{\triangle}_1$ and $f^{\triangle}_k$, respectively. Finally, $b_1$ and $b_m$ cannot be incident to the same face $f^{\triangle}$ of $\cal E^+$ obtained from $f$ by the merge of $\mathcal E$ and $A(\cal E)$. Indeed, if they were, the part of the boundary of $f^{\triangle}$ between them that does not contain the vertex $f_1=f_k$ would be a path whose vertices are all black, hence it would be the entire black path $P$. However, this would imply that the boundary of $f$ contains at least one edge between two red vertices, which is not possible. 

\begin{figure}[tb!]
	\centering
	\subfloat[\label{fi:add-b1-r1-A}]
	{\includegraphics[page=1,scale=.45]{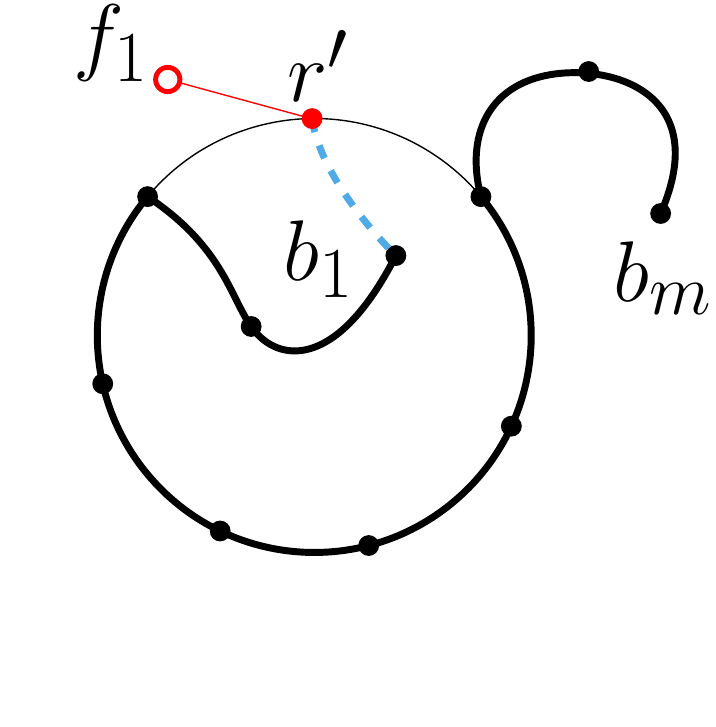}}
	\hfil
	\subfloat[\label{fi:add-b1-r1-B}]
	{\includegraphics[page=2,scale=.45]{add-b1-r1.pdf}}
	\hfil
	\subfloat[\label{fi:add-b1-r1-C}]
	{\includegraphics[page=3,scale=.45]{add-b1-r1.pdf}}
	\caption{Illustration for the insertion of the edge $(b_1,r')$ in $\mathcal E^+$. (a) $r'$ and $b_1$ share a face $f$ that is not red; (b) $r'$ and $b_1$ share a face $f$ that is red, and $b_m$ and $r''$ share a face different from $f$; (c) $b_1$, $b_m$, $r'$ and $r''$ are all incident to the unique red face $f_1=f_k$ of $\mathcal E$.}
	\label{fi:add-b1-r1}
\end{figure}

Third, we remove from $\mathcal{E}^+$ each vertex $f_i$ of $A(\mathcal E)$ corresponding to a red face of $\cal E$ and its incident edges; denote by $\mathcal{E}'$ the resulting embedded planar graph. For each $i=1,\dots,k$, we denote by $f^*_i$ the face of $\mathcal{E}'$ that used to contain the vertex $f_i$. Observe that $f^*_i$ is delimited by the same walk as the face $f_i$ of $\cal E$,  with the possible exceptions of $f^*_1$ and $f^*_k$; indeed, $f^*_1$ and $f^*_k$ might be incident to the edges $(b_1,r')$ and $(b_m,r'')$, respectively, which are not in $\cal E$. A key point is that, even if $(b_1,r')$ is incident to $f^*_1$ and/or $(b_m,r'')$ is incident to $f^*_k$, all the red vertices of $H$ adjacent to $f_1$ in $\mathcal A(\mathcal E)$ are incident to $f^*_1$ and all the red vertices of $H$ adjacent to $f_k$ in $\mathcal A(\mathcal E)$ are incident to $f^*_k$; in fact, the introduction of $A(\mathcal E)$ on top of $\cal E$ before introducing the edges $(b_1,r')$ and $(b_m,r'')$ served this purpose.

\begin{figure}[tb!]
	\centering
	\subfloat[]
	{\includegraphics[page=5,scale=.45]{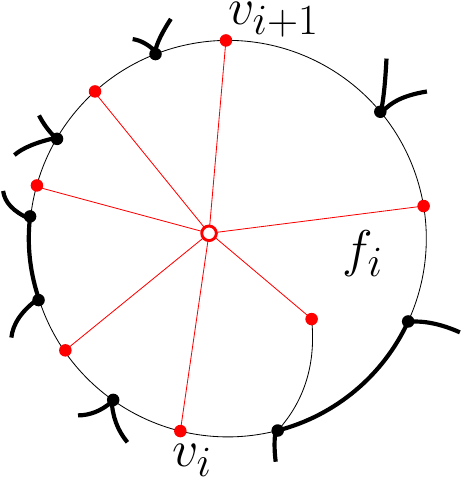}}
	\hfil
	\subfloat[]
	{\includegraphics[page=6,scale=.45]{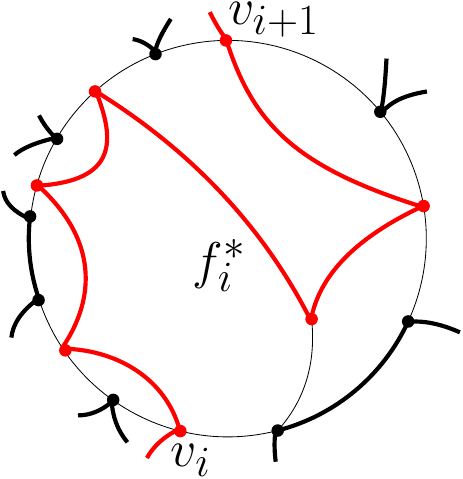}}
	\caption{Illustration for the insertion of the path $R$ in $\mathcal E'$. (a) A face $f_i$ with the vertex $f_i$ of $A(\mathcal E)$ and its incident edges drawn inside the face. (b) The path $R_i$ inserted inside $f^*_i$.}
	\label{fi:add-path}
\end{figure}

Fourth, we show how to add to $\mathcal{E}'$ a path $R$ from $r'$ to $r''$ spanning all the red vertices of $H$, so that the resulting embedded graph $\mathcal{E}_R$ is planar. The augmentation can be performed one face $f^*_i$ at a time; refer to~\cref{fi:add-path}. Let $v_1=r'$ and let $v_{k+1}=r''$. For $i=1,\dots,k$, we embed a path $R_i$ inside $f^*_i$, so that $R_i$ starts at $v_i$, touches every red vertex incident to $f^*_i$, and ends at $v_{i+1}$, and so that $R_i$ does not cross itself. This can be done by first visiting all the internal red vertices incident to one of the two walks between $v_i$ and $v_{i+1}$ delimiting $f^*_i$ and then vising all the internal red vertices incident to the other walk. Let $R$ be the union of the paths $R_i$, for $i= 1,\dots,k$. Since every red vertex is incident to at least one face $f_i$, and thus also to the face $f^*_i$ (this is trivial if $2\leq i\leq k-1$ and comes from the above argument if $i=1$ or $i=k$), and since $A(\mathcal{E})$ is connected, we get that $R$ is a connected graph spanning all the red vertices of $H$. Further, we have that $R$ is a path. Namely, for each red vertex $r$ of $H$, we have that $r$ is either incident to exactly one face $f^*_i$ or to two faces $f^*_i$ and $f^*_{i+1}$. In the first case, $r$ has degree exactly $2$ in $R_i$ (and, thus, in~$R$) with the exception of $r'$ and $r''$, which have degree $1$. In the second case, $r$ has degree $1$ in both $R_i$ and $R_{i+1}$, and thus it has degree exactly $2$ in $R$. Finally, $R$ is embedded planarly in $\mathcal{E}$, since the faces $f^*_i$ are disjoint from one another, except along their boundaries, and since each sub-path $R_i$ does not cross itself.

Let $\cal C$ be the cycle formed by the black path $P=(b_1,b_2,\dots,b_m)$, by the edge $(b_m,r'')$, by the path $R$, and by the edge $(r',b_1)$. We have that $\mathcal E_R$ is an embedded planar graph whose underlying graph is $G \cup E(\mathcal C)$. Further, $\cal C$ traverses all the vertices in $V_r$ consecutively and all the vertices in $V_b$ consecutively in the ordering $\pi_b$; interpreting $\cal C$ as the spine of a book embedding proves that $\langle G, \pi_b \rangle$ admits a \btpbef{}.
\end{fullproof}

In the following we call \emph{good} an embedding of $H$ satisfying the characterization of \cref{le:characterization-triconnected-onesidefixed}.

\begin{lemma} \label{le:book-from-good}
Given a good embedding of the black saturation $H$ of an instance $\langle G,\pi_b\rangle$ of the \btpbefP problem, it is possible to construct in $O(|G|)$ time a \btpbef of $\langle G,\pi_b\rangle$. 
\end{lemma}

\begin{proof}
The proof of sufficiency for the characterization in \cref{le:characterization-triconnected-onesidefixed} is constructive and can be refined into an $O(|G|)$-time algorithm. Namely, given a good embedding $\mathcal E$ of $H$, the caterpillar $A(\mathcal E)$ can be constructed and embedded into $\mathcal E$ in $O(|H|)\subseteq O(|G|)$ time by traversing the boundary of each face of $\mathcal E$, while counting the number of encountered red vertices; this results in an embedded graph $\cal E^+$. The insertion of the edge $(b_1,r')$ into $\cal E^+$ can clearly be done in time linear in the number of vertices of the face $b_1$ and $r'$ share, and likewise for the insertion of the edge $(b_m,r'')$ into $\cal E^+$. The removal of the edges of $A(\mathcal E)$ and of the vertices of $A(\mathcal E)$ corresponding to red faces of $\mathcal E$ is also easily done in linear time; this results in an embedded graph $\cal E'$. Furthermore, the insertion of the path $R$ into $\mathcal E'$ only requires to traverse the boundary of each face $f^*_i$ in order to define the path $R_i$, hence it can be performed in total linear time; this results in an embedded graph $\mathcal{E}_R$. Finally, a \btpbef $\Gamma_G$ of $\langle G,\pi_b\rangle$ can be directly recovered from $\mathcal{E}_R$. Namely, the order of the vertices along the spine of $\Gamma_G$  coincides with the order in which such vertices occur along the Hamiltonian cycle $\mathcal C=P\cup (b_m,r'')\cup R\cup (b_1,r')$ of $\mathcal{E}_R$, where $P$ is the black path of $H$. Further, the page assignment for the edges incident to each vertex can be derived from the clockwise order of the edges incident to the vertex in $\mathcal{E}_R$.
\end{proof}

\section{Simply-Connected Graphs}\label{se:simply}


The goal of this section is to prove that the black saturation $H$ of an instance $\langle G,\pi_b\rangle$ of \btpbef can be assumed to be ``almost'' biconnected; this greatly simplifies the search for a good embedding of $H$. Namely, although it is not always true that $H$ admits a good embedding if and only if its biconnected components admit good embeddings (with respect to the sub-instance of $\langle G,\pi_b\rangle$ they represent), we will prove that the existence of a good embedding of $H$ is equivalent to the existence of a good embedding for each biconnected component of $H$ augmented with some edges, called \emph{$rb$-trivial components}. In the following we make this argument precise. 

The \emph{block-cut-vertex tree} $T$ of a connected graph $G$~\cite{h-gt-69,ht-eagm-73} is the tree whose nodes are the blocks and the cut-vertices of $G$; a block is adjacent in $T$ to all the cut-vertices it contains. 

By Assumption \blue{A\ref{A1}}, we have that $H$ is connected. Consider the blocks of $H$ (see \cref{fig:block-cut-vertex-tree-biconnected}). Let $h_1,\dots,h_p$ be the \emph{$rb$-trivial} components, i.e., the blocks consisting of a single edge between a red and a black vertex, and let $H^-_1,\dots,H^-_q$ be the other \emph{non-$rb$-trivial} blocks. 
Recall that $P=(b_1,\dots,b_m)$. 
We present the following structural results (see \cref{fig:block-cut-vertex-tree-tree}). 

\begin{figure}[tb]\tabcolsep=4pt
	\centering
	\subfloat[\label{fig:block-cut-vertex-tree-biconnected}]{\includegraphics[scale=0.9]{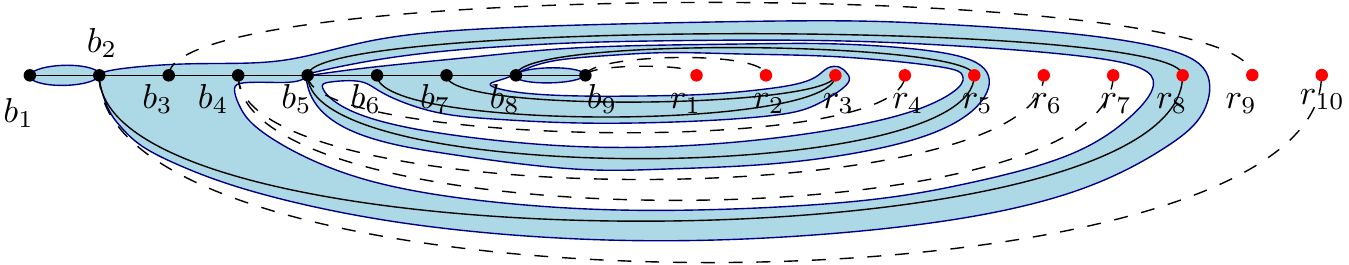}}\hfil
	\subfloat[\label{fig:block-cut-vertex-tree-tree}]{\includegraphics[scale=0.9]{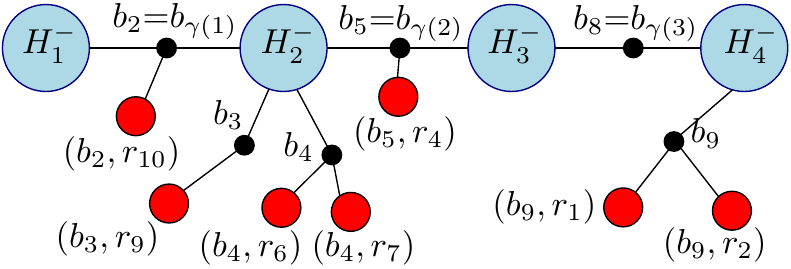}} \hfil
	\subfloat[\label{fig:block-cut-vertex-tree-augmented}]{\includegraphics[scale=0.9]{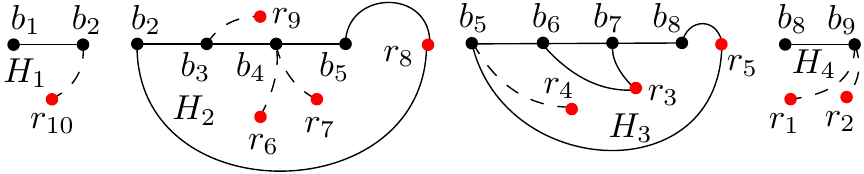}}
	\caption{(a) A planar embedding of $H$ associated with a \btpbef of $\langle G,\pi_b\rangle$. The blue regions enclose the non-$rb$-trivial components of $H$, while the dashed lines represent the $rb$-trivial components of $H$. (b) The block-cut-vertex tree $T$ of $H$. Blue, red, and black disks represent non-$rb$-trivial components, $rb$-trivial components, and cut-vertices of $H$, respectively. (c) The $rb$-augmented components of $H$.}\label{fig:block-cut-vertex-tree}
\end{figure}

\begin{observation}\label{obs:cutvertex-h-black}
Every cut-vertex of $H$ is a black vertex.	
\end{observation}

\begin{proof}
The observation descends from the fact that the black vertices induce a connected subgraph of $H$, namely the black path $P$.
\end{proof}

\begin{lemma}\label{le:bc-tree}
The block-cut-vertex tree $T$ of $H$ consists of:

\begin{itemize}
	\item a path $Q=(H^-_1,b_{\gamma(1)},H^-_2,b_{\gamma(2)},\dots,H^-_{q-1},b_{\gamma(q-1)},H^-_{q})$, where $1<\gamma(1)<\gamma(2)<\dots<\gamma(q-1)<m$, that contains all the non-$rb$-trivial components $H^-_j$ of $H$ and no $rb$-trivial component of $H$; in particular:
	\begin{itemize}
	\item $H^-_1$ contains the subpath of $P$ between $b_1$ and $b_{\gamma(1)}$; 
	\item for $j=2,\dots,q-1$, $H^-_j$ contains the subpath of $P$ between $b_{\gamma(j-1)}$ and $b_{\gamma(j)}$; and
	\item $H^-_q$ contains the subpath of $P$ between $b_{\gamma(q-1)}$ and $b_m$;
	\end{itemize}
	\item a set of leaves representing the $rb$-trivial components $h_1,\dots,h_p$; for each $rb$-trivial component $h_x=(b_i,r_h)$, the black vertex $b_i$ either belongs to a single non-$rb$-trivial component $H^-_j$ (and then $T$ contains a node $b_i$ adjacent to $H^-_j$ and the leaf representing $h_x$ is adjacent to $b_i$), or belongs to two non-$rb$-trivial components $H^-_j$ and $H^-_{j+1}$ (and then the leaf representing $h_x$ in $T$ is adjacent to $b_i=b_{\gamma(j)}$).  
\end{itemize} 
\end{lemma}

\begin{proof}
First note that, if $H$ is biconnected, then its block-cut-vertex tree $T$ consists of a single node, and the statement follows. Assume hence that $H$ is not biconnected. 

Let $b_i$ and $b_h$, with $i<h$, be any two black vertices of $H$ that belong to the same non-$rb$-trivial component $H^-_j$. Suppose, for a contradiction, that some vertices among $b_{i+1}, b_{i+2},\dots,b_{h-1}$ do not belong to $H^-_j$. Consider the subpath $P_{i,h}$ of the black path $P$ between $b_i$ and $b_h$. Then the subgraph of $H$ composed of $H^-_j$ and of the vertices and edges of $P_{i,h}$ that are not in $H^-_j$ is biconnected. However, this contradicts the fact that $H^-_j$ is a maximal biconnected subgraph of $H$ and proves that all the vertices $b_{i+1}, b_{i+2},\dots,b_{h-1}$ also belong to $H^-_j$. 

It follows that there exist subpaths $P_1,\dots,P_q$ of $P$ such that: 

\begin{enumerate}[(i)]
	\item the path $P_1$ is the subpath of $P$ between $b_1$ and a vertex $b_{\gamma(1)}$ with $\gamma(1)>1$;
	\item for $j=2,\dots,q-1$, the path $P_j$ is the subpath of $P$ between $b_{\gamma(j-1)}$ and a vertex $b_{\gamma(j)}$ with $\gamma(j)>\gamma(j-1)$; 
	\item the path $P_q$ is the subpath of $P$ between $b_{\gamma(q-1)}$ and $b_m$, where $m>\gamma(q-1)$;
	\item for $j=1,\dots,q$, the path $P_j$ belongs to a non-$rb$-trivial component $H^-_j$ of $H$; and 
	\item for any two distinct indices $i$ and $j$ in $\{1,\dots,q\}$, the non-$rb$-trivial components $H^-_i$ and $H^-_j$ are distinct.
\end{enumerate}

Note that, for $j=1,\dots,q-1$, the vertex $b_{\gamma(j)}$ is shared by $H^-_j$ and $H^-_{j+1}$, hence it is a cut-vertex. This implies that $T$ contains a path $Q=(H^-_1,b_{\gamma(1)},H^-_2,b_{\gamma(2)},\dots,H^-_{q-1},b_{\gamma(q-1)},H^-_{q})$ such that all the components $H^-_j$ of $H$ are non-$rb$-trivial. 

Consider any block $h_x$ of $H$ that is not among $H^-_1,H^-_2,\dots,H^-_{q}$. We prove that $h_x$ contains exactly one black vertex. Suppose, for a contradiction, that $h_x$ contains (at least) two black vertices $b_i$ and $b_l$; assume, w.l.o.g.\ that $i\leq l$. If $b_i$ and $b_l$ both belong to the same non-$rb$-trivial component $H^-_j$, then the subgraph of $H$ composed of $h_x$ and $H^-_j$ is biconnected, contradicting the fact that $H^-_j$ is a maximal biconnected subgraph of $H$. Otherwise, let $j$ be the largest index such that $b_i$ belongs to $H^-_j$ and let $k$ be the smallest index such that $b_l$ belongs to $H^-_k$; note that $j<k$ given that $i<l$ and given that $b_i$ and $b_l$ do not belong to the same non-$rb$-trivial component of $H$. Then the subgraph of $H$ composed of $h_x$ and of $H^-_j,H^-_{j+1},\dots,H^-_k$ is biconnected, contradicting the fact that $H^-_j$ is a maximal biconnected subgraph of $H$. This proves that $h_x$ contains exactly one black vertex $b_i$; since $H$ contains no edge between red vertices, it follows that $h_x$ is an edge $(b_i,r_h)$, where $r_h$ is a red vertex, i.e., $h_x$ is an $rb$-trivial component. This implies that $Q$ contains all the non-$rb$-trivial components of $H$. Further, by \cref{obs:cutvertex-h-black}, we have that $r_h$ is not a cut-vertex of $H$, and hence $b_i$ is. Finally, we have that $b_i$ belongs either to a single non-$rb$-trivial component $H^-_j$ or to two non-$rb$-trivial components $H^-_j$ and $H^-_{j+1}$ of $H$; in the former case, $T$ contains a node $b_i$ adjacent to $H^-_j$ and the leaf representing $h_x$ is adjacent to $b_i$, while in the latter case the leaf representing $h_x$ in $T$ is adjacent to  $b_i=b_{\gamma(j)}$. This concludes the proof.
\end{proof}


Before proceeding to the decomposition of $H$ into its ``almost'' biconnected components, we prove that the following simplification on the structure of $H$ can be assumed. 

\begin{lemma} \label{le:one-trivial}
For each black vertex $b_i$ of $G$, let $n_i$ be the number of incident $rb$-trivial components; further, if $n_i>0$, let $(b_i,r_{\lambda_i(1)}),\dots,(b_i,r_{\lambda_i(n_i)})$ be the $rb$-trivial components incident to $b_i$. Let $G'$ be the graph obtained from $G$ by removing the red vertices $r_{\lambda_i(2)},\dots,r_{\lambda_i(n_i)}$ together with their incident edges. Then $\langle G,\pi_b \rangle$ and $\langle G',\pi_b \rangle$ are equivalent instances of \btpbef. Moreover, given a \btpbef of $\langle G',\pi_b \rangle$, a \btpbef of $\langle G,\pi_b \rangle$ can be computed in $O(|G|)$ time.
\end{lemma}

\begin{proof}
From a \btpbef of $\langle G,\pi_b \rangle$, a \btpbef of $\langle G',\pi_b \rangle$ is obtained by removing, for each black vertex $b_i$ of $G$ with $n_i>1$, the red vertices $r_{\lambda_i(2)},\dots,r_{\lambda_i(n_i)}$ together with their incident edges. Conversely, from a \btpbef of $\langle G',\pi_b \rangle$, a \btpbef of $\langle G,\pi_b \rangle$ is obtained by placing, for each black vertex $b_i$ of $G$ with $n_i>1$, the red vertices $r_{\lambda_i(2)},\dots,r_{\lambda_i(n_i)}$ immediately to the right of $r_{\lambda_i(1)}$ and by drawing the edges $(b_i,r_{\lambda_i(2)}),\dots,(b_i,r_{\lambda_i(n_i)})$ on the same page as $(b_i,r_{\lambda_i(1)})$; this can easily be done in $O(|G|)$ time.
\end{proof}

In view of \cref{le:one-trivial}, in the rest of the paper we assume that each black vertex of $H$ has at most one incident $rb$-trivial component. 

By \cref{le:bc-tree}, the black vertex $b_i$ of each $rb$-trivial component $h_x=(b_i,r_h)$ either belongs to a single non-$rb$-trivial component $H^-_j$, or belongs to two non-$rb$-trivial components $H^-_j$ and $H^-_{j+1}$. In the former case, we assign $h_x$ to $H^-_j$, while in the latter case we arbitrarily assign $h_x$ to $H^-_j$ or $H^-_{j+1}$. For $j=1,\dots,q$, denote by $H_j$ the subgraph of $H$ which consists of $H^-_j$ and of the $rb$-trivial components of $H$ that have been assigned to $H^-_j$. We call \emph{$rb$-augmented components} of $H$ the subgraphs $H_1,H_2,\dots,H_q$ (refer to \cref{fig:block-cut-vertex-tree-augmented}). For $j=1,\dots,q$, let $G_j$ be the subgraph of $G$ induced by the vertices in $H_j$ (that is, $G_j$ is the graph obtained from $H_j$ by removing the edges that belong to the black path $P$). Finally, for $j=1,\dots,q$, let $\pi^j_b$ be the restriction of $\pi_b$ to the black vertices in $G_j$ (that is, $\pi^1_b=\langle b_1,\dots,b_{\gamma(1)}\rangle$, $\pi^j_b=\langle b_{\gamma(j-1)},\dots,b_{\gamma(j)}\rangle$ for each $j=2,\dots,q-1$, and $\pi^q_b=\langle b_{\gamma(q-1)},\dots,b_m\rangle$). We have the following.

\begin{lemma} \label{le:rb-augmented}
$\langle G,\pi_b\rangle$ is a positive instance of \btpbef if and only if $\langle G_j,\pi^j_b\rangle$ is a positive instance of \btpbef, for every $j=1,\dots,q$.
\end{lemma}

\begin{proof}
One direction is trivial. Namely, if $\langle G,\pi_b\rangle$ admits a \btpbef $\Gamma_G$, then, for every $j=1,\dots,q$, the restriction of $\Gamma_G$ to the vertices and edges of $G_j$ is a \btpbef of $\langle G_j,\pi^j_b\rangle$.

Next, assume that $\langle G_j,\pi^j_b\rangle$ admits a \btpbef $\Gamma_{G_j}$, for every $j=1,\dots,q$; denote by $\pi^j_r$ the order of the red vertices in $\Gamma_{G_j}$.  We construct a \btpbef $\Gamma_G$ of $\langle G,\pi_b\rangle$ as follows. 

First, the ordering of the vertices along the spine of $\Gamma_G$ is $\langle\pi^1_b,\pi^2_b,\dots,\pi^q_b,\pi^q_r,\pi^{q-1}_r,\dots,\pi^1_r\rangle$. This ordering is well-defined once the last vertex of $\pi^j_b$ is identified with the first vertex of $\pi^{j+1}_b$, for $j=1,\dots,q-1$. Indeed, any two graphs $G_j$ and $G_k$ with $k>j$ share vertices if and only if $k=j+1$; moreover, if $k=j+1$ then the only vertex shared by $G_j$ and $G_{j+1}$ is $b_{\gamma(j)}$, which is the last element of $\pi^j_b$ and the first element of $\pi^{j+1}_b$, by construction. 

Second, an edge of $G$ is assigned to the first (second) page of $\Gamma_G$ if and only if it is assigned to the first (second) page of the book embedding $\Gamma_{G_j}$ of the graph $G_j$ it belongs to. This assignment is well-defined, since each edge of $G$ belongs to one of the graphs $G_1,G_2,\dots,G_q$.

By \cref{le:bc-tree} and by construction, the ordering $\langle\pi^1_b,\pi^2_b,\dots,\pi^q_b\rangle$ coincides with $\pi_b$. Hence, we only need to prove that no two edges $e$ and $e'$ in the same page of $\Gamma_G$ cross; we assume that $e$ and $e'$ do not share end-vertices, as otherwise they do not cross. If $e$ and $e'$ belong to the same graph $G_j$, then they do not cross in $\Gamma_G$ as they do not cross in $\Gamma_{G_j}$. Suppose next that $e$ belongs to a graph $G_j$, while $e'$ belongs to a graph $G_k$ with $k>j$. By construction, the order of the end-vertices of $e$ and $e'$ along the spine is: the black end-vertex of $e$ first, then the black end-vertex of $e'$, then the red end-vertex of $e'$, and finally the red end-vertex of $e$; hence, $e$ and $e'$ do not cross. This concludes the proof of the lemma.
\end{proof}

We obtain the following.

\begin{corollary} \label{cor:rb-augmented-good-embeddings}
The black saturation $H$ admits a good embedding if and only if the black saturation $H_j$ of $\langle G_j,\pi^j_b\rangle$ admits a good embedding, for every $j=1,\dots,q$. Further, a good embedding of $H$ can be constructed in $O(|H|)$ time from good embeddings of $H_1,\dots,H_q$.
\end{corollary}

\begin{proof}
By \cref{le:characterization-triconnected-onesidefixed}, we have that $H$ admits a good embedding if and only if $\langle G,\pi_b\rangle$ is a positive instance of \btpbef, and we have that $H_j$ admits a good embedding if and only if $\langle G_j,\pi^j_b\rangle$ is a positive instance of \btpbef, for $j=1,\dots,q$. By \cref{le:rb-augmented}, we have that $\langle G,\pi_b\rangle$ is a positive instance of \btpbef if and only if $\langle G_j,\pi^j_b\rangle$ is a positive instance of \btpbef, for $j=1,\dots,q$. The first part of the statement follows.

By \cref{le:book-from-good}, for $j=1,\dots,q$, a \btpbef $\Gamma_j$ of $\langle G_j,\pi^j_b\rangle$ can be constructed in $O(|G_j|)$ time, and hence in $O(|G|)$ time over all the $rb$-augmented components $H_1,\dots,H_q$. The algorithm described in the proof of sufficiency of \cref{le:rb-augmented} constructs a \btpbef $\Gamma$ of $\langle G,\pi_b\rangle$ starting from the \btpbef $\Gamma_1,\dots,\Gamma_q$ of $H_1,\dots,H_q$ in $O(|G|)$ time. 

By \cref{le:one-trivial}, we can insert in $\Gamma$ in total $O(|G|)$ time all the $rb$-trivial components that have been possibly removed from $H$ because of the existence of other $rb$-trivial components incident to the same black vertices, while maintaining  $\Gamma$ a \btpbef.

Finally, we can draw in $O(|G|)$ time the edges of the black path $P$ along the spine of $\Gamma$. This results in a planar embedding of $H$ which is good, as it is the one associated with $\Gamma$. The second part of the statement follows by observing that $O(|G|)\subseteq O(|H|)$.
\end{proof}

\section{Properties and Classification of rb-augmented Components} \label{se:biconnected}

We now provide a linear-time algorithm that decides whether an $rb$-augmented component $H$ admits a good embedding and, in case it does, constructs one such embedding. This, together with \cref{le:book-from-good,le:rb-augmented}, results in a linear-time testing and embedding algorithm for the  \btpbefP problem.

Let $H^-$ be the biconnected graph obtained from $H$ by removing the degree-$1$ vertices (and their incident $rb$-trivial components). Observe that all these degree-$1$ vertices are red, which implies that $H$ and $H^-$ contain the same set of black vertices. The algorithm is based on a bottom-up traversal of the SPQR-tree $\mathcal{T}$ of~$H^-$, rooted at the Q-node $\rho$ corresponding to the first edge $(b_1,b_2)$ of $P$. 

Consider a node $\mu \in \mathcal{T}$ different from $\rho$. Observe that the poles of $\mu$ are a split pair not only in $H^-$, but also in $H$. We will denote by $H^-_\mu$ the pertinent graph of $\mu$, and by $H_\mu$ the subgraph of $H$ obtained from $H^-_\mu$ by adding each degree-$1$ vertex of $H$ that is adjacent to a vertex of $H^-_\mu$. Again, observe that $H_\mu$ and $H^-_\mu$ contain the same set of black vertices. Thus, when clear from the context, we refer also to $\pert{\mu}$ as the pertinent graph of $\mu$.
Finally, we denote by $\overline{H}_\mu$ (by $\overline{H}_\mu^-$) the subgraph of $H$ (of $H^-$) induced by the edges not in $H_\mu$ (not in $H_\mu^-$).

In the following subsections, we provide several concepts and tools.

In \cref{sse:node-classification}, we present a classification of the nodes $\mu \in \mathcal{T}$ into six \emph{node types}, based on the possible interactions between the black path $P$ and the graph $H_\mu$; refer to \cref{fig:Nodetypes}. Namely, 
$\pert{\mu}$ can either 
be ``touched once'' by $P$, when $P$ shares with $\pert{\mu}$ only one of the poles of $\mu$ and no edges,
or
be ``entered'' by $P$, when $P$ traverses just one of the poles of $\mu$ and ends in the interior of $\pert{\mu}$,
or
be ``touched twice'' by $P$, when $P$ shares with $\pert{\mu}$ both the poles of $\mu$ and no edges,
or
be ``traversed'' by $P$, when $P$ enters from a pole of $\mu$, exists from the other pole, and shares at least one edge with $\pert{\mu}$,
or
be ``bi-entered'' by $P$, when $P$ enters from a pole of $\mu$,
touches the other pole, and ends in the interior of $\pert{\mu}$, 
or
be ``touched twice and entered'' by $P$, when $P$ touches a pole of $\mu$ but contains no edge incident to it, enters from the other pole of $\mu$, and ends in the interior of~$\pert{\mu}$.

In \cref{sse:properties}, we present structural properties of the nodes of $\mathcal{T}$ based on their type, and study the
possible types and arrangements of the nodes having the same parent.

In \cref{sse:neat-embeddings}, we define the concept of {\em neat} embedding, which is a good embedding with additional properties. Informally, a good embedding is neat if
each $rb$-trivial component lies in a face that corresponds to a face of the embedding of the skeleton of the proper allocation node of its black vertex (see \cref{def:neat}). Neat embeddings have three important properties:
\begin{inparaenum}
	\item They are good.
	\item They are not restrictive, in the sense that if an $rb$-augmented component admits a good embedding, then it also admits a neat embedding.
	\item They decrease the degrees of freedom when embedding $rb$-trivial components incident to the poles of a node of $\mathcal T$. 
\end{inparaenum}

In \cref{sse:embedding-classification}, we classify the embeddings of $\pert{\mu}$ that occur in a neat embedding $\mathcal E$ of $H$.
Preliminarily, we introduce two important concepts:
\begin{inparaenum}[(i)]
\item the {\em extensibility} of an embedding and
\item the {\em auxiliary graph} of the embedding of a pertinent graph.
\end{inparaenum}
Let $\mu$ be a node of $\mathcal T$.
An embedding $\mathcal E$ of $H$ {\em extends} an embedding $\mathcal E_\mu$ of $\pert{\mu}$ if the restriction of $\cal E_\mu$ to $\pert{\mu}$ yields $\mathcal E_\mu$. 
An embedding ${\cal E}_\mu$ of $\pert{\mu}$ is {\em extensible} if there is a neat embedding of $H$ that extends ${\cal E}_\mu$.
Finally, the auxiliary graph $A(\mathcal E_\mu)$ is ``essentially'' the restriction of the auxiliary graph $A(\mathcal E)$ to an embedding $\mathcal E_\mu$ of $\pert{\mu}$.
After these preliminary definitions, for each of the six node types introduced in \cref{sse:node-classification}, we classify the extensible embeddings of $\pert{\mu}$ into a constant number of equivalence classes, called {\em embedding types}, based on several features of $A(\mathcal E_\mu)$. The three major features are the number of caterpillars of $A(\mathcal E_\mu)$, the number of outer faces of $\mathcal E_\mu$ that belong to $A(\mathcal E_\mu)$, and the existence of at least one internal face of $\mathcal E_\mu$ in $A(\mathcal E_\mu)$.

Finally, in \cref{sse:handling-rel-emb}, we first define as {\em relevant} the embeddings of $\pert{\mu}$ whose types are those determined in \cref{sse:embedding-classification}, and we argue that all the extensible embeddings are relevant. This proves that our classification of the embedding types is complete. The second part of the subsection is devoted to study the relationship between 
the type of the embedding of the pertinent graph of a node and the type of the embeddings of the pertinent graphs of its children. 

We start by showing that restricting a relevant embedding of $\pert{\mu}$ to $\pert{\lambda}$, where $\lambda$ is any child of $\mu$, yields again a relevant embedding. 
%
Then, we show a result that is fundamental for our algorithmic approach. Let $t$ be the type of a relevant  embedding $\mathcal E_\mu$ of $\pert{\mu}$ and let $\mathcal E_\lambda$ be the embedding of $\pert{\lambda}$ determined by $\mathcal E_\mu$.
Replace $\pert{\lambda}$ in $\mathcal E_\mu$ with an embedded graph $D$, possibly $D \neq \pert{\lambda}$, whose type and ``flip'' are the same of $\mathcal E_\lambda$. Let $\mathcal E^*_\mu$ be the embedding of the resulting embedded graph. We prove that $\mathcal E^*_\mu$ has type $t$.
In order to prove this result, we introduce the key concepts of {\em replacement graph} and {\em embedding-replacement}.
From an algorithmic perspective, this opens the possibility of replacing the embedding of an arbitrarily large subgraph with an embedding of a constant-size graph of the same type without altering the fact that the a given instance is positive or negative.

\subsection{Node Classification}\label{sse:node-classification}

\begin{figure}[tb]\tabcolsep=4pt
	\centering
\subfloat[RE\label{fig:Nodetypes-RE}]{\includegraphics[height=.2\textwidth,page=1]{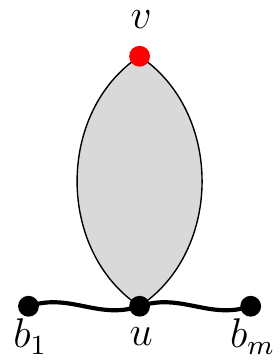}}\hfil
\subfloat[RF\label{fig:Nodetypes-RF}]{\includegraphics[height=.2\textwidth,page=2]{node-types}}\hfil
\subfloat[BE\label{fig:Nodetypes-BE}]{\includegraphics[height=.2\textwidth,page=3]{node-types}}\hfil
\subfloat[BP\label{fig:Nodetypes-BP}]{\includegraphics[height=.2\textwidth,page=4]{node-types}}\hfil
\subfloat[BB\label{fig:Nodetypes-BB}]{\includegraphics[height=.2\textwidth,page=6]{node-types}}\hfil
\subfloat[BF\label{fig:Nodetypes-BF}]{\includegraphics[height=.2\textwidth,page=5]{node-types}}
\caption{Taxonomy of the nodes in the SPQR-tree of $H^-$ based on the position of the black path.}\label{fig:Nodetypes}
\end{figure}

We subdivide the nodes of $\mathcal T$ into six classes. The first parameter for the classification is the color of the two poles of the node. In this respect, we observe the following.

\begin{lemma} \label{le:no-red-red}
There is no node in $\mathcal{T}$ whose poles are both red.
\end{lemma}
\begin{fullproof}
Suppose, for a contradiction, that there exists a node $\mu \in \mathcal{T}$ with two red poles. Since $H$ contains no edge between two red vertices, there is at least a black vertex in $H_\mu$ and at least a black vertex in $\overline{H}_\mu$. However, this contradicts the fact that the graph induced by the black vertices is a path, since any path in $H$ between these two black vertices contains one of the two poles of $\mu$, which are both red.
\end{fullproof}

Let $u,v$ be the poles of $\mu$. By~\cref{le:no-red-red}, we can assume that $u$ is black. We distinguish two cases, based on the color of $v$.

Suppose that $v$ is red.  We distinguish two subcases, based on whether the subpaths of $P$ that are separated by $u$ belong to $\pert{\mu}$ or not. In particular, let $u = b_i$, with $1 \leq i \leq m$. We consider the two subpaths $P_u = (b_1,\dots,b_i)$ and $Q_u=(b_i,\dots,b_m)$ of $P$. 
Observe that, due to the choice of the root of $\mathcal{T}$, the first edge $(b_1,b_2)$ of $P$ belongs to $\rest{\mu}$, and thus $P_u$ entirely belongs to $\rest{\mu}$.
If $Q_u$ belongs to $\rest{\mu}$, as in \cref{fig:Nodetypes-RE}, then we say that $\mu$ is of \textbf{type \atype{}}, where \textbf{R} stands for \emph{red}, meaning that $v$ is red, and \textbf{E} stands for \emph{empty}, meaning that $\mu$ does not contain any non-pole black vertex. 
If $Q_u$ belongs to $\pert{\mu}$, as in \cref{fig:Nodetypes-RF}, then $\mu$ is of \textbf{type \dtype{}}, where \textbf{F} stands for \emph{finishing}, meaning that $P$ ends in $\pert{\mu}$.

Suppose now that $v$ is black. We distinguish four subcases, based on whether the subpaths of $P$ that are separated by $u$ and $v$ belong to $\pert{\mu}$ or not. In particular, let $u = b_i$ and $v = b_j$, with $1 \leq i,j \leq m$; assume without loss of generality that $i < j$. Then $u$ and $v$ split $P$ into three subpaths (possibly composed of single vertices) $P_u = (b_1,\dots,b_i)$, $P_{uv} = (b_i,\dots,b_j)$, and $P_v = (b_j,\dots,b_m)$, each of which entirely belongs to either $\pert{\mu}$ or to $\rest{\mu}$. 
As in the previous case, $P_u$ entirely belongs to $\rest{\mu}$, due to the choice of the root of $\mathcal{T}$.
On the other hand, $P_{uv}$ and $P_v$ can independently belong to either $\pert{\mu}$ or $\rest{\mu}$, which defines the four subcases.

If both $P_{uv}$ and $P_v$ belong to $\rest{\mu}$, as in \cref{fig:Nodetypes-BE}, then $\mu$ is of \textbf{type BE}, where \textbf{B} stands for \emph{black}, meaning that $v$ is black, and \textbf{E} stands for \emph{empty}. 
If $P_{uv}$ belongs to $\pert{\mu}$ and $P_v$ belongs to $\rest{\mu}$, as in \cref{fig:Nodetypes-BP}, then $\mu$ is of \textbf{type BP}, where \textbf{P} stands for \emph{passing}, meaning that $P$ passes through $\pert{\mu}$ from pole to pole.
If both $P_{uv}$ and $P_v$ belong to $\pert{\mu}$, as in \cref{fig:Nodetypes-BB}, then $\mu$ is of \textbf{type BB}, where the second \textbf{B} stands for \emph{both}.
Finally, if $P_{uv}$ belongs to $\rest{\mu}$ and $P_v$ belongs to $\pert{\mu}$, as in \cref{fig:Nodetypes-BF}, then $\mu$ is of \textbf{type BF}, where \textbf{F} stands for \emph{finishing}.


\subsection{Properties of the nodes of $\mathcal T$} \label{sse:properties}

We start with a simple lemma concerning the structure of type~\RE-RE nodes.

\begin{lemma}\label{le:structure-node-RE} 
	Suppose that $\mu$ is of type~\RE-RE. Then, $H^-_\mu$ consists of an edge between the poles of $\mu$, and $H_\mu$ consists of a star centered at $u$ with at most two leaves,~which~are~red. 
\end{lemma}

\begin{fullproof}
	The first part of the statement follows from the observations that the red pole $v$ of $\mu$ is only connected to black vertices in $H$, and that the only black vertex in $\pert{\mu}$ is the pole $u$, by definition of {\bf type RE}. The second part of the statement then follows from the fact that the edge incident to $u$ other than $(u,v)$, if any, is an $rb$-trivial component of $H$.
\end{fullproof}

In the following, we discuss how nodes of different types can appear in the SPQR-tree $\mathcal{T}$ of~$H^-$. We start with a lemma that is a consequence of the fact that $\cal T$ is rooted at the Q-node~$\rho$ corresponding to the first edge $(b_1,b_2)$ of $P$.

\begin{lemma}\label{le:rooting-at-b1b2}
Let $\mu \neq \rho$ be a node of $\cal T$ such that $b_1 \in \pert{\mu}$. Then, $b_1$ is a pole of $\mu$. Also, $\mu$ is of type either \BE-BE, or \BF-BF, or \RE-RE. Finally, if $\mu$ is adjacent to $\rho$, then $\mu$ is of type \BF-BF.
\end{lemma}

\begin{proof}
Since $b_1 \in \pert{\mu}$ and since $(b_1,b_2)$ is an edge of $\rest{\mu}$, due to the rooting of $\mathcal{T}$, we have that $b_1$ is a pole of $\mu$; in particular, $b_1 = u$, by definition. The fact that $\mu$ is of type neither \RF-RF, nor \BP-BP, nor \BB-BB descends from the fact that in any of these types, the pole $u$ has one black neighbor in $\pert{\mu}$, while the only black neighbor of $b_1$ is $b_2$, and $(b_1,b_2)$ is an edge of $\rest{\mu}$.

Suppose that $\mu$ is adjacent to $\rho$. Thus, the poles of $\mu$ are $b_1$ and $b_2$, and hence $\mu$ is of type either \BE-BE or \BF-BF. However, if $\mu$ is of type \BE-BE, then $H$ does not contain black vertices other than $b_1$ and $b_2$. This contradicts Assumption (A2) and proves the statement.
\end{proof}

Next, we describe some properties of the skeletons of the nodes of $\mathcal T$. Consider any node $\mu\neq \rho$ of $\mathcal T$. In order to simplify the discussion, we extend the notion of {\em type} to the virtual edges of $\skel(\mu)$. Namely, any virtual edge $e$ in $\skel(\mu)$ corresponds to a child $\nu$ of $\mu$; then, the type of $e$ is the type of $\nu$. We have the following.

\begin{lemma} \label{le:final-edges}
	There is at most one virtual edge of $\skel(\mu)$ that is of type \RF-RF, \BF-BF, or \BB-BB. 
	\remove{Furthermore, if such a virtual edge exists, then it is incident to the end-vertex $y$ of~$\mathcal P(\mu)$.}
\end{lemma}

\begin{proof}
	The statement follows from the fact that each virtual edge of type \RF-RF, \BF-BF, or \BB-BB, contains $b_m$ as a non-pole vertex, by definition. 
	\remove{
		For the second part, suppose for a contradiction that there is a virtual edge $e$ whose type is \RF-RF, \BF-BF, or \BB-BB that is not incident to $y$. 
		
		If $e$ is not incident to any vertex of $\mathcal P(\mu)$, then there are two sub-paths of the black path $P$ in $\pert{\mu}$, each with at least one edge, which implies that both $b_1$ and $b_m$ are non-pole vertices of $\pert{\mu}$, contradicting the rooting of $\cal T$. Similarly, if $e$ is of type~\BF-BF and its pole that has no black neighbor in $\pert{e}$ is an internal vertex of $\mathcal P(\mu)$, then $\pert{\mu}$ contains two sub-paths of $P$, each with at least one edge, and a contradiction is derived as above.
		
		If $e$ is of type \RF-RF and it is incident to an internal vertex of $\mathcal P(\mu)$, then there is a black vertex with at least three black neighbors, a contradiction. Similarly, if $e$ is of type \BF-BF and its pole that has one black neighbor in $\pert{e}$ is an internal vertex of $\mathcal P(\mu)$, then there is a vertex with at least three black neighbors, a contradiction. 
		
		Finally, if $e$ is of type \BB-BB and is not the last edge of $\mathcal P(\mu)$, then its pole that has two black neighbors in $\pert{e}$ is also adjacent to a third black neighbor in the pertinent graph of the other virtual edge in $\mathcal P(\mu)$ incident to it, a contradiction.}
\end{proof}

\begin{lemma} \label{le:type-C-path}
	All the virtual edges of type \BP-BP or \BB-BB form a path 
	in $\skel(\mu)$ that: 
	(i) starts at a pole of $\mu$; 
	(ii) ends at a vertex $y$ of $\skel(\mu)$; and 
	(iii) contains all the black vertices of $\skel(\mu)$, except, possibly, for the other pole of $\mu$.
\end{lemma}

\begin{proof}
	Consider the subgraph $\skel_b(\mu)$ of $\skel(\mu)$ composed of the black vertices and of the type \BP-BP or \BB-BB virtual edges of $\skel(\mu)$. First, we have that every vertex of $\skel_b(\mu)$ has degree at most $2$; indeed, a vertex of $\skel_b(\mu)$ with three incident edges is a vertex of $H$ with at least three black neighbors, which is not possible. Second, we have that $\skel_b(\mu)$ contains no cycle, as such a cycle would correspond to a cycle of black vertices in $H$, which is not possible. It follows that $\skel_b(\mu)$ is a set of paths. Finally, if $\skel_b(\mu)$ has at least two connected components both of which are different from a single pole of $\mu$, then $\pert{\mu}$ would contain as non-pole vertices both the end-vertices $b_1$ and $b_m$ of the black path $P$, which is impossible because of the rooting of $\mathcal T$. 
\end{proof}

\subsection{Neat Embeddings}\label{sse:neat-embeddings}

We start with the following definition.

\begin{definition}\label{def:neat}
A good embedding $\mathcal E$ of $H$ is \emph{neat} if, for every $rb$-trivial component $(r,b)$ with $b$ in $H^-$, it  satisfies the following property: 
Let $\mu$ be the proper allocation node of $b$ in $\cal T$; then, the face of~$\mathcal E$ vertex $r$ is incident to corresponds to a face of the embedding of $\skel(\mu)$ determined by $\mathcal E$.
\end{definition}

We show in \cref{le:neat-embeddings} that $H$ admits a good embedding if and only if it admits a neat embedding. Thus, in the remainder of the section, we will focus our attention on~neat~embeddings.

%

\begin{lemma}\label{le:neat-embeddings}
The graph $H$ admits a good embedding if and only if it admits a neat embedding.
\end{lemma}

\begin{proof}
The necessity is trivial. In the following, we prove the sufficiency. Let $\cal E^*$ be a good embedding of $H$. 
Suppose that there exists (otherwise, there is nothing to prove) an $rb$-trivial component $(r,b)$, with $b$ in $H^-$, 
such that the face $f_r$ of $\cal E^*$ vertex $r$ is incident to does not correspond to any face of the embedding of $\skel(\mu)$ determined by $\cal E^*$, where $\mu$ is the proper allocation node of $b$ in $\cal T$. 
This implies that $f_r$ is an internal face of the embedding $\mathcal E_e$ of $\pert{e}$ determined by $\cal E^*$, where $e=(u_e=b,v_e)$ is a virtual edge of $\skel(\mu)$ incident to $b$. 

We show how to obtain a good embedding $\mathcal E'$ of $H$ in which the vertex $r$ is incident to a face of $\mathcal E'$ that corresponds to a face of the embedding of $\skel(\mu)$ determined by $\mathcal E'$.
We obtain $\mathcal E'$ from~$\mathcal E^*$, by placing $r$ inside a different face, while maintaining the rest of $\mathcal E^*$ unchanged (except, possibly, for the routing of a single edge belonging to $\pert{e}$). 
First, we show {\bf (i)} that removing $r$ and $(b,r)$ from $f_r$ yields a good embedding $\mathcal E^\circ$ of the resulting instance; see, e.g., \cref{fig:neat-2-a,fig:neat-2-b}.
Second, we show {\bf (ii)} how to reinsert $r$ and $(b,r)$ into $\mathcal E^\circ$ to obtain $\mathcal E'$, after possibly rerouting a single edge belonging~to~$\pert{e}$; see, e.g.,  \cref{fig:neat-2-c}.

Consider the auxiliary graph $A({\mathcal E^*})$ of $\mathcal E^*$. By \cref{le:characterization-triconnected-onesidefixed}, we have that $A({\mathcal E^*})$ is a caterpillar whose backbone $\mathcal B=(f_1,v_2,f_2,\dots,v_k,f_k)$ spans all the red faces of $\mathcal E^*$. Furthermore, there exist two distinct red vertices $r'$ and $r''$ that are leaves of $A({\mathcal E^*})$, whose neighbors in $A({\mathcal E^*})$ are $f_1$ and $f_k$, respectively, and such that $r'$ and $b_1$ are incident to the same face of $\mathcal  E^*$, and $r''$ and $b_m$ are incident to the same face of $\mathcal  E^*$.

\begin{figure}[t!]
	\centering
	\subfloat[$\mathcal E^*$]{
		\includegraphics[height=.3\textwidth,page=4]{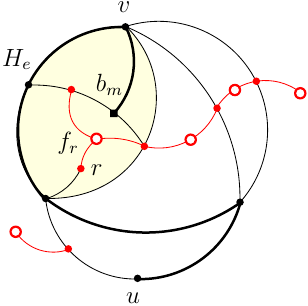} 
		\label{fig:neat-2-a}
	}\hfil
	\subfloat[$\mathcal E^\circ$]{
		\includegraphics[height=.3\textwidth,page=5]{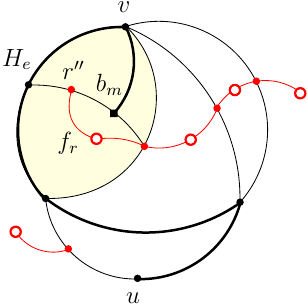} 
		\label{fig:neat-2-b}
	}\hfil
	\subfloat[$\mathcal E$]{
		\includegraphics[height=.3\textwidth,page=6]{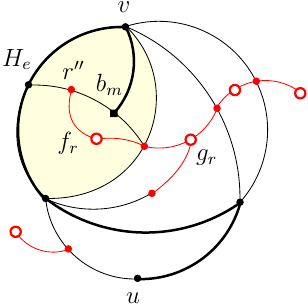} 
		\label{fig:neat-2-c}
	}
	\caption{Illustration for the proof of \cref{le:neat-embeddings} when $f_r$ is also red in $\mathcal E^\circ$ after the removal~of~$r$.}
	\label{fig:neat-2}
\end{figure}

We prove {\bf (i)}. Since $A({\mathcal E^*})$ is connected and since $r$ is only incident to $f_r$, we have that $f_r$ is red and that $r$ is a leaf of $A({\mathcal E^*})$ adjacent to $f_r$. Therefore, $A({\mathcal E^*})$ remains a caterpillar after the removal of $r$. 
Suppose first that $f_r$ remains red after the removal of $r$; refer to \cref{fig:neat-2-a,fig:neat-2-b}. Then
$A({\mathcal E^\circ})$ coincides with $A({\mathcal E^*}) \setminus r$, which implies that \cref{condition:caterpillar} of \cref{le:characterization-triconnected-onesidefixed} is satisfied.
If $r \notin \{r',r''\}$, then $A({\mathcal E^\circ})$ satisfies \cref{condition:endvertices} of \cref{le:characterization-triconnected-onesidefixed}.
If $r = r''$ (the case in which $r=r'$ is analogous), then~$f_r$ is an end-vertex of $A({\mathcal E^\circ})$ and, since $r$ is only incident to~$f_r$, it holds that $b_m$ is incident to~$f_r$. Observe that, since~$f_r$ remains red after the removal of~$r$, there exist at least two other red vertices different from $r$ incident to~$f_r$; further, all but at most one of such vertices are leaves of $A({\mathcal E^\circ})$ (note that, if the backbone of $A({\mathcal E^\circ})$ is not the single vertex $f_r$, then the neighbor of $f_r$ in the backbone of $A({\mathcal E^\circ})$ is the only red vertex incident to~$f_r$ that is not a leaf of $A({\mathcal E^\circ})$). Therefore, at least one red vertex different from $r$ and incident to~$f_r$ is a leaf of $A({\mathcal E^\circ})$; further, such a leaf shares the face $f_r$ with $b_m$, given that $b_m$ is incident to $f_r$, and thus it can be selected to play the role of $r''$ to satisfy \cref{condition:endvertices} of \cref{le:characterization-triconnected-onesidefixed}. 
If~$f_r$ is not red after the removal of $r$; refer to \cref{fig:neat-a,fig:neat-b}. Then there exists exactly one red vertex $r^\circ$ different from $r$ incident to~$f_r$. Since $r$ is a leaf of $A({\mathcal E^*})$, we have that~$f_r$ is an end-vertex of the backbone of $A({\mathcal E^*})$, that $r \in \{r',r''\}$, say $r=r''$, and that $b_m$ is incident to~$f_r$.
Thus, we have that $A({\mathcal E^\circ})$ coincides with $A({\mathcal E^*}) \setminus \{f_r,r\}$, and that it is a caterpillar, which 
satisfies \cref{condition:caterpillar} of \cref{le:characterization-triconnected-onesidefixed}. 
Further, the other neighbor of $r^\circ$ in $A({\mathcal E^*})$ is an end-vertex of the backbone of $A({\mathcal E^\circ})$.
Therefore, setting $r''= r^\circ$ satisfies \cref{condition:endvertices} of \cref{le:characterization-triconnected-onesidefixed}. 

\begin{figure}[h]
	\centering
	\subfloat[$\mathcal E^*$]{
		\includegraphics[height=.3\textwidth,page=1]{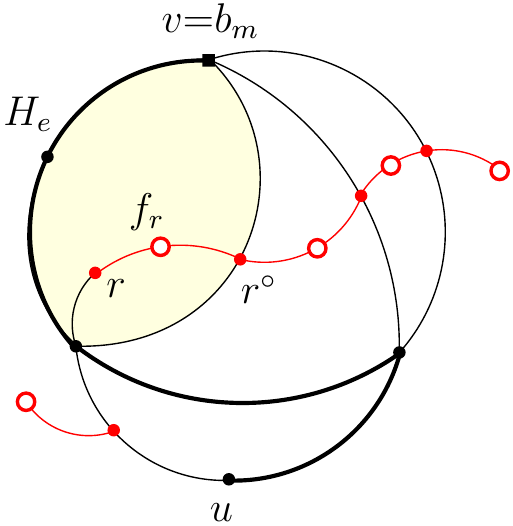} 
		\label{fig:neat-a}
	}\hfil
	\subfloat[$\mathcal E^\circ$]{
		\includegraphics[height=.3\textwidth,page=2]{neat} 
		\label{fig:neat-b}
	}\hfil
	\subfloat[$\mathcal E'$]{
		\includegraphics[height=.3\textwidth,page=3]{neat} 
		\label{fig:neat-c}
	}
	\caption{Illustration for the proof of \cref{le:neat-embeddings} when $f_r$ is not red in $\mathcal E^\circ$ after the removal of $r$.}
	\label{fig:neat}
\end{figure}

We prove {\bf (ii)}. 
Since $r$ is only incident to~$f_r$ in $\mathcal E^*$, we have that each face of $\mathcal E^*$ different from~$f_r$ is also a face of $\mathcal E^\circ$; in particular, 
each red face of $\mathcal E^*$ different from~$f_r$ is also a red face of $\mathcal E^\circ$.
Consider $\ell(\mathcal E_e)$ and $r(\mathcal E_e)$, and let $g_\ell$ and $g_r$ be the two faces of $\mathcal E^*$ corresponding to $\ell(\mathcal E_e)$ and $r(\mathcal E_e)$, respectively.
First, if one of $g_\ell$ and $g_r$ is red in $\mathcal E^*$ (and, thus, in $\mathcal E^\circ$, given that $f_r \neq g_\ell,g_r$), placing $r$ and $(b,r)$ inside such a face yields a good embedding of $H$; refer to \cref{fig:neat-2-c,fig:neat-c}. 
Thus, in the following we will assume that none of $g_\ell$ and $g_r$ is red. 

To continue the proof, we distinguish several cases based on the structure of $A(\mathcal E^*)$.

Suppose first that $H \setminus \pert{e}$ contains no red vertex.
We distinguish two cases based on whether~$b_1 \in \pert{e}$ or not.
Assume first that $b_1$ belongs to $\pert{e}$. By \cref{le:rooting-at-b1b2}, we have that $b_1$ is a pole of $e$ (possibly, $b_1 = b$), and $e$ is of type either \BE-BE, or \BF-BF, or \RE-RE. By \cref{le:structure-node-RE} and since~$f_r$ is an internal red face of $\mathcal E_e$, we have that $e$ is not of type \RE-RE.
By definition of type \BE-BE and \BF-BF nodes, we have that all the neighbors of $b_1$ in $\pert{e}$ are red. Let $u_\ell$ and $u_r$ be the neighbors of $b_1$ incident to $g_\ell$ and $g_r$, respectively.
We have that the internal faces of $\mathcal E_e$ incident to $b_1$ form a subpath of the backbone of $A(\mathcal E^\circ)$ connecting $u_\ell$ and $u_r$. No red vertex different from  $u_\ell$ and $u_r$ is incident to $g_\ell$ and $g_r$, as otherwise one of $g_\ell$ and $g_r$ would be red. Therefore, since $\mathcal E^\circ$ is a good embedding, we have that either $r'=u_\ell$ or $r'=u_r$, say $r'=u_\ell$, that is, $u_\ell$ is a leaf of $A(\mathcal E^\circ)$ adjacent to an end-vertex of the backbone of $A(\mathcal E^\circ)$ sharing a face with $b_1$. Therefore, by placing $r$ and $(b,r)$ inside $g_\ell$, we obtain the embedding $\mathcal E'$ of $H$, whose auxiliary graph is a caterpillar obtained from $A(\mathcal E^\circ)$ by adding the edges $(u_\ell,g_\ell)$ and $(g_\ell,r)$. This implies that \cref{condition:caterpillar} of \cref{le:characterization-triconnected-onesidefixed} is satisfied.
Further, by setting $r'=r$, we have that \cref{condition:endvertices} of \cref{le:characterization-triconnected-onesidefixed} is satisfied, since $b_1$ is incident to $g_\ell$.
Assume now that $b_1$ does not belong to $\pert{e}$. Since $H \setminus \pert{e}$ contains no red vertex and since $b_1$ has at least one red neighbor in $H^-$, we have that $b_1$ is adjacent to a red pole $v_e$ of $e$. Moreover, since $g_\ell$ and $g_r$ do not contain any other red vertex, as otherwise they would be red, we have that $v_e$ is a leaf of $A(\mathcal E^\circ)$ adjacent to an end-vertex of its backbone, given that $\mathcal E^\circ$ satisfies \cref{condition:caterpillar} of \cref{le:characterization-triconnected-onesidefixed}. By placing $r$ and $(b,r)$ inside one of the two outer faces of $\mathcal E_e$, we obtain the embedding $\mathcal E'$ of $H$ whose auxiliary graph satisfies the two conditions~of~\cref{le:characterization-triconnected-onesidefixed}.

Finally, consider the case that $H \setminus \pert{e}$ contains at least one red vertex.
Consider the path of $A(\mathcal E^*)$ connecting $r$ and a red vertex of $H \setminus \pert{e}$. Since none of $g_\ell$ and $g_r$ is red, we have that such a path must include a pole of $e$. Therefore, we can assume that the pole $v_e$ of $e$ different from $b$ is red, which implies that $e$ is of type \RF-RF, since $e$ is not of type \RE-RE as discussed above; refer to \cref{fig:neat2-a}. Further,~$v_e$ belongs to the backbone of $A(\mathcal E^*)$, and it is incident to a red face $f_{in}$ internal to $\mathcal E_e$ and to a red face $f_{out}$ not belonging to $\mathcal E_e$.
First, we claim that $\pert{e}$ contains the edge $(u_e,v_e)$.
Note that, since $b=u_e$ is incident to~$f_r$, there exists at least an internal face of $\mathcal E_e$ that is incident to $u_e$. This implies that there exist two distinct neighbors $u_\ell$ and $u_r$ of $u_e$ incident to~$g_\ell$ and $g_r$, respectively.
Since $u_e$ has exactly one black neighbor belonging to $\pert{e}$, due to the fact that $e$ is of type \RF-RF, at least one of $u_\ell$ and $u_r$ is red, say $u_r$. We have that $u_r$ must coincide with~$v_e$, as otherwise $g_r$ would be red, which proves the above claim. 
By the discussion above and since $H$ does not contain parallel edges, we have that $u_\ell$ is black. 
Let $f$ be the internal face of~$\mathcal E_e$ incident to $(u_e,v_e)$.
Note that $f$ is red in $\mathcal E^*$. In fact, the neighbor of $u_e$ preceding $v_e$ in $\mathcal E_e$ in clockwise order around $u_e$ is incident to $f$ and is red; the latter descends from the fact that this vertex is different from $u_\ell$, since $r$ appears between $u_\ell$ and $u_r$ in the clockwise order of the neighbors of $u_e$ in $\mathcal E_e$.
Since $f_{in}$ is the unique red face of $\mathcal E_e$ incident to $v_e$ and since $f$ is incident to $v_e$, we have $f=f_{in}$.
Suppose first that $f$ is still red in $\mathcal E^\circ$. 
\begin{figure}[t!]
	\centering
	\subfloat[$\mathcal E^\circ$]{
		\includegraphics[height=.3\textwidth,page=1]{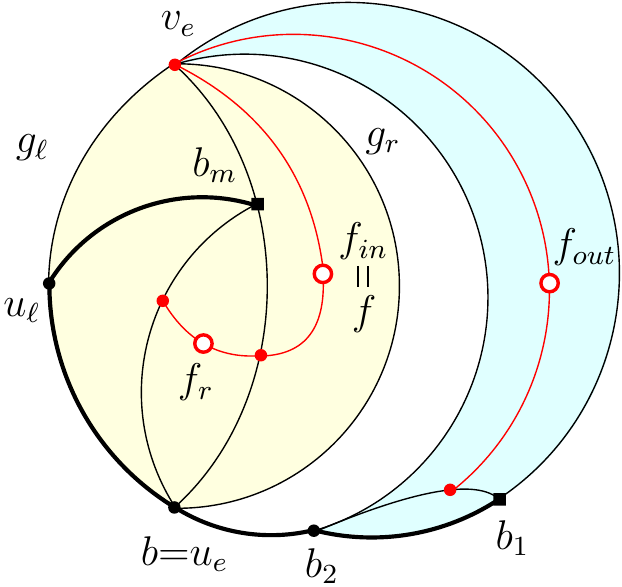} 
		\label{fig:neat2-a}
	}\hfil
	\subfloat[$\mathcal E'$]{
		\includegraphics[height=.3\textwidth,page=2]{neat2} 
		\label{fig:neat2-b}
	}
	\caption{Illustration for the proof of \cref{le:neat-embeddings} when $H \setminus \pert{e}$ contains at least one red vertex, none of $g_\ell$ and $g_r$ is red, and $f_r$ is red in $\mathcal E^\circ$.}
	\label{fig:neat2}
\end{figure}
We obtain $\mathcal E'$ as follows; refer to \cref{fig:neat2-b}. First, we reroute the edge $(u_e,v_e)$ in the interior of $g_\ell$, that is, we let $(u_e,v_e)$ be the edge immediately preceding $(u_e,u_\ell)$ clockwise around $u_e$. This merges $f$ with $g_r$ into a new red face $g'_r$ containing the same red vertices as $f$ in $\mathcal E^\circ$, which is now the right outer face of the resulting embedding of $\pert{e}$.
Also, it splits $g_\ell$, which was not red in $\mathcal E^\circ$, into two faces that are not red. This implies that the auxiliary graph of the resulting embedding is a caterpillar isomorphic to $A(\mathcal E^\circ)$ that satisfies the conditions of \cref{le:characterization-triconnected-onesidefixed}.
Second, we place $r$ and $(b,r)$ inside $g'_r$, thus obtaining a good embedding $\mathcal E'$ of $H$.
Finally, suppose that $f$ is not red in $\mathcal E^\circ$. This implies that $f=f_{in}=f_r$ and that $v_e$ is not incident to any internal red face of the embedding of $\pert{e}$ in~$\mathcal E^\circ$. Moreover, $f=f_r$ is an end-vertex of the backbone of $A(\mathcal E^*)$ incident to $v_e$ and to $r$, where $r''=r$ and $b_m$ is incident to~$f_r$. Therefore, $A(\mathcal E^\circ)$ coincides with $A(\mathcal E^*) \setminus \{f_r,r\}$ and has $r''=v_e$.
We obtain $\mathcal E'$ as follows. First, we reroute the edge $(u_e,v_e)$ in the interior of $g_\ell$. This merges $f$ with $g_r$ into a new face $g'_r$ incident to $b_m$ whose unique red vertex is $v_e$, which is now the right outer face of the resulting embedding of $\pert{e}$.
Also, it splits $g_\ell$, which was not red in $\mathcal E^\circ$, into two faces that are not red. This implies that the auxiliary graph of the resulting embedding is a caterpillar isomorphic to $A(\mathcal E^\circ)$ that satisfies the conditions of \cref{le:characterization-triconnected-onesidefixed}.
Second, we complete the construction of $\mathcal E'$, by placing $r$ and $(b,r)$ inside $g'_r$. We have that $A(\mathcal E')$ is a caterpillar, obtained from $A(\mathcal E^\circ)$, by adding the edges $(v_e,g'_r)$ and $(g'_r,r)$. This implies that \cref{condition:caterpillar} of \cref{le:characterization-triconnected-onesidefixed} is satisfied.
Further, by setting $r''=r$, we have that \cref{condition:endvertices} of \cref{le:characterization-triconnected-onesidefixed} is satisfied, since $b_m$ is incident to $g'_r$. This concludes the proof.
\end{proof}

\subsection{Embedding Classification}\label{sse:embedding-classification}

Consider a node $\mu$ of $\mathcal{T}$ different from $\rho$. We proceed with two definitions.

\begin{definition} \label{def:extends}
	An embedding  $\cal E$ of $H$ \emph{extends} an embedding ${\cal E}_\mu$ of $H_\mu$ if $\cal E$ coincides with ${\cal E}_\mu$ when restricted to the vertices and edges of $H_\mu$.
\end{definition}

\begin{definition} \label{def:extensible}
An embedding ${\cal E}_\mu$ of $H_\mu$ is called \emph{extensible} if there is a neat embedding of $H$ that extends ${\cal E}_\mu$.
\end{definition}

Let $\cal E$ be any embedding of $H$. Consider the embedding ${\cal E}_\mu$ of $H_\mu$ in $\cal E$, and the corresponding embedding ${\cal E}_\mu^-$ of $H_\mu^-$. The \emph{left} and \emph{right} outer faces $\ell(\mathcal E_\mu)$ and $r(\mathcal E_\mu)$, respectively, are the faces corresponding to the left and right outer face of ${\cal E}_\mu^-$, respectively. Often, we talk about the left and right outer faces of ${\cal E}_\mu$ even when an embedding $\cal E$ of $H$ is not specified or determined. 
Note that, when considering an embedding ${\cal E}_\mu$ of $H_\mu$, the vertices and edges that are incident to $\ell({\mathcal E}_\mu)$ and $r({\mathcal E}_\mu)$ are determined, with the exception of the (at most two) $rb$-trivial components that are incident to the outer face of ${\cal E}_\mu$ and to the poles $u_\mu$ and $v_\mu$ of $\mu$. In particular, if $\mathcal E$ is a neat embedding of $H$ extending ${\cal E}_\mu$, each of such components is incident either to $\ell({\mathcal E}_\mu)$, or to $r({\mathcal E}_\mu)$, or to an internal face of $\rest{\mu}$. We call \emph{undecided} such $rb$-trivial components. 

In the following, we present a classification of the possible embeddings~${\cal E}_\mu$ of~$H_\mu$ that occur in a neat embedding~$\cal E$ of~$H$, i.e., the extensible embeddings of $H_\mu$. Let $A(\cal E)$ be the auxiliary graph of $\cal E$. 
Ideally, we would like to classify such embeddings of $H_\mu$ based on the structure of the subgraph of $A(\cal E)$ induced by: 
(i) the vertices of $A(\mathcal E)$ corresponding to red vertices of $H_\mu$; and 
(ii) the vertices that correspond to faces of ${\mathcal E}_\mu$, including its left and right outer faces.
Then, when visiting a node $\mu$ in the bottom-up traversal of the SPQR-tree of $H^-$, we aim at testing which embedding types are possible for $H_\mu$.

There is one complication to this plan, though. 
Namely, an embedding ${\cal E}_\mu$ of $H_\mu$ does not always fully determine whether its left and right outer faces $\ell({\mathcal E}_\mu)$ and $r({\mathcal E}_\mu)$ are going to be red in a neat embedding $\mathcal E$ of $H$ extending ${\cal E}_\mu$.
For example, ${\cal E}_\mu$ might contain exactly one red vertex incident to $r({\mathcal E}_\mu)$; then $r({\mathcal E}_\mu)$ is red in an embedding $\cal E$ of $H$ extending ${\cal E}_\mu$ if and only if there is
either
a red vertex of $\rest{\mu}$
or
an undecided $rb$-trivial component of $\mu$ that is incident to $r({\mathcal E}_\mu)$ in~$\cal E$.
However, whether any of the above two cases occurs is not determined by the embedding ${\cal E}_\mu$ of~$\pert{\mu}$.
On the other hand, the existence of such a red vertex is guaranteed under some conditions, as described in the following lemma.

\begin{lemma} Suppose that $\mu$ is of type \ctype{} or \gtype{}. Then in any embedding $\cal E$ of $H$, at least one red vertex of $\rest{\mu}$ is incident to $\ell({\mathcal E}_\mu)$ and at least one red vertex of $\rest{\mu}$ is incident to~$r({\mathcal E}_\mu)$.
\end{lemma}

\begin{fullproof}
Let $\cal E$ be any embedding of $H$. Suppose, for a contradiction, that there exists no red vertex of $\rest{\mu}$ incident to $\ell({\mathcal E}_\mu)$ in $\cal E$. This implies that the portion of $\ell({\mathcal E}_\mu)$ that contains all the vertices of $\rest{\mu}$ incident to it induces a path of black vertices connecting the poles $u$ and $v$ of $\mu$. Since $\mu$ is of type BP or BB, there exists a path with the same property also in $\pert{\mu}$, namely $P_{u,v}$. Since these two paths create a cycle composed of black vertices, we have a contradiction. Hence, at least one red vertex of $\rest{\mu}$ is incident to $\ell({\mathcal E}_\mu)$ in $\cal E$. By a symmetric argument, at least one red vertex of $\rest{\mu}$ is incident to $r({\mathcal E}_\mu)$ in $\cal E$.
\end{fullproof}

In view of the above discussion, we classify the extensible embeddings ${\cal E}_\mu$ of $H_\mu$ based on the following auxiliary graph associated with $\mathcal E_\mu$.

\begin{definition}\label{def:auxiliary-graph-E-mu}
Let $\mathcal E_\mu$ be a extensible embedding of $\pert{\mu}$. The graph \emph{$A({\mathcal E}_\mu)$} contains a vertex for each red vertex of $\pert{\mu}$ that does not belong to an undecided $rb$-trivial component; furthermore, it contains a vertex for each internal red face of ${\mathcal E}_\mu$; finally, it contains a vertex for $\ell({\mathcal E}_\mu)$ (resp.\ for $r({\mathcal E}_\mu)$) if $\ell({\mathcal E}_\mu)$ (resp.\ $r({\mathcal E}_\mu)$) has at least two incident\footnote{For the sake of this definition, red vertices that belong to undecided $rb$-trivial components are not considered to be incident to $\ell({\mathcal E}_\mu)$ or $r({\mathcal E}_\mu)$.} red vertices of $\pert{\mu}$, or if $\ell({\mathcal E}_\mu)$ (resp.\ $r({\mathcal E}_\mu)$) has one incident red vertex of $\pert{\mu}$ and $\mu$ is of type \BP-BP or \BB-BB. The edge set of $A({\mathcal E}_\mu)$ contains an edge $(v,f)$ for each two vertices $v$ and $f$ that belong to the vertex set of $A({\mathcal E}_\mu)$ and such that $v$ is a vertex of $\pert{\mu}$ incident to a face $f$ of ${\mathcal E}_\mu$. 
\end{definition}

We give the following lemma about the structure~of~$A({\mathcal E}_\mu)$.

\begin{lemma}[Structure of $\mathbf{A({\mathcal E}_\mu)}$]\label{le:structure-A-mu}
	The graph $A({\mathcal E}_\mu)$ is composed of at most two caterpillars.
\end{lemma}

\begin{fullproof}
Since ${\mathcal E}_\mu$ is extensible, there exists a neat embedding $\cal E$ of $H$ that extends ${\mathcal E}_\mu$. Observe that each connected component of $A({\mathcal E}_\mu)$ is a caterpillar, since every edge of $A({\mathcal E}_\mu)$ also belongs to $A({\mathcal E})$ and since $A({\mathcal E})$ is a caterpillar, by~\cref{le:characterization-triconnected-onesidefixed}.

We now prove that there exist at most two connected components in $A({\mathcal E}_\mu)$. Suppose, for a contradiction, that there exist at least three such components. We first observe that, if there exists a component $C$ of $A({\mathcal E}_\mu)$ that contains neither $\ell({\mathcal E}_\mu)$, nor $r({\mathcal E}_\mu)$, nor the red pole of $\mu$ (if it exists), then $C$ is also a connected component of $A({\mathcal E})$, since any path in $A({\mathcal E})$ connecting a vertex of $\pert{\mu}$ with a vertex of $\rest{\mu}$ must contain one of these three vertices. Thus, we may assume that there exist exactly three connected components in $A({\mathcal E}_\mu)$, each containing one of these three vertices. However, this is not possible, since the red pole of $\mu$ is incident to both $\ell({\mathcal E}_\mu)$ and $r({\mathcal E}_\mu)$, and thus the three components would be merged into one. The statement follows.
\end{fullproof}

The \emph{backbones} of $A({\mathcal E}_\mu)$ are the paths obtained from the connected components of $A({\mathcal E}_\mu)$ by removing the leaves that correspond to red vertices of $\pert{\mu}$.

In the following subsections, we classify the extensible embeddings ${\mathcal E}_\mu$ of $\pert{\mu}$ in terms of the way in which $A({\mathcal E}_\mu)$ can participate to $A({\mathcal E})$, where $\mathcal E$ is a neat embedding of $H$ that extends ${\mathcal E}_\mu$.
Our classification is based on the following features:
\begin{enumerate}[(F1)]
\item the number of caterpillars $A({\mathcal E}_\mu)$ consists of;
\item the number of outer faces of $\mathcal E_\mu$ that belong to $A(\mathcal E_\mu)$; and
\item whether ${\cal E}_\mu$ has at least one internal red face $f$ or not.
\end{enumerate}


We are going to exploit the following lemma for the classification.

\begin{lemma}\label{le:internal-red-vertex}
If ${\mathcal E}_\mu$ contains an internal red vertex, then it also contains an internal red face and a red vertex incident to the outer face.
\end{lemma}
\begin{fullproof}
Let $w$ be an internal red vertex of ${\mathcal E}_\mu$. Suppose, for a contradiction, that no face incident to $w$ is red. Then the closed walk delimiting the face of ${\mathcal E}_\mu - w$ that used to contain $w$ is composed of black vertices only. However, this contradicts the fact that the black vertices induce a path in $H$. This proves the first part of the statement. 
For the second part, denote by $W$ the closed walk delimiting the outer face of ${\mathcal E}_\mu$. Note that $W$ contains at least a simple cycle, which contains $w$ in its interior. Thus, if there exists no red vertex incident to the outer face, then $W$ is composed of black vertices only, which again contradicts the fact that the black vertices induce a path in $H$.
\end{fullproof}

\subsubsection{type \protect\RE-RE Nodes}

\begin{wrapfigure}{r}{.2\textwidth}\tabcolsep=4pt
	\centering
	\vspace{-18mm}
	\includegraphics[page=1]{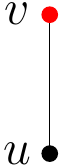}
	\caption{The pertinent graph of a type-\protect\RE-RE node.}\label{fig:Nodetypesw-gadgets}	
\end{wrapfigure}

type~\RE-RE nodes have a very simple structure, as in the following lemma.

\begin{lemma} \label{le:relevant-in-extensible-RE}
 	If $\mu$ is of type \RE-RE, then $\pert{\mu}$ has a unique extensible embedding (up to a relabeling of the neighbors of $u$). Furthermore, $A(\mathcal E_\mu)$ consists of a single red vertex.
 \end{lemma}

\begin{fullproof}
By \cref{le:structure-node-RE}, $H^-_\mu$ consists of an edge between the poles of $\mu$. Therefore, $A({\mathcal E}_\mu)$ contains just one red vertex, which corresponds to the red pole $v$ of $\mu$. 
\end{fullproof}

By \cref{le:relevant-in-extensible-RE}, we say that the unique extensible embedding of $\pert{\mu}$ a type \RE-RE node $\mu$ is of type \RE-RE; refer to \cref{fig:Nodetypesw-gadgets}.

\subsubsection{Type \protect\RF-RF Nodes}

We now discuss the case in which $v$ is red and some edges of $P$ belong to $\pert{\mu}$, that is, $\mu$ is of type~\RF-RF.

\begin{lemma}\label{le:structure-A-mu-RF}
	Suppose that $\mu$ is of type~\RF-RF. Then $A({\mathcal E}_\mu)$ is a caterpillar. Further, if $A({\mathcal E}_\mu)$ contains both $\ell(\mathcal E_\mu)$ and $r(\mathcal E_\mu)$, then the path between $\ell(\mathcal E_\mu)$ and $r(\mathcal E_\mu)$ in $A({\mathcal E}_\mu)$ is $(\ell(\mathcal E_\mu),v,r(\mathcal E_\mu))$.
\end{lemma}

\begin{fullproof}
Since ${\mathcal E}_\mu$ is extensible, by \cref{le:structure-A-mu}, we have that $A({\mathcal E}_\mu)$ is composed of at most two caterpillars. 
Suppose, for a contradiction, that $A({\mathcal E}_\mu)$ actually consists of two caterpillars. If one of them, say $C$, contains neither $\ell({\mathcal E}_\mu)$ nor $r({\mathcal E}_\mu)$, then, as in the proof of \cref{le:structure-A-mu}, we have that $C$ is also a connected component of $A({\mathcal E})$, a contradiction to the fact that $A({\mathcal E})$ is connected. Otherwise, one caterpillar contains $\ell({\mathcal E}_\mu)$ and the other one contains $r({\mathcal E}_\mu)$. However, $v$ is incident to both $\ell({\mathcal E}_\mu)$ and $r({\mathcal E}_\mu)$ in ${\mathcal E}_\mu$, hence it is adjacent to both $\ell({\mathcal E}_\mu)$ and $r({\mathcal E}_\mu)$ in $A({\mathcal E}_\mu)$, which implies that $\ell({\mathcal E}_\mu)$ and $r({\mathcal E}_\mu)$ are in the same connected component of $A({\mathcal E}_\mu)$, a contradiction that proves that $A({\mathcal E}_\mu)$ is a connected caterpillar. As just observed, if both $\ell(\mathcal E_\mu)$ and $r(\mathcal E_\mu)$ are nodes of $A({\mathcal E}_\mu)$, then $A({\mathcal E}_\mu)$ contains the path $(\ell(\mathcal E_\mu),v,r(\mathcal E_\mu))$.
\end{fullproof}

We now classify the possible types of ${\cal E}_\mu$ based on Features (F2)--(F3);
refer to~\cref{fig:RF} for a complete schematization and to~\cref{fig:RF-gadgets} for examples. Observe that the classification with respect to Feature (F1) is unique, as stated by \cref{le:structure-A-mu-RF}. There exist three possible options for Feature (F2), namely ${\cal E}_\mu$ can have $0$, $1$, or $2$ outer faces which are red. Further, there exist two possible options for Feature (F3), namely ${\cal E}_\mu$ can have an internal red face $f$ or not. This defines six types of embeddings for $\pert{\mu}$, called~\RFNz-{\bf RFN0},~\RFNi-{\bf RFN1}, \hbox{\RFNii-{\bf RFN2}}, \hbox{\RFIz-{\bf RFI0}}, \hbox{\RFIi-{\bf RFI1}}, and \hbox{\RFIii-{\bf RFI2}}, where the third character indicates whether ${\cal E}_\mu$ contains an internal red face (I) or not (N), and the fourth character indicates the number of red outer faces of ${\cal E}_\mu$. In the following, we refine the above classification. 

\begin{figure}[htb]\tabcolsep=6pt
	\centering
	\begin{tabular}{c c c c c c c}
		\includegraphics[height=.2\textwidth,page=15]{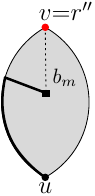}  &
		\includegraphics[height=.2\textwidth,page=14]{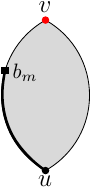}  &
		\includegraphics[height=.2\textwidth,page=16]{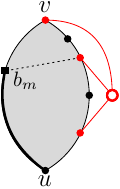}  &
		\includegraphics[height=.2\textwidth,page=17]{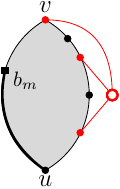}  &
		\includegraphics[height=.2\textwidth,page=19]{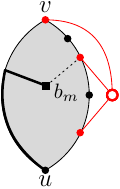} &
s		\includegraphics[height=.2\textwidth,page=18]{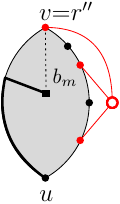}  &
		\includegraphics[height=.2\textwidth,page=20]{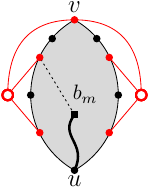} \\
		\protect\RFNza-RFN0A & \protect\RFNzb-RFN0B & \protect\RFNia-RFN1A & \protect\RFNib-RFN1B & \protect\RFNic-RFN1C & \protect\RFNid-RFN1D  & \protect\RFNii-RFN2
	\end{tabular}
	\begin{tabular}{c c c c}
		\includegraphics[height=.2\textwidth,page=6]{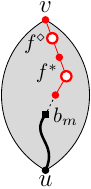}  &
		\includegraphics[height=.2\textwidth,page=8]{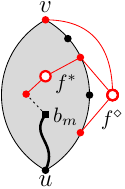} &
		\includegraphics[height=.2\textwidth,page=9]{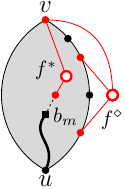}  &
		\includegraphics[height=.2\textwidth,page=12]{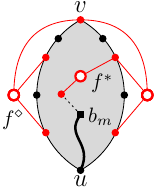} \\
		\protect\RFIz-RFI0 & \protect\RFIia-RFI1A & \protect\RFIib-RFI1B  & \protect\RFIii-RFI2
	\end{tabular}
	\caption{Embedding types for type~\protect\RF-RF nodes, without taking into account the existence of an undecided $rb$-trivial component incident to $u$.}\label{fig:RF}
\end{figure}

\begin{figure}[htb]\tabcolsep=6pt
	\centering
	\begin{tabular}{c c c c c c c}
		\includegraphics[height=.2\textwidth,page=15]{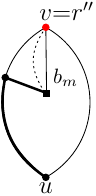}  &
		\includegraphics[height=.2\textwidth,page=14]{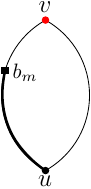}  &
		\includegraphics[height=.2\textwidth,page=16]{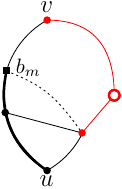}  &
		\includegraphics[height=.2\textwidth,page=17]{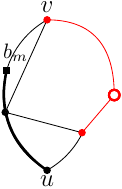}  &
		\includegraphics[height=.2\textwidth,page=19]{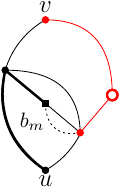}  &
		\includegraphics[height=.2\textwidth,page=18]{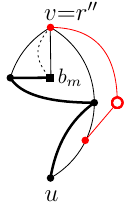}  &
		\includegraphics[height=.2\textwidth,page=20]{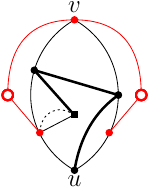} \\
		\protect\RFNza-RFN0A & \protect\RFNzb-RFN0B & \protect\RFNia-RFN1A & \protect\RFNib-RFN1B & \protect\RFNic-RFN1C & \protect\RFNid-RFN1D  & \protect\RFNii-RFN2
	\end{tabular}
	\begin{tabular}{c c c c}
		\includegraphics[height=.2\textwidth,page=6]{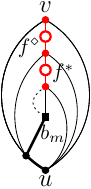}  &
		\includegraphics[height=.2\textwidth,page=8]{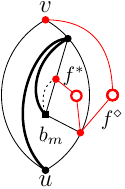} &
		\includegraphics[height=.2\textwidth,page=9]{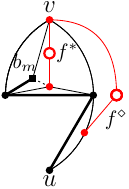}  &
		\includegraphics[height=.2\textwidth,page=12]{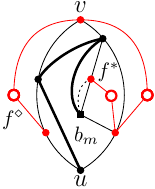} \\
		\protect\RFIz-RFI0 & \protect\RFIia-RFI1A & \protect\RFIib-RFI1B  & \protect\RFIii-RFI2
	\end{tabular}
	\caption{Embeddings of pertinent graphs of type \protect\RF-RF.}\label{fig:RF-gadgets}
\end{figure}

We first deal with embeddings with no internal red faces. We start with the following lemma.

\begin{lemma}\label{le:rf-endvertex-backbone-no-internal}
Suppose that $\mathcal E_\mu$ does not contain any internal red face. If $\pert{\mu}$ contains at least one red vertex different from $v$, then $r(\mathcal{E}_\mu)$ or $\ell(\mathcal{E}_\mu)$ is an end-vertex of the backbone of $A(\mathcal{E})$. 
\end{lemma}
\begin{fullproof}
Since $\mu$ is of type \RF-RF, we have that $b_m$ is a non-pole vertex of $\pert{\mu}$. By \cref{le:internal-red-vertex} and by the assumption that $\mathcal E_\mu$ does not contain any internal red face, we have that $\mathcal E_\mu$ does not contain any internal red vertex. By \cref{le:characterization-triconnected-onesidefixed}, we have that $A(\cal E)$ contains a leaf $r''$ that is adjacent to an end-vertex of the backbone of $A(\mathcal E)$ and that shares a face of $\mathcal E$ with $b_m$. 

Hence, either (1) $r''$ is a vertex of $\rest{\mu}$ incident to $r(\mathcal E_\mu)$ or $\ell(\mathcal E_\mu)$; or (2) $r''$ is a non-pole vertex of $\pert{\mu}$ incident to $r(\mathcal E_\mu)$ or $\ell(\mathcal E_\mu)$; or (3) $r''=v$. In cases (1) and (2), we have that $r''$ is (only) adjacent to $r(\mathcal E_\mu)$ or $\ell(\mathcal E_\mu)$ in $A(\mathcal{E})$, hence $r(\mathcal E_\mu)$ or $\ell(\mathcal E_\mu)$ is an end-vertex of the backbone of $A(\mathcal E)$. In case (3), by the assumption that $\pert{\mu}$ contains at least one red vertex different from $v$ and since such a vertex is incident to either $r(\mathcal E_\mu)$ or $\ell(\mathcal E_\mu)$, we have that $r(\mathcal E_\mu)$ or $\ell(\mathcal E_\mu)$ is red. Hence, $r''=v$ is (only) adjacent to $r(\mathcal E_\mu)$ or $\ell(\mathcal E_\mu)$ in $A(\mathcal{E})$, and again $r(\mathcal E_\mu)$ or $\ell(\mathcal E_\mu)$ is an end-vertex of the backbone of $A(\mathcal E)$.
\end{fullproof}

We refine the classification for the embedding types \RFNz-RFN0,~\RFNi-RFN1, and \RFNii-RFN2, according to the position of $b_m$ with respect to the red vertices of $\pert{\mu}$. Observe that, by \cref{le:internal-red-vertex}, all the red vertices are incident to the outer face of $\mathcal{E}_\mu$. Further, $b_m$ is adjacent to at least one red vertex in $\pert{\mu}^-$, since $b_m$ has only one black neighbor and it has degree at least~$2$~in~$\pert{\mu}^-$.

Suppose first that $\mathcal{E}_\mu$ is of type \RFNz-RFN0. In this case, $v$ is the only red vertex of $\pert{\mu}^-$, and thus $b_m$ is adjacent to $v$. If $b_m$ is an internal vertex of $\mathcal{E}_\mu$, then we say that $\mathcal{E}_\mu$ is of type \RFNza-{\bf RFN0A}, otherwise it is of type \RFNzb-{\bf RFN0B}. Note that, if $\mathcal{E}_\mu$ is of type \RFNza-RFN0A, then $v$ must coincide with the leaf $r''$ of $A(\mathcal{E})$ that shares a face with $b_m$ and that is adjacent to an end-vertex of the backbone of $A(\mathcal{E})$. 

Suppose next that $\mathcal{E}_\mu$ is of type \RFNi-RFN1 and assume without loss of generality that $r(\mathcal{E}_\mu)$ is red. We further distinguish four subtypes \RFNia-{\bf RFN1A}, \RFNib-{\bf RFN1B}, \RFNic-{\bf RFN1C}, and \RFNid-{\bf RFN1D}, based on whether $b_m$ is incident to $\ell(\mathcal{E}_\mu)$ (\RFNia-RFN1A and \RFNib-RFN1B) or not (\RFNic-RFN1C and \RFNid-RFN1D), and based on whether $b_m$ shares a face with at least one red vertex different from $v$ (\RFNia-RFN1A and \RFNic-RFN1C) or not (\RFNib-RFN1B and \RFNid-RFN1D).
Observe that, if $\mathcal{E}_\mu$ is of type \RFNid-RFN1D, then $v$ must coincide with the leaf $r''$ of $A(\mathcal{E})$ that shares a face with $b_m$ and that is adjacent to an end-vertex of the backbone of $A(\mathcal{E})$. Further, if $\mathcal{E}_\mu$ is of type either \RFNic-RFN1C or \RFNid-RFN1D, then $r(\mathcal{E}_\mu)$ must be an end-vertex of the backbone of $A(\mathcal{E})$. 

Suppose that $\mathcal{E}_\mu$ is of type \RFNii-RFN2. In this case, we do not perform any further refinement. Observe that $b_m$ must be adjacent to at least one red vertex that is different from $v$, since $v$ belongs to the backbone of $A(\mathcal{E})$ and thus $v \neq r''$.

We now turn our attention to embeddings with internal red faces. We have the following lemma.
 
\begin{lemma}\label{le:rf-endvertex-backbone}
Suppose that $\mathcal E_\mu$ contains an internal red face, and let $f^*$ and $f^\diamond$ be the end-vertices of the backbone $\mathcal B_\mu$ of the caterpillar $A(\mathcal E_\mu)$. We have that:
\begin{enumerate}
\item \label{le:rf-endvertex-backbone-bm} Let $\mathcal B$ be the backbone of $A(\mathcal E)$ and let $f^{\circ}$ be the end-vertex of $\mathcal B$ that is adjacent in $A(\mathcal E)$ to a leaf $r''$ of $A(\mathcal E)$ that shares a face with $b_m$. Then $f^{\circ}$ is one of $f^*$ and $f^\diamond$, say $f^{\circ}=f^*$; further, $f^{\circ}$ corresponds to an internal face of $\mathcal E_\mu$.
\item If $\mathcal E_\mu$ is of type \RFIz-RFI0, then $f^\diamond$ corresponds to an internal face of $\mathcal E_\mu$ incident to the red pole $v$, which is a leaf of $A(\mathcal E_\mu)$ (note that, in this case, $f^\diamond=f^*$ may hold).
\item If $\mathcal E_\mu$ is of type \RFIi-RFI1 or \RFIii-RFI2, then $f^\diamond$ corresponds to an outer face of $\mathcal E_\mu$.
\end{enumerate}
\end{lemma}


\begin{fullproof}
Assume that $\mathcal E_{\mu}$ contains an internal red face $f$, as otherwise there is nothing to prove. The proof distinguishes the case in which $f^*=f^\diamond$ from the one in which $f^*\neq f^\diamond$.

We first discuss the case in which $f^*=f^\diamond$. In this case $\mathcal E_{\mu}$ contains a single red face, which is adjacent to all the red vertices of $H^-_{\mu}$, including $v$. Since $f$ is an internal red face, we have that $f^*=f^\diamond=f$ and hence each of the outer faces of $\mathcal E_{\mu}$ is not red (that is, $\mathcal E_\mu$ is of type \RFIz-RFI0). It follows that the only red vertex that is incident to the outer face of $\mathcal E^-_{\mu}$ is $v$, which is a leaf given that $\mathcal E_{\mu}$ contains a single red face. This proves Item 2. We now prove that the choice $f^\circ=f^*=f^\diamond$ satisfies the requirements of Item 1. First, $v$ is the only neighbor of $f^\circ$ in $A(\mathcal E_\mu)$ which might be adjacent to a vertex different from $f^\circ$ in $A(\mathcal E)$, as all the other red vertices of $H^-_{\mu}$ are internal to $\mathcal E_{\mu}$. Hence, $f^\circ$ is also an end-vertex of $\mathcal B$. By \cref{le:characterization-triconnected-onesidefixed}, we have that $f^\circ$ is adjacent to a leaf that shares a face with either $b_1$ or $b_m$. If $f^\circ$ is adjacent to a leaf $r''$ that shares a face with $b_m$, then we are done, so assume that $f^\circ$ is adjacent to a leaf $r'$ that shares a face with $b_1$. By \cref{le:rooting-at-b1b2}, we have that $b_1$ is not a vertex of $\pert{\mu}$, hence $r'$ cannot be an internal vertex of $\mathcal E_{\mu}$, that is, we have $r'=v$. This implies that $\mathcal B$ consists of the vertex $f^\circ$ only, hence $f^\circ$ is also adjacent to a leaf $r''$ that shares a face with $b_m$, and the proof of Item 1 is completed. The statement of Item 3 is vacuous, given that $\mathcal E_\mu$ is of type \RFIz-RFI0.

In the remainder of the proof we assume that $f^*\neq f^\diamond$.

We first prove that (at least) one of $f^*$ and $f^\diamond$ corresponds to an internal red face of $\mathcal E_{\mu}$. Suppose, for a contradiction, that neither of $f^*$ and $f^\diamond$ corresponds to an internal red face of $\mathcal E_{\mu}$. Hence, ${\mathcal B}_{\mu}=(f^*,\dots,f,\dots,f^\diamond)$, where $f^*,f^\diamond\in \{\ell(\mathcal E_{\mu}),r(\mathcal E_{\mu})\}$. However, ${\mathcal B}_{\mu}$, together with the path $(\ell(\mathcal E_{\mu}),v,r(\mathcal E_{\mu}))$ which belongs to $A(\mathcal E_{\mu})$ by \cref{le:structure-A-mu-RF}, forms a cycle, contradicting the fact that $A(\mathcal E_{\mu})$ is a caterpillar. It follows that (at least) one of $f^*$ and $f^\diamond$, say $f^*$, corresponds to an internal red face of $\mathcal E_{\mu}$. 

We now argue that $\mathcal B$ has an end-vertex $f^\circ$ that corresponds to an internal red face of $\mathcal E_\mu$ and whose every adjacent leaf in $A(\mathcal E)$ is an internal vertex of $\mathcal E_\mu$. If $f^*$ is an end-vertex of $\mathcal B$ and every leaf adjacent to $f^*$ in $A(\mathcal E)$ is an internal vertex of $\mathcal E_\mu$, then we are done with $f^\circ=f^*$. Assume to the contrary that there exists a leaf $w$ of $A(\mathcal E)$ that is adjacent to $f^*$ in $A(\mathcal E)$ and that is incident to the outer face of $\mathcal E^-_\mu$. Since $f^*$ is an internal face of $\mathcal E_\mu$, it follows that $w$ is also a leaf of $A(\mathcal E_\mu)$ that is adjacent to $f^*$ in $A(\mathcal E_\mu)$. If $w\neq v$ and $w$ is incident to $\ell(\mathcal E_\mu)$ (to $r(\mathcal E_\mu)$), then $\ell(\mathcal E_\mu)$ (resp.\ $r(\mathcal E_\mu)$) belongs to $A(\mathcal E_\mu)$, and hence $w$ is adjacent to both $f^*$ and $\ell(\mathcal E_\mu)$ (resp.\ $r(\mathcal E_\mu)$) in $A(\mathcal E_\mu)$, which is not possible since $w$ has degree $1$ in $A(\mathcal E_\mu)$. If $w=v$, then, since $v$ is adjacent to $f^*$ and has degree $1$ in $A(\mathcal E_\mu)$, it follows that each of $\ell(\mathcal E_\mu)$ and $r(\mathcal E_\mu)$ is not red, and thus no vertex different from $v$ and incident to the outer face of $\mathcal E^-_\mu$ is red. Hence, $f^\diamond\neq f^*$ is an internal red face of $\mathcal E_\mu$; further, $f^\diamond$ has no adjacent leaf in $A(\mathcal E_\mu)$ that is incident to the outer face of $\mathcal E^-_\mu$, hence $f^\circ=f^\diamond$ is an end-vertex of $\mathcal B$ that corresponds to an internal red face of $\mathcal E_\mu$ and whose every adjacent leaf in $A(\mathcal E)$ is an internal vertex of $\mathcal E_\mu$, as requested.

We are now ready to prove Item 1. As proved above, $f^{\circ}\in \{f^*,f^\diamond\}$ is an end-vertex of $\mathcal B$ that corresponds to an internal red face of $\mathcal E_\mu$ and whose every adjacent leaf in $A(\mathcal E)$ is an internal vertex of $\mathcal E_\mu$. By \cref{le:rooting-at-b1b2}, we have that $b_1$ is not a vertex of $\pert{\mu}$, hence it cannot share a face of $\mathcal E$ with any leaf $r'$ of $A(\mathcal E)$ that is adjacent to $f^\circ$ in $A(\mathcal E)$. By \cref{le:characterization-triconnected-onesidefixed}, it follows that $f^{\circ}$ is adjacent in $A(\mathcal E)$ to a leaf $r''$ of $A(\mathcal E)$ that shares a face with $b_m$.

%

In the following assume, w.l.o.g., that $f^{\circ} = f^*$.

In order to prove Item~2, assume that $\mathcal E_{\mu}$ is of type \RFIz-RFI0. Suppose, for a contradiction, that $f^\diamond$ does not correspond to an internal face of $\mathcal E_\mu$ incident to the red pole $v$. By the definition of type  \RFIz-RFI0, the outer faces of $\mathcal E_{\mu}$ are not red. It follows that $f^\diamond$ corresponds to an internal face of $\mathcal E_\mu$; if $f^\diamond$ is not incident to $v$, then every red vertex incident to $f^\diamond$ is an internal vertex of $\mathcal E_{\mu}$. Thus, $f^\diamond$ is also an end-vertex of the backbone of $A(\mathcal E)$. However, by \cref{le:rooting-at-b1b2}, we have that $b_1$ is not a vertex of $\pert{\mu}$, hence it cannot share a face of $\mathcal E$ with any leaf $r'$ of $A(\mathcal E)$ that is adjacent to $f^\diamond$ in $A(\mathcal E)$, a contradiction to \cref{le:characterization-triconnected-onesidefixed}. This proves that $f^\diamond$ corresponds to an internal face of $\mathcal E_\mu$ that is incident to $v$. If $v$ is not a leaf of $A(\mathcal E_{\mu})$, we again have that $f^\diamond$ is an end-vertex of the backbone of $A(\mathcal E)$ and that any leaf $r'$ of $A(\mathcal E)$ that is adjacent to $f^\diamond$ in $A(\mathcal E)$ is an internal vertex of $\mathcal E_{\mu}$, hence it cannot share a face of $\mathcal E$ with $b_1$, a contradiction to \cref{le:characterization-triconnected-onesidefixed}. This proves that $v$ is a leaf of $A(\mathcal E_{\mu})$ and hence concludes the proof of Item 2.

In order to prove Item 3, assume that $\mathcal E_{\mu}$ is of type \RFIi-RFI1. Assume, w.l.o.g.\ up to symmetry, that the red outer face of $\mathcal E_{\mu}$ is $r(\mathcal E_{\mu})$. Suppose, for a contradiction, that $f^\diamond$ corresponds to an internal face of $\mathcal E_\mu$. Then the backbone of $A(\mathcal E_{\mu})$ is a path ${\mathcal B}_{\mu}=(f^*,\dots,r(\mathcal E_{\mu}),\dots,f^\diamond)$. Every red vertex that is incident to the outer face of $\mathcal E^-_{\mu}$ is a neighbor of $r(\mathcal E_{\mu})$ in $A(\mathcal E_{\mu})$, and thus it is not a leaf of $A(\mathcal E_{\mu})$ adjacent to $f^\diamond$. It follows that $f^\diamond$ is also an end-vertex of the backbone of $A(\mathcal E)$. Thus, any leaf $r'$ of $A(\mathcal E)$ that is adjacent to $f^\diamond$ is an internal vertex of $\mathcal E_{\mu}$, hence by \cref{le:rooting-at-b1b2} it cannot share a face of $\mathcal E$ with $b_1$, a contradiction to \cref{le:characterization-triconnected-onesidefixed}. Similarly, assume that $\mathcal E_{\mu}$ is of type \RFIii-RFI2 and consider any good embedding $\mathcal E$ of $H$ extending $\mathcal E_{\mu}$. If $f^\diamond$ corresponds to an internal face of $\mathcal E_\mu$, then the backbone of $A(\mathcal E_{\mu})$ is a path ${\mathcal B}_{\mu}$ that has $\ell(\mathcal E_{\mu})$ and $r(\mathcal E_{\mu})$ as internal vertices. Every red vertex that is incident to the outer face of $\mathcal E^-_{\mu}$ is a neighbor of $\ell(\mathcal E_{\mu})$ or $r(\mathcal E_{\mu})$ in $A(\mathcal E_{\mu})$, and thus it is not a leaf of $A(\mathcal E_{\mu})$ adjacent to $f^\diamond$. It follows that $f^\diamond$ is also an end-vertex of the backbone of $A(\mathcal E)$. Thus, any leaf $r'$ of $A(\mathcal E)$ that is adjacent to $f^\diamond$ is an internal vertex of $\mathcal E_{\mu}$, hence by \cref{le:rooting-at-b1b2} it cannot share a face of $\mathcal E$ with $b_1$, a~contradiction~to~\cref{le:characterization-triconnected-onesidefixed}.\end{fullproof}


If $\mu$ is of type \RFIi-RFI1, then consider the internal red face $f^*$ of $\mathcal E_\mu$ that corresponds to an end-vertex of the backbone of $A(\mathcal E_\mu)$; this exists by \cref{le:rf-endvertex-backbone}. By \cref{le:structure-A-mu-RF}, we have that $A({\cal E}_\mu)$ is a caterpillar; then there exists exactly one path in $A(\mathcal E_\mu)$ connecting $f^*$ and $v$. Two cases are possible: Either this path contains the red outer face of ${\cal E}_\mu$ (type \hbox{\RFIia-{\bf RFI1A}}) or not (type \hbox{\RFIib-{\bf RFI1B}}). Note that, if $\mu$ is of type \RFIz-RFI0 or \RFIii-RFI2, then this further distinction is not meaningful. In fact, if $\mu$ is of type \RFIz-RFI0, then the outer faces are not red; while if $\mu$ is of type \RFIii-RFI2, then both the outer faces are red, and thus the path connecting $f^*$ and $v$ must contain one of such faces, as otherwise $v$ would be incident to three red faces, contradicting the fact that $A(\mathcal E_\mu)$ is a caterpillar.
The above discussion leads to the following.

\begin{lemma} \label{le:relevant-in-extensible-RF}
	If $\mu$ is of type \RF-RF, any extensible embedding of $\pert{\mu}$ is of one of the types \RFNza-RFN0A,~\RFNzb-RFN0B,~\RFNia-RFN1A,~\RFNib-RFN1B,~\RFNic-RFN1C,~\RFNid-RFN1D, \RFNii-RFN2, \RFIz-RFI0, \RFIia-RFI1A, \RFIib-RFI1B, and \RFIii-RFI2.
\end{lemma}

\subsubsection{Type \protect\BE-BE Nodes}

First, we prove that the type~\BE-BE nodes have a simple structure.

\begin{lemma}\label{le:structure-node-BE}
	Suppose that $\mu$ is of type~\BE-BE. Then the graph $\pert{\mu}^-$ consists of a set of length-$2$ paths between the poles of $\mu$. The middle vertices of these paths are red. Further, the graph $\pert{\mu}$ is the same as $\pert{\mu}^-$ plus, for each pole of $\mu$, at most one $rb$-trivial component incident to it.
\end{lemma}

\begin{fullproof}
	By definition of type~\BE-BE, all the non-pole vertices of $\pert{\mu}$ and of $\pert{\mu}^-$ are red. Since there exists no edge between two red vertices, each non-pole vertex of $\pert{\mu}^-$ is connected to both~$u$ and~$v$. Further, the two poles $u$ and $v$ are not connected by an edge in $\pert{\mu}$. In fact, if edge $(u,v)$ exists in $H$, then it coincides with path $P_{u,v}$, and thus it belongs to $\rest{\mu}$. Hence, there exists at least one non-pole vertex, and the statement for $\pert{\mu}^-$ follows. The statement for $\pert{\mu}$ descends from its definition.
\end{fullproof}

If there exists only one length-$2$ path in $\pert{\mu}$ between the poles of $\mu$, then we say that $\mu$ is of \emph{type~\BE-BE slim}, while if there exists more than one of such paths, then $\mu$ is of \emph{type~\BE-BE fat}; refer to \cref{fig:BE-gadgets} for an example.

 \begin{figure}[tb]\tabcolsep=4pt
 	\centering
 	\subfloat[\BE-BE-slim]{
 	\includegraphics[height=.2\textwidth,page=7]{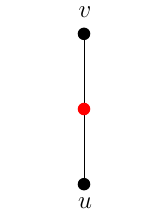}
 }\hfil
 	\subfloat[\BE-BE-fat]{
	\includegraphics[height=.2\textwidth,page=8]{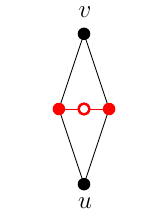}
}
 	\caption{Embeddings of pertinent graphs of type \protect\BE-BE.}\label{fig:BE-gadgets}	
 \end{figure}

We now describe the structure of $A({\mathcal E}_\mu)$ for a node $\mu$ of type~\BE-BE.

\begin{lemma}\label{le:structure-A-mu-BE}
Suppose that $\mu$ is of type~\BE-BE. Let $u_{\ell}$ and $u_r$ be the red neighbors of $u$ in $\pert{\mu}^-$ that are incident to $\ell(\mathcal E_\mu)$ and $r(\mathcal E_\mu)$, respectively. Then $A({\mathcal E}_\mu)$ is a path between $u_\ell$ and $u_r$ that contains neither $\ell(\mathcal E_\mu)$ nor $r(\mathcal E_\mu)$. 
\end{lemma}

\begin{fullproof}
By \cref{le:structure-node-BE}, we have that $u_{\ell}$ and $u_r$ are the only red vertices of $\pert{\mu}^-$ that are incident to $\ell({\mathcal E}_\mu)$ and $r({\mathcal E}_\mu)$, respectively. Since $\ell({\mathcal E}_\mu)$ and $r({\mathcal E}_\mu)$ are not incident to any other red vertex of $\pert{\mu}$ (recall that the red vertices belonging to undecided $rb$-trivial components are incident neither to $\ell({\mathcal E}_\mu)$ nor to $r({\mathcal E}_\mu)$), we have that neither $\ell({\mathcal E}_\mu)$ nor $r({\mathcal E}_\mu)$ belongs to $A({\mathcal E}_\mu)$.

Thus, if $\mu$ is of type~\BE-BE~slim, then $A({\mathcal E}_\mu)$ only contains the middle red vertex $u_{\ell}=u_r$ of the unique path between the poles. Further, if $\mu$ is of type~\BE-BE~fat, then every internal face of ${\mathcal E}_\mu$ is incident to exactly two red vertices of $\pert{\mu}$, while each of these vertices is incident to exactly two internal faces of ${\mathcal E}_\mu$, except for $u_\ell$ and $u_r$, which are incident only to one internal face; this constitutes the path $A({\mathcal E}_\mu)$. The statement follows.
\end{fullproof}

\cref{le:structure-A-mu-BE} shows that the structure of $A(\mathcal E_\mu)$ is the same in any planar embedding $\mathcal E_\mu$ of $\pert{\mu}$. In fact, concerning Feature (F1), $A({\mathcal E}_\mu)$ consists of only one caterpillar; concerning Feature (F2), none of the outer faces is red; concerning Feature (F3), there exists an internal red face of  ${\cal E}_\mu$ if and only if $\mu$ is of type \BE-BE fat. 
We state the following.

\begin{lemma} \label{le:relevant-in-extensible-BE}
 	If $\mu$ is of type \BE-BE, then $\pert{\mu}$ has a unique extensible embedding (up to a relabeling of the neighbors of $u$ and of the neighbors of $v$).
 \end{lemma}

By \cref{le:relevant-in-extensible-BE}, we say that the unique extensible embedding of $\pert{\mu}$ is of type \BE-BE slim (resp.\ fat) if $\mu$ is of type \BE-BE slim (resp.\ fat).

\subsubsection{Type \protect\BP-BP Nodes and type \protect\BB-BB Nodes}

We now discuss the case in which both $u$ and $v$ are black and the path $P_{uv}$ belongs to $\pert{\mu}$. This corresponds to the types~\BP-BP and \BB-BB. We start with the following lemma on the structure of $A({\mathcal E}_\mu)$.

\begin{lemma}\label{le:structure-A-mu-BP-BB}
	Suppose that $\mu$ is of type~\BP-BP or \BB-BB, then the following hold: 
	\begin{enumerate}
		\item If $\mu$ is of type~\BB-BB, then $A(\mathcal E_\mu)$ consists of either one or two caterpillars, while if $\mu$ is of type~\BP-BP, then $A(\mathcal E_\mu)$ consists of either zero, one, or two caterpillars.
		\item \label{le:structure-A-mu-BP-BB-item-no-traversing} If $A(\mathcal E_\mu)$ consists of one caterpillar, then its backbone starts at either $\ell({\mathcal E}_\mu)$ or $r({\mathcal E}_\mu)$; also, $A(\mathcal E_\mu)$ does not contain both $\ell({\mathcal E}_\mu)$ and $r({\mathcal E}_\mu)$.
		\item \label{le:structure-A-mu-BP-BB-outer-face} If $A(\mathcal E_\mu)$ consists of two caterpillars, then the backbone of one caterpillar starts at $\ell({\mathcal E}_\mu)$ and the backbone of the other one starts at $r({\mathcal E}_\mu)$; further, the backbone of one of the two caterpillars is a single vertex. 
		\item \label{le:structure-A-mu-BP-BB-item-bm-in-pertinent} If one of the caterpillars composing $A(\mathcal E_\mu)$, say $C$, contains a vertex corresponding to an internal face of $\cal E_\mu$, then $b_m$ belongs to $\pert{\mu}$ (in particular, $b_m=v$ if $\mu$ is of type~\BP-BP, while $b_m\notin\{u,v\}$ if $\mu$ is of type~\BB-BB); further, the end-vertex of the backbone of $C$ that corresponds to an internal face of $A(\mathcal E_\mu)$ is adjacent to a leaf $r''$ of $A(\mathcal E_\mu)$ that shares a face with $b_m$.
		\item If $\mu$ is of type~\BB-BB, then an end-vertex of the backbone of a caterpillar composing $A(\mathcal E_\mu)$ is adjacent to a leaf that shares a face with $b_m$.
	\end{enumerate}  
\end{lemma}

\begin{fullproof}
By \cref{le:structure-A-mu}, we have that $A({\mathcal E}_\mu)$ is composed of at most two caterpillars. 

To prove the first item of the statement, we prove that if $\mu$ is of type~\BB-BB, then $A(\mathcal E_\mu)$ is not empty. Namely, if $\pert{\mu}$ contains no red vertex, then it is a black path, hence $\mu$ is of type~\BP-BP. It follows that, if $\mu$ is of type~\BB-BB, then $\pert{\mu}$ contains at least one red vertex and hence $A({\mathcal E}_\mu) \neq \emptyset$.

In the following, we assume, without loss of generality, that $A({\mathcal E}_\mu)$ contains at least one caterpillar, as otherwise there is nothing more to prove. 
Thus, $\pert{\mu}$ contains at least one red vertex, which implies that $\pert{\mu}$ also contains a red vertex that is incident to one of $\ell({\mathcal E}_\mu)$ and $r({\mathcal E}_\mu)$, by \cref{le:internal-red-vertex}. Hence, such an outer face belongs to $A({\mathcal E}_\mu)$, by \cref{def:auxiliary-graph-E-mu}.

Next, we prove the second item of the statement.
If $A({\mathcal E}_\mu)$ consists of one caterpillar $C$, then $C$ contains at least one of the outer faces, as discussed above. Further, $C$ does not contain both $\ell({\mathcal E}_\mu)$ and $r({\mathcal E}_\mu)$, since $P_{u,v}$ ``separates'' $\ell({\mathcal E}_\mu)$ from $r({\mathcal E}_\mu)$. It follows that $A({\mathcal E}_\mu)$ contains exactly one of $\ell({\mathcal E}_\mu)$ and $r({\mathcal E}_\mu)$, say $r({\mathcal E}_\mu)$. We show that $r({\mathcal E}_\mu)$ is an end-vertex of the backbone of $C$.
Suppose, for a contradiction, that $r({\mathcal E}_\mu)$ is not an end-vertex of $C$. Then both the end-vertices of $C$ correspond to internal red faces of $\mathcal E_\mu$. However, by \cref{condition:endvertices} of \cref{le:characterization-triconnected-onesidefixed}, one of them must be adjacent to a leaf $r'$ that shares a face with $b_1$. We claim that, under this condition, $b_1$ belongs to $\pert{\mu}$. Namely, if $b_1$ does not belong to $\pert{\mu}$, then the only face that can be shared by $b_1$ and $r'$ is $r({\mathcal E}_\mu)$, given that $r'$ is a non-pole vertex of $\pert{\mu}$. However, this is not possible as otherwise $r({\mathcal E}_\mu)$, which is the only red face incident to $r'$, would be an end-vertex of the backbone of $C$. Thus, $b_1$ belongs to $\pert{\mu}$, contradicting \cref{le:rooting-at-b1b2}. This proves the second item of the statement.

Assume now that $A({\mathcal E}_\mu)$ consists of two caterpillars $C$ and $C'$. Suppose, for a contradiction, that one of them, say $C$, contains neither $\ell({\mathcal E}_\mu)$ nor $r({\mathcal E}_\mu)$. Then we have that $C$ is also a connected component of $A({\mathcal E})$, a contradiction to the fact that $A({\mathcal E})$ is connected. It follows that each caterpillar composing $A({\mathcal E}_\mu)$ contains exactly one of $\ell({\mathcal E}_\mu)$ and $r({\mathcal E}_\mu)$. 
To prove the third item, observe that if one of the caterpillars does not start at an outer face, or if both caterpillars have at least one vertex on their backbone that corresponds to an internal red face of $\pert{\mu}$, then there exist two internal red faces of $\pert{\mu}$ that correspond to end-vertices of the backbone of a caterpillar. With the same arguments as in the case in which $A({\mathcal E}_\mu)$ consists of one caterpillar, we can prove that this implies that $b_1$ belongs to $\pert{\mu}$, which contradicts \cref{le:rooting-at-b1b2}.

We now prove the fourth item. Let $C$ be a caterpillar composing $A(\mathcal E_\mu)$ that contains a vertex corresponding to an internal face of $\cal E_\mu$. Then one of the end-vertices of the backbone of $C$ corresponds to an internal face $f$ of $\cal E_\mu$, due to the second and the third item of this lemma. Further, $f$ is also an end-vertex of the backbone of $A(\mathcal E)$. In particular, $f$ is the end-vertex of the backbone of $A(\mathcal E)$ that is adjacent to a leaf $r''$ sharing a face with $b_m$, as otherwise $b_1$ would belong to $\pert{\mu}$, as proved above, which would contradict \cref{le:rooting-at-b1b2}. Note that $r'' \in \pert{\mu}$, since $f$ is a face of $\cal E_\mu$, and thus $r'' \in A(\cal E_\mu)$. Finally, with the same argument used to prove that $r' \in \pert{\mu}$ implies $b_1 \in \pert{\mu}$, we can prove that $r'' \in \pert{\mu}$ implies $b_m \in \pert{\mu}$. This concludes the proof of the fourth item.

To prove the fifth item, observe that if $\mu$ is of type \BB-BB, then $b_m$ is a non-pole vertex of $\pert{\mu}$. If $\cal E_\mu$ contains at least an internal red face, then the proof of the statement follows from the fourth item of this lemma. If $\cal E_\mu$ does not contain any internal red face, then $A(\mathcal{E}_\mu)$ is composed of either one or two caterpillars, each consisting of a star centered at one of the outer faces of $\cal E_\mu$. Since there exists at least a red vertex of $\pert{\mu}$ sharing a face with $b_m$, and since every red vertex is a leaf of $A(\mathcal{E}_\mu)$ incident to an outer face, the statement follows.
\end{fullproof}

\begin{figure}[htb]\tabcolsep=4pt
	\centering
	\begin{tabular}{c c c c c}
		\includegraphics[height=.2\textwidth,page=1]{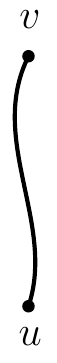} &
		\includegraphics[height=.2\textwidth,page=3]{BP} &
		\includegraphics[height=.2\textwidth,page=2]{BP} &
		\includegraphics[height=.2\textwidth,page=5]{BP} &
		\includegraphics[height=.2\textwidth,page=4]{BP}
		\\
		{\protect \BPi-BP1} & {\protect \BPii-BP2} & {\protect \BBii-BB2} & {\protect \BPiii-BP3} & {\protect \BBiii-BB3}
	\end{tabular}
	\begin{tabular}{c c c c}
		\includegraphics[height=.2\textwidth,page=7]{BP} &
		\includegraphics[height=.2\textwidth,page=6]{BP} &
		\includegraphics[height=.2\textwidth,page=9]{BP} &
		\includegraphics[height=.2\textwidth,page=8]{BP}
		\\
		{\protect \BPiv-BP4} & {\protect \BBiv-BB4} & {\protect \BPv-BP5} & {\protect \BBv-BB5}
	\end{tabular}
	\caption{Embedding types for type~\protect\BP-BP and~\protect\BB-BB nodes. The order of the figures has been chosen so to highlight the correspondence between the embedding types for type~\protect\BP-BP and type~\protect\BB-BB nodes, based on the features of our classification.}\label{fig:BPBB}
\end{figure}

\begin{figure}[htb]\tabcolsep=4pt
	\centering
	\begin{tabular}{c c c c c}
		\includegraphics[height=.2\textwidth,page=1]{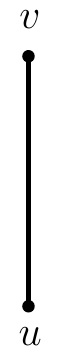} &
		\includegraphics[height=.2\textwidth,page=3]{BP-gadgets} &
		\includegraphics[height=.2\textwidth,page=2]{BP-gadgets} &
		\includegraphics[height=.2\textwidth,page=5]{BP-gadgets} &
		\includegraphics[height=.2\textwidth,page=4]{BP-gadgets}
		\\
		{\protect \BPi-BP1} & {\protect \BPii-BP2} & {\protect \BBii-BB2} & {\protect \BPiii-BP3} & {\protect \BBiii-BB3}
	\end{tabular}
	\begin{tabular}{c c c c}
		\includegraphics[height=.2\textwidth,page=7]{BP-gadgets} &
		\includegraphics[height=.2\textwidth,page=6]{BP-gadgets} &
		\includegraphics[height=.2\textwidth,page=9]{BP-gadgets} &
		\includegraphics[height=.2\textwidth,page=8]{BP-gadgets}
		\\
		{\protect \BPiv-BP4} & {\protect \BBiv-BB4} & {\protect \BPv-BP5} & {\protect \BBv-BB5}
	\end{tabular}
	\caption{Embeddings of pertinent graphs of type \protect\BB-BB.}\label{fig:BPBB-gadgets}
\end{figure}

We now classify the embedding types for type~\protect\BP-BP and~\protect\BB-BB nodes; 
refer to~\cref{fig:BPBB} for a complete schematization and to~\cref{fig:BPBB-gadgets} for examples.

Suppose first that $\mu$ is of type~\BP-BP. There exist three possible options for Feature (F1); namely, by~\cref{le:structure-A-mu-BP-BB}, we have that $A({\mathcal E}_\mu)$ consists of either zero, or one, or two caterpillars. 

If $A({\mathcal E}_\mu)$ does not contain any caterpillar, then we say that ${\cal E}_{\mu}$ is of type \BPi-{\bf BP1}. Hence, ${\mathcal E}_\mu$ does not have any red outer or internal face (so Features (F2) and (F3) are uniquely determined). We further observe the following.

\begin{observation}\label{obs:bp1-black-path}
If $\pert{\mu}$ admits an embedding of type \BPi-BP1, then $H^-_{\mu}$ is a path composed of black vertices between the poles of $\mu$ whose internal vertices are not incident to any $rb$-trivial component.
\end{observation}

By \cref{obs:bp1-black-path}, if $\mu$ admits an embedding of type \BPi-BP1, then $H^-_{\mu}$ has a unique embedding. Therefore, to ease the description, in the following we also say that $\mu$ is of type \BPi-BP1.

If $A({\mathcal E}_\mu)$ consists of one caterpillar, then by~\cref{le:structure-A-mu-BP-BB} it contains one red outer face, say $r({\mathcal E}_\mu)$, and does not contain the other one (so Feature (F2) is uniquely determined). We distinguish two embedding types according to Feature (F3); namely, if $A({\mathcal E}_\mu)$ does not contain any internal red face, then we say that ${\mathcal E}_\mu$ is of type \BPii-{\bf BP2}, otherwise we say that ${\mathcal E}_\mu$ is of type \BPiv-{\bf BP4}. We remark that, by~\cref{le:structure-A-mu-BP-BB}, one end-vertex of the backbone $\mathcal B_{\mu}$ of $A({\mathcal E}_\mu)$ is $r({\mathcal E}_\mu)$. Further, if ${\mathcal E}_\mu$ is of type \BPiv-BP4, then the end-vertex of $\mathcal B_{\mu}$ different from $r({\mathcal E}_\mu)$, say $f^*$, is an internal face of ${\mathcal E}_\mu$ that does not share any leaf of $A({\mathcal E}_\mu)$ with the outer face of ${\mathcal E}_\mu$. Namely, by hypothesis we have $\mathcal B_{\mu}=(r({\mathcal E}_\mu),\dots,f,\dots,f^*)$, where $f$ is an internal red face of ${\mathcal E}_\mu$. If $f^*=r({\mathcal E}_\mu)$, then $A({\mathcal E}_\mu)$ contains a cycle, contradicting the fact that $A({\mathcal E}_\mu)$ is a caterpillar. Further, $f^*\neq \ell({\mathcal E}_\mu)$, given that $A({\mathcal E}_\mu)$ does not contain $\ell({\mathcal E}_\mu)$. It follows that $f^*$ is an internal face of ${\mathcal E}_\mu$.  If $f^*$ is adjacent to a leaf $w$ incident to $r({\mathcal E}_\mu)$, then $w$ is adjacent to $f^*$ and to $r({\mathcal E}_\mu)$ in $A({\mathcal E}_\mu)$, which contradicts the fact that $w$ is a leaf of $A({\mathcal E}_\mu)$. Finally, $f^*$ is not adjacent to a leaf $w$ incident to $\ell({\mathcal E}_\mu)$, as otherwise $\ell({\mathcal E}_\mu)$ would be red.


If $A({\mathcal E}_\mu)$ consists of two caterpillars, then by~\cref{le:structure-A-mu-BP-BB} we have that $A({\mathcal E}_\mu)$ contains both $\ell({\mathcal E}_\mu)$ and $r({\mathcal E}_\mu)$ (so Feature (F2) is uniquely determined). We distinguish two embedding types according to Feature (F3); namely, if $A({\mathcal E}_\mu)$ does not contain any internal red face, then we say that ${\mathcal E}_\mu$ is of type \BPiii-{\bf BP3}, otherwise we say that ${\mathcal E}_\mu$ is of type \BPv-{\bf BP5}. We remark that, by~\cref{le:structure-A-mu-BP-BB}, the backbone $\mathcal B^\ell_{\mu}$ of one caterpillar composing $A({\mathcal E}_\mu)$ starts at $\ell({\mathcal E}_\mu)$, while the backbone $\mathcal B^r_{\mu}$ of the other caterpillar composing $A({\mathcal E}_\mu)$ starts at $r({\mathcal E}_\mu)$. Further, again by~\cref{le:structure-A-mu-BP-BB}, either $\mathcal B^\ell_{\mu}$ or $\mathcal B^r_{\mu}$, say $\mathcal B^\ell_{\mu}$, is a single vertex. If ${\mathcal E}_\mu$ is of type \BPv-BP5, then it can be proved similarly as for type \BPiv-BP4 that the end-vertex of $\mathcal B^r_{\mu}$ different from $r({\mathcal E}_\mu)$, say $f^*$, is an internal face of ${\mathcal E}_\mu$ that does not share any leaf of $A({\mathcal E}_\mu)$ with the outer face of ${\mathcal E}_\mu$.



The above discussion leads to the following.

\begin{lemma} \label{le:relevant-in-extensible-BP}
	If $\mu$ is of type \BP-BP, any extensible embedding of $\pert{\mu}$ is of one of the types \BPi-BP1,~\BPii-BP2, \BPiii-BP3, \BPiv-BP4, and \BPv-BP5.
\end{lemma}

If $\mu$ is of type~\BB-BB, then similarly to the previous case we distinguish four embedding types for $\pert{\mu}$; we call them types {\bf \BBii-BB2}, $\dots$, {\bf \BBv-BB5} in analogy with the types \BPii-BP2, $\dots$, \BPv-BP5. We observe that there exists no analogous of the embedding type \BPi-BP1, given that $A({\mathcal E}_\mu)$ contains either one or two caterpillars, by~\cref{le:structure-A-mu-BP-BB}.

\begin{lemma} \label{le:relevant-in-extensible-BB}
	If $\mu$ is of type \BB-BB, any extensible embedding of $\pert{\mu}$ is of one of the types~\BBii-BB2, \BBiii-BB3, \BBiv-BB4, and \BBv-BB5.
\end{lemma}

\subsubsection{Type \protect\BF-BF Nodes}

Finally, we discuss the type~\BF-BF, in which both $u$ and $v$ are black and the path $P_{uv}$ does not belong to $\pert{\mu}$, while the path $P_{v}$ belongs to $\pert{\mu}$. We start with the following lemma on the structure of $A({\mathcal E}_\mu)$.

\begin{lemma}\label{le:structure-A-mu-BF}
	Suppose that $\mu$ is of type~\BF-BF. Then $A({\mathcal E}_\mu)$ is a caterpillar containing a path $Q_u$ between the red neighbor $u_\ell$ of $u$ incident to $\ell(\mathcal E_\mu)$ and the red neighbor $u_r$ of $u$ incident to $r(\mathcal E_\mu)$; further, $Q_u$ passes through all the internal faces of ${\mathcal E}_\mu$ incident to~$u$ and through no other face.
\end{lemma}

\begin{fullproof}
First note that all the neighbors of $u$ in $\pert{\mu}$ are red, by the definition of type~\BF-BF. It follows that $A({\mathcal E}_\mu)$ contains a path $Q_u$ between $u_\ell$ and $u_r$ passing through all the internal faces of ${\mathcal E}_\mu$ incident to~$u$ and through no other face. 

We now show that $A({\mathcal E}_\mu)$ consists of only one caterpillar. By \cref{le:structure-A-mu}, we have that $A({\mathcal E}_\mu)$ is composed of at most two caterpillars. Suppose, for a contradiction, that $A({\mathcal E}_\mu)$ consists of two caterpillars. Let $C$ be the caterpillar that contains $Q_u$, and let $C'$ be the other caterpillar. We have that $C'$ contains neither $u_\ell$ nor $u_r$. Moreover, it does not contain any outer face of ${\mathcal E}_\mu$, as otherwise it would contain  $u_\ell$ or $u_r$. It follows that $C'$ is also a connected component of $A(\mathcal E)$, a contradiction.
\end{fullproof}

If $Q_u$ contains at least one vertex corresponding to a red face of $\mathcal E_\mu$, then we denote by $f_\ell$ and $f_r$ the vertices of $Q_u$ adjacent to $u_\ell$ and $u_r$, respectively (possibly $f_\ell=f_r$). 

We now classify the embedding types for type~\protect\BF-BF nodes; refer to~\cref{fig:NodetypesBF} for a complete schematization and to~\cref{fig:NodetypesBF-gadgets} for examples. Observe that the classification with respect to Feature (F1) is unique, as stated by \cref{le:structure-A-mu-BF}. As for the case in which $\mu$ is of type \RF-RF, we have that Features (F2) and (F3) determine six possible types for $\mathcal E_{\mu}$, called~\BFNz-{\bf BFN0},~\BFNi-{\bf BFN1}, \hbox{\BFNii-{\bf BFN2}}, \hbox{\BFIz-{\bf BFI0}}, \hbox{\BFIi-{\bf BFI1}}, and \hbox{\BFIii-{\bf BFI2}}, where the third character indicates whether ${\cal E}_\mu$ contains an internal red face (I) or not (N), and the fourth character indicates the number of red outer faces of ${\cal E}_\mu$.

\begin{figure}[htb]\tabcolsep=6pt
	\centering
	\begin{tabular}{c c c c c c c}
		\includegraphics[height=.2\textwidth,page=15]{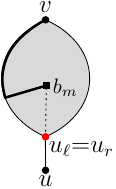}  &
		\includegraphics[height=.2\textwidth,page=14]{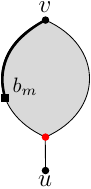}  &
		\includegraphics[height=.2\textwidth,page=16]{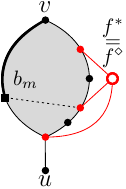}  &
		\includegraphics[height=.2\textwidth,page=17]{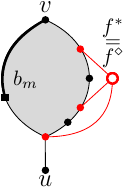}  &
		\includegraphics[height=.2\textwidth,page=19]{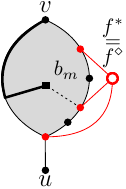} &
		\includegraphics[height=.2\textwidth,page=18]{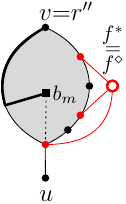}  &
		\includegraphics[height=.2\textwidth,page=20]{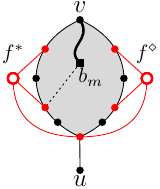} \\
		\protect\BFNza-BFN0A & \protect\BFNzb-BFN0B & \protect\BFNia-BFN1A & \protect\BFNib-BFN1B & \protect\BFNic-BFN1C & \protect\BFNid-BFN1D  & \protect\BFNii-BFN2
	\end{tabular}
	\begin{tabular}{c c}
		\includegraphics[height=.2\textwidth,page=21]{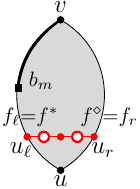}  &
		\includegraphics[height=.2\textwidth,page=22]{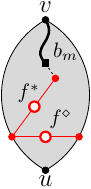}  \\
		\protect\BFIza-BFI0A & \protect\BFIzb-BFI0B
	\end{tabular}
	\begin{tabular}{c c c c c c c}
		\includegraphics[height=.2\textwidth,page=23]{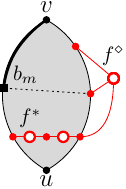} &
		\includegraphics[height=.2\textwidth,page=24]{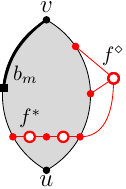} &
		\includegraphics[height=.2\textwidth,page=25]{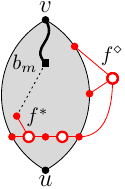} &
		\includegraphics[height=.2\textwidth,page=26]{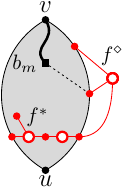} &
		\includegraphics[height=.2\textwidth,page=27]{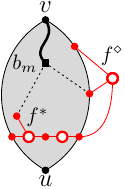} &
		\includegraphics[height=.2\textwidth,page=28]{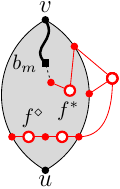} &
		\includegraphics[height=.2\textwidth,page=29]{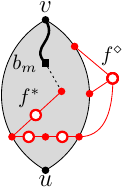} \\
		\protect\BFIia-BFI1A & \protect\BFIib-BFI1B  & \protect\BFIic-BFI1C & \protect\BFIid-BFI1D & \protect\BFIie-BFI1E & \protect\BFIif-BFI1F & \protect\BFIig-BFI1G
	\end{tabular}
		\begin{tabular}{c c}
		\includegraphics[height=.2\textwidth,page=30]{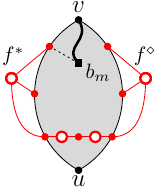}  &
		\includegraphics[height=.2\textwidth,page=31]{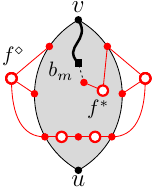}  \\
		\protect\BFIiia-BFI2A & \protect\BFIiib-BFI2B
	\end{tabular}
	\caption{Embedding types for type~\protect\BF-BF nodes, without taking into account the existence of undecided $rb$-trivial components incident to $u$.}\label{fig:NodetypesBF}
\end{figure}

\begin{figure}[htb]\tabcolsep=6pt
	\centering
	\begin{tabular}{c c c c c c c}
		\includegraphics[height=.2\textwidth,page=15]{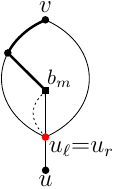}  &
		\includegraphics[height=.2\textwidth,page=14]{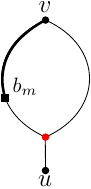}  &
		\includegraphics[height=.2\textwidth,page=16]{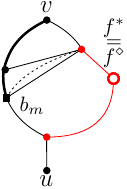}  &
		\includegraphics[height=.2\textwidth,page=17]{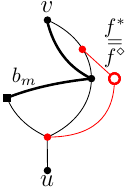}  &
		\includegraphics[height=.2\textwidth,page=19]{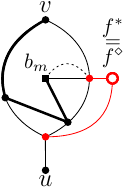} &
		\includegraphics[height=.2\textwidth,page=18]{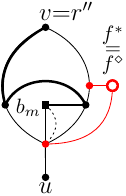}  &
		\includegraphics[height=.2\textwidth,page=20]{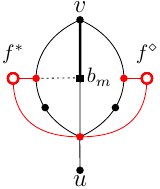} \\
		\BFNza-BFN0A & \BFNzb-BFN0B & \BFNia-BFN1A & \BFNib-BFN1B & \BFNic-BFN1C & \BFNid-BFN1D  & \BFNii-BFN2
	\end{tabular}
	\begin{tabular}{c c}
		\includegraphics[height=.2\textwidth,page=21]{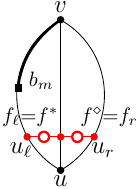}  &
		\includegraphics[height=.2\textwidth,page=22]{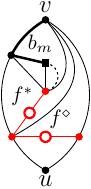}  \\
		\BFIza-BFI0A & \BFIzb-BFI0B
	\end{tabular}
	\begin{tabular}{c c c c c c c}
		\includegraphics[height=.2\textwidth,page=23]{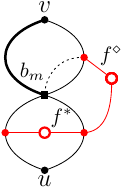} &
		\includegraphics[height=.2\textwidth,page=24]{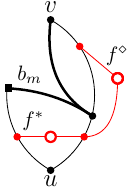} &
		\includegraphics[height=.2\textwidth,page=25]{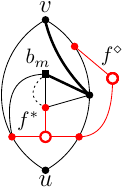} &
		\includegraphics[height=.2\textwidth,page=26]{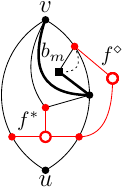} &
		\includegraphics[height=.2\textwidth,page=27]{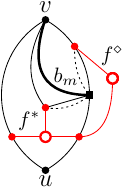} &
		\includegraphics[height=.2\textwidth,page=28]{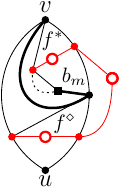} &
		\includegraphics[height=.2\textwidth,page=29]{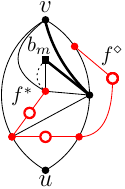} \\
		\BFIia-BFI1A & \BFIib-BFI1B  & \BFIic-BFI1C & \BFIid-BFI1D & \BFIie-BFI1E & \BFIif-BFI1F & \BFIig-BFI1G
	\end{tabular}
	\begin{tabular}{c c}
		\includegraphics[height=.2\textwidth,page=30]{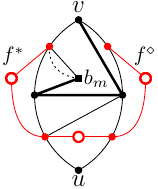}  &
		\includegraphics[height=.2\textwidth,page=31]{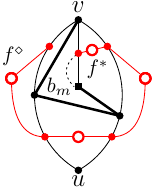}  \\
		\BFIiia-BFI2A & \BFIiib-BFI2B
	\end{tabular}
	\caption{Embeddings of pertinent graphs of type \protect\BF-BF.}\label{fig:NodetypesBF-gadgets}
\end{figure}

%
%

We now refine the above classification. Suppose first that $\mathcal E_\mu$ contains no internal red face, which implies that it is of type \BFNz-BFN0,~\BFNi-BFN1, or \BFNii-BFN2. We observe that, in this case, the path $Q_u$ of \cref{le:structure-A-mu-BF} consists of only one vertex $u_r = u_\ell$, since $\mathcal E_\mu$ does not contain any internal red face. Thus, $\pert{\mu}^-$ consists of an edge between $u$ and $u_r = u_\ell$, and of the pertinent $\pert{\nu}^-$ of a child $\nu$ of $\mu$ in $\mathcal{T}$ whose poles are $v$ and $u_r = u_\ell$. Note that $\nu$ is of type \RF-RF, and that the embedding of $\pert{\nu}$ that is contained in $\mathcal E_\mu$ is of type either \RFNz-RFN0,~\RFNi-RFN1, or \RFNii-RFN2. Hence, we refine the classification for these types by defining the types \BFNza-{\bf BFN0A}, \BFNzb-{\bf BFN0B},~\BFNia-{\bf BFN1A}, \BFNib-{\bf BFN1B}, \BFNic-{\bf BFN1C}, \BFNid-{\bf BFN1D}, and \BFNii-{\bf BFN2}, based on the type of the embedding of $\pert{\nu}$.

In order to refine the classification for the case in which $\mathcal E_\mu$ contains at least one internal red face, we are going to exploit the following lemma (which, in fact, holds in the presence of any, possibly outer, red face). Denote by $\mathcal B_\mu$ the backbone of the caterpillar $A(\mathcal E_\mu)$. 

\begin{lemma}\label{le:bf-endvertex-backbone}
	Suppose that $\mathcal E_\mu$ contains a red face and let $f^*$ and $f^\diamond$ be the end-vertices of the backbone $\mathcal B_\mu$ of the caterpillar $A(\mathcal E_\mu)$. 
	Then, up to a relabeling of $f^*$ and $f^\diamond$, we have that:
	\begin{enumerate}
	\item \label{le:bf-endvertex-backbone-bm-shares} $b_m$ shares a face with a leaf of $A(\mathcal E_\mu)$ adjacent to $f^*$; and
	\item \label{le:bf-endvertex-backbone-endvertex} $f^\diamond$ is either an outer face of ${\mathcal E}_\mu$ or is an internal face of ${\mathcal E}_\mu$ incident to $u_{\ell}$ (to $u_r$); in the latter case, $u_{\ell}$ (resp.\ $u_r$) is a leaf of $A({\mathcal E}_\mu)$.
	\end{enumerate} 
\end{lemma}

\begin{proof}
Suppose first that $\mathcal E_\mu$ contains no internal red face. By \cref{le:characterization-triconnected-onesidefixed}, there exists a leaf $r''$ of $A(\mathcal E)$ that shares a face with $b_m$. If $b_m$ is an internal vertex of $\mathcal E_\mu$, then $r''$ is a vertex of $\pert{\mu}$. By \cref{le:internal-red-vertex} and by the assumption that $\mathcal E_\mu$ contains no internal red face, we have that $r''$ is incident to the outer face of $\mathcal E_\mu$. Let $f^*$ be the red outer face of $\mathcal E_\mu$ that is incident to $r''$. If $\mathcal E_\mu$ is of type~\BFNi-BFN1 (in particular, of type~\BFNic-BFN1C or~\BFNid-BFN1D, given that $b_m$ is internal), then the statement follows with $f^\diamond=f^*$. If $\mathcal E_\mu$ is of type~\BFNii-BFN2, then the statement follows by letting $f^\diamond$ be the red outer face of $\mathcal E_\mu$ that is different from $f^*$. Suppose now that $b_m$ is incident to an outer face of $\mathcal E_\mu$. If this outer face is red (hence $\mathcal E_\mu$ is of type~\BFNic-BFN1C, or~\BFNid-BFN1D, or~\BFNii-BFN2), then the proof is completed as the in the case in which $b_m$ is an internal vertex of $\mathcal E_\mu$, where the leaf used in place of $r''$ is any leaf of $A(\mathcal E_\mu)$ that is incident to the red outer face $b_m$ is incident to. If $b_m$ is incident to an outer face of $\mathcal E_\mu$ which is not red (hence $\mathcal E_\mu$ is of type~\BFNia-BFN1A or~\BFNib-BFN1B), say $\ell(\mathcal E_\mu)$, then the statement follows with $f^*=f^\diamond=r(\mathcal E_\mu)$ and with $u_\ell=u_r$ as the leaf of $A(\mathcal E_\mu)$ that shares a face with $b_m$ and that is adjacent to $f^*$ in $A(\mathcal E_\mu)$.

In the following, suppose that $\mathcal E_\mu$ contains internal red faces. Consider an end-vertex of $A({\mathcal E}_\mu)$, say $f^*$. We say that $f^*$ is a \emph{bad face} if it is an internal face of ${\mathcal E}_\mu$ and, if it is incident to $u_{\ell}$ (to $u_r$), then $u_{\ell}$ (resp.\ $u_r$) is not a leaf of $A({\mathcal E}_\mu)$. Suppose that $f^*$ is a bad face. Then the following statements hold.

\begin{enumerate}[(S1)]
\item {\em $f^*$ is an end-vertex of the backbone $\mathcal{B}$ of $A({\mathcal E})$.} Namely, if $f^*$ is not an end-vertex of $\mathcal{B}$, then there exists a leaf $r^* \neq u_\ell, u_r$ of $A({\mathcal E}_\mu)$ adjacent to $f^*$ that belongs to $\mathcal{B}$. Let $f^+$ be the face other than $f^*$ that is adjacent to $r^*$ in $\mathcal{B}$. Note that $f^+$ is red in ${\mathcal E}$ but not in ${\mathcal E}_\mu$; thus, $f^+$ is one of the outer faces of ${\mathcal E}_\mu$, say $\ell({\mathcal E}_\mu)$. This implies that $r^*$ is incident to $\ell({\mathcal E}_\mu)$. Since $r^* \neq u_\ell$, it follows that $f^+=\ell({\mathcal E}_\mu)$ is red in ${\mathcal E}_\mu$, a contradiction. 
\item {\em Every leaf $w$ of $A({\mathcal E})$ adjacent to $f^*$ is an internal vertex of ${\mathcal E}_\mu$.} Namely, suppose for a contradiction that $w$ is incident to, say, the left outer face of ${\mathcal E}_\mu$; then $\ell({\mathcal E}_\mu)$ is red (given that both $u_\ell$ and $w$ are incident to it), and hence $w$ is adjacent to both $\ell({\mathcal E}_\mu)$ and $f^*$, hence it is not a leaf, a contradiction.
\item {\em No leaf $r'$ of $A({\mathcal E})$ adjacent to $f^*$ shares a face $f$ with $b_1$.} Namely, suppose for a contradiction that such a leaf $r'$ exists. By (S2), $r'$ is an internal vertex of ${\mathcal E}_\mu$. Thus, $b_1$ belongs to $\pert{\mu}$, which implies that $b_1$ coincides with $u$ (by \cref{le:rooting-at-b1b2}) and hence that $f$ belongs to $Q_u$; since $r'$ is a leaf of $A({\mathcal E}_\mu)$ adjacent to $f^{*}$ and since all the faces of ${\mathcal E}_\mu$ belonging to $Q_u$ are red, it follows that $f=f^*$. However, since $f^*$ is an end-vertex of $\mathcal B_\mu$, we have $f^*=f_{\ell}$ or $f^*=f_r$; this implies that $u_\ell$ or $u_r$ is a leaf, a contradiction. 
\item {\em A leaf $r''$ of $A({\mathcal E})$ adjacent to $f^*$ shares a face $f$ with $b_m$.} Namely, by (S1) we have that $f^*$ is an end-vertex of the backbone $\mathcal{B}$ of $A({\mathcal E})$; then the statement follows by \cref{condition:endvertices} of \cref{le:characterization-triconnected-onesidefixed} and (S3).  
\end{enumerate}

We prove Properties~1 and~2 of the lemma's statement. 
We first show that at least one of $f^\diamond$ and $f^*$ is not a bad face. 
Namely, suppose for a contradiction that both such faces are  bad. By (S1), $f^\diamond$ and $f^*$ are the end-vertices of $\mathcal B$. By (S3), none of them is adjacent to a leaf of $A(\mathcal E)$ that shares a face with $b_1$. Thus, \cref{condition:endvertices} of \cref{le:characterization-triconnected-onesidefixed} is not satisfied with respect to $b_1$, a contradiction.
Next, we discuss the case in which exactly one of $f^\diamond$ and $f^*$ is a bad face, say $f^*$. By (S4), $f^*$ is adjacent to a leaf of $A(\mathcal E_\mu)$ sharing a face with $b_m$. Therefore, since $f^\diamond$ is not bad, it follows that $f^*$ and $f^\diamond$ satisfy Properties~1 and~2 of the lemma's statement.
Finally, we discuss the case in which none of $f^\diamond$ and $f^*$ is a bad face. 
First, consider the case in which $b_m$ is incident to one of the outer faces of $\mathcal E_\mu$, say $\ell(\mathcal E_\mu)$. Since none of $f^\diamond$ and $f^*$ is a bad face, it follows that one of $f^\diamond$ and $f^*$, say $f^*$, either coincides with $\ell(\mathcal E_\mu)$ or it holds that $u_\ell$ is a leaf of $A(\mathcal E_\mu)$ adjacent to $f^*$. In both cases, there exists a leaf of $A(\mathcal E_\mu)$ adjacent to $f^*$ that shares the outer face $\ell(\mathcal E_\mu)$ with $b_m$.
Therefore, since $f^\diamond$ is not a bad face, it follows that $f^*$ and $f^\diamond$ satisfy Properties~1 and~2 of the lemma's statement.
Second, consider the case in which $b_m$ is an internal vertex of $\mathcal E_\mu$.  Recall that, by  \cref{condition:endvertices} of \cref{le:characterization-triconnected-onesidefixed}, there exists a leaf $r''$ of $A(\mathcal E)$ that is adjacent to an end-vertex of $\mathcal B$ and that shares a face $f$ with $b_m$. Since $b_m$ is an internal vertex of $\mathcal E_\mu$, it follows that $r''$ belongs to $\pert{\mu}$ and that $f$ is an internal face of $\mathcal E_\mu$. Since $r''$ is adjacent to an end-vertex of $\mathcal B$, it cannot be adjacent to an internal vertex of $\mathcal B_\mu$.
Hence, $r''$  is a leaf of $A(\mathcal E_\mu)$ adjacent to either $f^\diamond$ and $f^*$, say $f^*$, sharing the face $f$ with $b_m$. Therefore, since $f^\diamond$ is not bad, it follows that $f^*$ and $f^\diamond$ satisfy Properties~1 and~2 of the lemma's statement. This concludes the proof.
\end{proof}

We are now ready to refine the classification for the case in which $\mathcal E_\mu$ contains at least one internal red face. 

We start with type \BFIz-BFI0. If $b_m$ is incident to one of $\ell(\mathcal E_\mu)$ or $r(\mathcal E_\mu)$, say $\ell(\mathcal E_\mu)$, and $f^*$ and $f^\diamond$ both belong to $Q_u$ (that is, the backbone $\mathcal B_\mu$ of $A(\mathcal E_\mu)$ is the subpath of $Q_u$ between $f_\ell$ and $f_r$), then we say that $\mathcal E_\mu$ is of {\bf type~\BFIza-BFI0A}. Otherwise, we say that $\mathcal E_\mu$ is of {\bf type~\BFIzb-BFI0B}.

We now deal with type \BFIi-BFI1. In this case, only one of the two outer faces of $\mathcal E_\mu$ is red, say $r(\mathcal E_\mu)$; then $\mathcal B_\mu$ contains the edge $(r(\mathcal E_\mu),u_r)$ and, in case $Q_u$ contains at least one vertex corresponding to a red face of $\mathcal E_\mu$, the subpath of $Q_u$ between $u_r$ and $f_\ell$. 

We first discuss the case in which $\mathcal B_\mu$ does not contain any other vertex; since $\mathcal E_\mu$ contains internal red faces, it follows that $Q_u$ contains at least one vertex corresponding to a red face of~$\mathcal E_\mu$. Further, we can assume that $f^\diamond=r(\mathcal E_\mu)$ and $f^*=f_\ell$. Suppose that $b_m$ is incident to $\ell(\mathcal E_\mu)$. Then we say that $\mathcal E_\mu$ is of {\bf type~\BFIia-BFI1A} if $b_m$ shares a face with a leaf of $A(\mathcal E_\mu)$ incident to $r(\mathcal E_\mu)$, and it is of {\bf type~\BFIib-BFI1B} otherwise. Suppose next that $b_m$ is not incident to $\ell(\mathcal E_\mu)$; by \cref{le:bf-endvertex-backbone}, we have that $b_m$ shares a face with a leaf of $A(\mathcal E_\mu)$ adjacent to $f^*$ or $f^\diamond$. Then we say that $\mathcal E_\mu$ is of {\bf type~\BFIic-BFI1C} or {\bf type~\BFIid-BFI1D} if $b_m$ only shares a face with a leaf of $A(\mathcal E_\mu)$ incident to
$f^*=f_\ell$ or incident to $f^\diamond=r(\mathcal E_\mu)$, respectively, while  $\mathcal E_\mu$ is of {\bf type~\BFIie-BFI1E} if $b_m$ shares faces both with a leaf of $A(\mathcal E_\mu)$ incident to $f^*=f_\ell$ and with a leaf of $A(\mathcal E_\mu)$ incident to $f^\diamond=r(\mathcal E_\mu)$. 

Next, we discuss the case in which $\mathcal B_\mu$ contains vertices different from $r(\mathcal E_\mu)$ and not belonging to $Q_u$. If there exists a red vertex different from $u_r$ that is incident to $r(\mathcal E_\mu)$ and belongs to $\mathcal B_\mu$, then we say that $\mathcal E_\mu$ is of {\bf type~\BFIif-BFI1F}, otherwise we say that $\mathcal E_\mu$ is of {\bf type~\BFIig-BFI1G}. 
Observe that, by \cref{le:structure-A-mu-BF}, if $\mathcal E_\mu$ is of type~\BFIif-BFI1F, then $f_\ell$ is an end-vertex of $\mathcal B_\mu$ and $u_\ell$ is a leaf of $A(\mathcal E_\mu)$ adjacent to it; on the other hand, if $\mathcal E_\mu$ is of type~\BFIig-BFI1G, then $r(\mathcal E_\mu)$ is an end-vertex of $\mathcal B_\mu$ and $u_\ell$ is a vertex of $\mathcal B_\mu$.

We now consider the case in which $\mathcal E_\mu$ is of type \BFIii-BFI2. Then, $\mathcal B_\mu$ contains the edge $(r(\mathcal E_\mu),u_r)$, the path $Q_u$, and the edge $(u_\ell, \ell(\mathcal E_\mu))$. If $\mathcal B_\mu$ does not contain any other vertex, then we say that $\mathcal E_\mu$ is of {\bf type \BFIiia-BFI2A}; otherwise, it is of {\bf type \BFIiib-BFI2B}. 
Observe that, by \cref{le:structure-A-mu-BF}, if $\mathcal E_\mu$ is of type~\BFIiib-BFI2B, then one of the outer faces of $\mathcal E_\mu$ is an end-vertex of $\mathcal B_\mu$, while the other one is not.

%
%


The above discussion leads to the following.

\begin{lemma} \label{le:relevant-in-extensible-BF}
	If $\mu$ is of type \BF-BF, any extensible embedding of $\pert{\mu}$ is one of the types~\BFNza-BFN0A, \BFNzb-BFN0B, \BFNia-BFN1A, \BFNib-BFN1B, \BFNic-BFN1C,  \BFNid-BFN1D, \BFNii-BFN2, 
		\BFIza-BFI0A, \BFIzb-BFI0B,
		\BFIia-BFI1A, \BFIib-BFI1B, \BFIic-BFI1C, \BFIid-BFI1D, \BFIie-BFI1E, \BFIif-BFI1F, \BFIig-BFI1G,
		\BFIiia-BFI2A, and~\BFIiib-BFI2B.
\end{lemma}

This concludes our classification of the extensible embeddings of  $\pert{\mu}$.
By \cref{le:structure-A-mu-RF,le:relevant-in-extensible-RE,le:structure-A-mu-BE,le:structure-A-mu-BP-BB,le:structure-A-mu-BF}, we have the following.

\begin{observation}\label{obs:two-components}
If $A(\mathcal{E}_{\mu})$ contains two connected components, then $\mu$ is of type~\BP-BP or type~\BB-BB.
\end{observation}

\subsection{Handling Relevant Embeddings}\label{sse:handling-rel-emb}

Consider a node $\mu$ of $\mathcal{T}$ different from $\rho$. We begin with the following definition.

\begin{definition} \label{def:relevant}
An embedding $\mathcal E_{\mu}$ of $\pert{\mu}$ is \emph{ relevant} if it has a type defined in \cref{sse:embedding-classification}. 
\end{definition}

In view of the discussion of \cref{sse:embedding-classification}, we have the following.

\begin{lemma} \label{le:relevant-in-extensible}
Every extensible embedding of $\pert{\mu}$ is relevant.
\end{lemma}

\begin{fullproof}
The lemma descends from \cref{le:relevant-in-extensible-RE,le:relevant-in-extensible-RF,le:relevant-in-extensible-BE,le:relevant-in-extensible-BP,le:relevant-in-extensible-BF}.
\end{fullproof}

A property we are going to use is that relevant embeddings of $\pert{\mu}$ are composed of relevant embeddings of the pertinent graphs of the children of $\mu$. Together with \cref{le:relevant-in-extensible}, this allows us to only focus on the relevant embeddings of the pertinent graph of each node of the SPQR-tree.

\begin{lemma} \label{le:relevant-relevant}
Consider any relevant embedding $\mathcal E_{\mu}$ of $\pert{\mu}$. Then, for any child $\nu$ of $\mu$, the restriction $\mathcal E_{\nu}$ of $\mathcal E_{\mu}$ to $\pert{\nu}$ defines a relevant embedding of $\pert{\nu}$.
\end{lemma}

\begin{fullproof}
If $\mathcal E_{\mu}$ is extensible, then $\mathcal E_{\nu}$ is extensible as well. Hence, by \cref{le:relevant-in-extensible}, we have that $\mathcal E_{\nu}$ is a relevant embedding of $\pert{\nu}$. 

If $\mathcal E_{\mu}$ is not extensible then, in order to apply a similar argument, we change $\rest{\mu}$ so that~$\mathcal E_{\mu}$ is extensible. More precisely, we define a new instance $\langle G',\pi'_b\rangle$ of the \btpbefP problem such that: (i) the black saturation $H'$ of $\langle G',\pi'_b\rangle$ is an $rb$-augmented component; (ii) $H'$ contains $\pert{\mu}$ as a subgraph and there exists no edge with an end-vertex in $\pert{\mu}-\{u,v\}$ and an end-vertex in $H'-\pert{\mu}$; and (iii) there exists a neat embedding $\mathcal E'$ of $H'$ that extends $\mathcal E_{\mu}$. This implies that $\mathcal E_{\mu}$ is extensible (with respect to the instance $\langle G',\pi'_b\rangle$). Hence, the embedding $\mathcal E_{\nu}$ of $\pert{\nu}$ is also extensible (still with respect to $\langle G',\pi'_b\rangle$). By~\cref{le:relevant-in-extensible}, we have that $\mathcal E_{\nu}$ is a relevant embedding of $\pert{\nu}$. In  particular, since $\mathcal E_\mu$ is relevant,
each $rb$-trivial component $(r,b)$ of $\pert{\mu}$, where $b$ is black and $b$ is different from the poles of $\mu$, lies in a face~of~$\mathcal E_\mu$ that corresponds to a face of the embedding (determined by $\mathcal E_\mu$) of the skeleton of the proper allocation node of $b$.

We construct $H'$ as follows (the definition of the instance $\langle G',\pi'_b\rangle$ is implied by the construction of $H'$). We initialize $H'=H$. We then remove from $H'$ the vertices and edges not in $\pert{\mu}$, and we insert vertices and edges according to the following case distinction. 

\begin{figure}[tb!]
	\centering
	\subfloat[Case~1]{
		\includegraphics[height=.2\textwidth,page=1]{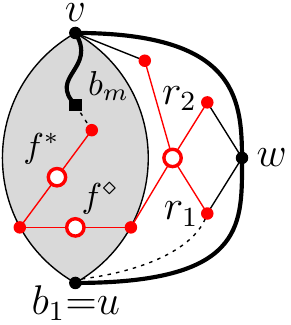} 
		\label{fig:case1-extended}
	}
	\subfloat[Case~2]{
		\includegraphics[height=.2\textwidth,page=2]{extended} 
		\label{fig:case2-extended}
	}
	\subfloat[Case~3]{
		\includegraphics[height=.2\textwidth,page=3]{extended} 
		\label{fig:case3-extended}
	}
\end{figure}
		
\emph{Case 1:} If $\mu$ is of type~\BE-BE or~\BF-BF, then we add to $H'$ a path $(u,w,v)$, where $w$ is a black vertex, and two red vertices $r_1$ and $r_2$ that are only adjacent to $w$; see \cref{fig:case1-extended}. Now the black vertices of $H'$ induce a path in which $b_1=u$; further, $H'$ contains at least three black vertices, namely $u$, $w$, and $v$, and at least three red vertices, namely $r_1$, $r_2$, and any neighbor of $u$ in $\pert{\mu}$. A neat embedding $\mathcal E'$ of $H'$ is constructed starting from $\mathcal E_{\mu}$ by embedding the path $(u,w,v)$ in such a way that the undecided $rb$-trivial components of $\mathcal E_{\mu}$, if any, as well as the $rb$-trivial components $(w,r_1)$ and $(w,r_2)$, are all incident to the same outer face $x(\mathcal E_\mu)$, with $x \in \{r,\ell\}$, of $\mathcal{E}_\mu$. 
In particular, if $\mu$ is of type~\BE-BE or \BFNz-BFN0, then we set $x(\mathcal E_\mu)$ to any of $\ell(\mathcal E_\mu)$ or $r(\mathcal E_\mu)$. Otherwise, we select $x(\mathcal E_\mu)$ as follows. Let $f^*$ and $f^\diamond$ be the end-vertices of the backbone $\mathcal B_\mu$ of the caterpillar $A(\mathcal E_\mu)$. Recall that, by \cref{le:bf-endvertex-backbone}, up to a relabeling of $f^*$ and $f^\diamond$, we have that
$f^\diamond$ is either an outer face of ${\mathcal E}_\mu$ (then we set $x(\mathcal E_\mu) = f^\diamond$) or is an internal face of ${\mathcal E}_\mu$ incident to $u_{\ell}$ (to $u_r$), which is a leaf of $A({\mathcal E}_\mu)$ (then we set $x = \ell$ or $x=r$ depending on whether $u_{\ell}$ or $u_{r}$ is a leaf of $A({\mathcal E}_\mu)$ adjacent to $f^\diamond$, respectively).
The SPQR-tree of $H'$ is rooted at the edge $(b_1,w)$. The proper allocation node of both $u=b_1$ and $w$ is the root Q-node, while the proper allocation node of $v$ is the S-node child of the root Q-node. Hence, the faces of $\mathcal{E}'$ in which $r_1$, $r_2$, and the undecided $rb$-trivial component of $\mathcal E_\mu$ have been embedded correspond to faces of the unique embedding of the skeleton of the proper allocation node of $u$, $v$ and $w$.

\emph{Case 2:} If $\mu$ is of type~\RE-RE or~\RF-RF, then we add to $H'$ a path $(u,w,w',v)$, where $w$ and $w'$ are black vertices, and two red vertices $r_1$ and $r_2$ that are only adjacent to $w'$; see \cref{fig:case2-extended}.
Now, the black vertices of $H'$ induce a path in which $b_1=w'$; further, $H'$ contains at least three black vertices, namely $u$, $w$, and $w'$, and at least three red vertices, namely $v$, $r_1$, and $r_2$. A neat embedding $\mathcal E'$ of $H'$ is constructed starting from $\mathcal E_{\mu}$ by embedding the path $(u,w,w',v)$ in such a way that the undecided $rb$-trivial components of $\mathcal E_{\mu}$, if any, as well as the $rb$-trivial components $(w',r_1)$ and $(w',r_2)$, are all incident to the same outer face $x(\mathcal E_\mu)$, with $x \in \{r,\ell\}$, of $\mathcal{E}_\mu$.
In particular, if $\mu$ is of type~\RE-RE, \RFNz-RFN0, or \RFIz-RFI0, then we set $x(\mathcal E_\mu)$ to any of $\ell(\mathcal E_\mu)$ or $r(\mathcal E_\mu)$. Further, if $\mu$ is of type~\RFNi-RFN1 or \RFIi-RFI1, then we set $x(\mathcal E_\mu)$ to the one of $\ell(\mathcal E_\mu)$ or $r(\mathcal E_\mu)$ that is red in $\mathcal E_\mu$. Suppose next that $\mu$ is of type~\RFNii-RFN2. Then, by \cref{le:structure-A-mu-RF}, we have that $\ell(\mathcal{E}_\mu)$ and $r(\mathcal{E}_\mu)$ are the end-vertices of $\mathcal{B}_\mu$. If $b_m$ is adjacent to a red vertex incident to $\ell(\mathcal{E}_\mu)$, then we set $x=r$, otherwise $b_m$ is adjacent to a red vertex incident to $r(\mathcal{E}_\mu)$, since $b_m$ has at least one red neighbor in $\pert{\mu}$, and we set $x=\ell$. Finally, suppose that $\mu$ is of type~\RFIii-RFI2. Then, by \cref{le:rf-endvertex-backbone}, we have that $b_m$ shares a face with a red vertex incident to an end-vertex of $\mathcal{B}_\mu$ that corresponds to an internal face of $\mathcal{E}_\mu$, while the other end-vertex of $\mathcal{B}_\mu$ corresponds to an outer face of $\mathcal{E}_\mu$, say $\ell(\mathcal{E}_\mu)$. Then, we set $x=\ell$.
The SPQR-tree of $H'$ is rooted at the edge $(w',w)$. The proper allocation node of both $w'=b_1$ and $w$ is the root Q-node, while the proper allocation node of both $u$ and $v$ is the S-node child of the root Q-node. Hence, the faces of $\mathcal{E}'$ in which $r_1$, $r_2$, and the undecided $rb$-trivial component of $\mathcal E_\mu$ have been embedded correspond to faces of the unique embedding of the skeleton of the proper allocation node of $w'$, $u$ and $v$.

\emph{Case 3:} If $\mu$ is of type~\BP-BP or~\BB-BB, then we add to $H'$ a path $(u,w,r,v)$, where $w$ is a black vertex and $r$ is a red vertex, and two red vertices $r_1$ and $r_2$ that are only adjacent to~$w$; see \cref{fig:case3-extended}.
Now, the black vertices of $H'$ induce a path in which $b_1=w$; further, $H'$ contains at least three black vertices, namely $u$, $w$, and $v$, and at least three red vertices, namely $r$, $r_1$, and $r_2$. A neat embedding $\mathcal E'$ of $H'$ is constructed starting from $\mathcal E_{\mu}$ by embedding the path $(u,w,r,v)$ in such a way that the undecided $rb$-trivial components of $\mathcal E_{\mu}$, if any, as well as the $rb$-trivial components $(w,r_1)$ and $(w,r_2)$, are all incident to the same outer face $x(\mathcal E_\mu)$, with $x \in \{r,\ell\}$, of $\mathcal{E}_\mu$.
In particular, if $\mu$ is of type~\BPi-BP1 or ~\BPiii-BP3, then we set $x(\mathcal E_\mu)$ to any of $\ell(\mathcal E_\mu)$ or $r(\mathcal E_\mu)$. Also,  if $\mu$ is of type~\BPii-BP2, \BBii-BB2, \BPiv-BP4, or~\BBiv-BP4, then we set $x(\mathcal E_\mu)$  to be the unique red outer face of $\mathcal E_\mu$. Further,  if $\mu$ is of type~\BPv-BP5 or \BBv-BB5, then we set $x(\mathcal E_\mu)$  to be the outer face of $\mathcal E_\mu$ that forms the backbone of one of the two caterpillars of $A(\mathcal E_\mu)$. Finally,  if $\mu$ is of type~\BBiii-BB3, then we set $x(\mathcal E_\mu)$  to be the outer face of $\mathcal E_\mu$ that is not incident to any red vertex that shares a face with $b_m$.
The SPQR-tree of $H'$ is rooted at the edge $(w,u)$. The proper allocation node of both $w=b_1$ and $u$ is the root Q-node, while the proper allocation node of $v$ is the S-node child of the root Q-node. Hence, the faces of $\mathcal{E}'$ in which $r_1$, $r_2$, and the undecided $rb$-trivial component of $\mathcal E_\mu$ have been embedded correspond to faces of the unique embedding of the skeleton of the proper allocation node of $w$, $u$ and $v$.

In all the three considered cases, ${\mathcal E}'$ satisfies the conditions of \cref{le:characterization-triconnected-onesidefixed}. 
First, $A({\mathcal E}')$ is a caterpillar. In particular, in \emph{Case 3}, if $A((\mathcal E_\mu)$ is composed of two caterpillars, then they are joined into a single one by means of the vertex $r$; whereas in the other cases the backbone of $A({\mathcal E}')$ coincides with the one of $A((\mathcal E_\mu)$, except, possibly, for $x(\mathcal E_\mu)$.
Second, $x(\mathcal E_\mu)$ is always an end-vertex of the backbone of $A({\mathcal E}')$, and both $b_1$ and $r_1$ are incident to $x(\mathcal E_\mu)$, which implies that we can set $r'=r_1$.
Finally, $b_m$ shares a face with a leaf adjacent to the other end-vertex $f^*$ of the backbone of $A({\mathcal E}')$. In fact, $b_m$ belongs to $\pert{\mu}$ and $f^*$ is also an end-vertex of the backbone of $A(\mathcal E_\mu)$. Therefore, a leaf $r''$ of $A({\mathcal E}')$ incident to $f^*$ that shares a face with $b_m$ exists,~already,~in~$A(\mathcal E_\mu)$.
\end{fullproof}

We will exploit the following definition.

\begin{definition}\label{def:flipped}
	Let $\mathcal{E}_\mu$ be a relevant embedding of $\pert{\mu}$. We say that $\mathcal{E}_\mu$ is \emph{$x$-flipped}, with $x \in \{\ell, r\}$, if one of the following conditions holds:
	\begin{enumerate}[(a)]
		\item the outer face $x(\mathcal{E}_{\mu})$ is red, and $\mathcal{E}_{\mu}$ has one of the types \RFNi-RFN1, \RFIi-RFI1, \BPii-BP2, \BBii-BB2, \BPiv-BP4, \BBiv-BB4, \BFNi-BFN1, and \BFIi-BFI1; or
		\item the outer face $x(\mathcal{E}_{\mu})$ is the only vertex of the backbone of one of the two caterpillars composing $A(\mathcal{E}_{\mu})$, and $\mathcal{E}_{\mu}$ has one of the types \BPv-BP5 or \BBv-BB5; or
		\item the outer face $x(\mathcal{E}_{\mu})$ is an end-vertex of the backbone of $A(\mathcal{E}_{\mu})$, and $\mathcal{E}_{\mu}$ has one of the types \RFIii-RFI2 and \BFIiib-BFI2B; or
		\item the vertex $u_x$ is a leaf adjacent to an end-vertex of the backbone of $A(\mathcal{E}_{\mu})$, and $\mathcal{E}_{\mu}$ has the type \BFIzb-BFI0B; or
		\item the vertex $b_m$ is incident to the outer face $x(\mathcal{E}_{\mu})$, and $\mathcal{E}_{\mu}$ has one of the types \RFNzb-RFN0B, \BFNzb-BFN0B, and \BFIza-BFI0A; or
		\item there exists a red vertex incident to the outer face $x(\mathcal{E}_{\mu})$ that shares a face with $b_m$, and $\mathcal{E}_{\mu}$ has one of the types \RFNii-RFN2, \BBiii-BB3, \BFNii-BFN2, and \BFIiia-BFI2A;
		\item $\mathcal{E}_{\mu}$ has one of the other types, namely \RE-RE, \BE-BE, \RFNza-RFN0A, \RFIz-RFI0, \BPi-BP1, \BPiii-BP3, and \BFNza-BFN0A.
	\end{enumerate}
\end{definition}

We remark that, according to \cref{def:flipped}\blue{(g)}, an embedding $\mathcal E_\mu$ of $\pert{\mu}$ of type \RE-RE, \BE-BE, \RFNza-RFN0A, \RFIz-RFI0, \BPi-BP1, \BPiii-BP3, and \BFNza-BFN0A is both $r$- and $\ell$-flipped. This is also true if $\mathcal E_\mu$ is of type \BFIiia-BFI2A and $b_m$ shares a face both with a leaf of $A(\mathcal E_\mu)$ incident to $f^*$ and with a leaf of $A(\mathcal E_\mu)$ incident to $f^\diamond$.

\remove{
In the following we establish that the type of a relevant embedding is the only information needed to determine whether that embedding is extensible.

\begin{lemma}[Equivalence]\label{le:equivalence}
Let $\mathcal{E}_1$ and $\mathcal{E}_2$ be two relevant embeddings of $\pert{\mu}$ with the same type. Then $\mathcal{E}_1$ is extensible if and only if $\mathcal{E}_2$ is extensible.
\end{lemma}

Rather than proving \cref{le:equivalence}, we state and prove the following stronger version of it, which we will also exploit in order to improve the efficiency of our algorithm. 
}

Next, we present some definitions and tools that we will exploit to manipulate relevant embeddings.

\begin{definition}\label{def:replacement}
Let $u$ and $v$ be the poles of $\mu$. A \emph{replacement-graph} for $\pert{\mu}$ is a graph \mbox{$D$ such that:} 

\begin{enumerate}[{r.}1]
		\item \label{cond:replacement-uv} the vertex set of $D$ contains $u$ and $v$; 
		\item \label{cond:replacement-color} each vertex of $D$ is colored either red or black and the vertices $u$ and $v$ have the same color as in $H$; 
		\item \label{cond:replacement-rb-trivial} there exist no $rb$-trivial component incident to $u$ or $v$ in $D$; and
		
		
		\item \label{cond:replacement-rb-augmented} denote by $K$ the graph obtained from $H$ by removing the vertices and the edges of $\pert{\mu}$---except for $u$, $v$, and their incident $rb$-trivial components, if any---and by inserting $D$ in the resulting graph, while identifying the vertices $u$ and $v$ in the two graphs; then $K$ is an $rb$-augmented component. 
	\end{enumerate} 
\end{definition}

\begin{figure}[t!]
	\centering
	\subfloat[Graph $H$]{
		\includegraphics[scale=.9,page=1]{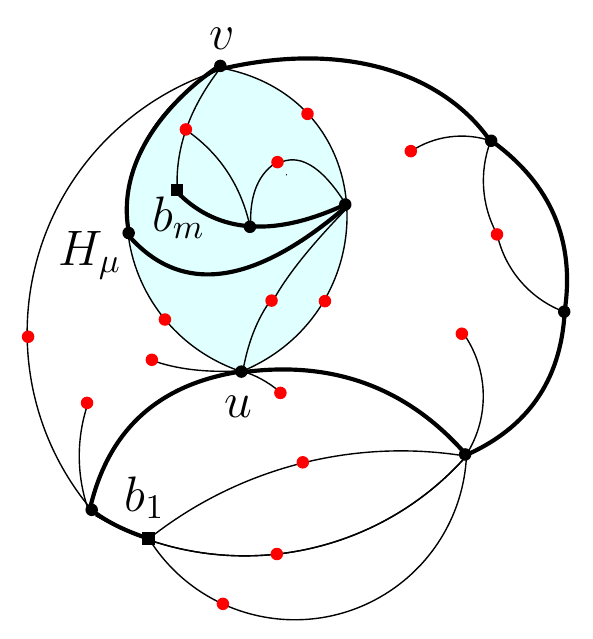} 
		\label{fig:replacement-a}
	}
	\hfil
	\subfloat[Replacement-graph $D$]{
	\includegraphics[scale=.9,page=5]{replacement} 
	\label{fig:replacement-b}
	}
\hfil
	\subfloat[Graph $K$]{
		\includegraphics[scale=.9,page=2]{replacement} 
		\label{fig:replacement-c}
	}
	\caption{Illustrations for the definitions of replacement-graph and embedding-replacement. 
	(a) A black saturation $H$ with a planar embedding $\mathcal E$; let $\mathcal E_\mu$ be the embedding of the graph $\pert{\mu}$ in the shaded-cyan region.
	(b) A replacement-graph $D$ for the graph $\pert{\mu}$ in (a) with embedding $\mathcal E_\nu$.
	(c) The graph $K$ obtained by an embedding-replacement of $\mathcal E_\mu$ with  $\mathcal E_\nu$ in $\mathcal E$.
	}
	\label{fig:replacement}
\end{figure}

\noindent
In other words, $K$ is the black saturation of a new instance of the \btpbefP problem, such that $\rest{\mu}$ 
coincides with $K - (D - \{u,v\})$. Let $\mathcal T_{K^-}$ be the SPQR-tree of the biconnected graph $K^-$ obtained from $K$ by removing the vertices of degree $1$ (i.e., its $rb$-trivial components). By rooting $\mathcal T_{K^-}$ at the Q-node corresponding to the edge $(b_1,b_2)$ (that is, at the same Q-node at which $\mathcal T$ is rooted), we have that $\mathcal T_{K^-}$ contains a node $\nu$ with poles $u$ and $v$. Note that $\mathcal T$ and $\mathcal T_{K^-}$ coincide, except for the subtrees rooted at $\mu$ and at $\nu$, respectively, and that the pertinent graph $K_{\nu}$ of $\nu$ coincides with $D$. Hence, in the following, by \emph{type} of the replacement-graph $D$, we mean the type (i.e., \RE-RE, \RF-RF, \BE-BE, \BP-BP, \BB-BB, or \BF-BF)  of the node $\nu$ in $\mathcal T_{K^-}$. Further, we extend the notions of \emph{embedding type} and of \emph{relevant embedding} to embeddings of replacement-graphs.

Let $\mathcal{E}_{\nu}$ be a relevant embedding of $K_{\nu}$ with the same type as $\mathcal{E}_{\mu}$. We say that $\mathcal{E}_{\nu}$ and $\mathcal{E}_{\mu}$ \emph{have the same flip} if they are either both $r$-flipped or both $\ell$-flipped.


Let $\mathcal E$ be an embedding of $H$ that extends $\mathcal{E}_{\mu}$, and let $\mathcal{E}_{\nu}$ be a relevant embedding of $K_{\nu}$ with the same type and the same flip as $\mathcal{E}_{\mu}$. An \emph{embedding-replacement of} $\mathcal{E}_{\mu}$ \emph{with} $\mathcal{E}_{\nu}$ \emph{in} $\mathcal E$ is an operation defined as follows; refer again to \cref{fig:replacement}. We delete from $\mathcal E$ the edges of $\pert{\mu}$ and the vertices of $\pert{\mu}$ different from $u$ and $v$ and from the $rb$-trivial components incident to such vertices, if any, and we insert $\mathcal{E}_{\nu}$ inside the region in which the deleted vertices and edges of $\pert{\mu}$ used to lie, identifying the vertices $u$ and $v$ of $\mathcal{E}_{\nu}$ and of $\mathcal E$. This is done in such a way that each $rb$-trivial component incident to a poles of $\mu$ 
that lie inside $x(\mathcal E_\mu)$ in $\mathcal E$ lies inside $x(\mathcal E_\nu)$ in the resulting embedding $\mathcal E'$ of $K$, with $x \in \{\ell,r\}$.
By construction, $\mathcal E'$ extends $\mathcal{E}_{\nu}$. 
Furthermore, each $rb$-trivial component of $H$ incident to $u$ or $v$, if any, lies in $\mathcal E'$ either in~$\ell(\mathcal E_\nu)$, or in $r(\mathcal E_\nu)$, or in a face (delimited by edges of $\rest{\mu}$) that is also a face of $\cal E$.

\remove{
We are now ready to state and prove the stronger version of \cref{le:equivalence}.

\begin{lemma}[Stronger Equivalence]\label{le:equivalence-stronger}
Let $\mathcal{E}_{\mu}$ be a relevant embedding of $\pert{\mu}$ and $\mathcal{E}_{\nu}$ be a relevant embedding of $K_{\nu}$ with the same type and the same flip as $\mathcal{E}_{\mu}$. Then $\mathcal{E}_{\mu}$ is extensible (with respect to $H$) if and only if $\mathcal{E}_{\nu}$ is extensible (with respect to $K$). 
\end{lemma}

\begin{fullproof}

\end{fullproof}
}

We conclude the section with the following. 

\begin{lemma}[Child Replacement]\label{le:child-replacement}
Let $\mu$ be a node of $\mathcal{T}$ and let $\lambda$ be a child of $\mu$ in $\mathcal{T}$.
Let $\mathcal{E}_\mu$ be an embedding of $\pert{\mu}$, and let $\mathcal{E}_{\lambda}$ be the embedding of $\pert{\lambda}$ such that $\mathcal{E}_\mu$ extends $\mathcal{E}_{\lambda}$. 
Also, let $D$ be a replacement-graph for $\pert{\lambda}$ and let $\mathcal{E}_{D}$ be an embedding of $D$ with the same type and the same flip as $\mathcal{E}_\lambda$.
Let $H_\mu^*$ be the graph obtained from $\pert{\mu}$ by replacing $\pert{\lambda}$ with $D$ and let $\mathcal{E}^*_\mu$ be the embedding of $H^*_\mu$ obtained from $\mathcal{E}_\mu$ by performing an embedding-replacement of $\mathcal{E}_{\lambda}$ with $\mathcal{E}_{D}$; refer to \cref{fig:child-replacement}.
Then $\mathcal{E}_\mu$ and $\mathcal{E}^*_\mu$ have the same type and the same flip.
\end{lemma}
\begin{fullproof}	
First, we observe that the type of $\mu$ after the replacement remains the same. In fact, 
the poles of $\mu$ are the same after the replacement, by \cref{def:replacement}.
There exists a path composed of black vertices in $H_\mu$ if and only if there exists a path composed of black vertices in $H^*_\mu$, due to the fact that $H_\lambda$ contains a black path composed of black vertices if and only if $D$ contains such a path.
Finally, $b_m$ is a pole or an non-pole vertex of $\mu$ before the replacement if and only if it is a pole or an non-pole vertex of $\mu$, respectively, due to the fact that the type of $\mathcal E_\lambda$ and $\mathcal E_D$ coincide.

According to the definitions of type and flip we have to prove that:

\begin{enumerate}
	\item Each connected component of $A(\mathcal E^*_\mu)$ is a caterpillars and the 
	number of components composing $A(\mathcal E_{\mu})$ and $A(\mathcal E^*_{\mu})$ is the same---Feature (F1);
	
	\item Face $x(\mathcal E_{\mu})$, with $x \in \{\ell,r\}$, belongs to $A(\mathcal E_{\mu})$ if and only if $x(\mathcal{E^*}_{\mu})$ belongs to $A(\mathcal E^*_{\mu})$---Feature (F2) and \cref{def:flipped} of $x$-flipped. Moreover, face $x(\mathcal E_{\mu})$ is an end-vertex of a backbone of $A(\mathcal E_{\mu})$ if and only if $x(\mathcal{E^*}_{\mu})$ is an end-vertex of a backbone of $A(\mathcal E^*_{\mu})$;
	
	\item Face $f_x(\mathcal E_{\mu})$, with $x \in \{\ell,r\}$, is an end-vertex of a backbone of $A(\mathcal E_{\mu})$ if and only if $f_x(\mathcal E^*_{\mu})$ is an end-vertex of a backbone of $A(\mathcal E^*_{\mu})$, where $f_x(\mathcal E_{\mu})$ and $f_x(\mathcal E^*_{\mu})$ are the faces the play the role of $f_x$ in $\mathcal E_{\mu}$ and $\mathcal E^*_{\mu}$, respectively. Note that these faces may only exist if $\mu$ is of type \BF-BF or \BE-BE-fat, and thus there exists at least one red vertex incident to each outer face of $\mathcal E_{\mu}$ and $\mathcal E^*_{\mu}$;
	
	\item The end-vertex of a backbone of $A(\mathcal E_{\mu})$ that is adjacent to the leaf $r''$ of $A(\mathcal E_{\mu})$ that shares a face with $b_m$ in $\mathcal E_{\mu}$ is $\ell(\mathcal E_{\mu})$, $r(\mathcal E_{\mu})$, $f_\ell(\mathcal E_{\mu})$, $f_r(\mathcal E_{\mu})$, or is none of such faces, if and only if the end-vertex of a backbone of $A(\mathcal E^*_{\mu})$ that is adjacent to the leaf $r''$ of $A(\mathcal E^*_{\mu})$ that shares a face with $b_m$ in $\mathcal E^*_{\mu}$ is $\ell(\mathcal E^*_{\mu})$, $r(\mathcal E^*_{\mu})$, $f_\ell(\mathcal E^*_{\mu})$, $f_r(\mathcal E^*_{\mu})$, or none of such faces, respectively. Note that the faces $f_x(\mathcal E_\mu)$ and $f_x(\mathcal E^*_{\mu})$, with $x \in \{\ell,r\}$ may exist only if $\mu$ is of type \BF-BF and \BE-BE-fat;
	

	\item The backbone of $A(\mathcal E_{\mu})$ contains at least an internal face of $\mathcal E_{\mu}$ if and only if the backbone of $A(\mathcal E^*_{\mu})$ contains at least an internal face of $\mathcal E^*_{\mu}$---Feature (F3);

	\item If $\mathcal E_\mu$ is of type \RFNi-RFN1, or \BFNi-BFN1, or \BFIi-BFI1, then $b_m$ is incident to the outer face of $\mathcal E_\mu$ that does not belong to $A(\mathcal E_\mu)$ if and only if $b_m$ is incident to the outer face of $\mathcal E^*_\mu$ that does not belong to $A(\mathcal E^*_\mu)$;

	\item If $\mu$ is of type \BF-BF or \RF-RF
	and
	$\mathcal E_\mu$ is of type neither \RFNi-RFN1, nor \BFNi-BFN1, nor \BFIi-BFI1, then
	 $b_m$ is an internal vertex of $\mathcal E_\mu$ if and only if it is an internal vertex of $\mathcal E^*_\mu$; furthermore, if $b_m$ is not an internal vertex of $\mathcal E_\mu$, then $b_m$ is incident to the outer face $x(\mathcal E_\mu)$ if and only if it is incident to the outer face $x(\mathcal E^*_\mu)$, with $x \in \{\ell,r\}$.
	
	\item If $\mathcal E_\mu$ is of type \RFNi-RFN1, then~$b_m$ shares a face with at least a red vertex $w$ different from $v$  in $\mathcal E_\mu$ if and only if~$b_m$ shares a face with at least a red vertex $w'$ different from $v$ in $\mathcal E^*_\mu$.
	Also, if $\mathcal E_\mu$ is of type \BFIi-BFI1, then~$b_m$ shares a face with a leaf incident to the red outer face of $\mathcal E_\mu$ if and only if~$b_m$ shares a face with a leaf incident to the red outer face of $\mathcal E^*_\mu$; and
	
	\item If $\mathcal E_\mu$ is of type \RFIi-RFI1, then the degree of $v$ in $A(\mathcal E_\mu)$ is the same as in $A(\mathcal E^*_\mu)$. 
\end{enumerate}

By hypothesis, we have that all the properties listed above hold when 
$\lambda$, $\mathcal E_\lambda$, and $\mathcal E_D$ are considered in place of
$\mu$, $\mathcal E_\mu$, and $\mathcal E^*_\mu$, respectively. 
We also observe the following facts.

\begin{fact}\label{prop:external}	
The part of $A(\mathcal E_\mu)$ that depends on the pertinent graphs of the children of $\mu$ different from $\lambda$ is the same as in $A(\mathcal E^*_\mu)$.
\end{fact}

\begin{fact}\label{prop:outer-vertices}
If $\mu$ is of type either \BP-BP or \BB-BB, then there exists at least one red vertex incident to $x(\mathcal E_\lambda)$, with $x \in \{\ell, r\}$, if and only if there exists at least one red vertex incident to $x(\mathcal E_D)$.
Also, if $\mu$ is of type neither \BP-BP nor \BB-BB, then there exist 1 or more than one red vertices incident to $x(\mathcal E_\lambda)$, with $x \in \{\ell, r\}$, if and only if there exist 1 or more than one red vertices incident~to~$x(\mathcal E_D)$.
\end{fact}

\noindent \cref{prop:external} holds by construction, while \cref{prop:outer-vertices} descends from Items~2 and from \cref{def:auxiliary-graph-E-mu}.

We now prove that Items~1--9 hold for $\mu$. Let $\mathcal S_\mu$ be the embedding of $\skel(\mu)$ determined by $\mathcal E_\mu$, and let $e_\lambda$ be the virtual edge of $\skel(\mu)$ that corresponds to $\lambda$. 

Item~1 for $\mu$ descends from
Items~1,~2,~3,~5,and~9 for $\lambda$ and from \cref{prop:external}.

We prove Item~2 for $\mu$. First observe that, if $e_\lambda$ is not incident to an outer face of $\mathcal S_\mu$, then Item~2 trivially holds. Otherwise, Item~2 for $\mu$ descends from \cref{prop:external,prop:outer-vertices}.

\begin{figure}[t!]
	\centering
	\subfloat[$\mathcal E_\mu$]{
		\includegraphics[height=.3\textwidth,page=1]{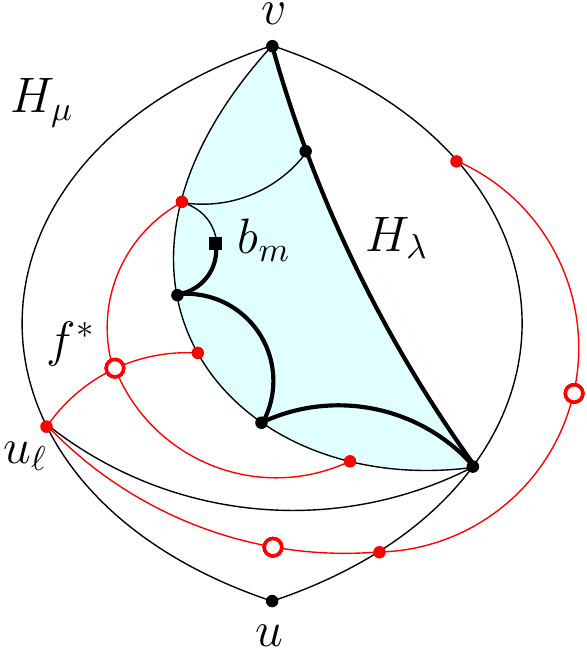} 
		\label{fig:child-replacement-a} 
	}
    \hfil
	\subfloat[$\mathcal E^*_\mu$]{
		\includegraphics[height=.3\textwidth,page=2]{child-replacement}
		\label{fig:child-replacement-b} 
	}
	\caption{Illustration for the proof that Item~3 holds for $\mu$. The type of $\mathcal E_\mu$ is \BFIig-BFI1G and the type of $\mathcal E_\lambda$ is \BBii-BB2.}	
	\label{fig:child-replacement}
\end{figure}

We prove Item~3 for $\mu$. First observe that, if $e_\lambda$ is not incident to a face of $\mathcal S_\mu$ that contains the neighbor $u_x$ of $u$ that is incident to $x(\mathcal E_\mu)$, then Item~3 trivially holds. Otherwise, Item~3 for $\mu$ descends from \cref{prop:external,prop:outer-vertices}. As an example, suppose that $\mathcal E_\mu$ is of type \BFIig-BFI1G and that it is flipped so that $u_\ell$ is incident to two red faces in $\mathcal E_\mu$; refer to \cref{fig:child-replacement}. Let $f$ be the face of $\mathcal S_\mu$ that is shared by $u_x$ and $e_\lambda$. If the face $f^*$ of $\mathcal E^*_\mu$ corresponding to $f$ is not red, then the there exists no non-pole red vertex of $D$ that is incident to $f^*$. Thus, either $\mathcal E_D$ is of type \BPii-BP2, \BPiv-BP4, \BBii-BB2, \BBiv-BB4, or $u_\ell$ is a pole of $D$ and $\mathcal E_D$ is of type \RFNz-RFN0, \RFIz-RFI0, \RFNi-RFN1, or \RFIi-RFI1. However, by \cref{prop:external,prop:outer-vertices}, it follows that the face of $\mathcal E_\mu$ corresponding to $f$ is not red, a contradiction to the fact that $\mathcal E_\mu$ is of type \BFIig-BFI1G.

We prove Item~4 for $\mu$. 
Suppose first that the end-vertex of the backbone of $A(\mathcal E_\mu)$ that is adjacent to a leaf $r''$ of $A(\mathcal E_\mu)$ that shares a face with $b_m$ is an outer face of $\mathcal E_\mu$, say $\ell(\mathcal E_\mu)$.
By Item~2 for $\mu$, we have that $\ell(\mathcal E^*_\mu)$ is an end-vertex of the backbone of $A(\mathcal E^*_\mu)$. 
Since $b_m$ belongs to $\pert{\mu}$, then $\mu$ is of type either \BB-BB, or \RF-RF, or \BF-BF. 
First, if each of $b_m$ and $r''$ either does not belong to $\pert{\lambda}$ or is a pole of $\pert{\lambda}$, then Item~4 for $\mu$ follows by \cref{prop:external}.
If both $b_m$ and $r''$ belong to $\pert{\lambda}$, then Item~4 for $\mu$ follows from Item~4 for $\lambda$.
Suppose now that $b_m$ is a non-pole vertex of $\pert{\lambda}$, while $r''$ does not belong to $\pert{\lambda}$. 
By \cref{prop:external}, this implies that $r''$ is incident to $\ell(\mathcal E^*_\mu)$. 
Then, the face $g$ of $\mathcal E_\mu$ that is shared by $b_m$ and $r''$ corresponds to a face  $f$ of $\mathcal S_\mu$, which is incident to $e_\lambda$. 
We claim that $g$ is red. Since $b_m$ is a non-pole vertex of $\pert{\lambda}$ that is incident to the outer face of $\mathcal E_\lambda$ that correspond to $f$, there exists two neighbors of $b_m$ that are incident to $g$. Since $b_m$ has exactly one black neighbor, at least one of these two vertices is red. This vertex and $r''$ imply that $g$ is red. 
The above claim, the fact that $g$ contains $r''$, and the fact that $r''$  is a leaf of $A(\mathcal E_\mu)$ incident to $\ell(\mathcal E_\mu)$, by hypothesis, implies that $g$ coincides with $\ell(\mathcal E_\mu)$. Therefore, $b_m$ is incident to $\ell(\mathcal E_\mu)$. 
By Items~6 or~7 for $\lambda$, we have that $b_m$ is also incident to $\ell(\mathcal E^*_\mu)$, when $\lambda$ is of type either \BF-BF or \RF-RF.
This implies that Item~4 holds for $\mu$ in this case. Suppose, finally, that $\lambda$ is of type \BB-BB. Note that, the type of $\mathcal E_\lambda$ is neither \BBiv-BB4 nor \BBv-BB5, as otherwise $\ell(\mathcal E_\mu)$ would not be the end-vertex of the spine of $A(\mathcal E_\mu)$
that is incident to a leaf of $A(\mathcal E_\mu)$ sharing a face with $b_m$, by \cref{le:bf-endvertex-backbone}. If $\mathcal E_\mu$ is of type either \BBii-BB2 or \BBiii-BB3, then there exists a red neighbor of $b_m$ in $\pert{\lambda}$ that is incident to an outer face of  $\mathcal E_{\lambda}$, since $b_m$ has exactly one black neighbor. By Item~4 for $\lambda$, we have that this red vertex is incident to the outer face of $\mathcal E_\lambda$ that corresponds to $\ell(\mathcal E^*_\mu)$. Thus, we can use this red vertex to play the role of $r''$ in $\mathcal E^*_\mu$.

We prove Item~5 for $\mu$. Suppose first that $\mathcal E_\lambda$ contains at least one internal red face, which implies that $\mathcal E_\mu$ contains at least one internal red face. By Item~5 for $\lambda$, we have that $\mathcal E_D$ contains at least one internal red face, which implies that $\mathcal E^*_\mu$ contains at least one internal red face. 
Suppose then that $\mathcal E_\lambda$ contains at least one outer red face that corresponds to an internal red face of $\mathcal E_\mu$. By \cref{prop:external,prop:outer-vertices}, we have that $\mathcal E_D$ contains at least one outer red face that corresponds to an internal red face of $\mathcal E^*_\mu$.
Suppose that there exists an internal red face of $\mathcal E_\mu$ that contains no vertex of $\pert{\lambda}$. Then, by \cref{prop:external}, such a face is also red in $\mathcal E^*_\mu$. This implies that Item~5 holds~for~$\mu$.

We prove Item~6 for $\mu$. If $b_m$ either does not belong to $\pert{\lambda}$ or is a pole of $\lambda$, then Item~6 descends from \cref{prop:external}. Otherwise, $b_m$ is a non-pole vertex of $\pert{\lambda}$.  Suppose that $b_m$ is incident to the outer face $x(\mathcal E_\mu)$ of $\mathcal E_\mu$, with $x \in \{\ell,r\}$, that does not belong to $A(\mathcal E_\mu)$.
Then, $b_m$ is incident to the outer face $y(\mathcal E_\lambda)$ of $\mathcal E_\lambda$ that corresponds to $x(\mathcal E_\mu)$, with $y \in \{\ell,r\}$. This implies that $y(\mathcal E_\lambda)$ does not belong to $A(\mathcal E_\lambda)$, as otherwise $x(\mathcal E_\mu)$ would belong to $A(\mathcal E_\mu)$. By Item~2~and either Items~6 or~7 for $\lambda$, we have that $b_m$ in incident to the outer face of $\mathcal E_D$ that does not belong to $\mathcal E_D$, which corresponds to the outer face $x(\mathcal E^*_\mu)$ of $\mathcal E^*_\mu$. Thus, Item~6 holds for $\mu$.

The proof that Item~7 holds for $\mu$ is analogous to the one that Item~6 holds for $\mu$ and it exploits \cref{prop:external}, and Item~2~and either Item~6~or~7 for $\lambda$.

We prove Item~8 for $\mu$. 
First, if each of $b_m$ and $w$ either does not belong to $\pert{\lambda}$ or is a pole of $\pert{\lambda}$, then Item~8 for $\mu$ follows by \cref{prop:external}.
If both  $b_m$ and $w$ belong to $\pert{\lambda}$, then Item~8 for $\mu$ follows from Item~7 or~8 for $\lambda$.
Suppose now that $b_m$ is a non-pole vertex of $\pert{\lambda}$, while $w$ does not belong to $\pert{\lambda}$. 
Note that, $e_\lambda$ and a virtual edge $e_w$ such that $w$ belongs to $\pert{e_w}$ share a face $f$ of $\mathcal S_\mu$. In particular, $b_m$ and $w$ are incident to the face of $\mathcal E_\mu$ that corresponds to $f$.
By \cref{prop:external}, $w$ is incident to the face of $\mathcal E^*_\mu$ that corresponds to $f$. The fact that also $b_m$ is incident to such a face follows from Items~6 or~7 for $\lambda$.

Finally, we prove Item~9 for $\mu$. Note that, since $\mathcal E_\mu$ is of type \RFIi-RFI1 and by Item~2 for $\mu$, we have that $v$ is adjacent to exactly one outer face both in $A(\mathcal E_\mu)$ and in $A(\mathcal E^*_\mu)$. We show that $v$ is incident to an internal red face $g$ in $\mathcal E_\mu$ if and only if it is incident to an internal red face in $\mathcal E^*_\mu$. 
The proof is analogous to the one of Item~3, where $v$ plays the role of $u_\ell$. Namely, if $e_\lambda$ is not incident to any face of $\mathcal S_\mu$ that is incident to $v$, this hols by \cref{prop:external}.
If $g$ is internal to $\mathcal E_{\lambda}$, then $v$ is a pole of $e_\lambda$ and this holds by Item~9 for $\lambda$.
Otherwise, $g$ corresponds to a face $g'$ of $\mathcal S_\mu$. Then, this holds by \cref{prop:external,prop:outer-vertices}. This concludes the proof.
\end{fullproof}


\section{Testing and Embedding Algorithm}\label{se:spqr-tree}

In this section, we show how to solve in linear time the \btpbef problem. In particular, we prove the following. 

\begin{theorem}\label{th:main}
Let $\langle G, \pi_b\rangle$ be an instance of the \btpbefP problem, where $G$ is an $n$-vertex bipartite planar graph and $\pi_b$ is a linear ordering of the black vertices of $G$. There exists an $O(n)$-time algorithm that tests whether $\langle G, \pi_b\rangle$  admits a solution and, if any, constructs one.
\end{theorem}

\begin{proof}
The algorithm of \cref{th:main} consists of the following steps, all of which can be performed in linear time. 
First, construct the black saturation $H$ of $\langle G, \pi_b\rangle$, by adding an edge between any two black vertices of $G$ that are consecutive in $\pi_b$. 
Second, test whether $H$ is planar~\cite{ht-ept-74}, which, by \cref{obs:gprime-planar}, is a necessary condition for $\langle G, \pi_b\rangle$ to be a positive instance of the \btpbef problem.
Third, compute the block-cut-vertex tree of $H$~\cite{ht-eagm-73}, and construct the $rb$-augmented components for~$H$ as described in \cref{se:simply}. Recall that, by 
\cref{le:characterization-triconnected-onesidefixed,le:rb-augmented,le:neat-embeddings}, $\langle G,\pi_b\rangle$ is a positive instance of the \btpbef problem if and only if each $rb$-augmented component admits a neat embedding. 
Fourth, consider each $rb$-augmented component, which we still denote by $H$, and compute in $O(|H|)$ time whether it admits a neat embedding, by \cref{algo:neat} whose proof is contained in the remainder of the section. This concludes the test on whether $\langle G, \pi_b\rangle$ is a positive instance of the \btpbefP problem.
 
Finally, if the test above succeeds, we compute a \btpbef of $\langle G, \pi_b\rangle$ in linear time as follows.
First, we exploit \cref{th:neat-construction} in \cref{sse:embeddingconstruction} to compute
a neat embedding of $H$ in $O(|H|)$ time, for each $rb$-augmented component $H$.
Second, by applying \cref{cor:rb-augmented-good-embeddings}, we obtain a neat embedding of the black saturation of $\langle G, \pi_b\rangle$ in $O(|G|)$ time. 
Finally, by applying \cref{le:book-from-good}, we obtain a  \btpbef of $\langle G, \pi_b\rangle$ from the constructed neat embedding in $O(|G|)$ time. This concludes the proof.
\end{proof}

\cref{th:main}, together with \cref{le:book-from-good,cor:equivalence-fixed,cor:rb-augmented-good-embeddings}, implies the following.

\begin{corollary}\label{co:two-layer-fixed-algorithm}
The {\sc 2-Level Quasi-Planarity with Fixed Order} problem is linear-time solvable. Further, given a bipartite graph $K=(U_b,U_r,E_K)$ and a total order $\sigma_b$ of the vertices in $U_b$, a $2$-level quasi-planar drawing in which the vertices in $U_b$ lie along a straight line in the order $\sigma_b$ can be constructed in linear time, if it exists.
\end{corollary}

Let $H$ be an $rb$-augmented component and let $P=(b_1,b_2,\dots,b_m)$ be the black path of $H$. In the following, we show how to test in linear time whether $H$ admits a neat embedding (see \cref{algo:neat}). 
First, we compute the biconnected graph $H^-$ obtained from $H$ by removing its degree-$1$ vertices; we construct the SPQR-tree $\mathcal{T}$ of $H^-$ and root it at the Q-node $\rho$ corresponding to the edge $(b_1,b_2)$. This can be done in linear time~\cite{DBLP:conf/icalp/BattistaT90,DBLP:conf/gd/GutwengerM00}.
We will traverse $\mathcal{T}$ bottom-up and compute, for each node $\mu \neq \rho$ of $\mathcal{T}$, the types of the relevant embeddings that $\pert{\mu}$ admits, using a dynamic-programming approach. We say that an embedding type is \emph{admissible} for $\mu$ if $\pert{\mu}$ admits a relevant embedding of that type. \cref{le:relevant-in-extensible-RE,le:relevant-in-extensible-RF,le:relevant-in-extensible-BE,le:relevant-in-extensible-BP,le:relevant-in-extensible-BB,le:relevant-in-extensible-BF} imply the following.

\begin{property}\label{pr:admissible-constant}
The set of admissible types for a node $\mu$ has  $O(1)$ size.
\end{property}

\noindent 
If there is a node $\mu$ of $\mathcal T$ for which the set of admissible types is empty, then we will reject the instance (\cref{le:empty-admissible}). Otherwise, we will conclude that $H$ admits a neat embedding (\cref{le:some-admissible}). 

\begin{lemma} \label{le:empty-admissible}
If the set of admissible types for a non-root node $\mu$ of $\mathcal T$ is empty, then $H$ does not admit any neat embedding.
\end{lemma}

\begin{proof}
Suppose for a contradiction that $H$ admits a neat embedding $\mathcal E$  and that there exists a node $\mu$ of $\mathcal T$ whose set of admissible types is empty. By \cref{def:extensible}, the restriction of $\mathcal E$ to $\pert{\mu}$ defines an extensible embedding $\mathcal E_\mu$ of $\pert{\mu}$. By \cref{le:relevant-in-extensible}, we have that $\mathcal E_\mu$ is relevant, hence the set of admissible types for $\mu$ is not empty, a contradiction.
\end{proof}

\begin{lemma} \label{le:some-admissible}
If the set of admissible types for the child of the root $\rho$ of $\mathcal T$ is not empty, then $H$ admits a neat embedding.
\end{lemma}

\begin{proof}
Let $\mu$ be the child of the root $\rho$ of $\mathcal T$. By \cref{le:rooting-at-b1b2}, we have that $\mu$ is of type \BF-BF and the pole $u$ coincides with $b_1$. Since $\mu$ has at least one admissible type, we have that $\pert{\mu}$ admits a relevant embedding $\mathcal{E}_\mu$. We prove that $\mathcal{E}_\mu$ is extensible, by showing that it can be augmented to a neat embedding $\mathcal E$ of $H$. For this, we route the edge $(u,v)$ in the outer face of~$\mathcal{E}_\mu$.  Moreover, we place the undecided $rb$-trivial component incident to~$u$ (if any) or to~$v$ (if any) in one of the two faces incident to $(u,v)$ as described below, based on the type of $\mathcal{E}_\mu$. Observe that both such faces correspond to faces of the unique embedding of the skeleton of $\rho$, which is the proper allocation node of $u$ and $v$. Therefore, in order to prove that $\cal E$ is neat, it suffices to show $\cal E$ is good, i.e., it satisfies the \cref{cond:a,cond:b} of \cref{le:characterization-triconnected-onesidefixed}. 

\begin{itemize}
	\item {\em Suppose first that the type of $\mathcal{E}_\mu$ is \BFNz-BFN0.} Refer to~\cref{fig:from-child-zero-n}. Note that, in this case, there exists an undecided $rb$-trivial component incident to $u$ and an undecided $rb$-trivial component incident to $v$, as otherwise $H$ does not contain at least three red vertices. In fact, by \cref{le:internal-red-vertex}, we have that $\mathcal{E}_\mu$ contains no internal red vertex, since no internal face of $\mathcal{E}_\mu$ is red; furthermore, $\mathcal{E}_\mu$ does not contain any red vertex incident to the outer face other than $u_\ell=u_r$ (this equality is discussed before~\cref{le:bf-endvertex-backbone}) and those of the undecided $rb$-trivial components, since no outer face of $\mathcal{E}_\mu$ is red. To obtain $\mathcal{E}$, we let the two undecided $rb$-trivial components lie inside the same face $f \in \{\ell(\mathcal E_\mu),r(\mathcal E_\mu)\}$ incident to $(u,v)$.
We have that $\mathcal{E}$ satisfies \cref{cond:a}. In fact, 
$A(\mathcal{E})$ is a star centered at~$f$, thus it is also a caterpillar and its (trivial) backbone contains the unique red face $f$ of $\mathcal{E}_\mu$. Also, $\mathcal{E}$ satisfies \cref{cond:b}. In fact, 
we can set $f=f_1=f_k$; 
also, we can select the leaf $r'$ of $A(\mathcal{E})$ sharing a face with $b_1=u$
to be the red vertex of an undecided $rb$-trivial component; 
finally, we can select the leaf $r''$ of $A(\mathcal{E})$ sharing a face with $b_m$ to be the red vertex $u_\ell=u_r$ (indeed, note that $b_m$ shares a face of $\mathcal{E}_\mu$, and hence of $\mathcal{E}$, with $u_\ell=u_r$ both if the type of $\mathcal{E}_\mu$ is~\BFNza-BFN0A and if it is~\BFNzb-BFN0B).

\begin{figure}[htb]\tabcolsep=6pt
	\centering
	\subfloat[]{
		\includegraphics[height=.2\textwidth,page=33]{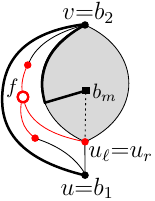} 
		\label{fig:from-child-zero-n}
	}
		\subfloat[]{
		\includegraphics[height=.2\textwidth,page=34]{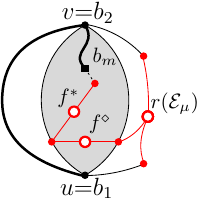} 
		\label{fig:from-child-zero-i}
	}
		\subfloat[]{
		\includegraphics[height=.2\textwidth,page=35]{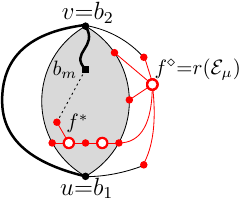} 
		\label{fig:from-child-nonz}
	}
	\caption{Extending an embedding $\mathcal E_\mu$ of the pertinent graph of the child $\mu$ of the root $\rho$ of $\mathcal T$ to a neat embedding $\mathcal E$ of $H$. In (a) the type of $\mathcal E_\mu$ is~\protect\BFNz-BFN0 (in particular,~\protect\BFNza-BFN0A), in (b) the type of $\mathcal E_\mu$ is~\protect\BFIz-BFI0 (in particular,~\protect\BFIzb-BFI0B), while in (c) the type of $\mathcal E_\mu$ is of one of the remaining types~\protect\BF-BF (in particular,~\protect\BFIid-BFI1D).}\label{fig:from-child}
\end{figure}

\item {\em Suppose next that the type of $\mathcal{E}_\mu$ is~\BFIz-BFI0.} Refer to~\cref{fig:from-child-zero-i}.
Let $f^*$ and $f^\diamond$ be the end-vertices of the backbone $\mathcal B_\mu$ of the caterpillar $A(\mathcal E_\mu)$. 
By \cref{le:bf-endvertex-backbone}, we have  that (up to exchanging $f^*$ with $f^\diamond$ and/or $u_r$ with $u_\ell$):
\begin{enumerate}
	\item $b_m$ shares a face with a leaf of $A(\mathcal E_\mu)$ adjacent to $f^*$; and
	\item $f^\diamond$ is an internal face of ${\mathcal E}_\mu$ incident to $u_r$, and $u_r$ is a leaf of $A({\mathcal E}_\mu)$.
\end{enumerate}	

To obtain $\mathcal{E}$, we distinguish two cases based on whether 
 undecided $rb$-trivial components exist or not.
Suppose first that no undecided $rb$-trivial component exists.
We have that $\mathcal{E}$ satisfies \cref{cond:a}, since $A(\mathcal{E}) = A(\mathcal{E}_{\mu})$ is a caterpillar by \cref{le:structure-A-mu-BF}. 
Also, $\mathcal{E}$ satisfies \cref{cond:b}. In fact, 
we can set $f_1=f^\diamond$ and $f_k=f^*$.
Also, we can set $u_r$ to be the leaf $r'$ of $A(\mathcal{E})$ sharing a face with $b_1=u$.
Finally, we can set as $r''$ the leaf of $A(\mathcal{E}_{\mu})$ adjacent to $f^*$ and sharing a face with $b_m$.
Suppose now that undecided $rb$-trivial components exist.
To obtain $\mathcal{E}$, we let the undecided $rb$-trivial components lie inside $r(\mathcal E_{\mu})$.
We have that $\mathcal{E}$ satisfies \cref{cond:a}.
In fact, $\mathcal{E}$ contains the same red faces as $\mathcal{E}_{\mu}$ plus $r(\mathcal E_{\mu})$.
The auxiliary graph  $A(\mathcal{E})$ is obtained by adding to $A(\mathcal{E}_{\mu})$  the edge $(u_r,r(\mathcal E_{\mu}))$ and edges connecting $r(\mathcal E_{\mu})$ to the red vertex of each undecided $rb$-trivial component. Clearly, we obtained a caterpillar whose backbone is the same as the one of $A(\mathcal{E}_{\mu})$ plus the path $(f^\diamond,u_r,r(\mathcal E_{\mu}))$. 
Also, $\mathcal{E}$ satisfies \cref{cond:b}. In fact, 
we can set $f_1=r(\mathcal E_{\mu})$ and $f_k=f^*$.
Also, we can set the red vertex of one undecided $rb$-trivial component to be
the leaf $r'$ of $A(\mathcal{E})$ sharing a face with $b_1=u$.
Finally, we can set as $r''$ the leaf of $A(\mathcal{E}_{\mu})$ adjacent to $f^*$ and sharing a face with $b_m$.

\item {\em Suppose finally that $\mathcal{E}_{\mu}$  is of one of the remaining types \BF-BF.} Refer to~\cref{fig:from-child-nonz}.
We have that at least one of $\ell(\mathcal E_{\mu})$ or $r(\mathcal E_{\mu})$ is red. Suppose, w.l.o.g., that $r(\mathcal E_{\mu})$ is red.
To obtain $\mathcal{E}$, we let the undecided $rb$-trivial components, if any, lie inside $r(\mathcal E_{\mu})$.
We have that $\mathcal{E}$ satisfies \cref{cond:a}.
In fact, $\mathcal{E}$ contains the same red faces as $\mathcal{E}_{\mu}$.
The auxiliary graph  $A(\mathcal{E})$ is obtained from $A(\mathcal{E}_{\mu})$ by connecting $r(\mathcal E_{\mu})$ to the red vertex of each undecided $rb$-trivial component, if any. Since $A(\mathcal{E}_{\mu})$ is a caterpillar by \cref{le:structure-A-mu-BF}, we obtained a caterpillar, whose backbone is the same as the one of $A(\mathcal{E}_{\mu})$. Also, $\mathcal{E}$ satisfies \cref{cond:b}. In fact, by \cref{le:bf-endvertex-backbone}, up to a relabeling of $f^*$ and $f^\diamond$, we have that:
	\begin{enumerate}
		\item $b_m$ shares a face with a leaf of $A(\mathcal E_\mu)$ adjacent to $f^*$; and
		\item $f^\diamond$ is either an outer face of ${\mathcal E}_\mu$ or is an internal face of ${\mathcal E}_\mu$ incident to $u_{\ell}$ (to $u_r$); in the latter case, $u_{\ell}$ (resp.\ $u_r$) is a leaf of $A({\mathcal E}_\mu)$.
	\end{enumerate} 

Then we can set $f_1=f^\diamond$ and $f_k=f^*$. Also, if $f^\diamond$ is an outer face of ${\mathcal E}_\mu$, we can set 
any red leaf of $A({\mathcal E}_\mu)$ adjacent to $f^\diamond$ to be
the leaf $r'$ of $A(\mathcal{E})$ sharing a face with $b_1=u$; otherwise, if $f^\diamond$ is an internal face of ${\mathcal E}_\mu$ we can set the one between $u_{\ell}$ and $u_r$ that is a leaf of $A({\mathcal E}_\mu)$ adjacent to $f^\diamond$ to be
the leaf $r'$ of $A(\mathcal{E})$ sharing a face with $b_1=u$. Finally, we can set as $r''$ the leaf of $A(\mathcal{E}_{\mu})$ adjacent to $f^*$ and sharing a face with $b_m$.
\end{itemize}
This completes the case distinction and hence the proof of the lemma.
\end{proof}

\subsection{The Bottom-up Traversal}\label{sse:testingalgorithm} 
\newcommand{\plus}{plus}

Let $\mu$ be a node $\mathcal T$ and let $e_\mu$ be the virtual edge representing $\mu$ in the skeleton of the parent of $\mu$. We denote the poles of $\mu$ (the end-vertices of $e_\mu$) as $u_\mu$ and $v_\mu$. 
In particular, $u_\mu$ is a black vertex, and 
if $v_\mu$ is also black, then $u_\mu$ precedes $v_\mu$ along $P$.
Also, we say that an admissible type for $\mu$ is an admissible~type~for~$e_\mu$.


%
%

By \cref{le:empty-admissible,le:some-admissible}, in order to test whether $H$ admits a neat embedding, we will compute the set of admissible types for all the non-root nodes of $\mathcal T$. Note that we are also interested in constructing a neat embedding of $H$, if any. Such a task would be easily achieved if we could store an embedding of $\pert{\mu}$, for each admissible type for $\mu$, however this would imply a quadratic running time. In order to achieve a linear running time, for each node $\mu$ and for each admissible type $t_\mu$ for $\mu$, we store in $\mu$ 
the following pieces of information, whose size is $O(|\skel(\mu)|)$.

\begin{definition}\label{def:embedding-configuration}
	The following pieces of information form an \emph{embedding configuration} for $\mu$:
	\begin{enumerate}[(i)]
		\item a planar embedding $\mathcal S_\mu$ of $\skel(\mu)$;
		\item for each vertex $w$ of $\skel(\mu)$ that is incident to an $rb$-trivial component and such that $\mu$ is the proper allocation node of $w$, a designated face $f_w$ of $\mathcal S_\mu$ incident to $w$,
		\item an admissible type $t_i$, for each virtual edge $e_i$ of $\skel(\mu)$, and
		\item a label $x_i \in \{\ell,r\}$, for each virtual edge $e_i$ of $\skel(\mu)$.
	\end{enumerate}
\end{definition}


In \cref{sse:embeddingconstruction}, we will show how the computed information can be exploited to construct a neat embedding $\cal E$ of $H$, if any, in linear time (see \cref{th:neat-construction}).

In the bottom-up traversal of $\mathcal T$, we  distinguish four cases, based on whether $\mu$ is a Q-, S-, P-, or an R-node.  Observe that the leaves of $\mathcal T$ are all Q-nodes.

Suppose that $\mu$ is a Q-node.  
If $v_\mu$ is black, then 
the only admissible type $t_\mu$ for $\mu$ is \BPi-BP1.
Otherwise, $v_\mu$ is red, and the only admissible type $t_\mu$ for $\mu$ is \RE-RE. 
Since in both cases, $\mathcal S(t_\mu)$ is the unique embedding of $\skel(\mu)=(u_\mu,v_\mu)$, $\skel(\mu)$ does not contain any virtual edges, and $\mu$ is not the proper allocation node of any vertex, we do not need to store any embedding configuration~for~$\mu$.

If $\mu$ is not a Q-node, then it has $k\geq 2$ children $\nu_1,\dots,\nu_k$.  For the sake of readability, we let $e_i = e_{\nu_i}$, $u_i = u_{\nu_i}$, and $v_i = v_{\nu_i}$. We will compute the set of admissible types and embedding configurations for $\mu$, assuming to have already computed the same information for $\nu_1,\dots,\nu_k$. A key technical ingredient in our approach would be the construction of embeddings of $\pert{\mu}$, obtained by substituting each virtual edge $e$ of $\skel(\mu)$ with an embedding of its pertinent graph $\pert{e}$ whose type is one of the admissible types for $e$. A naive implementation of this technique would, however, immediately imply a super-linear running time. Therefore, in order to achieve a linear running time and in virtue of \cref{le:child-replacement}, rather than substituting $e$ with an embedding of $\pert{e}$, we will substitute it with an embedding of a constant-size replacement-graph for~$\pert{e}$ (recall \cref{def:replacement}) whose type is the same as the one of the embedding of $\pert{e}$. To this aim, we will exploit a dictionary $\mathcal D$ that associates each embedding type $t$ with a plane graph of constant size, called \emph{$t$-gadget}, whose embedding is of type $t$. This graph will be used as a replacement-graph for $\pert{e}$.
The dictionary $\mathcal D$ contains the plane graphs illustrated in \cref{fig:Nodetypesw-gadgets,fig:RF-gadgets,fig:BE-gadgets,fig:BPBB-gadgets,fig:NodetypesBF-gadgets}.

\begin{lemma}\label{obs:constant-types}
The dictionary $\mathcal D$ contains a $t$-gadget, for each embedding type~$t$, and has $O(1)$~size. 
\end{lemma}

\begin{proof}
The first part of the statement can be observed by comparing the graphs in \cref{fig:Nodetypesw-gadgets,fig:RF-gadgets,fig:BE-gadgets,fig:BPBB-gadgets,fig:NodetypesBF-gadgets} with the definitions of the types \RE-RE, \RF-RF, \BE-BE, \BB-BB/\BP-BP, and \BF-BF, respectively. For instance, consider the plane graph $\mathcal D(t)$  in \cref{fig:NodetypesBF-gadgets} labeled as \BFIig-BFI1G.
Since its poles are black, since there exists no black path between them, and since $b_m$ is an internal vertex of $\mathcal D(t)$, we have that the type of $\mathcal D(t)$ is \BF-BF. Furthermore, the auxiliary graph of $\mathcal D(t)$ is a caterpillar whose backbone contains two internal red faces and one outer face. Finally, one end-vertex of the backbone of $\mathcal D(t)$ is an outer face, while the other end-vertex is an internal face, different from $f_\ell$ and $f_r$, that is incident to a leaf of the caterpillar that shares a face with $b_m$. Therefore, $\mathcal D(t)$ has type $t=$ BFI1G, and hence $\mathcal D(t)$ can serve as a $t$-gadget.
The second part of the statement follows from  \cref{pr:admissible-constant} and from the fact that each $t$-gadget has constant size, by construction.
\end{proof}


Given an embedding configuration $\cal C_\mu$ for $\mu$, the \emph{realization} of $\cal C_\mu$ is the plane graph $R({\cal C_\mu})$ obtained as follows.
First, initialize $R({\cal C_\mu}) = \mathcal S_\mu$.
Second, for each vertex $w$ of $\skel(\mu)$ 
such that~$f_w$ is defined, place the $rb$-trivial component incident to $w$ inside~$f_w$.
Finally, for each virtual edge~$e_i$ of $\skel(\mu)$, replace~$e_i=(u_i,v_i)$ in~$R({\cal C_\mu})$ with an $x_i$-flipped copy of the $t_i$-gadget, in such a way that $u_{i}$ and $v_{i}$ are identified with the vertices $u$ and $v$ of the $t_i$-gadget, respectively.

\begin{lemma} \label{le:rc-replacement-graph}
Given an embedding configuration $\cal C_\mu$, its realization $R({\cal C_\mu})$ is a replacement-graph for $\pert{\mu}$. Also, the types of $R({\cal C_\mu})$ and $\mu$ coincide.
\end{lemma}

\begin{proof}
We first show that the types of $R({\cal C_\mu})$ and $\mu$ coincide.
Recall that the type of $R({\cal C_\mu})$ is determined by 
\begin{inparaenum}[(a)]
	\item  whether $v_{\mu}$ is a red vertex,
	\item  whether a black path exists in $R({\cal C_\mu})$ between $u_{\mu}$ and $v_{\mu}$, and 
	\item the number of black neighbors of $u_{\mu}$ and $v_{\mu}$ in $R({\cal C_\mu})$.
\end{inparaenum}
Since $R({\cal C_\mu})$ is obtained by replacing each virtual edge $e_i$ of $\skel(\mu)$ with the $t_i$-gadget for $e_i$ in $\mathcal C_\mu$ and $t_i$ is an admissible type for $e_i$, we have that $R({\cal C_\mu})$ and $\pert{\mu}$ are equivalent with respect to properties (a), (b),~and~(c). 

Next, we show that $R({\cal C_\mu})$ is a replacement-graph for $\pert{\mu}$. Recall the conditions of \cref{def:replacement} of a replacement-graph.
By the construction	of $R(\mathcal C_\mu)$, we immediately have that the vertex set of $R(\mathcal C_\mu)$ contains $u_\mu$ and $v_\mu$ (Condition r.\ref{cond:replacement-uv}), that each vertex of $R(\mathcal C_\mu)$ is colored either red or black and the vertices $u_\mu$ and $v_\mu$ have the same color as in $H$ (Condition r.\ref{cond:replacement-color}), and that (iii) there exists no $rb$-trivial component incident to $u_\mu$ or $v_\mu$ in $R(\mathcal C_\mu)$ (Condition r.\ref{cond:replacement-rb-trivial}).
Finally, let $K$ the graph obtained from $H$ by removing the vertices and the edges of $\pert{\mu}$ except for $u_\mu$, $v_\mu$, and their incident $rb$-trivial components, if any, and by inserting $R(\mathcal C_\mu)$ in the resulting graph, while identifying the vertices $u_\mu$ and $v_\mu$ in the two graphs; then $K$ is an $rb$-augmented component. 
Namely, the black vertices induce a path in $K$ since $R({\cal C_\mu})$ and $\mu$ have the same type. Furthermore, by the definition of the $t_i$-gadget, there exists no edge connecting two red vertices. Moreover, the fact that the graph $K^-$ obtained from $K$ by removing all the degree-$1$ vertices is biconnected immediately descends from the fact that the graph $H^-$ obtained from $H$ by removing all the degree-$1$ vertices is biconnected, that $\skel(\mu)$ is $u_\mu v_\mu$-biconnectible, and that each $t_i$-gadget is $uv$-biconnectible. This concludes the proof that Condition r.\ref{cond:replacement-rb-augmented} of \cref{def:replacement} is satisfied and hence the proof of the lemma.
\end{proof}

The type of an embedding configuration can be detected efficiently, as in the following lemma.

\begin{lemma} \label{le:compute-type}
Given an embedding configuration $\cal C_\mu$, it is possible to determine in $O(|\skel(\mu)|)$ time whether the embedding of $R({\cal C_\mu})$ is relevant and, if so, to return its type.
\end{lemma}

\begin{fullproof}
To prove the statement, we first show that the plane graph $R({\cal C_\mu})$ can be constructed in $O(|\skel(\mu)|)$ time.
We visit $\skel(\mu)$ and, for each virtual edge $e_i$, we replace $e_i$ with the $x_i$-flipped copy of the $t_i$-gadget. Since each such a replacement can be done in $O(1)$ time, the construction of $R({\cal C_\mu})$ can be performed in $O(|\skel(\mu)|)$ time. 

Second, by \cref{le:rc-replacement-graph}, we have that $R({\cal C_\mu})$ is a replacement-graph. Hence, we determine the type of $R({\cal C_\mu})$ (i.e., \RF-RF, \BE-BE, \BP-BP, \BB-BB, or \BF-BF) by checking whether $v_{\mu}$ is a red vertex, by checking whether a black path exists in $R({\cal C_\mu})$ between $u_{\mu}$ and $v_{\mu}$, and by counting the number of black neighbors of $u_{\mu}$ and $v_{\mu}$; clearly, these computations can be performed in $O(|\skel(\mu)|)$ time. Note that the type of $R({\cal C_\mu})$ is not~\RE-RE, given that $\mu$ is not a Q-node, hence $\skel(\mu)$ contains more than one virtual edge, while by \cref{le:structure-node-RE} a replacement-graph of type~\RE-RE consists of a single edge.

Next, we determine whether the embedding $\mathcal E_R$ of $R({\cal C_\mu})$ is relevant. In order to do that, we construct the auxiliary graph $A(\mathcal E_R)$ in $O(|\skel(\mu)|)$ time as follows. We initialize $A(\mathcal E_R)$ to the set of red vertices of $R({\cal C_\mu})$. Then we traverse the boundary of each face $f$ of $\mathcal E_R$; if more than one red vertex is encountered during this traversal, then we insert a new vertex $f$ in $A(\mathcal E_R)$ and we connect it to all the red vertices incident to $f$. Note that $A(\mathcal E_R)$ has $O(|\skel(\mu)|)$ size. 

We check whether $A(\mathcal E_R)$ consists of at most two caterpillars. If that is not the case, then by \cref{le:structure-A-mu}, we conclude that the embedding $\mathcal E_R$ of $R({\cal C_\mu})$ is not relevant. Otherwise, we distinguish some cases based on the type of $R({\cal C_\mu})$.

Suppose that the type of $R({\cal C_\mu})$ is~\RF-RF. Let $b^*_m$ be the end-vertex of the black path of $R({\cal C_\mu})$ different from $u_{\mu}$. We check whether $A(\mathcal E_R)$ is a single caterpillar. By \cref{le:structure-A-mu-RF}, if the check fails, we conclude that $\mathcal E_R$ is not relevant. Assume now that $A(\mathcal E_R)$ is a single caterpillar. If $\mathcal E_R$ does not contain any internal red face, then $\mathcal E_R$ is always relevant, by~\cref{le:rf-endvertex-backbone-no-internal}. Assume now that $\mathcal E_R$ contains internal red faces. We check whether: (i) one of the end-vertices of the backbone of $A(\mathcal E_R)$ corresponds to an internal face of $\mathcal E_R$ that is incident to a leaf $r''$ that shares a face with $b^*_m$; and (ii) the other end-vertex of the backbone of $A(\mathcal E_R)$ either corresponds to an outer face of $\mathcal E_R$, or corresponds to an internal face of $\mathcal E_R$ incident to $v_{\mu}$; in the latter case, we also check whether $v_{\mu}$ is a leaf of $A(\mathcal E_R)$. By \cref{le:rf-endvertex-backbone}, if the checks succeed, then we conclude that $\mathcal E_R$ is relevant, otherwise we conclude that $\mathcal E_R$ is not relevant.

If the type of $R({\cal C_\mu})$ is~\BE-BE, then we check whether $A(\mathcal E_R)$ is a path between the red neighbors $u_\ell$ and $u_r$ of $u_{\mu}$ that are incident to the outer faces of $\mathcal E_R$. By \cref{le:structure-A-mu-BE}, if the check succeeds, then we conclude that $\mathcal E_R$ is relevant, otherwise we conclude that $\mathcal E_R$ is not relevant.

Suppose that the type of $R({\cal C_\mu})$ is \BP-BP or \BB-BB. We first check whether each red outer face is an end-vertex of the backbone of a caterpillar composing $A(\mathcal E_R)$. Further, if $A(\mathcal E_R)$ consists of two caterpillars, then we check whether one of them is a star. Suppose now that $\mathcal E_R$ contains internal red faces. If the type of $R({\cal C_\mu})$ is \BP-BP, then we check whether $b_m=v_{\mu}$ and whether the backbone of a caterpillar has an end-vertex that corresponds to an internal red face of $\mathcal E_R$ that is incident to a leaf of $A(\mathcal E_R)$ that shares a face with $v_{\mu}$. If the type of $R({\cal C_\mu})$ is \BB-BB, let $b^*_m$ be the end-vertex of the black path of $R({\cal C_\mu})$ different from $u_{\mu}$; then we check whether the backbone of a caterpillar has an end-vertex that corresponds to an internal red face of $\mathcal E_R$ that is incident to a leaf of $A(\mathcal E_R)$ that shares a face with $b^*_m$. By \cref{le:structure-A-mu-BP-BB}, if all checks succeed, then we conclude that $\mathcal E_R$ is relevant, otherwise we conclude that $\mathcal E_R$ is not relevant.

Suppose that the type of $R({\cal C_\mu})$ is \BF-BF. Let $b^*_m$ be the end-vertex of the black path of $R({\cal C_\mu})$ different from $u_{\mu}$. We check whether $A(\mathcal E_R)$ is a single caterpillar. By \cref{le:structure-A-mu-BF}, if the check fails, we conclude that $\mathcal E_R$ is not relevant. Assume now that $A(\mathcal E_R)$ is a single caterpillar. If $\mathcal E_R$ does not contain any internal red face, then $\mathcal E_R$ is always relevant. Assume now that $\mathcal E_R$ contains internal red faces. Again by \cref{le:structure-A-mu-BF}, we have that $A(\mathcal E_R)$ contains a path between the neighbors $u_\ell$ and $u_r$ of $u_{\mu}$ incident to the outer faces of $\mathcal E_R$. We check whether: (i) one of the end-vertices of the backbone of $A(\mathcal E_R)$ corresponds to a face of $\mathcal E_R$ that is incident to a leaf $r''$ that shares a face with $b^*_m$; and (ii) the other end-vertex of the backbone of $A(\mathcal E_R)$ either corresponds to an outer face of $\mathcal E_R$, or corresponds to an internal face of $\mathcal E_R$ incident to $u_\ell$ or $u_r$; in the latter case, we also check whether  $u_\ell$ (resp.\ $u_r$) is a leaf of $A(\mathcal E_R)$. By \cref{le:bf-endvertex-backbone}, if the checks succeed, then we conclude that $\mathcal E_R$ is relevant, otherwise we conclude that $\mathcal E_R$ is~not~relevant.

Note that all the above checks can be easily performed in $O(|\skel(\mu)|)$ time. If we concluded that $\mathcal E_R$ is relevant, then its type can be detected within the same time bound by checking whether the corresponding definition is satisfied. For example, if the type of $R({\cal C_\mu})$ is \BP-BP, we check whether $A(\mathcal E_R)$ does not contain any vertex (then the type of $\mathcal E_R$ is \BPi-BP1), whether it consists of one star (then the type of $\mathcal E_R$ is \BPii-BP2), whether it consists of two stars (then the type of $\mathcal E_R$ is \BPiii-BP3), whether it consists of a single caterpillar different from a star (then the type of $\mathcal E_R$ is \BPiv-BP4), and whether it consists of two caterpillars, one of which is different from a star (then the type of $\mathcal E_R$ is \BPv-BP5).
\end{fullproof}

The following lemma will be crucial to show the correctness of our approach.

\begin{lemma} \label{le:relevant-in-configuration}
An embedding type $t$ is admissible for $\mu$ if and only if there exists an embedding configuration $\cal C_\mu$ for $\mu$ such that the type of $R(\cal C_\mu)$ is $t$.
\end{lemma}

\begin{fullproof}
Suppose that there exists an admissible type $t$ for $\mu$ and let $\mathcal E_\mu$ be an embedding of $\pert{\mu}$ of type $t$. We construct an embedding configuration $\mathcal C_\mu$ as follows.
First, we set the embedding $\cal S_\mu$ of $\skel(\mu)$ in $\mathcal C_\mu$ to be the one determined by $\mathcal E_\mu$. 
Second, for each vertex $w$ of $\skel(\mu)$ that is incident to an $rb$-trivial component $(w,r)$ and such that  $\mu$ is the proper allocation node of $w$, we set the designated face $f_w$ to be the face of $\mathcal S_\mu$ that corresponds to the face of $\mathcal E_\mu$ incident to $w$ and $r$.
Third, for each virtual edge $e_i$ of $\mu$, let $\mathcal E_i$ be the embedding of $\pert{e_i}$ in $\mathcal E_\mu$. Note that $\mathcal E_i$ is relevant, by \cref{le:relevant-relevant}. We set $t_i$ in $\mathcal C_\mu$ to be the type of $\mathcal E_i$ and $x_i$ in $\mathcal C_\mu$ to be the flip of $\mathcal E_i$, for each virtual edge $e_i$ of $\skel(\mu)$.
This completes the construction of $\mathcal C_\mu$.
Next we show that the type of $R(\mathcal C_\mu)$ is $t$. 
Let $k$ be the number of virtual edges of $\skel(\mu)$. Also, let $\mathcal E^0_\mu = \mathcal E_\mu$. 
For $i=1,\dots,k$, consider the embedding $\mathcal E^i_\mu$ obtained from $\mathcal E^{i-1}_\mu$ by replacing the pertinent graph $\mathcal E_i$ of the virtual edge $e_i$ of $\skel(\mu)$ with an $x_i$-flipped copy of the $t_i$-gadget.
By construction, we have that $\mathcal E^k_\mu = R(\mathcal C_\mu)$. The type of $\mathcal E^0_\mu$ is $t$, by hypothesis. Furthermore, by \cref{le:child-replacement}, if the type of $\mathcal E^{i-1}_\mu$ is $t$, then also the type of $\mathcal E^{i}_\mu$ is $t$. Thus, the type of $ R(\mathcal C_\mu)$ is $t$. 

Suppose now that there exists an embedding configuration $\mathcal C_\mu$ for $\mu$ such that the type of $R(\mathcal C_\mu)$ is $t$. 
We show how to construct an embedding $\mathcal E_\mu$ of $\pert{\mu}$ of type $t$.
Let $\mathcal E^0_\mu = R(\mathcal C_\mu)$.
For $i=1,\dots,k$, consider the embedding $\mathcal E^i_\mu$ obtained from $\mathcal E^{i-1}_\mu$ by replacing the $t_i$-gadget for the virtual edge $e_i$ of $\skel(\mu)$ with an $x_i$-flipped copy of an embedding of $\pert{e_i}$ of type $t_i$.
Note that the embedding $\mathcal S_\mu$ of $\skel(\mu)$ in $\mathcal C_\mu$ is planar, by \cref{def:embedding-configuration}. Moreover, each $rb$-trivial component incident to a vertex $w$ of $\skel(\mu)$ such that $\mu$ is the proper allocation node of $w$ is incident to a face of $\mathcal E_\mu$ corresponding to a face of $\mathcal S_\mu$. Therefore, we have that $\mathcal E^k_\mu = \mathcal E_\mu$ is a planar embedding of $\pert{\mu}$. The type of $\mathcal E^0_\mu$ is $t$, by hypothesis. Furthermore, by \cref{le:child-replacement}, if the type of $\mathcal E^{i-1}_\mu$ is $t$, then also the type of $\mathcal E^{i}_\mu$ is $t$. Thus, the type of $\mathcal E_\mu$ is $t$. This concludes the proof.
\end{fullproof}

In the following subsections, we show how to compute, in $O(|\skel(\mu)|)$ time, for each non-root node $\mu$ of $\mathcal T$, a set $\mathcal V_\mu$ of $O(1)$ embedding configurations for $\mu$ such that a type $t$ is admissible for $\mu$ if and only if there exists an embedding configuration $\mathcal C_\mu \in \mathcal V_\mu$ such that the embedding of $R(\mathcal C_\mu)$ is of type $t$.  This is shown in \cref{le:valid-set-s-vm-red} when $\mu$ is an S-node, in \cref{le:valid-set-p-vm-red,le:valid-set-p-vm-black} when $\mu$ is a P-node, and in \cref{lem:emb-configurations-bm-out,lem:emb-configurations-bm-pole,lem:emb-configurations-bm-internal} when $\mu$ is an R-node.
The listed lemmas, together with \cref{le:compute-type,le:relevant-in-configuration}, yield the following.

\begin{theorem}\label{algo:neat}
There exists an $O(n)$-time algorithm that tests whether 
 an $n$-vertex $rb$-augmented component admits a neat embedding.
\end{theorem}

\subsubsection{S-nodes}\label{se:s-nodes}

Suppose that $\mu$ is an S-node. 
By \cref{le:relevant-in-configuration}, in order to compute the admissible types of $\mu$, we construct the set of all the embedding configurations for $\mu$, and for each such an embedding configuration $\cal C_\mu$ we compute the type of the embedding of $R(\cal C_\mu)$, using \cref{le:compute-type}.

Recall that, by the definition of the SPQR-tree $\mathcal T$ we adopted, we have that $\skel(\mu)$ consists of a path of two edges, $e_1=(u_\mu,w)$ and $e_2=(w,v_\mu)$. Also, by \cref{remark:pan}, $\mu$ is the proper allocation node of $w$, whereas it is the proper allocation node of neither $u_\mu$ nor $v_\mu$; thus, an embedding configuration for $\mu$ needs to specify a designated face only for the vertex $w$.
We create all the embedding configurations for~$\mu$ obtained by combining all the following choices.
We set $\cal S_\mu$ to the unique embedding of $\skel(\mu)$, we choose $t_1$ and $t_2$ as each combination of admissible types for $e_1$ and for $e_2$, respectively, we choose $x_1$ and $x_2$ as each combination of $\ell$ and $r$, and we set the designated face of $w$ as each of $\ell(\mathcal S_\mu)$ or $r(\mathcal S_\mu)$, if there exists an $rb$-trivial component incident to $w$. Clearly, by \cref{pr:admissible-constant} and since $\skel(\mu)$ has $O(1)$ size, the above set of embedding configurations has $O(1)$ size. Further, since $\skel(\mu)$ has $O(1)$ size and by \cref{le:compute-type}, we~have~the~following.


\begin{lemma} \label{le:valid-set-s-vm-red}
	Suppose that $\mu$ is an S-node.
	There exists a set $\mathcal V_\mu$ of embedding configurations for~$\mu$ such that:
	\begin{enumerate}[\bf (i)]
		\item a type $t$ is admissible for $\mu$ if and only if there exists an embedding configuration $\mathcal C_\mu \in \mathcal V_\mu$ such that the embedding of $R(\mathcal C_\mu)$ is of type $t$;
		\item $|\mathcal V_\mu| \in O(1)$; and
		\item $\mathcal V_\mu$ can be constructed in $O(1)$ time.
	\end{enumerate}
\end{lemma}

%
%
\subsubsection{P-nodes}\label{se:p-nodes}

Suppose that $\mu$ is a P-node. 
Similarly to the case of S-nodes, we exploit \cref{le:relevant-in-configuration,le:compute-type} to compute in linear time the admissible types for $\mu$. 
Differently from the S-node case, however, $\skel(\mu)$ has not necessarily $O(1)$ size; further, the number of embeddings of $\skel(\mu)$, and hence the number of embedding configurations for $\mu$, is factorial in $|\skel(\mu)|$. Nonetheless, we are able to 
efficiently construct a constant-size subset $\cal V_\mu$ of the embedding configurations for $\mu$ with following property: For every admissible type $t$ for $\mu$, there exists an embedding configuration $\cal C_\mu \in \cal V_\mu$ such that the type of the embedding of $R(\cal C_\mu)$ is $t$. 
By \cref{remark:pan}, $\mu$ is the proper allocation node of neither $u_\mu$ nor $v_\mu$; thus, an embedding configuration for $\mu$ does not need to specify a designated face for any vertex of $\skel(\mu)$.
We distinguish two cases based on whether $v_\mu$ is red or black. Suppose first that $v_\mu$~is~red.

\begin{lemma}\label{le:p-node-R}
If $\mu$ is a P-node such that $v_\mu$ is red, then the type of $\mu$ is~\RF-RF and it has exactly two children. One of them is of type~\RE-RE and the other is of type~\RF-RF.
\end{lemma}

\begin{fullproof}
Since $\mu$ is a P-node with $v_\mu$ red, every child of $\mu$ has $v_\mu$ as a red pole, hence its type is \RE-RE or \RF-RF. By~\cref{le:final-edges}, $\mu$ has at most one child of type \RF-RF. By~\cref{le:structure-node-RE}, every child of $\mu$ of type \RE-RE is an edge, hence there is at most one such a child. It follows that $\mu$ has at most two children, one of type~\RF-RF and one of type~\RE-RE. Furthermore, since $\mu$ is a P-node, it has at least two children, from which the statement follows. 
\end{fullproof}

Let $e_1$ and $e_2$ be the virtual edges of $\skel(\mu)$.
We create all the embedding configurations for~$\mu$ obtained by combining all the following choices.
We set $\cal S_\mu$ to each of the two permutations of the virtual edges of $\skel(\mu)$, we choose $t_1$ and $t_2$ as each of the possible combinations of the admissible types for $e_1$ and for $e_2$, respectively, we choose $x_1$ and $x_2$ as each of the possible combinations of $\ell$ and $r$. Clearly, by \cref{pr:admissible-constant} and since $\skel(\mu)$ has $O(1)$ size, the above set of embedding configurations has $O(1)$ size. 
Therefore, the set $\cal V_\mu$ contains all the embedding configurations for $\mu$ and can be constructed in $O(1)$ time.
The above discussion yields the following.

\begin{lemma} \label{le:valid-set-p-vm-red}
	Suppose that $\mu$ is a P-node such that $v_\mu$ is red.
	There exists a set $\mathcal V_\mu$ of embedding configurations for $\mu$ such that:
	\begin{enumerate}[\bf (i)]
		\item a type $t$ is admissible for $\mu$ if and only if there exists an embedding configuration $\mathcal C_\mu \in \mathcal V_\mu$ such that the embedding of $R(\mathcal C_\mu)$ is of type $t$;
		\item $|\mathcal V_\mu| \in O(1)$; and
		\item $\mathcal V_\mu$ can be constructed in $O(|\skel(\mu)|)$ time.
	\end{enumerate}
\end{lemma}

Suppose now that $v_\mu$ is black.

\begin{lemma}\label{le:p-node-B}
If $\mu$ is a P-node such that $v_\mu$ is black, then $\mu$ has at most two children that are not of \hbox{type~\BE-BE slim}; each of them is of type~\BP-BP,~\BB-BB, or~\BF-BF. Also, if there exist exactly two such children, then one is of type~\BP-BP and the other is of type~\BF-BF.
\end{lemma}

\begin{fullproof}
Since $\mu$ is a P-node with $v_\mu$ black, every child of $\mu$ has $v_\mu$ as a black pole, hence its type is \BE-BE, \BP-BP, \BF-BF, or \BB-BB. By~\cref{le:type-C-path}, there is at most one child of $\mu$ whose type is \BP-BP or \BB-BB. Also, by~\cref{le:final-edges}, there is at most one child of $\mu$ whose type is \BF-BF or \BB-BB. Hence, if $\mu$ has two children that are not of \hbox{type~\BE-BE}, these are one of type \BP-BP and the other of type~\BF-BF.  
\end{fullproof}

Let $e_1,\dots,e_k$ be the virtual edges of $\skel(\mu)$.
We distinguish three cases, based on the number~$h$ of virtual edges of $\mu$ that are not of type \BE-BE slim. 
Note that, by \cref{le:p-node-B}, $h \leq 2$.

Suppose first that $h=0$. 
Observe that, in this case, the type of $\mu$ is \BE-BE fat.
Then the set $\cal V_\mu$ contains a unique embedding configuration for~$\mu$ obtained by setting $\cal S_\mu$ to an arbitrary permutation of the virtual edges of $\skel(\mu)$, and by choosing $t_i=$\BE-BE slim and $x_i=\ell$, for $i=1,\dots,k$.
Clearly, such an embedding configuration can be constructed in $O(k)$ time.

Suppose now that $h=1$. W.l.o.g., let $e_1$ be the virtual edge of $\skel(\mu)$ whose type is not \BE-BE slim.
Then the set $\cal V_\mu$ contains the embedding configurations for~$\mu$~obtained by combining all the following choices.
\begin{itemize}[$\circ$]

\item For any pair of non-negative integers $a$ and $b$ such that $a+b = k-1$ and $a \leq 1$, we set $\cal S_\mu$ to be any embedding of $\skel(\mu)$ such that $a$ virtual edges of type \BE-BE precede $e_1$ (and, thus, $b = k-a-1$ virtual edges of type \BE-BE follow $e_1$).
\item We choose $t_1$ as each of the admissible types for $e_1$, and we set $t_i=$\BE-BE slim, for $i=2,\dots,k$.
\item We choose $x_1$ as each of $\ell$ and $r$, and
$x_i=\ell$, for $i=2,\dots,k$.
\end{itemize}
Clearly, the above set of embedding configurations has $O(1)$ size, since $a \leq 1$ and since $e_1$ has a constant number of admissible types, by \cref{pr:admissible-constant}. 
Therefore, the set $\cal V_\mu$ can be constructed in $O(k)$ time.

Finally, suppose that $h=2$. W.l.o.g., let $e_1$ and $e_2$ be the virtual edges of $\skel(\mu)$ whose type is not \BE-BE slim.
Then the set $\cal V_\mu$ contains the embedding configurations for~$\mu$~obtained by combining all the following choices.

\begin{itemize}[$\circ$]
	\item For any triple of non-negative integers $a$, $b$, and $c$ such that $a+b+c = k-2$, $a \leq 1$, and $b \leq 1$, we set $\cal S_\mu$ to be any embedding of $\skel(\mu)$ such that $a$ virtual edges of type \BE-BE precede $e_1$ and $b$ virtual edges of type \BE-BE lie between $e_1$ and $e_2$ (and, thus, $c = k-2-a-b$ virtual edges of type \BE-BE follow $e_2$).
	\item We choose $t_1$ and $t_2$
	as each of the possible combinations of the admissible types for $e_1$ and for $e_2$, respectively, and $t_i=$\BE-BE slim, for $i=3,\dots,k$.
	\item We choose $x_1$ and $x_2$ as each of the possible combinations of $\ell$ and $r$, and
	$x_i=\ell$, for $i=3,\dots,k$.
\end{itemize}

We have the following.

\begin{lemma} \label{le:valid-set-p-vm-black}
Suppose that $\mu$ is a P-node such that $v_\mu$ is black.
	 There exists a set $\mathcal V_\mu$ of embedding configurations for $\mu$ such that:
\begin{enumerate}[\bf (i)]
	\item a type $t$ is admissible for $\mu$ if and only if there exists an embedding configuration $\mathcal C_\mu \in \mathcal V_\mu$ such that the embedding of $R(\mathcal C_\mu)$ is of type $t$;
	\item $|\mathcal V_\mu| \in O(1)$; and
	\item $\mathcal V_\mu$ can be constructed in $O(|\skel(\mu)|)$ time.
\end{enumerate}
\end{lemma}

\begin{proof}
Let $\mathcal V_\mu$ be the set constructed as described above. Clearly, $\mathcal V_\mu$  has $O(1)$ size, since $a, b \leq 1$ and since $e_1$ and $e_2$ have a constant number of admissible types, by \cref{pr:admissible-constant}. 
Therefore, the set $\cal V_\mu$ can be constructed in $O(k)$ time, with $k \in O(|\skel(\mu)|)$.
	
If there exists an embedding configuration~$\cal C_\mu$ for $\mu$ in the set $\cal V_\mu$ such that the embedding of $R(\cal C_\mu)$ is of type $t$, then we have that $t$ is admissible for $\mu$, by \cref{le:relevant-in-configuration}.

Consider any admissible type $t$ for $\mu$ and let $\mathcal E_\mu$ be a type-$t$ embedding of $\pert{\mu}$. We show that we inserted into $\mathcal V_\mu$ an embedding configuration~$\cal C_\mu$ such that the embedding of $R(\cal C_\mu)$~is~of~type~$t$.

Suppose first that $v_\mu$ is red. Then by construction we inserted into $\cal V_\mu$ all the embedding configurations for $\mu$. Thus, the statement follows by \cref{le:relevant-in-configuration}.

Suppose now that $v_\mu$ is black. We distinguish three cases, based on the number~$h \leq 2$ of virtual edges of $\skel(\mu)$ that are not of type \BE-BE slim. If $h=0$, then the type of $\mu$ is \BE-BE fat, and thus the only admissible type for $\mu$ is \BE-BE fat, by \cref{le:relevant-in-extensible-BE}. By construction, we inserted into $\cal V_\mu$ one embedding configuration $\cal C_\mu$ such that the embedding of $R(\cal C_\mu)$ is of type \BE-BE fat.

Suppose that $h=1$. 
Consider any two embedding configurations $\mathcal C'$ and $\mathcal C''$ for $\mu$ such that (i) the number of 
\BE-BE slim virtual edges preceding $e_1$ in the embedding of $\skel(\mu)$ is the same in $\mathcal C'$ as in $\mathcal C''$ and (ii) the type $t_1$ and the flip $x_1$ for $e_1$ are the same in $\mathcal C'$ as in $\mathcal C''$. Then the type of the embedding of $R(\mathcal C')$ is the same as the type of the embedding of  $R(\mathcal C'')$. Therefore, the type of an embedding configuration is determined only by the above criteria.

Let $a'$ and $b' = k - a' -1$ be the number of type \BE-BE slim virtual edges of $\skel(\mu)$ that precede and follow $e_1$, respectively, in the embedding of $\skel(\mu)$ determined by $\mathcal E_\mu$. W.l.o.g., we can assume that $a' \leq b'$, as otherwise we can flip $\mathcal E_\mu$, which does not alter its type.
Let $t^*_1$ be the type of the embedding of $\pert{e_1}$ in $\mathcal E_\mu$ and $x^*_1$ be its flip.

We claim that $a' \leq 1$. Observe that the claim implies the statement. Namely, if $a' \leq 1$, then by construction, the set $\cal V_\mu$ contains an embedding configuration $\mathcal C_\mu$ for $\mu$ such that the number of type \BE-BE slim virtual edges preceding $e_1$ in the embedding of $\skel(\mu)$ is $a=a'$, such that $t_1=t^*_1$, and such that $x_1=x^*_1$. Hence, the type of the embedding of $R(\mathcal C_\mu)$ is $t$.

We now prove the claim. Suppose, for a contradiction, that $a' > 1$ and thus $b' > 1$; refer to \cref{fig:troppi-be}. 
\begin{figure}[tb!]
	\centering
	\subfloat[]{
	\includegraphics[height=.25\textwidth,page=1]{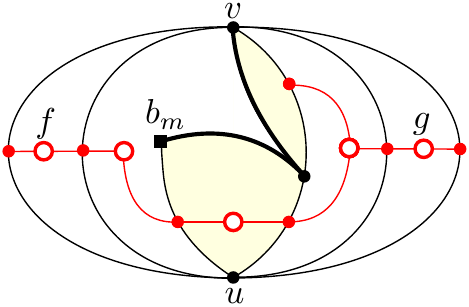} 
	\label{fig:troppi-be}
	}
	\subfloat[]{
	\includegraphics[height=.25\textwidth,page=2]{troppi-be} 
	\label{fig:troppi-be-h-2}
	}	
	\caption{Illustration for the proof of \cref{le:valid-set-p-vm-black}.	(a) 
Proof of $a' \leq 1$ when $h=1$. (b) Proof of $b' \leq 1$ when $h=2$.}
\end{figure}
Since the pertinent graphs of any two virtual edges of $\skel(\mu)$ that are consecutive in the embedding of $\skel(\mu)$ determined by $\mathcal E_{\mu}$ delimit an internal red face of $\mathcal E_\mu$, we have that the internal faces $f$ and $g$ of $\mathcal E_\mu$ delimited by the pertinent graphs of the leftmost two and of the rightmost two type \BE-BE slim virtual edges, respectively, are red.
We distinguish two cases based on whether there exists a black path in $\pert{\mu}$ between $u_\mu$ and $v_\mu$ or not. 

In the former case, both the type of $e_1$ and $\mu$ is either \BP-BP or \BB-BB. Hence, $\ell(\mathcal E_\mu)$ and $r(\mathcal E_\mu)$ belong to $A(\mathcal E_\mu)$, by \cref{def:auxiliary-graph-E-mu} and since there exists exactly one red vertex incident to each of $\ell(\mathcal E_\mu)$ and $r(\mathcal E_\mu)$. This implies that $A(\mathcal E_\mu)$ consists of two caterpillars whose backbones contain at least two vertices each, namely $f$ and $\ell(\mathcal E_\mu)$, and $g$ and $r(\mathcal E_\mu)$, respectively. This contradicts Item~\ref{le:structure-A-mu-BP-BB-item-no-traversing} of \cref{le:structure-A-mu-BP-BB}.

In the latter case, both the type of $e_1$ and $\mu$ is \BF-BF. Hence, neither $\ell(\mathcal E_\mu)$ nor $r(\mathcal E_\mu)$ belongs to $A(\mathcal E_\mu)$, by \cref{def:auxiliary-graph-E-mu}. Since $\mathcal E_\mu$ contains at least two internal red faces, namely $f$ and $g$, we have that the type of $\mathcal E_\mu$ is either BFI0A-\BFIza or BFI0B-\BFIzb. Furthermore, since $b_m$ is not incident to an outer face of $\mathcal E_\mu$, we have that the type of $\mathcal E_\mu$ is BFI0B-\BFIzb. Recall that $u_\ell$ and $u_r$ are the unique red vertices incident to $\ell(\mathcal E_\mu)$ and $r(\mathcal E_\mu)$, respectively. Since $f$ (resp. $g$) is the only red face incident to $u_\ell$ (resp. $u_r$), we have that $f$ and $g$ are the end-vertices of $A(\mathcal E_\mu)$ and that $u_\ell$ and $u_r$ are the only leaves incident to them. Since $b_m$ shares a face with neither $u_\ell$ nor $u_r$, we have a contradiction to Item~\ref{le:bf-endvertex-backbone-bm-shares} of \cref{le:bf-endvertex-backbone}.

 
Suppose that $h=2$. 
Similarly to the case $h=1$, given an embedding configuration $\mathcal C_\mu$, the type of an embedding of $R(\mathcal C_\mu)$ is only determined by the following criteria:
(i) the number of type \BE-BE slim virtual edges of $\skel(\mu)$ preceding $e_1$, following $e_1$ and preceding $e_2$, and following $e_2$ in the embedding $\mathcal S_\mu$ of $\skel(\mu)$,
(ii) the type $t_1$ (resp., $t_2$) and the flip $x_1$ (resp., $x_2$) for $e_1$ (resp., for $e_2$). 

Let $a'$, $b'$, and $c' = k-a'-b'-2$ be the number of type \BE-BE slim virtual edges of $\skel(\mu)$ that precede $e_1$, lie between $e_1$ and $e_2$, and follow $e_2$, respectively, in the embedding of $\skel(\mu)$ determined by $\mathcal E_\mu$. W.l.o.g., we can assume that $a' \leq c'$, as otherwise we can flip $\mathcal E_\mu$, which does not alter its type.
Let $t^*_1$ (resp., $t^*_2$) be the type of the embedding of $\pert{e_1}$ (resp., $\pert{e_2}$) in $\mathcal E_\mu$ and $x^*_1$ (resp., $x^*_2$) be its flip.

We claim that $a' \leq 1$ and $b' \leq 1$. Observe that the claim implies the statement. Namely, if $a' \leq 1$ and $b' \leq 1$, then by construction, the set $\cal V_\mu$ contains an embedding configuration $\mathcal C_\mu$ for $\mu$ such that the number of type \BE-BE slim virtual edges preceding $e_1$, lying between $e_1$ and $e_2$, and following $e_2$ in the embedding of $\skel(\mu)$ is $a=a'$, $b=b'$, and $c=c'$, and such that $t_1=t^*_1$, $t_2=t^*_2$, $x_1=x^*_1$, and $x_2=x^*_2$. Hence, the type of the embedding of $R(\mathcal C_\mu)$ is $t$.

Recall that, by \cref{le:p-node-B}, one of $e_1$ and $e_2$ is of type \BP-BP and the other one is of type \BF-BF, and thus $\mu$ is of type \BB-BB. The claim that $a' \leq 1$ can be proved similarly to the case in which $h=1$ and the type of $\mu$ is \BB-BB. Suppose, for a contradiction, that $b' > 1$; refer to \cref{fig:troppi-be-h-2}. Let $f$ be any internal red face of $\mathcal E_{\mu}$ delimited by the pertinent graphs of two \BE-BE slim virtual edges lying between $e_1$ and $e_2$. Denote by $w$ and $z$ the two red vertices incident to $f$. By \cref{le:structure-A-mu-BP-BB}, we have that $f$ belongs to the backbone $\mathcal B$ of a caterpillar composing $A(\mathcal E_\mu)$ having an outer face of $\mathcal E_\mu$ as an end-vertex; let $f^*$ be the other end-vertex of $\mathcal B$. Note that, when walking along $\mathcal B$ from the outer face to $f^*$, the vertices of the pertinent graph of the type \BF-BF virtual edge are encountered before $w$, and then $f$ is encountered (up to a relabeling of $w$ and $z$). Then either $f = f^*$ and $z$ is the only leaf incident to it, or $z$ belongs to $\mathcal B$. In both cases, there exists no leaf that is adjacent to $f^*$ in $A(\mathcal E_\mu)$ and that shares a face with $b_m$, since $b_m$ is a non-pole vertex of the type \BF-BF virtual edge. This contradicts Item \ref{le:structure-A-mu-BP-BB-item-bm-in-pertinent} of \cref{le:structure-A-mu-BP-BB} and concludes the proof.
\end{proof}

\subsubsection{R-nodes}\label{se:r-nodes}

Suppose that $\mu$ is an R-node. This case is much more difficult to handle than the S- and P-node cases, given that there might be a super-constant number of children of $\mu$ that are not of the ``trivial'' types \RE-RE and \BE-BE. In particular, while~\cref{le:final-edges} ensures that there is at most one child of $\mu$ which is of type \RF-RF, \BF-BF, or \BB-BB, there can be, in general, any number of children of type \BP-BP. Therefore, in order to {\em efficiently} construct the set of admissible types for $\mu$, we cannot construct all the (possibly exponentially many) embedding configurations for $\mu$ and just apply \cref{le:compute-type} to determine their types. However, we are able to efficiently construct a subset $\cal V_\mu$ of the embedding configurations of $\mu$ having size $O(|\skel(\mu)|)$ such that, for each admissible type $t$ of $\mu$, there exists an embedding configuration $\cal C_\mu \in \cal V_\mu$ such that the type of the embedding of $R(\cal C_\mu)$ is $t$. 



Recall that $\skel(\mu)$ has a unique planar embedding with its poles on the outer face, up to a flip.
Thus, we set $\cal S_\mu$ to be such an embedding for all the embedding configurations in $\cal V_\mu$. 
Also, recall that the left and the right outer face of $\cal S_\mu$ are denoted as $\ell(\cal S_\mu)$ and $r(\cal S_\mu)$, respectively.
We observe that there are some faces of $\cal S_\mu$ that correspond to red faces in any embedding of~$\pert{\mu}$. Faces of this type are called {\em intrinsically red} and are formally defined as follows. 

\begin{definition} \label{def:intrinsically-red}
	A face $f$ of $\cal S_\mu$ is \emph{intrinsically red} if:
	\begin{itemize}[-]
		\item it is incident to at least two red vertices of $\skel(\mu)$, or
		\item it is incident to at least two virtual edges of type \BE-BE or \BF-BF, or
		\item it is incident to a red vertex of $\skel(\mu)$ and to a virtual edge of type \BE-BE or \BF-BF, or
		\item it is an outer face, it is incident to a red vertex of $\skel(\mu)$, and the type of $\mu$ is \BP-BP~or~\BB-BB,~or  
		\item it is an outer face, it is incident to a virtual edge of type \BE-BE or \BF-BF, and the type of $\mu$ is \BP-BP or \BB-BB.
	\end{itemize} 
\end{definition}

In the remainder of the discussion, we assume that each face of $\mathcal S_{\mu}$ is labeled either as intrinsically red or as non-intrinsically red. Indeed, determining all the intrinsically red faces can be done in total $O(|\skel(\mu)|)$ time, by simply traversing the vertices and the virtual edges incident to each face of $\skel(\mu)$. 

Contrary to the intrinsically red faces, a non-intrinsically red face $f$ of $\cal S_\mu$ may result in a red face of a relevant embedding $\cal E_\mu$ of $\pert{\mu}$ or not; in the positive case, we refer to such a face as \emph{reddened} in $\mathcal E_\mu$. 
In particular, the fact that a non-intrinsically red face $f$ is reddened in $\mathcal E_\mu$ only depends on the type and the flip of the embedding of the pertinent graphs of the virtual edges of $\skel(\mu)$ bounding $f$ and 
and on the placement of the $rb$-trivial components incident to the vertices of $f$, if any. 

We now state necessary and sufficient conditions for a face of $\skel(\mu)$ to be reddened in a relevant embedding. 
We define the following three sets of types:
The set $\mathcal T_0$ contains the types \BFNz-BFN0, \BFIz-BFI0, \BPi-BP1, \RFNz-RFN0 and \RFIz-RFI0.
The set $\mathcal T_1$ contains the types \BPii-BP2, \BBii-BB2, \BPiv-BP4, \BBiv-BB4, \RFNi-RFN1, \RFIi-RFI1, \BFNi-BFN1, and \BFIi-BFI1.
The set $\mathcal T_2$ contains the types \BPiii-BP3, \BBiii-BB3, \BPv-BP5, \BBv-BB5, \RFNii-RFN2, \RFIii-RFI2, \BFNii-BFN2, and \BFIii-BFI2.
Note that, for each type $t$ in $\mathcal T_k$ and for each relevant embedding $\mathcal E_\mu$ of $\pert{\mu}$ whose type is $t$, we have that $A(\mathcal E_\mu)$ contains exactly $k$ outer faces of $\mathcal E_\mu$, by definition.

\begin{lemma}[Reddened conditions]\label{le:iff-reddened}
Let $f$ be a non-intrinsically red face of $\mathcal S_\mu$, let $\mathcal E_\mu$ be a relevant embedding of $\pert{\mu}$, and let $f^*$ be the face of $\mathcal E_\mu$ that corresponds to $f$.
Also, for a virtual edge $e$ of $\skel(\mu)$, we denote by $t_e$ and $x_e$ the type and the flip, respectively, of the embedding $\mathcal E_e$ of $\pert{e}$ determined by $\mathcal E_\mu$. Then $f$ is reddened in $\mathcal E_\mu$, i.e., $f^*$ is red, if and only if one of the following conditions holds:
\begin{enumerate}[(1)]
	\item \label{cond:iff-reddened-rb-trivial} there exists an $rb$-trivial component incident $f^*$; or
	\item \label{cond:iff-reddened-virtual-edges} there exists at least one virtual edge $e$ incident to $f$ such that 
	either 
	\begin{inparaenum}[(\ref{cond:iff-reddened-virtual-edges}a)]
	\item \label{cond:iff-reddened-virtual-edges-T2} $t_e \in \mathcal T_2$, or
	\item \label{cond:iff-reddened-virtual-edges-T1} $t_e \in \mathcal T_1$ and $x_e$ is such that the outer face of $\mathcal E_e$ that belongs to $A(\mathcal E_e)$ corresponds~to~$f^*$.
	\end{inparaenum}	
\end{enumerate}	
\end{lemma}

\begin{proof}
$(\Longleftarrow)$ We first prove the sufficiency. 

{\em Suppose first that there exists at least one virtual edge $d$ incident to $f$ whose type is neither \BP-BP nor \BB-BB}; by \cref{le:type-C-path}, this is always the case if $f$ is an internal face of $\mathcal S_\mu$, while it might not be the case if $f$ is an outer face of $\mathcal S_\mu$. The assumption implies that there exists at least one red vertex $r_d$ belonging to $\pert{d}$ that is incident to $f^*$. This vertex is one of the poles of $d$, if the type of $d$ is \RE-RE or \RF-RF, or is a non-pole vertex of $\pert{d}$, if the type of $d$ is \BE-BE or \BF-BF.

\begin{itemize}
	\item If \cref{cond:iff-reddened-rb-trivial} holds, there exists a red vertex $r$ incident to $f^*$ that belongs to an $rb$-trivial component. Since $r \neq r_d$, we have that $f^*$ is red.
	\item If one of \cref{cond:iff-reddened-virtual-edges-T2,cond:iff-reddened-virtual-edges-T1} holds, by \cref{def:auxiliary-graph-E-mu} of $A(\mathcal E_e)$, we have that, if the type of $e$ is neither \BP-BP nor \BB-BB, then there exist at least two red vertices of  $\pert{e}$ that are incident to $f^*$ (and thus $f^*$ is red), while if the type of $e$ is \BP-BP or \BB-BB, then there exists at least one (non-pole) red vertex of $\pert{e}$ that is incident to $f^*$; in the latter case, we have that $e \neq d$ and thus $f^*$ is red as it is incident to $r_d$ and to such a red vertex of $\pert{e}$.
\end{itemize}

{\em Suppose next that every virtual edge incident to $f$ is of type either \BP-BP or \BB-BB}; then the type of  $\mu$ is either \BP-BP or \BB-BB, hence, in order to prove that $f^*$ is red, it suffices to prove that there exists at least one red vertex incident to it. 

\begin{itemize}
	\item If \cref{cond:iff-reddened-rb-trivial} holds, there exists a red vertex $r$ incident to $f^*$ that belongs to an $rb$-trivial component, hence $f^*$ is red.
	\item If one of \cref{cond:iff-reddened-virtual-edges-T2,cond:iff-reddened-virtual-edges-T1} holds, by \cref{def:auxiliary-graph-E-mu} of $A(\mathcal E_e)$, we have that there exists at least one (non-pole) red vertex of $\pert{e}$ that is incident to $f^*$, hence $f^*$ is red.
\end{itemize}

$(\Longrightarrow)$ Second, we prove the necessity. Suppose hence that $f^*$ is red.

{\em Suppose that there exist two red vertices $r_1$ and $r_2$ incident to $f^*$.} This is always the case if $f$ is an internal face of $\mathcal S_\mu$, while it might not be the case if $f$ is an outer face of $\mathcal S_\mu$.
First, if any of $r_1$ and $r_2$ belongs to an $rb$-trivial component, then \cref{cond:iff-reddened-rb-trivial} holds.
If both $r_1$ and $r_2$ belong to the pertinent graph of the same virtual edge $e$ incident to $f$ or if at least one of them belongs to the pertinent graph of a virtual edge of type \BP-BP or \BB-BB incident to $f$, then either \cref{cond:iff-reddened-virtual-edges-T2} or \cref{cond:iff-reddened-virtual-edges-T1} holds, by \cref{def:auxiliary-graph-E-mu}.
Therefore, each of $r_1$ and $r_2$ is either a vertex incident to $f$ or belongs to the pertinent graph of a virtual edges incident to $f$, whose type is neither \BP-BP nor \BB-BB. However, this contradicts the fact that $f$ is a non-intrinsically red face of $\mathcal S_\mu$.

{\em Suppose next that there exists a single red vertex $r$ incident to $f^*$.} Since $f^*$ is red, by assumption, it follows that $f$ is an outer face of $\mathcal S_\mu$ and that the type of $\mu$ is either \BP-BP or \BB-BB. Since $f$ is not intrinsically red, it follows that every virtual edge incident to $f$ is of type either \BP-BP or \BB-BB. Hence $r$ is either a vertex of an $rb$-trivial component (and thus \cref{cond:iff-reddened-rb-trivial} is satisfied), or is a vertex in the pertinent graph of a type \BP-BP or \BB-BB virtual edge incident to $f$ (and thus one of \cref{cond:iff-reddened-virtual-edges-T2,cond:iff-reddened-virtual-edges-T1} is satisfied).
\end{proof}

\begin{remark}\label{rm:reddened}
Whether a non-intrinsically red face of $\skel(\mu)$ is reddened in a relevant embedding~$\mathcal E_\mu$ of $\pert{\mu}$ only depends on the embedding configuration for $\mu$ determined by $\mathcal E_\mu$.
\end{remark}

Next, we show that not too many non-intrinsically red faces of $\cal S_\mu$ are reddened in $\cal E_\mu$.

First, we consider an internal face $f$ of $\cal S_\mu$ that is reddened in some relevant embedding ${\cal E}_\mu$ of $\pert{\mu}$.
Let $f_N$ be the face of ${\cal E}_\mu$ that corresponds to $f$. By~\cref{le:structure-A-mu}, the auxiliary graph $A({\cal E}_\mu)$ is composed of at most two caterpillars. 
We denote by \backbone the backbone of the caterpillar containing $f_N$. We claim that \backbone ends either in $f_N$, or ``close to it'', i.e., in the pertinent graph of one of the virtual edges incident to $f$. We formalize this in the following lemma;~refer~to~\cref{fig:FN}.

\begin{figure}[tb!]
	\centering
	\subfloat[$\cal S_\mu$]{
		\includegraphics[height=.3\textwidth,page=3]{FN} 
		\label{fig:FN-skeleton}
	}\hfil
	\subfloat[$f_s=f_N$]{
		\includegraphics[height=.3\textwidth,page=1]{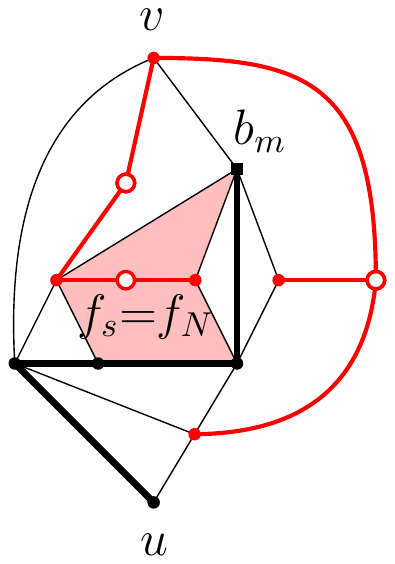} 
		\label{fig:FN-fsEQUALfN}
	}\hfil
	\subfloat[$f_s \neq f_N$]{
		\includegraphics[height=.3\textwidth,page=2]{FN} 
		\label{fig:FN-fsNOTEQUALfN}
	}
\caption{(a) Planar embedding $\cal S_\mu$ of the skeleton of an R-node $\mu$ of type \protect\RF-RF; the virtual edge $e$ corresponding to the unique non-Q-node child of $\mu$ is dashed; the face $f$ of $\cal S_\mu$ is not intrinsically red since the type of $e$ is \BP-BP. (b)-(c) Two distinct relevant embeddings of $\pert{\mu}$ containing a red face $f_N$ corresponding to $f$, which differ in the embedding type of $\pert{e}$ (type \protect\BPiii-BP3 in (b), and \protect\BPiv-BP4 in (c)). Face $f_N$ is an end-vertex of $\cal B_\mu$ in (b), while it is not~in~(c).}
\label{fig:FN}
\end{figure}

\begin{lemma} \label{le:few-non-intrinsically}
There is an end-vertex $f_s$ of the backbone \backbone containing $f_N$ such that either 
$f_s=f_N$ or there exists a virtual edge $e$ incident to $f$ such that
all the vertices, except for $f_N$, of the sub-path of \backbone between~$f_N$ and~$f_s$ correspond to vertices and internal faces of the embedding $\mathcal E_e$ of~$\pert{e}$~determined~by~$\mathcal E_\mu$. 
\end{lemma}

\begin{fullproof}
Suppose, for a contradiction, that the statement is false. Then there exist two distinct red faces $f'_N$ and $f''_N$ of $\mathcal E_\mu$ such that \backbone  $=(\dots,f'_N,v',f_N,v'',f''_N,\dots)$. Let $\mathcal B'_N$ (let $\mathcal B''_N$) be the subpath of \backbone between $f_N$ and an end-vertex of \backbone comprising $f'_N$ (resp.\ comprising $f''_N$). Assume that $v'$ belongs to the pertinent graph $H_e$ of a virtual edge $e$ of type \BP-BP or \BB-BB incident to $f$; let $\mathcal E_e$ be the embedding of $H_e$ determined by $\mathcal E_\mu$. Then, by \cref{le:structure-A-mu-BP-BB}, we have that $f'_N$ is a face of $\mathcal E_e$; again by the same lemma, all the internal faces of $\mathcal E_\mu$ in $\mathcal B'_N$ belong to $\mathcal E_e$, a contradiction. We can hence assume that $v'$ (and, analogously, $v''$) does not belong to the pertinent graph of a virtual edge of type \BP-BP or \BB-BB. 

Since $f$ is a non-intrinsically red face, there is at most one red vertex in $\skel(\mu)$ on the boundary of $f$. 
We distinguish the case in which there is no red vertex in $\skel(\mu)$ on the boundary of $f$ from the case in which there is one such vertex. 

If no red vertex of $\skel(\mu)$ is on the boundary of $f$, then there is no virtual edge incident to~$f$ whose type is \RE-RE or \RF-RF. Furthermore, since $f$ is a non-intrinsically red face, there is at most one virtual edge incident to~$f$ whose type is \BE-BE or \BF-BF. Indeed, since $f$ is an internal face of $\cal S_\mu$, then there is one edge whose type is \BE-BE or \BF-BF, as otherwise all the virtual edges incident to~$f$ would be of type \BP-BP or \BB-BB; these edges would form a cycle, which is not possible by~\cref{le:type-C-path}. Denote by $e^*$ the single virtual edge incident to $f$ whose type is \BE-BE or \BF-BF, if any. All the virtual edges incident to $f$ and different from $e^*$ are of type \BP-BP or \BB-BB. Since neither $v'$ nor $v''$ belongs to the pertinent graph of a virtual edge of type \BP-BP or \BB-BB, it follows that $v'$ and $v''$ belong to $\pert{e^*}$. Since $v'$ and $v''$ are distinct, it follows that the type of $e^*$ is \BF-BF. Let $\mathcal E_{e^*}$ be the embedding of $H_{e^*}$ determined by $\mathcal E_\mu$. By \cref{le:structure-A-mu-BF},  all the internal red faces of $\mathcal E_\mu$ in $\mathcal B'_N$ different from $f_N$  or all the internal red faces of $\mathcal E_\mu$ in $\mathcal B''_N$ different from $f_N$ belong to $\mathcal E_{e^*}$, a contradiction.

We now discuss the case in which there is a red vertex $w$ of $\skel(\mu)$ on the boundary of $f$. Since $f$ is a non-intrinsically red face, there is no virtual edge incident to $f$ whose type is \BE-BE or \BF-BF. The two virtual edges $e_1$ and $e_2$ of $\skel(\mu)$ that are on the boundary of $f$ and that are incident to $w$ are of type \RE-RE or \RF-RF. All the virtual edges incident to $f$ and different from $e_1$ and $e_2$ are hence of type \BP-BP or \BB-BB. Since neither $v'$ nor $v''$ belongs to the pertinent graph of a virtual edge of type \BP-BP or \BB-BB, we have that one between $e_1$ and $e_2$ is of type \RF-RF, while the other one is of type \RE-RE. Indeed, by \cref{le:final-edges}, we have that $e_1$ and $e_2$ are not both of type \RF-RF; further, $e_1$ and $e_2$ are not both of type \RE-RE, as otherwise we would have $v'=v''=w$, while $v'$ and $v''$ are distinct. Assume hence that the type of $e_1$ is \RF-RF, while the type of $e_2$ is \RE-RE. Then both $v'$ and $v''$ belong to $\pert{e_1}$ (possibly one of them coincides with $w$). By 
\cref{le:structure-A-mu-RF}, all the internal faces of $\mathcal E_\mu$ in $\mathcal B'_N$ or all the internal faces of $\mathcal E_\mu$ in $\mathcal B''_N$ belong to $\mathcal E_{e_1}$, a contradiction.
\end{fullproof}

Next, we prove that at most one internal non-intrinsically red face of $\cal S_\mu$ is reddened in $\cal E_\mu$.

\begin{lemma} \label{le:non-intrinsically-where}
There exists at most one internal non-intrinsically red face of $\cal S_\mu$ that is reddened in $\cal E_\mu$. Furthermore, such a face is incident to a virtual edge $e$ of $\skel(\mu)$ such that $b_m$ belongs~to~$\pert{e}$.
\end{lemma}

\begin{proof}
First, suppose, for a contradiction, that there exist two internal non-intrinsically red faces $f'$ and $f''$ of $\cal S_\mu$ that are reddened in $\cal E_\mu$. Also, let $f'_N$ and $f''_N$ be the faces of $\cal E_\mu$ that correspond to $f'$ and $f''$, respectively.
By \cref{le:few-non-intrinsically}, there exists an end-vertex $f'_s$ (resp.,\ $f''_s$) of a backbone \backbone of $A(\mathcal E_\mu$) such that either $f'_s=f'_N$ (resp.,\ $f''_s=f''_N$) or all the vertices, except for $f'_N$ (resp.,\ except for $f''_N$), of the sub-path of \backbone between~$f'_N$ and~$f'_s$ (resp.,\ between~$f''_N$ and~$f''_s$) correspond to vertices and internal faces of the embedding of the pertinent graph of a virtual edge incident to $f'$ (resp.,\ to $f''$) in $\cal E_\mu$.
Therefore, $f''_N$ does not belong to the subpath of \backbone between $f'_N$ and $f'_s$, and $f'_N$ does not belong to the subpath of \backbone between $f''_N$ and $f''_s$.
Since $f'_N \neq f''_N$, we have that $f'_s \neq f''_s$.
Also, by \cref{le:few-non-intrinsically}, we have that $f'_s$ and $f''_s$ are both internal faces of $\mathcal E_\mu$.

%

By definition, the only embedding types for which two end-vertices of the backbones of $A(\mathcal E_\mu)$ are internal faces of $\mathcal E_\mu$ are \RFIz-RFI0, \BFIz-BFI0, \BFIif-BFI1F, and \BE-BE fat. The latter is excluded, however, since $\mu$ is an R-node. 

\begin{itemize}
	\item Suppose that the type of $\mathcal E_\mu$ is \RFIz-RFI0. 
	Since $\mu$ is an R-node, we have that $u_\mu$ has at least two neighbors in $\pert{\mu}$, exactly one of which is black. This implies that there exists a red neighbor of $u_\mu$ incident to an outer face of $\mathcal E_\mu$. This neighbor must coincide with $v_\mu$, as otherwise one of the outer faces of $\mathcal E_\mu$ would be red and the type of $\mathcal E_\mu$ could not be \RFIz-RFI0. However, the existence of the edge $(u_\mu,v_\mu)$ implies that $\mu$ is a P-node, a contradiction. 
	\item Suppose that the type of $\mathcal E_\mu$ is \BF-BF. By \cref{le:structure-A-mu-BF}, we have that $A(\mathcal E_\mu)$ is a single caterpillar. Also, by Item~\ref{le:bf-endvertex-backbone-endvertex} of \cref{le:bf-endvertex-backbone} and since $f'_s$ and $f''_s$ are both internal faces of $\mathcal E_\mu$,  we have that at least one of them, say $f'_s$, is incident to $u_\mu$.  
	This implies that $f'_N$ and $f'$ are also incident to $u_\mu$ in $\mathcal E_\mu$ and $\mathcal S_\mu$, respectively. 
	Consider the two virtual edges incident to $u_\mu$ and to $f'$. Since the type of $\mu$ is \BF-BF, $u_\mu$ does not have any black neighbor in $\pert{\mu}$, and thus these virtual edges are of type either \RF-RF, or \RE-RE, or \BF-BF, or \BE-BE. This contradicts the fact that $f'$ is non-intrinsically red. 
\end{itemize}

This concludes the proof that there exist no two internal non-intrinsically red faces $f'$ and $f''$ of $\cal S_\mu$ that are reddened in $\cal E_\mu$. Suppose now that an internal non-intrinsically red face  $f'$  that is reddened in $\cal E_\mu$ exists. 
Since $\mu$ is an R-node, its type is either \BP-BP, or \BB-BB,  or \BF-BF, or \RF-RF.
The second part of the statement follows 
from Item~\ref{le:rf-endvertex-backbone-bm} of \cref{le:rf-endvertex-backbone}, if the type of $\mu$ is \RF-RF;
from 
Item~\ref{le:structure-A-mu-BP-BB-item-bm-in-pertinent} of \cref{le:structure-A-mu-BP-BB}, if the type of $\mu$ is \BP-BP or \BB-BB; and
from Item~\ref{le:bf-endvertex-backbone-bm-shares} of \cref{le:bf-endvertex-backbone}, if the type of $\mu$ is \BF-BF.
This concludes the proof.
\end{proof}

Clearly, since $\mathcal S_\mu$ has two outer faces, the number of outer faces that can be reddened in a relevant embedding of $\pert{\mu}$ is also constant. However, in the next lemma, we prove a stronger result.

\begin{lemma}\label{lem:R-external-reddened}
There exists at most one outer non-intrinsically red face $f$ of $\cal S_\mu$ that is reddened in a relevant embedding $\cal E_\mu$ of $\pert{\mu}$. 
\end{lemma}

\begin{proof}
Since $\mu$ is an R-node, its type is neither \RE-RE, by \cref{le:structure-node-RE}, nor \BE-BE, by \cref{le:structure-node-BE}.

Suppose that the type of $\mu$ is \BP-BP or \BB-BB. Then one of the outer faces of $\mathcal S_\mu$ is intrinsically red, by definition and since its type is not \BPi-BP1, by \cref{obs:bp1-black-path}. Thus, the statement trivially follows in this case.

Suppose that the type of $\mu$ is \RF-RF. We claim that also in this case one of the outer face of $\mathcal S_\mu$ is intrinsically red. Namely, first observe that the red pole $v_\mu$ of $\mu$ is incident to both the outer faces of $\mathcal S_\mu$. Also, the black pole $u_\mu$ of $\mu$ is incident to at most one type \BP-BP virtual edge. Since $\mu$ is an R-node, $u_\mu$ is incident to at least two virtual edges; hence, $u_\mu$ is incident to a virtual edge $e$ that is not of type  \BP-BP and that is incident to an outer face of $\mathcal S_\mu$, say to $\ell(\mathcal S_\mu)$.
The other end-vertex of $e$ is different from $v_\mu$, as otherwise $\mu$ would be a P-node. Then either the type of $e$ is \BE-BE or \BF-BF, or the other pole of $e$ is a red vertex of $\skel(\mu)$ (different from $v_\mu$) that is incident to
$\ell(\mathcal S_\mu)$; in all the cases, it follows that $\ell(\mathcal S_\mu)$ is intrinsically red.

Suppose finally that the type of $\mu$ is \BF-BF. We claim that also in this case one of the outer faces of $\mathcal S_\mu$ is intrinsically red. Namely, since $\mu$ is an R-node, it has two distinct red neighbors $u_\ell$ and $u_r$ that are incident to $\ell(\mathcal S_\mu)$ and $r(\mathcal S_\mu)$, respectively; note that the virtual edges $(u_\mu,u_\ell)$ and $(u_\mu,u_r)$ are of type \RE-RE, since $u_\mu$ has no black neighbor in $\pert{\mu}$. Also, the pole $v_\mu$ of $\mu$ has one black neighbor in $\pert{\mu}$. Hence, one of the two outer faces of $\mathcal S_\mu$, say $\ell(\mathcal S_\mu)$, is such that the pertinent graph of the virtual edge $e$ incident to $v_\mu$ and to $\ell(\mathcal S_\mu)$ contains no black neighbor of $v_\mu$. If the other end-vertex of $e$ is not $u_\ell$, then the claim follows as in the case in which the type of $\mu$ is \RF-RF.  Otherwise, the virtual edges $(u_\mu,u_\ell)$ and $e=(v_\mu,u_\ell)$ are the only virtual edges incident to $\ell(\mathcal S_\mu)$. However, this implies  that $\ell(\mathcal S_\mu)$ can not be reddened. Namely, $(u_\mu,u_\ell)$ and $(v_\mu,u_\ell)$ are both of type \RE-RE; this was proved above for $(u_\mu,u_\ell)$ and it follows from the fact that the pertinent graph of $(v_\mu,u_\ell)$ contains no black neighbor of $v_\mu$ for $(v_\mu,u_\ell)$. Finally, the only black vertices incident to $\ell(\mathcal S_\mu)$ are the poles of $\mu$, and thus their $rb$-trivial components, if any, are undecided. This concludes the proof.
\end{proof}

In order to determine the desired set of embedding configurations for $\mu$, we distinguish three cases, based on the position of $b_m$ with respect to $\pert{\mu}$.

\smallskip
\noindent\fbox{\bf Case 1: $b_m$ does not belong to $\pert{\mu}$.}
We start with a structural lemma.


\begin{lemma}\label{R:bm-external}
Let $\mu$ be an R-node. Suppose that $b_m$ does not belong to $\pert{\mu}$ and that $\pert{\mu}$ admits a relevant embedding $\cal E_\mu$.
For a virtual edge $e$ of $\skel(\mu)$, let $\mathcal E_e$ be the embedding of $\pert{e}$ determined by $\cal E_\mu$.
Then the following conditions are satisfied:
\begin{enumerate}
	\item \label{R:bm-external-BP} the type of $\mu$ is \BP-BP;
	\item 
	\label{R:bm-external-BP-types}
	the type of $\cal E_\mu$ is either \BPii-BP2 or \BPiii-BP3;
	\item \label{R:bm-external-no-intrinsically-red} $\skel(\mu)$ contains no internal intrinsically red face;
	\item 
	\label{R:bm-external-possible-virtual-edges}
	the virtual edges of $\skel(\mu)$ are of type \RE-RE, \BE-BE slim, or \BP-BP;
	\item \label{R:bm-external-BP-virtual-edges} for each type \BP-BP virtual edge $e$ of $\skel(\mu)$, the type of $\mathcal E_e$ is either \BPi-BP1 or \BPii-BP2; 
	\item \label{R:bm-external-external-face} for each virtual edge $e$ such that the type of $\mathcal E_e$ is \BE-BE slim or \BPii-BP2, we have that $e$ is incident to an outer face $x(\cal S_\mu)$ of $\cal S_\mu$, with $x \in \{\ell,r\}$; also, if the type of $\mathcal E_e$ is \BPii-BP2, then~$\mathcal E_e$ is flipped in such a way that the red vertices of $\pert{e}$ are incident to $x(\cal S_\mu)$; and
	\item 
	\label{R:bm-external-internal-rb-trivial}
	each $rb$-trivial component incident to a vertex of $\skel(\mu)$ is incident to an outer face of $\mathcal E_\mu$.
\end{enumerate}
\end{lemma}

\begin{proof}
\cref{R:bm-external-BP}: Since $b_m$ is not a vertex of $\pert{\mu}$, the type of $\mu$ is neither \RF-RF, nor \BB-BB, nor \BF-BF. Also, since $\mu$ is an R-node, its type is neither \RE-RE, by \cref{le:structure-node-RE}, nor \BE-BE, by \cref{le:structure-node-BE}. Thus, the type of $\mu$ is \BP-BP.

\cref{R:bm-external-BP-types}: First, the type of $\mathcal E_\mu$ is not \BPi-BP1, as otherwise, by \cref{obs:bp1-black-path}, $\pert{\mu}^-$ would be a path composed of black vertices between the poles of $\mu$, and hence $\mu$ would be either a Q-node or an S-node, while it is an R-node.
Also, the type of $\mathcal E_\mu$ is neither \BPiv-BP4 nor \BPv-BP5, as otherwise, by Item~\ref{le:structure-A-mu-BP-BB-item-bm-in-pertinent} of \cref{le:structure-A-mu-BP-BB}, $b_m$ would belong to $\pert{\mu}$. Thus, the type of $\mathcal E_\mu$ is either \BPii-BP2 or \BPiii-BP3.

Before discussing the remaining conditions, we observe that \cref{R:bm-external-BP-types} and the definition of types \BPii-BP2 and \BPiii-BP3 imply that there exists no internal red face in $\mathcal E_\mu$. Furthermore, by \cref{le:internal-red-vertex}, there exists no internal red vertex in $\mathcal E_\mu$.

\cref{R:bm-external-no-intrinsically-red} is then satisfied as otherwise $\mathcal E_\mu$ would contain an internal red face.

\cref{R:bm-external-possible-virtual-edges,R:bm-external-BP-virtual-edges,R:bm-external-external-face}: Consider first a virtual edge $e$ of $\skel(\mu)$ that is not incident to an outer face of $\mathcal S_\mu$. We have that either the type of $e$ is \RE-RE or it is \BP-BP and the type of its unique embedding is \BPi-BP1, as otherwise, by definition, $\pert{e}$ would contain a non-pole red vertex, and thus $\mathcal E_{\mu}$ would contain an internal red vertex, which we proved not to be the case. Consider now a virtual edge $e$ of $\skel(\mu)$ that is incident to an outer face $x(\cal S_\mu)$ of $\cal S_\mu$, with $x \in \{\ell,r\}$. Since $b_m$ does not belong to $\pert{\mu}$, we have that the type of $e$ is neither \RF-RF, nor \BB-BB, nor \BF-BF; further, the type of $\mathcal E_e$ is neither \BPiv-BP4 nor \BPv-BP5, as otherwise $b_m$ would belong to $\pert{e}$, and thus to $\pert{\mu}$. Moreover, the type of $\mathcal E_e$ is not \BE-BE fat, as otherwise $\mathcal E_e$, and thus $\mathcal E_\mu$, would contain an internal red face, which we proved not to be the case. Finally, the type of $\mathcal E_e$ is not \BPiii-BP3, as otherwise $\mathcal E_\mu$ would contain an internal red vertex, which we proved not to be the case; in fact, since $\mu$ is an R-node, the face of $\mathcal S_\mu$ incident to $e$ and different from $x(\cal S_\mu)$ is an internal face of $\mathcal S_\mu$, and $\mathcal E_e$ has red vertices incident to both its outer faces, by definition. For the same reason, if the type of $e$ is \BPii-BP2, then~$\mathcal E_e$ is flipped in such a way that the red vertices of $\pert{e}$ are incident to $x(\cal S_\mu)$.

Finally, \cref{R:bm-external-internal-rb-trivial} is satisfied as otherwise $\mathcal E_\mu$ would contain an internal red face.
\end{proof}

In order to compute the set of admissible types for $\mu$, we first use \cref{R:bm-external} to determine several cases in which the set of admissible types for $\mu$ is empty. To do so, we check in $O(|\skel(\mu)|)$ time whether one of the following conditions is fulfilled.
\begin{enumerate}
\item $\cal S_\mu$ contains an internal intrinsically-red face (see \cref{R:bm-external-no-intrinsically-red} of \cref{R:bm-external}),
\item there exists at least one virtual edge $e$ of $\skel(\mu)$ that is not of type \RE-RE, \BE-BE slim, or \BP-BP (see \cref{R:bm-external-possible-virtual-edges} of \cref{R:bm-external})
\item $\skel(\mu)$ contains a type-\BP-BP virtual edge whose pertinent graph admits no relevant embedding of type \BPi-BP1 or \BPii-BP2
(see \cref{R:bm-external-BP-virtual-edges} of \cref{R:bm-external}),
\item $\skel(\mu)$ contains a virtual edge not incident to an outer face of $\mathcal S_\mu$ whose pertinent graph only admits a relevant embedding of type \BPii-BP2 (see \cref{R:bm-external-external-face} of \cref{R:bm-external}), or
\item there exists an $rb$-trivial component that is incident to an internal vertex of $\skel(\mu)$ (see \cref{R:bm-external-internal-rb-trivial} of \cref{R:bm-external}). 
\end{enumerate}

If none of the above tests succeeds, then we insert into $\cal V_\mu$ a single embedding configuration defined as follows. 
For each virtual edge $e_i$ of $\skel(\mu)$, we set $t_i$ to \RE-RE or \BE-BE slim, if the type of $e_i$ is \RE-RE or \BE-BE slim, respectively. Further, consider each virtual edge $e_i$ of type \BP-BP. By \cref{R:bm-external-BP-virtual-edges} of \cref{R:bm-external}, we have that $t_i$ may be one of \BPi-BP1 or \BPii-BP2; however, by \cref{obs:bp1-black-path}, the set of admissible types for $e_i$ contains either \BPi-BP1, if the type of $\mu$ is \BPi-BP1, or \BPii-BP2, otherwise. Therefore, we set $t_i$ to be the only admissible type for $e_i$. Note that, by \cref{R:bm-external-possible-virtual-edges} of \cref{R:bm-external}, we have exhaustively considered all the virtual edges of $\skel(\mu)$. Furthermore, for each virtual edge $e_i$ that is not incident to an outer face of $\cal S_\mu$ (and thus the type of $e_i$ is \RE-RE or \BPi-BP1, by \cref{R:bm-external-BP-virtual-edges} of \cref{R:bm-external}), we arbitrarily set $x_i$ to either $\ell$ or~$r$; further, for each virtual edge $e_i$ that is incident to $x({\cal S_\mu})$, with $x \in \{\ell,r\}$, we set $x_i=x$ to comply with \cref{R:bm-external-BP-virtual-edges} of \cref{R:bm-external}. Finally, for each vertex $w$ of $\skel(\mu)$ that is incident to an $rb$-trivial component, we set the designated face $f_w$ of $w$ to the outer face of $\cal S_\mu$ vertex $w$ is incident to (this incidence is guaranteed by \cref{R:bm-external-internal-rb-trivial} of \cref{R:bm-external}).
By the above discussion, we have the following.

\begin{lemma}\label{lem:emb-configurations-bm-out}
Suppose that $\mu$ is an R-node such that $b_m$ does not belong to $\pert{\mu}$. There exists a set $\mathcal V_\mu$ of embedding configurations for $\mu$ such that:
\begin{enumerate}[\bf (i)]
	\item a type $t$ is admissible for $\mu$ if and only if there exists an embedding configuration $\mathcal C_\mu \in \mathcal V_\mu$ such that the embedding of $R(\mathcal C_\mu)$ is of type $t$;
	\item $|\mathcal V_\mu|=1$; and
	\item $\mathcal V_\mu$ can be constructed in $O(|\skel(\mu)|)$ time.
\end{enumerate}
\end{lemma}

\smallskip

\noindent\fbox{\bf Case 2: $b_m$ is a vertex of $\skel(\mu)$.} We start with a structural lemma.

\begin{lemma}\label{R:bm-pole}
	Let $\mu$ be an R-node. Suppose that $b_m$ is a vertex of $\skel(\mu)$ and that $\pert{\mu}$ admits a relevant embedding $\cal E_\mu$. For a virtual edge $e$ of $\skel(\mu)$, let $\mathcal E_e$ be the embedding of $\pert{e}$ determined by $\cal E_\mu$.
	Then the following conditions are satisfied:
	\begin{enumerate}
		\item \label{R:bm-pole-BP} the type of $\mu$ is \BP-BP if $b_m = v_{\mu}$ is a pole of $\mu$, and is either \RF-RF, or \BF-BF, or \BB-BB, otherwise;
		\item \label{R:bm-pole-possible-virtual-edges}
		the virtual edges of $\skel(\mu)$ are of type \RE-RE, \BE-BE slim, \BE-BE fat, or \BP-BP; 
		\item \label{R:bm-pole-BP-virtual-edges} \label{R:bm-pole-incident-bp} there exists at most one type \BP-BP virtual edge $e$ of $\skel(\mu)$ such that the type of $\mathcal E_e$ is \BPiv-BP4 or \BPv-BP5; further, there exists a single virtual edge $e_m$ of type \BP-BP incident to $b_m$; finally, if a virtual edge $e$ exists such that the type of $\mathcal E_e$ is \BPiv-BP4 or \BPv-BP5, then $e=e_m$.
	\end{enumerate}
\end{lemma}

\begin{proof}
	If $b_m$ is a pole of $\mu$, then
	the proof for \cref{R:bm-pole-BP} is the same as the one for \cref{R:bm-external-BP} of \cref{R:bm-external}.
	Otherwise, the type of $\mu$ is either \RF-RF, or \BF-BF, or \BB-BB, since $b_m$ is a non-pole vertex of $\pert{\mu}$.

	Also, \cref{R:bm-pole-possible-virtual-edges} holds, since the type of each virtual edge $e$ of $\skel(\mu)$ is neither \RF-RF, nor \BB-BB, nor \BF-BF, as $b_m$ belongs to $\skel(\mu)$.

	Finally, we prove that \cref{R:bm-pole-BP-virtual-edges} hold. First, $b_m$ is incident to exactly one virtual edge $e_m$ of type \BP-BP, since it has exactly one black neighbor in $\pert{\mu}$, given that the type of $\mu$ is not~\BE-BE by~\cref{R:bm-pole-BP}. Suppose that a type \BP-BP virtual edge $e$ of $\skel(\mu)$ such that the type of $\mathcal E_e$ is \BPiv-BP4 or \BPv-BP5 exists, as otherwise there is nothing else to prove. By Item~\ref{le:structure-A-mu-BP-BB-item-bm-in-pertinent} of \cref{le:structure-A-mu-BP-BB}, one of the poles of $e$ is $b_m$. Hence, $e=e_m$ and there exists no virtual edge $e' \neq e$ of $\skel(\mu)$ such that the type of $\mathcal E_{e'}$ is \BPiv-BP4 or \BPv-BP5. This concludes the proof.
\end{proof}



Differently from to the case in which $b_m \notin \pert{\mu}$, we do not perform any preliminary test. In fact, all the conditions of \cref{R:bm-pole} only descend from the structure of $\pert{\mu}$ and not from its embedding; therefore, they are guaranteed to hold.

In the following, we exploit the conditions of \cref{le:non-intrinsically-where,lem:R-external-reddened,R:bm-pole} to compute a constant-size set $\cal V_\mu$ of embedding configurations for $\mu$. 
In particular, by \cref{le:non-intrinsically-where}, there exists at most one internal non-intrinsically red face of $\mathcal S_\mu$ that may be reddened in a relevant embedding of $\pert{\mu}$ and such a face has to be incident to $b_m$. Moreover, such a face must be one of the two faces $f_m$ and $g_m$ of $\mathcal S_{\mu}$ incident to $e_m$. Namely, consider any face $f\notin \{f_m,g_m\}$ of $\mathcal S_{\mu}$ that is incident to $b_m$; the two virtual edges that lie along the boundary of $f$ and that are incident to $b_m$ are of type~\RE-RE or \BE-BE, by~\cref{R:bm-pole-possible-virtual-edges,R:bm-pole-BP-virtual-edges} of \cref{R:bm-pole}, hence $f$ is intrinsically red. Furthermore, by \cref{lem:R-external-reddened}, we have that at most one of the outer faces of $\mathcal S_{\mu}$, say $r(\mathcal S_{\mu})$, is non-intrinsically red and is allowed to be reddened in a relevant embedding of $\pert{\mu}$.


We will exploit the fact that, by \cref{rm:reddened}, whether a non-intrinsically red face in the set $\{f_m,g_m,r(\mathcal S_\mu)\}$
is reddened in some relevant embedding $\mathcal E_\mu$ of $\pert{\mu}$ only depends on the embedding configuration for $\mu$ determined by  
$\mathcal E_\mu$.  
%
Hence, our strategy to compute the embedding configurations to be added to $\cal V_\mu$ is to consider the six pairs $(x,y)$, with $x \in \{\emptyset, f_m, g_m \}$ and $y \in \{\emptyset, r(\mathcal S_\mu)\}$. For each pair $(x,y)$, we will exploit \cref{le:iff-reddened} to construct $O(1)$ embedding configurations each of which is such that, in the corresponding realization, the reddened faces of $\mathcal S_\mu$ are exactly $x$ and $y$ (where $x =\emptyset$ means that none of $f_m$ and $g_m$ is reddened, and $y=\emptyset$ means that $r(\mathcal S_\mu)$ is not reddened). In this case, we say that the pair is \emph{satisfied} by the embedding configuration.
Recall that \BPi-BP1 belongs to $\mathcal T_0$, that \BPii-BP2 and \BPiv-BP4 belong to~$\mathcal T_1$, and that \BPiii-BP3 and \BPv-BP5 belong to $\mathcal T_2$.

Note that some pairs may not be satisfied by any embedding configuration. For example, this is the case for the pair $(f_m,\emptyset)$, when $f_m$ coincides with $r(\mathcal S_\mu)$. In fact, $x = f_m$ requires $f_m = r(\mathcal S_\mu)$ to be reddened, while $y = \emptyset$ requires $f_m=r(\mathcal S_\mu)$ not to be reddened. Another example is a pair $(f_m, \cdot)$, when $f_m$ is non-intrinsically red, all the type \BP-BP virtual edges incident to it are of type \BPi-BP1, and no vertex incident to $f_m$ is also incident to an $rb$-trivial component, by \cref{le:iff-reddened}. On the other hand, some pairs may be equivalent in the restrictions they impose, for example the pairs $(\emptyset, \cdot)$ and $(f_m, \cdot)$, when $f_m$ is intrinsically red.  


Let $(x,y)$ be one of the six pairs. We now describe in detail how to construct a set of $O(1)$ embedding configurations such that $(x,y)$ is satisfied by each of such configurations. In particular, these embedding configurations will differ only in the type $t_m$ and the flip $x_m$ chosen for the virtual edge $e_m$. We discuss the case in which 
\begin{inparaenum}[\bf (i)]
	\item $x, y \neq \emptyset$, 
	\item $x \neq y$, and 
	\item both $x$ and $y$ are non-intrinsically red, that is, when both $x$ and $y$ need to be reddened;  
\end{inparaenum}
the other cases are analogous (in fact, simpler).
The construction consists of three phases.
\smallskip

\paragraph{Satisfying the pair (x,y).}
We first compute the set $S_x$ (resp. $S_y$) that contains all the type \BP-BP virtual edges incident to $x$ (resp. to $y$) that are not of type \BPi-BP1, and all the black vertices of $\skel(\mu)$ incident to $x$ (resp. to $y$) that are incident to an $rb$-trivial component. By \cref{le:iff-reddened}, the constructed sets contain exactly those elements that can be used to ensure that $(x,y)$ is satisfied by the embedding configuration.

If $S_x$ (resp. $S_y$) is empty, then we abort the construction of the embedding configurations for the pair $(x,y)$, since it is not possible to ensure that $x$ (resp. $y$) is reddened; in this case, $(x,y)$ cannot be satisfied.

Suppose next that $S_x = S_y$ and $|S_x| = |S_y| = 1$. If the only element in $S_x = S_y$ is a black vertex or a virtual edge  that does not have \BPiii-BP3 or \BPv-BP5 among its admissible types, then we abort the construction of the embedding configurations for the pair $(x,y)$, since it is not possible to ensure that both $x$ and $y$ are reddened, by \cref{le:iff-reddened}; also in this case, $(x,y)$ cannot be satisfied. 
Otherwise, $S_x$ contains exactly one virtual edge $e_i$, which has \BPiii-BP3 or \BPv-BP5 among its admissible types. If $e_i \neq e_m$, then 
\BPiii-BP3 is admissible for $e_i$,
by \cref{R:bm-pole-incident-bp} of \cref{R:bm-pole}.
We set $t_i=$~BP3 and we arbitrarily set $x_i$ to either $\ell$ or~$r$; whereas, if $e_i=e_m$, then we make up to four independent choices by setting $t_i$ as each of the admissible types for $e_i$ among \BPiii-BP3 and \BPv-BP5, and by setting $x_i$ as each of $\ell$ and $r$. This ensures that both $x$ and $y$ are reddened, by \cref{cond:iff-reddened-virtual-edges-T2} of \cref{le:iff-reddened}. Thus, $(x,y)$ is satisfied by any embedding configuration that complies with one of these choices.

In all the other cases (that is, if $|S_x| \geq 1$, $|S_y| \geq 1$, and it holds that $|S_x| \geq 2$, or $|S_y| \geq 2$, or $S_x \neq S_y$), we can find two distinct elements, one in $S_x$ and one in $S_y$. We use these elements to ensure that both $x$ and $y$ are reddened. Namely, consider the element that belongs to $S_x$, the discussion for the one belonging to $S_y$ being analogous. If this element is a vertex $w$, then we select $x$ to be the designated face $f_w$ of $w$, by \cref{cond:iff-reddened-rb-trivial}
 of \cref{le:iff-reddened}. Otherwise, this element is a virtual edge $e_i$, and we proceed as follows.

\begin{itemize}
	\item Suppose that $e_i \neq e_m$.
	If \BPii-BP2 is admissible for $e_i$, then we set $t_i=$~BP2 and we select $x_i$ so that the red outer face of the $t_i$-gadget will correspond to $x$ in the realization of an embedding configuration, by \cref{cond:iff-reddened-virtual-edges-T1} of \cref{le:iff-reddened}.
	Otherwise, \BPiii-BP3 is admissible for $e_i$, by \cref{R:bm-pole-BP-virtual-edges} of \cref{R:bm-pole}. If the face $z$ of $\mathcal S_\mu$ incident to $e_i$ different from $x$ is either intrinsically red or coincides with $y$, then we set $t_i=$~BP3 and we arbitrarily set $x_i$ to either $\ell$ or $r$. Else, we abort the construction of the embedding configurations for the pair $(x,y)$, since it is not possible to ensure that $x$ is reddened, while avoiding that $z$ is reddened, by \cref{cond:iff-reddened-virtual-edges-T2} of \cref{le:iff-reddened}.
	\item Suppose now that $e_i = e_m$. Let  $z \in \{f_m,g_m\}$ be the face of $\mathcal S_\mu$ incident to $e_i$ different from~$x$.
	If $z$ is intrinsically red or $z = y$, then we make up to six independent choices by selecting $t_i$ as each of the admissible types for $e_i$ (recall that the type of $e_i$ is not \BPi-BP1, since $e_i \in S_x$) and by setting $x_i$ to each of the flips such that a red outer face of the $t_i$-gadget will correspond to $x$ in the realization of an embedding configuration  (observe that, if $t_i \in \{\textrm{BP2, BP4}\} \subseteq \mathcal T_1$, then we select only one flip for $x_i$, while if $t_i \in \{\textrm{BP3, BP5}\} \subseteq \mathcal T_2$, then we select both possible flips for $x_i$).
	If $z$ is non-intrinsically red and $z \neq y$, if the set of admissible types for $e_i$ contains neither \BPii-BP2 nor \BPiv-BP4, we abort the construction of the embedding configurations for the pair $(x,y)$, since it is not possible to ensure that $x$ is reddened, while avoiding that $z$ is reddened, by \cref{cond:iff-reddened-virtual-edges-T2} of \cref{le:iff-reddened}. Otherwise, the set of admissible types for $e_i$ contains \BPii-BP2, \BPiv-BP4, or both, and we make up to two independent choices by setting $t_i$ as each of the admissible types for $e_i$ among \BPii-BP2 and \BPiv-BP4, and by setting $x_i$ so that the red outer face of the $t_i$-gadget will correspond to~$x$ in the realization of an embedding configuration, by \cref{cond:iff-reddened-virtual-edges-T1} of \cref{le:iff-reddened}.
\end{itemize}

If the above procedure has not aborted, the above ``partial'' embedding configurations satisfy the pair $(x,y)$. Note that, so far, we only set the types and the flips for at most two virtual edges, one in order to redden $x$ and one in order to redden $y$; moreover, multiple independent choices have been done only if one of these virtual edges coincides with $e_m$. To complete the construction of the embedding configurations for $\mu$, we have to select a type and a flip for the remaining virtual edges in each of the partial embedding configurations constructed above. We first deal with edges of type~\BP-BP.

\smallskip

\paragraph{Handling the remaining type \BP-BP virtual edges.}
Let $e_i$ be any virtual edge of $\skel(\mu)$ of type \BP-BP that is not of type \BPi-BP1 and that has not been processed so far to satisfy the pair $(x,y)$. 
Let $f_i$ and $g_i$ be the two faces of $\mathcal S_\mu$ that are incident to $e_i$. 

If both $f_i$ and $g_i$ are non-intrinsically red faces different from $x$ and $y$, then we abort the construction of the embedding configurations for the pair $(x,y)$, since the type of  $e_i$ is not \BPi-BP1, hence at least one of $f_i$ and $g_i$ is reddened because of the red vertices of $H_{e_i}$, by \cref{cond:iff-reddened-virtual-edges-T2} of \cref{le:iff-reddened}.

If exactly one of $f_i$ and $g_i$, say $f_i$, is a non-intrinsically red face different from $x$ and $y$, then we distinguish the case in which $e_i \neq e_m$ from the one in which $e_i = e_m$. 

\begin{itemize}
	\item If $e_i \neq e_m$, we check whether the type \BPii-BP2 is admissible for $e_i$. If not, we abort the construction of the embedding configurations for the pair $(x,y)$, since $f_i$ is reddened because of the red vertices of $H_{e_i}$, by \cref{cond:iff-reddened-virtual-edges-T2} of \cref{le:iff-reddened}; otherwise, we set $t_i=$~BP2 and we set $x_i$ so that the red outer face of the $t_i$-gadget will correspond to~$g_i$ in the realization of the embedding configuration, by \cref{cond:iff-reddened-virtual-edges-T1} of \cref{le:iff-reddened}.
	\item If $e_i = e_m$, we check whether the types \BPii-BP2 or \BPiv-BP4 are admissible for $e_i$. If not, we abort the construction of the embedding configurations for the pair $(x,y)$, since $f_i$ is reddened because of the red vertices of $H_{e_i}$, by \cref{cond:iff-reddened-virtual-edges-T2} of \cref{le:iff-reddened}; otherwise, we make up to two independent choices by setting $t_i$ as each of the admissible types for $e_i$ among \BPii-BP2 and \BPiv-BP4, and by setting $x_i$ so that the red outer face of the $t_i$-gadget will correspond to~$g_i$ in the realization of an embedding configuration, by \cref{cond:iff-reddened-virtual-edges-T1} of \cref{le:iff-reddened}.  
\end{itemize}

If none of $f_i$ and $g_i$ is a non-intrinsically red face different from $x$ and $y$ (that is, $f_i$ and $g_i$ will correspond to red faces in the realization of the embedding configuration, independently of the choice of $t_i$ and $x_i$), then we distinguish two cases.
If $e_i \neq e_m$, we arbitrarily set $t_i$ to any of the admissible types for $e_i$ (these are \BPii-BP2 and/or \BPiii-BP3, by \cref{R:bm-pole-incident-bp} of \cref{R:bm-pole}) and we arbitrarily set $x_i$ to either $\ell$ or $r$. 
If $e_i = e_m$, we make up to eight independent choices by setting $t_i$ to each of the admissible types for $e_i$ and by setting $x_i$ to each of $\ell$ and $r$. 

If the above procedure has not aborted, all the above ``partial'' embedding configurations satisfy the pair $(x,y)$ and avoid that any face of $\skel(\mu)$ different from $x$ and $y$ is reddened. It remains to deal with the virtual edges that are not of type \BP-BP.
\smallskip

\paragraph{Handling the remaining virtual edges.}
Finally, for each virtual edge $e_i$ of $\skel(\mu)$ that is of type \RE-RE, \BE-BE slim, \BE-BE fat, or \BPi-BP1, 
in each of the so-far constructed partial embedding configurations for $\mu$,
we set $t_i$ to \RE-RE, \BE-BE slim, \BE-BE fat, or \BPi-BP1, respectively, and we arbitrarily set $x_i$ to either $\ell$ or~$r$. 
First note that all the types \RE-RE, \BE-BE slim, \BE-BE fat, or \BPi-BP1 belong to $\mathcal T_0$, and thus they do not contribute to redden any face of $\mathcal S_\mu$, by \cref{cond:iff-reddened-virtual-edges} of \cref{le:iff-reddened}.
Furthermore, by \cref{R:bm-pole-possible-virtual-edges} of \cref{R:bm-pole}, we have exhaustively considered all the virtual edges of $\skel(\mu)$. This concludes the construction of the (up to eight) embedding configurations for the pair $(x,y)$. 

We summarize the above discussion in the following.

\begin{lemma}\label{lem:emb-configurations-bm-pole}
	Suppose that $\mu$ is an R-node such that $b_m$ is a vertex of $\skel(\mu)$. There exists a set $\mathcal V_\mu$ of embedding configurations for $\mu$ such that:
	\begin{enumerate}[\bf (i)]
		\item a type $t$ is admissible for $\mu$ if and only if there exists an embedding configuration $\mathcal C_\mu \in \mathcal V_\mu$ such that the embedding of $R(\mathcal C_\mu)$ is of type $t$;
		\item $|\mathcal V_\mu|=O(1)$; and
		\item $\mathcal V_\mu$ can be constructed in $O(|\skel(\mu)|)$ time.
	\end{enumerate}
\end{lemma}

\noindent\fbox{\bf Case 3: $b_m$ belongs to $\pert{\mu}$ but not to $\skel(\mu)$.} 
In this case, $b_m$ is a non-pole vertex of a virtual edge $e_m$ of $\skel(\mu)$. 
We start with a structural lemma.

\smallskip

\begin{lemma}\label{R:bm-internal}
	Let $\mu$ be an R-node. Suppose that $b_m$ is a non-pole vertex of a virtual edge $e_m$ of $\skel(\mu)$ and that $\pert{\mu}$ admits a relevant embedding $\cal E_\mu$. For a virtual edge $e$ of $\skel(\mu)$, let $\mathcal E_e$ be the embedding of $\pert{e}$ determined by $\cal E_\mu$.
	Then the following conditions are satisfied:
	\begin{enumerate}
		\item \label{R:bm-internal-BP} the type of $\mu$ is either \RF-RF, or \BF-BF, or \BB-BB;
	\item  \label{R:bm-internal-BP-virtual-edges-bm-not-in-skeleton} the type of $e_m$ is either \RF-RF, or \BF-BF, or \BB-BB;
	\item \label{R:bm-internal-possible-virtual-edges-bm-not-in-skeleton}
	the virtual edges of $\skel(\mu)$ different from $e_m$ are of type \RE-RE, \BE-BE slim, \BE-BE fat, or \BP-BP; and
	\item \label{R:bm-internal-em-types} there exists no type \BP-BP virtual edge $e$ of $\skel(\mu)$ such that the type of $\mathcal E_e$ is \BPiv-BP4 or \BPv-BP5.
	\end{enumerate}
\end{lemma}

\begin{proof}
Items~\ref{R:bm-internal-BP} and~\ref{R:bm-internal-BP-virtual-edges-bm-not-in-skeleton} follow from the fact that $b_m$ is a non-pole vertex of $\pert{\mu}$ and $\pert{e_m}$, respectively.
Items~\ref{R:bm-internal-possible-virtual-edges-bm-not-in-skeleton} and~\ref{R:bm-internal-em-types} follow from the fact that, for any virtual edge $e \neq e_m$, we have that $b_m$ does not belong to $\pert{e}$ and from Item~\ref{le:structure-A-mu-BP-BB-item-bm-in-pertinent} of \cref{le:structure-A-mu-BP-BB}.
\end{proof}

The algorithm to compute a constant-size set $\cal V_\mu$ of embedding configurations for $\mu$ is similar to the one for the case in which $b_m$ is a vertex of $\skel(\mu)$ (Case~2). The only differences lie in the choices we perform for the type $t_m$ and the flip $x_m$ for the virtual edge $e_m$. In fact, by \cref{R:bm-pole-possible-virtual-edges} of \cref{R:bm-pole} and by \cref{R:bm-internal-possible-virtual-edges-bm-not-in-skeleton} of \cref{R:bm-internal}, all the virtual edges of $\skel(\mu)$ different from $e_m$ are of type \RE-RE, \BE-BE slim, \BE-BE fat, or \BP-BP, both in Case~2 and in Case~3. Furthermore, if any of such virtual edges is of type \BP-BP, then the set of its admissible types contains neither \BPiv-BP4 nor \BPv-BP5, both in Case~2 and in Case~3,
by \cref{R:bm-pole-incident-bp} of \cref{R:bm-pole} and
by \cref{R:bm-internal-em-types} of \cref{R:bm-internal}, respectively.
However, differently from Case 2, where the type of $e_m$ is \BP-BP, by \cref{R:bm-pole-incident-bp} of \cref{R:bm-pole}, we have that, in Case 3, the type of $e_m$ is either \RF-RF, or \BF-BF, or \BB-BB, by \cref{R:bm-internal-BP-virtual-edges-bm-not-in-skeleton} of \cref{R:bm-internal}.
On the other hand, in Case 3, we can perform operations on $e_m$ analogous to those performed in Case 2. 
In particular, whenever the type of $e_m$ in Case 2 was \BPi-BP1, we select the type $t_m$ for $e_m$ among all the possible admissible types for $e_m$ belonging to $\mathcal T_0$; whenever the type $t_m$ selected for $e_m$ in Case 2 was either \BPii-BP2 or \BPiv-BP4, we select $t_m$ among all the possible admissible types for $e_m$ belonging to $\mathcal T_1$; finally, whenever the type $t_m$ selected for $e_m$ in Case 2 was either \BPiii-BP3 or \BPv-BP5, we select $t_m$ among all the possible admissible types for $e_m$ belonging to $\mathcal T_2$. Next, we describe the construction of the embedding configurations in greater detail.

In particular, we do not perform any preliminary test, since all the conditions of \cref{R:bm-internal} only descend from the structure of $\pert{\mu}$. Also, for each of the six pairs $(x,y)$, we construct $O(1)$ embedding configurations that satisfy the pair $(x,y)$, as follows.
Again, we discuss only the case in which 
\begin{inparaenum}[\bf (i)]
	\item $x, y \neq \emptyset$, 
	\item $x \neq y$, and 
	\item both $x$ and $y$ are non-intrinsically red, that is, when both $x$ and $y$ need to be reddened.  
\end{inparaenum}
The construction consists of three phases.
\smallskip

\paragraph{Satisfying the pair (x,y).}
We first compute the set $S_x$ (resp. $S_y$) that contains all the virtual edges incident to to $x$ (resp. to $y$) whose set of admissible types contains at least one type not in $\mathcal T_0$, and all the black vertices of $\skel(\mu)$ incident to $x$ (resp. to $y$) that are incident to an $rb$-trivial component. 
Recall that, by \cref{le:iff-reddened}, the constructed sets contain exactly those elements that can be used to ensure that $(x,y)$ is satisfied by the embedding configuration. 
Thus, if $S_x$ (resp. $S_y$) is empty, then we abort the construction of the embedding configurations for the pair $(x,y)$.
Also, suppose that $S_x = S_y$ and $|S_x| = |S_y| = 1$. If the only element in $S_x = S_y$ is a black vertex or a virtual edge whose set of admissible types does not contain any type in $\mathcal T_2$, then we abort the construction of the embedding configurations for the pair $(x,y)$. Otherwise, $S_x$ contains exactly one virtual edge $e_i$. If $e_i \neq e_m$, then 
we proceed as in Case 2; whereas, if $e_i=e_m$, then we construct $O(1)$ embedding
configurations obtained by setting $t_i$ as each of the admissible types for $e_i$ belonging to $\mathcal T_2$, and by setting $x_i$ as each of $\ell$ and $r$. 
In all the other cases, we can find two distinct elements (each being either a virtual edge or a vertex) such that one element belongs to $S_x$ and the other belongs to $S_y$. We use these elements to ensure that both $x$ and $y$ are reddened. Namely, consider the element that belongs to $S_x$, the discussion for the one belonging to $S_y$ being analogous. If this element is a vertex $w$ or if it is a virtual edge $e_i \neq e_m$, then we proceed as in Case 2.
Suppose now that the selected element in $S_x$ is $e_m$. Let $z \in \{f_m,g_m\}$ be the face of $\mathcal S_\mu$ incident to $e_m$ different from~$x$.
If $z$ is intrinsically red or $z = y$, then we construct $O(1)$ embedding configurations obtained by selecting $t_m$ as each of the admissible types for $e_m$ that are not in $\mathcal T_0$ and by setting $x_m$ to each of the flips such that a red outer face of the $t_m$-gadget corresponds to $x$.
Else, $z$ is a non-intrinsically red face different from $y$. If there is no admissible types for $e_m$ belonging to $\mathcal T_1$, we abort the construction of the embedding configurations for the pair $(x,y)$. Otherwise, we construct $O(1)$ embedding configurations obtained by setting $t_m$ as each of the admissible types for $e_m$ belonging to $\mathcal T_1$, and by setting $x_m$ so that the red outer face of the $t_m$-gadget corresponds to~$x$.

If the above procedure has not aborted, the above ``partial'' embedding configurations satisfy the pair $(x,y)$. To complete the construction of the embedding configurations for $\mu$, we have to select a type and a flip for the remaining virtual edges. We proceed as follows. 

\smallskip

\paragraph{Handling the remaining \BP-BP virtual edges.}
This phase is identical to corresponding one in Case 2. Observe, however, that in Case 3, $e_m$ is not handles in this phase, since it is not of type \BP-BP, by \cref{R:bm-internal-BP-virtual-edges-bm-not-in-skeleton} of \cref{R:bm-internal}.

\paragraph{Handling $e_m$.}
Recall that $x \in \{f_m,g_m\}$. Suppose $x = g_m$, the case $x = f_m$ being symmetric. If $f_m$ is a non-intrinsically red face different from $y$, then we check whether the set of admissible types for $e_m$ contains a at least a type not in $\mathcal T_2$. If not, we abort the construction of the embedding configurations for the pair $(x,y)$.
Otherwise, we construct $O(1)$ embedding configurations obtained by setting $t_m$ as each of the admissible types for $e_m$ belonging to either $\mathcal T_0$ or $\mathcal T_1$, and by setting $x_m$ as each of the flips such that an outer face of the $t_m$-gadget that is not red corresponds to $f_m$.

If $f_m$ is intrinsically red or $f_m=y$, then we construct $O(1)$ embedding configurations obtained by setting $t_m$ to each of the admissible types for $e_m$ and by setting $x_m$ to each of $\ell$ and $r$. 

If the above procedure has not aborted in the last two phases, all the above ``partial'' embedding configurations satisfy the pair $(x,y)$ and avoid that any face of $\skel(\mu)$ different from $x$ and $y$ is reddened. 
\smallskip

\paragraph{Handling the remaining virtual edges.}
This phase is identical to corresponding one in Case 2, and concludes the construction of the $O(1)$ embedding configurations for the pair $(x,y)$.

We summarize the above discussion in the following.

\begin{lemma}\label{lem:emb-configurations-bm-internal}
	Suppose that $\mu$ is an R-node such that $b_m$ belongs to $\pert{\mu}$ but not to $\skel(\mu)$. There exists a set $\mathcal V_\mu$ of embedding configurations for $\mu$ such that:
	\begin{enumerate}[\bf (i)]
		\item a type $t$ is admissible for $\mu$ if and only if there exists an embedding configuration $\mathcal C_\mu \in \mathcal V_\mu$ such that the embedding of $R(\mathcal C_\mu)$ is of type $t$;
		\item $|\mathcal V_\mu|=O(1)$; and
		\item $\mathcal V_\mu$ can be constructed in $O(|\skel(\mu)|)$ time.
	\end{enumerate}
\end{lemma}

\subsection{Constructing a Neat Embedding}\label{sse:embeddingconstruction}	

In this section, we show how the embedding configurations stored in the non-root nodes of~$\cal T$ by the testing algorithm in \cref{sse:testingalgorithm} can be exploited to efficiently construct a neat embedding~of~$H$.

\begin{theorem}\label{th:neat-construction}
	There exists an $O(n)$-time algorithm that constructs a neat embedding of an $n$-vertex $rb$-augmented component, if one exists.  
\end{theorem}

If the testing algorithm in \cref{sse:testingalgorithm} succeeds, then, by \cref{le:empty-admissible}, there exists at least one admissible type $t$ for the unique child $\tau$ of the root of $\mathcal T$. Let $\cal C_\tau$ be an embedding configuration for $\tau$ such that the type of the embedding of $R(\mathcal C_\tau)$ is $t$, which exists by \cref{le:relevant-in-configuration}.
We construct a neat embedding $\cal E$ of $H$ by means of a top-down traversal of $\mathcal T$ starting from $\tau$.

We initialize $\cal E$ to the embedding $\cal S_\tau$ of~$\skel(\tau)$ in $\mathcal C_\tau$ with an arbitrary flip.
For each node $\mu$ encountered in the traversal we assume that the following invariants hold:
\begin{enumerate}[\bf (i)]
\item $\mu$ is associated with a type $t_\mu$ and with an embedding configuration $\mathcal C_\mu$ in $\mathcal V_\mu$ such that the type of the embedding of $R(\mathcal C_\mu)$ is $t_\mu$,
\item $\cal E$ contains a copy of
$\skel(\mu)$ whose embedding $\cal S_\mu$ is the one indicated in $\mathcal C_\mu$, and
\item the flip of $\cal S_\mu$ in $\cal E$ is the one indicated in the embedding configuration $\mathcal C_\xi$ associated with the parent $\xi$ of $\mu$.
\end{enumerate}

When considering a non-leaf node $\mu$ we perform the following operations.
First, for every vertex $w$ of $\skel(\mu)$ that is incident to an $rb$-trivial component, we embed such a component into the face of $\cal E$ corresponding to the face of $\cal S_\mu$ that has been selected as the designated face $f_w$ of $w$ in $\mathcal C_\mu$.
Second, for each virtual edge~$e_i$ of~$\skel(\mu)$ that is not a Q-node, consider the admissible type $t_i$ for $e_i$ in $\mathcal C_\mu$.
We select an embedding configuration $\mathcal C_i$ in $\mathcal V_{\nu_i}$, where $\nu_i$ is the child of $\mu$ in~$\mathcal T$ corresponding to $e_i$, such that the type of the embedding of $R(\mathcal C_i)$ is $t_i$. This configuration exists by \cref{le:empty-admissible} and since $t_i$ is admissible for $e_i$.
We associate $\nu_i$ with $t_i$ and $\mathcal C_i$, and we replace  $e_i$ in $\cal E$ with the embedding $\mathcal S_{\nu_i}$ of~$\skel(\nu_i)$ in $\mathcal C_i$, flipped as indicated by the variable $x_i$ in~$\mathcal C_\mu$.
This guarantees that the invariants hold for $\nu_i$.

When all the non-leaf nodes of $\mathcal T$ have been considered, $\cal E$ is an embedding of $\pert{\tau}$ whose type is $t$. Recall, in fact, that the skeleton of each leaf Q-node contains the real edge corresponding to it.
To obtain a neat embedding of $H$, it only remains to augment $\cal E$ with a drawing of the reference edge $(b_1,b_2)$ corresponding to $\rho$ and of the undecided $rb$-trivial components incident to~$b_1$ and~$b_2$, if any. This can be done in constant time by embedding $(b_1,b_2)$ in the outer face of~$\cal E$ and by inserting the undecided $rb$-trivial components, if any, into one of the two faces incident to $(b_1,b_2)$ as discussed in the proof of \cref{le:some-admissible}. This concludes the proof of \cref{th:neat-construction}.

\section{Conclusions} \label{se:conclusions}
In this paper, we have proved that the {\sc $2$-Level Quasi-Planarity} problem is \NPC, while it is linear-time solvable if the order of the vertices in one level is given as part of the input. The equivalence of the $2$-level quasi-planar drawings with two different types of layouts, namely the bipartite $2$-page book embeddings and the $(2,2)$-track layouts, seems to indicate that the problem we solved in this paper is central in this area. 

Trying to extend our hardness result to the {\sc Quasi-Planarity} problem (without the $2$-level constraint) is tempting. One natural idea would be to add edges to the \emph{frame} (see \cref{se:complexity}), so that, in any quasi-planar drawing of the augmented graph, the drawing of the frame is topologically equivalent to its unique $2$-level quasi-planar drawing. However, it is not clear to us whether this can be done and, even if it can, how to ensure that the graph that we assemble to the frame (i.e., the graph from the {\sc Leveled Planarity} instance) can be forced to have a ``leveled'' structure. Hence, determining the computational complexity of the {\sc Quasi-Planarity} problem remains an elusive goal, in our opinion the most attractive \mbox{one on ``almost''-planar graphs.}



\protect\

\bibliographystyle{abbrv}
\bibliography{bibliography}

\end{document}